\pdfoutput=1

\documentclass[aps,prb,notitlepage,showkeys,showpacs,superscriptaddress,nofootinbib,longbibliography,preprintnumbers,11pt]{revtex4-2}

\usepackage[utf8]{inputenc}

\usepackage{here}
\usepackage{ulem}

\usepackage{url}

\usepackage{dsfont}

\newcounter{n}

\usepackage[top=30truemm,bottom=30truemm,left=25truemm,right=25truemm]{geometry}

\usepackage{spectralsequences}
\usepackage{tikz}
\usetikzlibrary{intersections}
\usetikzlibrary{calc}

\usepackage{xcolor}

\newcommand{\R}{\mathbb{R}}
\newcommand{\Z}{\mathbb{Z}}

\newcommand{\zmod}[1]{\mathbb{Z}_{#1}}

\newcommand{\tr}[1]{\mathrm{tr}\left(#1\right)}
\newcommand{\U}[1]{\mathrm{U}\left(#1\right)}

\newcommand{\coho}[3]{\mathrm{H}^{#1}(#2,#3)}
\newcommand{\cohoZ}[2]{\mathrm{H}^{#1}(#2,\mathbb{Z})}

\newcommand{\cohoU}[2]{\mathrm{H}^{#1}(#2,\U{1})}

\newcommand{\rep}[1]{\mathrm{Rep}(#1)}

\newcommand{\abs}[1]{\left|#1\right|}

\usepackage{tikz-cd}

\usepackage[colorlinks=true,linktoc=page,citecolor=red,linkcolor=blue]{hyperref}	
\usepackage{slashed}

\usepackage{amsthm}
\usepackage{thmtools}
\usepackage{blindtext}
\usepackage{braket}

\declaretheoremstyle[
       shaded={bgcolor=\color{rgb}{0.9,0.9,0.9}}  
]{theorem}

\declaretheoremstyle[
       shaded={bgcolor=\color{rgb}{0.9,0.9,0.9}}
]{question}

\declaretheoremstyle[
       shaded={bgcolor=\color{rgb}{0.9,0.9,0.9}}  
]{remark}

\declaretheoremstyle[
       shaded={bgcolor=\color{rgb}{0.9,0.9,0.9}}  
]{proposition}
\declaretheorem[style=theorem]{proposition}

\declaretheoremstyle[
       shaded={bgcolor=\color{rgb}{0.9,0.9,0.9}}  
]{definition}

\declaretheoremstyle[
       shaded={bgcolor=\color{rgb}{0.9,0.9,0.9}}  
]{assumption}

\declaretheoremstyle[
       shaded={bgcolor=\color{rgb}{0.9,0.9,0.9}}  
]{conjecture}
\declaretheorem[style=conjecture]{conjecture}

\declaretheoremstyle[
       shaded={bgcolor=\color{rgb}{0.9,0.9,0.9}}  
]{corrorary}

\declaretheoremstyle[
       shaded={bgcolor=\color{rgb}{0.9,0.9,0.9}}  
]{axiom}

\declaretheoremstyle[
       shaded={bgcolor=\color{rgb}{0.9,0.9,0.9}}  
]{lemma}

\def\Z{{\mathbb{Z}}}

\def\R{{\mathbb{R}}}
\def\C{{\mathbb{C}}}

\def\a{{\alpha}}

\def\>{{\geq }}
\def\<{{\leq }}

\newcommand{\cA}{\mathcal{A}}
\newcommand{\cB}{\mathcal{B}}
\newcommand{\cC}{\mathcal{C}}
\newcommand{\cD}{\mathcal{D}}

\newcommand{\cH}{\mathcal{H}}
\newcommand{\cI}{\mathcal{I}}
\newcommand{\cL}{\mathcal{L}}
\newcommand{\cM}{\mathcal{M}}
\newcommand{\cN}{\mathcal{N}}
\newcommand{\cO}{\mathcal{O}}
\newcommand{\cT}{\mathcal{T}}

\newcommand{\Hom}{\text{Hom}}
\newcommand{\End}{\text{End}}
\newcommand{\Aut}{\text{Aut}}
\newcommand{\Vect}{\text{Vec}}
\newcommand{\Rep}{\text{Rep}}
\newcommand{\Mod}{\text{Mod}}
\newcommand{\Fun}{\text{Fun}}
\newcommand{\Simp}{\text{Simp}}
\def\1{\mathds{1}}

\newcommand{\SPT}{\mathsf{SPT}}
\newcommand{\phys}{\text{phys}}
\newcommand{\TY}{\text{TY}}

\usepackage{mathrsfs}
\usepackage{amsmath,amssymb}

\usepackage{esvect}

\usepackage{adjustbox}

\usepackage[status=draft,inline,nomargin]{fixme}
\fxusetheme{color}
\FXRegisterAuthor{ki}{aki}{KI}
\FXRegisterAuthor{shu}{shuenv}{\color{red}Shu}

\usetikzlibrary{decorations.pathmorphing,shapes}
\newcounter{sarrow}
\newcommand\xrsquigarrow[1]{%
\stepcounter{sarrow}%
\begin{tikzpicture}[decoration=snake]
\node (\thesarrow) {\strut#1};
\draw[->,decorate] (\thesarrow.south west) -- (\thesarrow.south east);
\end{tikzpicture}%
}

\usepackage{xurl}

\begin{document}

\title{1+1d SPT phases with fusion category symmetry:\\ interface modes and non-abelian Thouless pump}

\author{Kansei Inamura}
\email{k\_inamura@issp.u-tokyo.ac.jp}
\affiliation{Institute for Solid State Physics, University of Tokyo, Kashiwa, Chiba 277-8581, Japan}

\author{Shuhei Ohyama}
\email{shuhei.ohyama@riken.jp}
\affiliation{RIKEN Center for Emergent Matter Science, Wako, Saitama, 351-0198, Japan}
\date{\today} 

\begin{abstract}
We consider symmetry protected topological (SPT) phases with finite non-invertible symmetry $\mathcal{C}$ in 1+1d.
In particular, we investigate interfaces and parameterized families of them within the framework of matrix product states.
After revealing how to extract the $\mathcal{C}$-SPT invariant, we identify the algebraic structure of symmetry operators acting on the interface of two $\mathcal{C}$-SPT phases.
By studying the representation theory of this algebra, we show that there must be a degenerate interface mode between different $\mathcal{C}$-SPT phases.
This result generalizes the bulk-boundary correspondence for ordinary SPT phases.
We then propose the classification of one-parameter families of $\mathcal{C}$-SPT states based on the explicit construction of invariants of such families.
Our invariant is identified with a non-abelian generalization of the Thouless charge pump, which is the pump of a local excitation within a $\mathcal{C}$-SPT phase.
Finally, by generalizing the results for one-parameter families of SPT phases, we conjecture the classification of general parameterized families of general gapped phases with finite non-invertible symmetries in both 1+1d and higher dimensions.
\end{abstract}

\maketitle

\setcounter{tocdepth}{3}
\tableofcontents

\section{Introduction and Summary}
\label{sec: Introduction and Summary}

\noindent{\bf Introduction.}
The theoretical study of topological phases has provided various insights into the physics of many-body systems. 
One of the most fundamental subjects of this research is a class of quantum states known as invertible states (a.k.a. short-range entangled states).
An invertible state is a state that is realized as the unique ground state of a gapped Hamiltonian.\footnote{
Here, we would like to clarify the usage of the term ``invertible state." 
Generally, SPT phases protected by non-invertible symmetry do not possess a group structure. 
This means that for a given non-invertible symmetric state, 
there may not necessarily be a non-invertible symmetric state that acts as its inverse. 
However, we will take the position that the concept of invertible states is independent of symmetry.
Therefore, we will refer to a state that is invertible when disregarding symmetry as an invertible state.
}
By imposing symmetry, invertible states exhibit various quantum phases, 
which have been actively studied as Symmetry Protected Topological (SPT) Phases \cite{Gu_2009,Pollmann:2009mhk,Pollmann_2010,Chen:2010zpc,Chen_2011_complete,Chen:2011bcp,Chen:2011pg,Schuch_2011}.
An interesting physical property of SPT phases is the existence of gapless edge modes. 
The non-triviality of the SPT phases and these edge modes are related to each other, 
and their relationship is referred to as bulk-boundary correspondence or anomaly inflow.
Furthermore, it is known that topological transport phenomena can occur by adiabatically and periodically driving an invertible state, 
and this is referred to as generalized Thouless pump phenomena \cite{Thouless83,PhysRevB.82.115120,Kitaev2011SCGP,Kitaev2013SCGP,Kitaev2015IPAM,Kapustin:2020mkl,Shiozaki:2021weu,Hermele2021CMSA,Wen:2021gwc,Spodyneiko:2023vsw,Ohyama:2022cib}.

Another trend in the study of quantum many-body systems is the generalization of the concept of symmetry~\cite{Gaiotto:2014kfa}.
A typical example of this is the notion known as non-invertible symmetry (a.k.a. categorical symmetry), see, e.g., \cite{Cordova:2022ruw, McGreevy:2022oyu, Schafer-Nameki:2023jdn, Brennan:2023mmt, Bhardwaj:2023kri, Luo:2023ive, Shao:2023gho, Carqueville:2023jhb, Iqbal:2024pee} for recent reviews. 
This framework regards invariance under operations on the theory, such as gauging and duality operations, as a symmetry of the theory.
While conventional symmetries are described by the mathematical structure of groups, 
non-invertible symmetries are described by certain types of categories. 
In particular, the structure of finite non-invertible symmetries in 1+1 dimensions is described by fusion categories \cite{Bhardwaj:2017xup, Chang:2018iay, Thorngren:2019iar}.
In this context, it is natural to consider the classification of SPT phases protected by non-invertible symmetry.

Tensor networks are an efficient way to describe highly entangled states. 
They have been applied not only to numerical computations but also to the classification and description of topological phases \cite{Chen:2010zpc,Chen_2011_complete,Chen:2011bcp,Schuch_2011,Pollmann:2009mhk,Cirac:2020obd,Bultinck:2015bot,Fidkowski_2011}.
Furthermore, it has been pointed out that tensor networks are a useful platform for representing (or microscopically realizing) non-invertible symmetries \cite{Molnar:2022nmh,Garre-Rubio:2022uum,Lootens:2021tet,Lootens:2022avn, Gorantla:2024ocs}. 
Therefore, tensor networks are an excellent tool for studying topological phases that possess non-invertible symmetries.

As pioneering work in the classification of 1+1d SPT phases protected by non-invertible symmetry, 
a method using effective theories described by TQFT has been proposed in \cite{Thorngren:2019iar}. 
This paper argues that the classification of fiber functors of the fusion category corresponds to the classification of SPT phases.\footnote{See Sec.~\ref{sec: Fiber functors and SPT phases} for a brief review of this claim.}
However, this paper focuses on continuous systems using field theory and does not discuss lattice systems. 
Recent studies have also presented several examples of lattice models that exhibit SPT phases with non-invertible symmetry \cite{Fechisin:2023dkj, Seifnashri:2024dsd, Jia:2024bng, Choi:2024rjm, Li:2024fhy}.
These studies construct generalizations of the cluster model and investigate their phase structures. 
Additionally, a general theory for classifying gapped systems with symmetry, described by tensor networks known as matrix product operators (MPO), 
has also been proposed in \cite{Garre-Rubio:2022uum}. 
This paper analyzes how matrix product operators (MPO) act on matrix product states (MPS) 
and points out that phases can be distinguished by the categorical data that appear when MPOs act on an MPS.


For SPT phases protected by non-invertible symmetry, 
the presence of edge modes and topological pumping phenomena are expected, similar to conventional SPT phases. 
However, systematic studies of these phenomena have not been carried out in previous research. 
Furthermore, since gapped lattice models are expected to be described by TQFT in the low-energy limit, 
it is expected that the data of a fiber functor can be extracted from the ground states of lattice systems. 
However, it has not been clearly established in previous research how to directly extract the data of the fiber functor from lattice systems.


In this study, we consider SPT phases protected by non-invertible symmetries within the framework of MPS. 
We first propose a method to extract the fiber functor from invertible states realized in lattice systems. 
We then investigate the symmetry algebra acting on the interface of SPT phases, generalizing the anomaly inflow to non-invertible symmetries.
We also explore the classification of generalized Thouless pumps and their relation to interface modes.
A more detailed summary of the paper is provided below.


\vspace{10pt}
\noindent{\bf Outline and Summary.}
The structure of the paper is as follows. If the reader is familiar with the contents of Sec.~\ref{sec: Preliminaries}, each section can be read almost independently.
In this paper, we focus only on bosonic systems. 

In Sec.~\ref{sec: Preliminaries}, we provide a brief review of fusion categories and fiber functors.
It is believed that SPT phases with fusion category symmetry $\cC$ are classified by fiber functors of $\mathcal{C}$. 
We also review the interpretation of this classification from the perspective of Topological Quantum Field Theory (TQFT).
In addition, we explain some mathematical facts on fusion categories used in the main part of the paper 
and provide a brief review of the relation between tensor networks and fusion categories.

In Sec.~\ref{sec: Fiber functors from injective MPS}, we describe how to extract data of the fiber functor from a $\mathcal{C}$-symmetric MPS.  
After some preliminary discussions on MPSs and their transfer matrices, we demonstrate that the data of the fiber functor can be extracted as the triple inner product of infinite MPSs. 
A key to this is the analogy between the partition functions of TQFT with symmetry defect and MPS. 
Additionally, we introduce the abelianization of the fiber functor.
As a result, we can define an invariant that takes values in (abelian) cohomology.
As an example, we explicitly compute the data of a fiber functor from the $G \times \Rep(G)$-symmetric cluster state~\cite{Brell_2015, Fechisin:2023dkj}.

In Sec.~\ref{sec: Interface modes of SPT phases with fusion category symmetries}, we analyze the interface modes between $\mathcal{C}$-symmetric SPT phases. 
This analysis is a generalization of that of edge modes in conventional SPT phases with group symmetry.
We first derive the symmetry operators that act on the interface.
These operators form an algebra, which we refer to as the interface algebra.
The analysis of interface modes is reduced to classifying the irreducible representations of this algebra. 
We show that the interface algebra between different SPT phases does not have one-dimensional representations.
Consequently, an interface between different SPT phases gives rise to degenerate ground states.
This can be regarded as the bulk-boundary correspondence in $\mathcal{C}$-symmetric SPT phases.
On the other hand,  the self-interface within the same SPT phase always has one-dimensional representations, which are in one-to-one correspondence with the automorphisms of a fiber functor.
As an example, we study interfaces of $\Rep(D_8)$-symmetric SPT phases.

In Sec.~\ref{sec: Parameterized family and Thouless pump}, we investigate $S^1$-parameterized families of $\mathcal{C}$-symmetric invertible states. 
Based on a careful analysis of gauge redundancy of $\cC$-symmetric MPS tensors, we define an invariant that detects the non-triviality of an $S^1$-family. 
We will see that this invariant takes values in the group of automorphisms of a fiber functor.
Additionally, we show that the non-triviality of an $S^1$-family gives rise to the Thouless pump of non-abelian charges.\footnote{Here, a charge refers to a one-dimensional representation of the self-interface algebra. In general, the charges form a non-abelian group, which is isomorphic to the group of automorphisms of a fiber functor.}
This non-abelian Thouless pump is in contrast to conventional Thouless pumps, which provide pumping of abelian charges.
As an example, we construct non-trivial $S^1$-families of $\Rep(G)$-symmetric invertible states.
 
In Sec.~\ref{sec: Generalizations}, we consider parameterized families of general $\cC$-symmetric gapped systems, not limited to invertible states. 
The central hypothesis is that the moduli space of $\mathcal{C}$-symmetric gapped systems in a gapped phase labeled by a $\cC$-module category $\mathcal{M}$ is given by the classifying space $B\underline{\Fun}_{\mathcal{C}}(\mathcal{M},\mathcal{M})^{\text{inv}}$ of categorical group $\underline{\Fun}_{\mathcal{C}}(\mathcal{M},\mathcal{M})^{\text{inv}}$. 
Here, $\underline{\Fun}_{\mathcal{C}}(\mathcal{M},\mathcal{M})^{\text{inv}}$ is a 2-group consisting of invertible $\cC$-module functors and invertible $\cC$-module natural transformations. 
Based on this hypothesis, it is expected that $X$-parameterized families of a general gapped phase $\cM$ are classified by the non-abelian \v{C}ech cohomology $\check{\mathrm{H}}^1(X, \underline{\Fun}_{\mathcal{C}}(\mathcal{M},\mathcal{M})^{\text{inv}})$. 
We will present examples supporting this conjecture by considering some special cases. 
Moreover, we extend this consideration to higher dimensions. 
Specifically, in the 2+1-dimensional case, it is expected that the moduli space of \(\mathcal{C}\)-symmetric gapped systems in a non-chiral gapped phase labeled by a $\cC$-module 2-category $\mathcal{M}$ is given by the classifying space $B\underline{\underline{\Fun}}_{\mathcal{C}}(\mathcal{M},\mathcal{M})^{\text{inv}}$ of the categorical 2-group $\underline{\underline{\Fun}}_{\mathcal{C}}(\mathcal{M},\mathcal{M})^{\text{inv}}$. 
As in the 1+1-dimensional case, we provide several examples supporting this conjecture.

Various technical details are relegated to Appendices.
In App.~\ref{sec: Tambara-Yamagami categories}, we review the basics of the Tambara-Yamagami categories.
In App.~\ref{sec: Fiber functors of Tambara-Yamagami categories}, we review the classification of fiber functors of the Tambara-Yamagami categories.
In App.~\ref{sec: L-symbols for fiber functors of Tambara-Yamagami categories}, we compute the data (called $L$-symbols) of fiber functors of non-anomalous $\Z_2 \times \Z_2$ Tambara-Yamagami categories, i.e., $\Rep(D_8)$, $\Rep(Q_8)$, and $\Rep(H_8)$.
In App.~\ref{sec: Interface algebra for Tambara-Yamagami categories}, we enumerate irreducible representations of the interface algebras for SPT phases with $\Rep(D_8)$, $\Rep(Q_8)$, and $\Rep(H_8)$ symmetries.
In App.~\ref{sec: GxRep(G) cluster state}, we show that the $G \times \Rep(G)$-symmetric cluster state is obtained as the ground state of a $\Rep(G)$-symmetric model discussed in \cite{Inamura:2021szw}.
In App.~\ref{sec: Weak completeness relation}, we show some identity of tensor network representations of fusion category symmetry.
In App.~\ref{sec: cohomology}, we provide a detailed computation of a cohomology group associated with $\Rep(D_8)$.
In App.~\ref{sec: Computation of pump invariant}, we explicitly compute the invariants of $S^1$-parameterized families of $\Rep(D_8)$-symmetric invertible states.

\section{Preliminaries}
\label{sec: Preliminaries}
In this section, we briefly review mathematical tools that we will use to study SPT phases with fusion category symmetries.
Throughout the paper, vector spaces are always finite-dimensional and the base field is always the field $\C$ of complex numbers.
We denote the group cohomology by $\mathrm{H}_{\text{gp}}$ to distinguish it from other cohomology.

\subsection{Fusion categories}
\label{sec: Fusion categories and fiber functors}
Finite symmetries of 1+1d unitary bosonic systems are generally described by unitary multifusion categories \cite{Bhardwaj:2017xup, Chang:2018iay}.
In what follows, (multi)fusion categories always refer to unitary ones.
Each object $x$ of a multifusion category $\cC$ labels a topological line $\cL_x$, and each morphism $\mu \in \Hom_{\cC}(x, y)$ labels a topological junction between topological lines $\cL_x$ and $\cL_y$.
Here, $\Hom_{\cC}(x, y)$ denotes the finite dimensional $\C$-vector space of all morphisms from $x$ to $y$.
An object $x$ is said to be simple if its endomorphism space $\End_{\cC}(x) = \Hom_{\cC}(x, x)$ is one-dimensional.
Correspondingly, a topological line $\cL_x$ labeled by a simple object $x$ is indecomposable.
The set of (representatives of isomorphisms classes of) simple objects is denoted by $\text{Simp}(\cC)$.
The number of elements in $\Simp(\cC)$ is finite.

A multifusion category $\cC$ is equipped with various data that encode the algebraic structures of topological lines and topological junctions between them.
Concretely, a multifusion category $\cC$ consists of the data listed below.
We refer the reader to \cite{EGNO2015} for a more complete description of multifusion categories.

\vspace{10pt}
\noindent{\bf Fusion coefficients.} 
The fusion of two topological lines $\cL_x$ and $\cL_y$ defines a tensor product $x \otimes y$ of objects $x$ and $y$.
The object $x \otimes y$ can always be decomposed into a finite direct sum of simple objects as
\begin{equation}
x \otimes y \cong \bigoplus_{z \in \text{Simp}(\cC)} N_{xy}^z ~ z, \quad N_{xy}^z \in \Z_{\geq 0},
\label{eq: fusion rules}
\end{equation}
where the sum on the right-hand side is taken over all simple objects of $\cC$.
The non-negative integer $N_{xy}^z$ is called a fusion coefficient.
The unit of the above fusion rule is denoted by $\1 \in \cC$ and is called a unit object, which corresponds to the trivial topological line.
The unit object $\1$ is equipped with isomorphisms $l_x: \1 \otimes x \rightarrow x$ and $r_x: x \otimes \1 \rightarrow x$, referred to as the left and right unitors, that satisfy appropriate coherence conditions. 
One can always take $l_x$ and $r_x$ to be the identity morphism by identifying $\1 \otimes x$ and $x \otimes \1$ with $x$ \cite{EGNO2015}.
When the unit object $\1$ is simple, the multifusion category $\cC$ is called a fusion category. 
In what follows, we will restrict our attention to the case where $\cC$ is fusion.

\vspace{10pt}
\noindent{\bf Associator/$F$-symbols.}
The fusion of topological lines is associative only up to isomorphism. 
Namely, there is a natural isomorphism
\begin{equation}
\alpha_{x, y, z}: (x \otimes y) \otimes z \rightarrow x \otimes (y \otimes z),
\label{eq: associator}
\end{equation}
which relates two different ways of fusing three topological lines into a single line.
The above isomorphism is called an associator and satisfies the pentagon equation represented by the following commutative diagram~\cite{Moore:1988qv}:
\begin{equation}
\adjincludegraphics[valign = c, width = 0.7\linewidth]{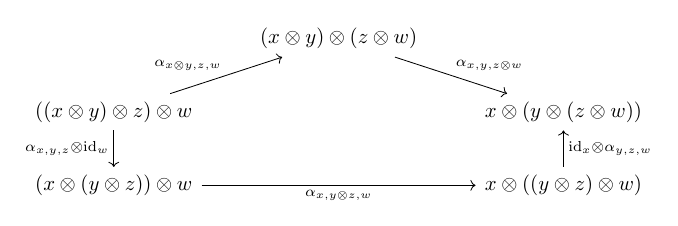}
\end{equation}
More explicitly, the associator \eqref{eq: associator} can be represented by a set of complex numbers $(F^{xyz}_w)_{(u; \mu, \nu), (v; \rho, \sigma)}$ defined by
\begin{equation}
\adjincludegraphics[valign = c, width = 2cm]{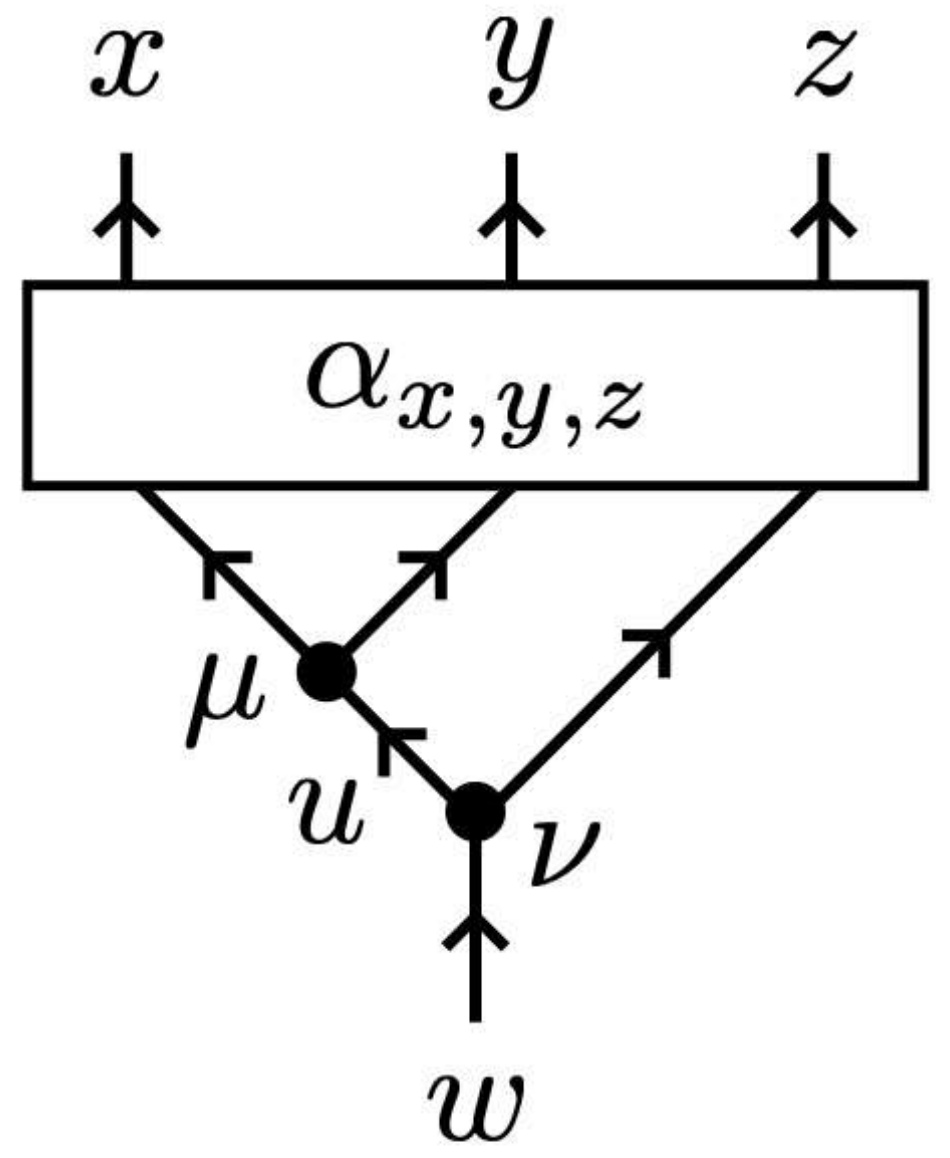} ~ = \sum_{v \in \text{Simp}(\cC)} \sum_{\rho, \sigma} (F^{xyz}_w)_{(u; \mu, \nu), (v; \rho, \sigma)} ~
\adjincludegraphics[valign = c, width = 1.7cm]{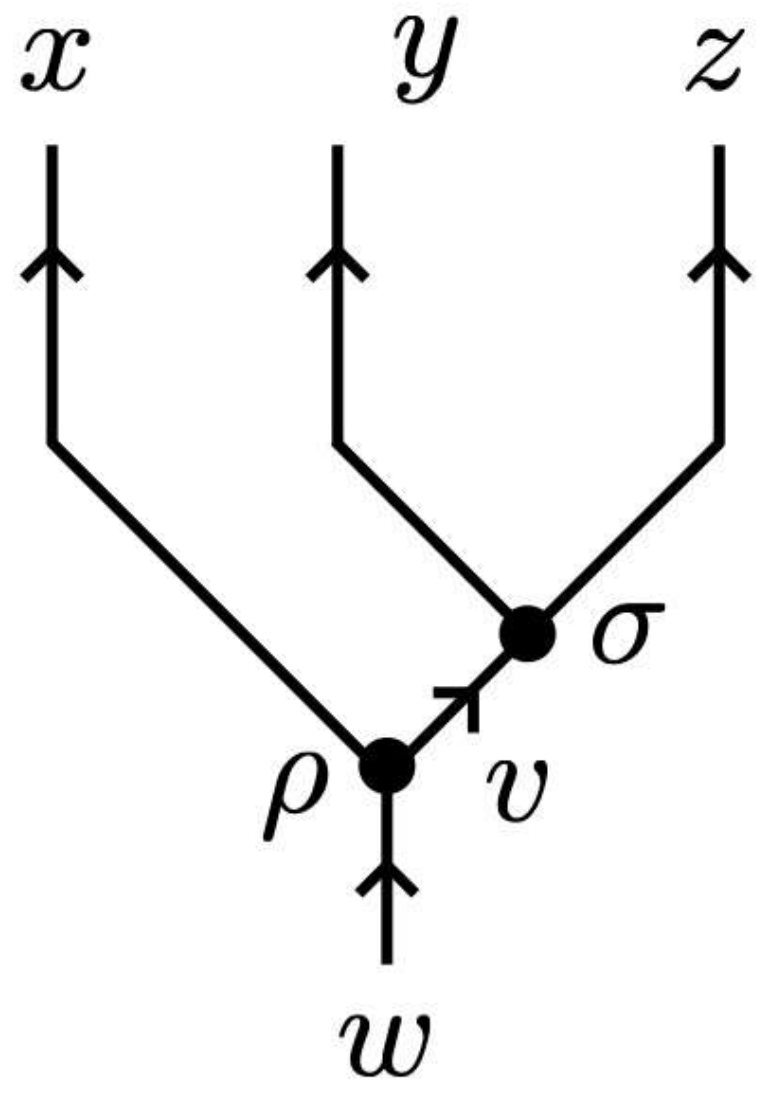},
\label{eq: F-symbols}
\end{equation}
where the summation on the right-hand side is taken over all simple objects of $\cC$ and all basis vectors of the Hom spaces.
The complex numbers $(F^{xyz}_w)_{(u; \mu, \nu), (v; \rho, \sigma)}$ are called $F$-symbols.
The $F$-symbols depend on the choice of basis vectors of the Hom spaces.
When $\cC$ is unitary, one can choose these bases so that $F^{xyz}_w$ becomes a unitary matrix for all simple objects $x, y, z, w \in \cC$.
In the physics literature, the associator in Eq.~\eqref{eq: F-symbols} is often chosen to be the identity, which is always possible due to the Mac Lane strictness theorem~\cite{MacLane1998}.
We emphasize that the $F$-symbols can be non-trivial even if the associator $\alpha_{x, y, z}$ is the identity.

\vspace{10pt}
\noindent{\bf Quantum dimension.}
For a unitary fusion category $\cC$, the quantum dimension $\dim(x)$ of a simple object $x$ is the unique positive real number that satisfies
\begin{equation}
\dim(x) \dim(y) = \sum_{z \in \Simp(\cC)} N_{xy}^z \dim(z),
\end{equation}
where $N_{xy}^z$ is the fusion coefficient in Eq.~\eqref{eq: fusion rules}.
We note that $\dim(x)$ is not necessarily an integer.
The quantum dimension of a non-simple object $\bigoplus_i x_i \in \cC$ is given by the sum of the quantum dimensions $\sum_i \dim(x_i)$.
The quantum dimension of $x$ agrees with that of the dual $x^*$.

\vspace{10pt}
The simplest example of a fusion category is the category $\Vect$ of finite dimensional $\C$-vector spaces.
Objects and morphisms of this category are finite dimensional vector spaces and linear maps between them.
The tensor product of objects is defined by the ordinary tensor product of vector spaces.
In particular, the unit object of $\Vect$ is a one-dimensional vector space $\C$, which is the unique simple object up to isomorphism.
The associators and $F$-symbols are all trivial.
The quantum dimension of an object $x \in \Vect$ is given by the dimension of the vector space $x$, which is a non-negative integer.

\subsection{Fiber functors and SPT phases}
\label{sec: Fiber functors and SPT phases}
In this subsection, we recall the definition of fiber functors and their relation to SPT phases with fusion category symmetries.

A fiber functor of a fusion category $\cC$ is defined as a tensor functor from $\cC$ to $\Vect$.
Let us unpack this definition.
First of all, a fiber functor $F: \cC \rightarrow \Vect$ is a functor, that is, it maps an object $x \in \cC$ to a vector space $F(x) \in \Vect$ and maps a morphism $\mu \in \Hom_{\cC}(x, y)$ to a linear map $F(\mu) \in \Hom_{\Vect}(F(x), F(y))$.
The map on morphisms should preserve the composition law and the identity, i.e.,
\begin{equation}
F(\nu \circ \mu) = F(\nu) \circ F(\mu), \quad F(\text{id}_x) = \text{id}_{F(x)}
\end{equation}
for any composable morphisms $\mu, \nu$ and any object $x \in \cC$.
The functor $F: \cC \rightarrow \Vect$ is called a fiber functor if it is equipped with a natural set of isomorphisms
\begin{equation}
J_{x, y}: F(x) \otimes F(y) \rightarrow F(x \otimes y), \quad \varphi: \C \rightarrow F(\1)
\label{eq: fiber functor isom}
\end{equation}
that satisfy the following commutative diagram:\footnote{We do not specify the coherence conditions on $\varphi$ here because we will not use them in what follows.}
\begin{equation}
\adjincludegraphics[valign = c, width = 0.95\linewidth]{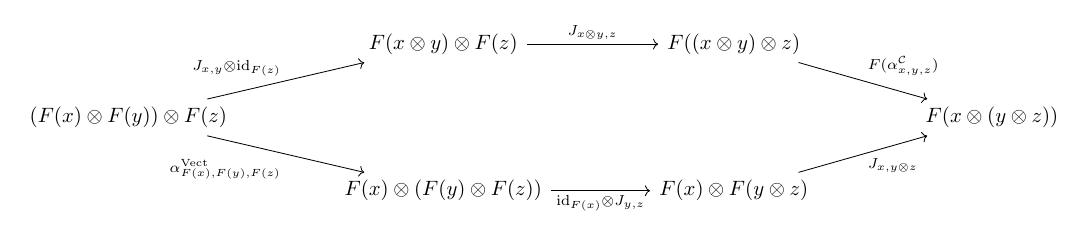}
\label{eq: fiber functor cd}
\end{equation}
where $\alpha^{\cC}$ and $\alpha^{\Vect}$ denote the associators of $\cC$ and $\Vect$, respectively.
The isomorphism $\varphi: \C \rightarrow F(\1)$ can always be chosen to be the identity.
Therefore, a fiber functor is characterized by a set of vector spaces $\{F(x) \mid x \in \cC\}$ together with isomorphisms $\{J_{x, y}: F(x) \otimes F(y) \rightarrow F(x \otimes y) \mid x, y \in \cC\}$ making the diagram~\eqref{eq: fiber functor cd} commute.

Equivalently, a fiber functor can also be defined as a $\cC$-module category whose simple object is unique up to isomorphism.
Here, a $\cC$-module category $\cM$ is a semisimple category equipped with a tensor functor $F: \cC \rightarrow \End(\cM)$, where $\End(\cM)$ is the category of endofunctors of $\cM$.
Indeed, when the simple object of $\cM$ is unique, $\cM$ is equivalent to $\Vect$ as a semisimple category, and hence the tensor functor $F: \cC \rightarrow \End(\cM) \cong \Vect$ becomes a fiber functor.
Conversely, a fiber functor $F: \cC \rightarrow \Vect$ allows us to endow the target category $\Vect$ with a $\cC$-module structure.
Thus, a fiber functor of $\cC$ can be identified with a $\cC$-module category with a single simple object.

Since a fiber functor $F$ preserves the tensor product of objects up to isomorphism $J$, it follows that
\begin{equation}
\quad F(x) \otimes F(y) \cong \bigoplus_{z \in \Simp(\cC)} N_{xy}^z F(z),
\end{equation}
where $N_{xy}^z$ is the fusion coefficient of $\cC$.
The above equation implies
\begin{equation}
\dim(F(x)) \dim(F(y)) = \sum_{z \in \Simp(\cC)} N_{xy}^z \dim(F(z)).
\end{equation}
Recalling that the quantum dimension is the unique positive one-dimensional representation of the fusion rules, we find 
\begin{equation}
\dim(F(x)) = \dim(x),
\end{equation}
which means that $F(x) \in \Vect$ is a vector space of dimension $\dim(x)$, i.e., $F(x) \cong \C^{\dim(x)}$.
In particular, since the dimension of a vector space $F(x)$ must be a positive integer, a fusion category $\cC$ admits a fiber functor only if the quantum dimension of every object $x \in \cC$ is a positive integer, namely, $\dim(x) \in \Z_{\geq 1}$.

We emphasize that not every fusion category admits a fiber functor.
In general, a fusion category $\cC$ admits a fiber functor if and only if $\cC$ is equivalent to the category $\Rep(H)$ of representations of a finite dimensional semisimple Hopf algebra $H$ \cite{EGNO2015}.\footnote{Any fusion category that may not admit a fiber functor is equivalent to the category of representations of a weak Hopf algebra \cite{Ostrik2003, Hayashi1999}, which is a generalization of a Hopf algebra \cite{BNS1999}.}
Here, $\Rep(H)$ denotes the category whose objects are finite dimensional representations of $H$ and morphisms are intertwiners between them.
A fusion category $\cC$ is said to be non-anomalous if it admits a fiber functor.
Otherwise, $\cC$ is said to be anomalous \cite{Thorngren:2019iar}.

For later convenience, let us introduce a matrix representation of the isomorphism $J_{x, y}$.
To this end, we first choose a basis $\{(\phi_x)_i \mid i = 1, 2, \cdots, \dim(x)\}$ of the vector space $F(x)$ for each simple object $x \in \cC$.
We also fix a basis $\{(\phi_{xy}^z)_{\mu} \mid \mu = 1, 2, \cdots, N_{xy}^z\}$ of the Hom space $\Hom_{\cC}(x \otimes y, z)$.
For these bases, one can write down the matrix elements of the isomorphism $J_{x, y}$ as follows:
\begin{equation}
F((\phi_{xy}^z)_{\mu}) \circ J_{x, y} ((\phi_x)_i \otimes (\phi_y)_j) = \sum_{k} J^{xy}_{(i, j), (z; k, \mu)} (\phi_z)_{k}, \quad J^{xy}_{(i, j), (z; k, \mu)} \in \C.
\label{eq: J}
\end{equation}
The set of complex numbers $\{J^{xy}_{(i, j), (z; k, \mu)} \mid z \in \Simp(\cC), ~ i = 1, 2, \cdots, \dim(x), ~ j = 1, 2, \cdots, \dim(y), ~ k = 1, 2, \cdots, \dim(z), ~ \mu = 1, 2, \cdots, N_{xy}^z\}$ form a matrix representation of $J_{x, y}$. 
That is, each $J^{xy}_{(i, j), (z; k, \mu)}$ can be regarded as a component of a square matrix $J^{xy}$ of size $\dim(x) \dim(y)$.

A simple example of a fiber functor is the forgetful functor of the representation category $\Rep(H)$ of a Hopf algebra $H$.
The forgetful functor maps an object $x \in \Rep(H)$ to the underlying vector space of $x$ and maps a morphism to the underlying linear map.
The structure isomorphisms $J$ and $\varphi$ associated with the forgetful functor are trivial.\footnote{The complex numbers $\{J^{xy}_{(i, j), (z; k, \mu)}\}$ can be non-trivial. This is analogous to the situation where $F$-symbols of a fusion category $\cC$ can be non-trivial even if the associator is trivial (i.e., $\cC$ is strict).}

\vspace{10pt}
\noindent{\bf Classification of SPT phases.}
An SPT phase is a symmetric gapped phase that has a unique ground state on any space without boundaries.
It is well known that 1+1d bosonic SPT phases with finite internal symmetry group $G$ are classified by the second group cohomology of $G$ \cite{Chen:2010zpc, Schuch_2011, Pollmann_2010}.
More generally, it was proposed in \cite{Thorngren:2019iar} that 1+1d bosonic SPT phases with symmetry $\cC$ are classified by fiber functors of $\cC$.
In particular, a fusion category symmetry $\cC$ admits an SPT phase if and only if $\cC$ is non-anomalous.
We note that the set of SPT phases with symmetry $\cC$ does not have a natural group structure in general.

The data of a fiber functor can be directly interpreted in terms of topological field theory that describes the low energy limit of an SPT phase \cite{Thorngren:2019iar, Inamura:2021wuo}.
Concretely, in the context of topological field theory, the vector space $F(x)$ is understood as the state space of the $x$-twisted sector on a circle.
Similarly, the isomorphism $J_{x, y}: F(x) \otimes F(y) \rightarrow F(x \otimes y)$ can be understood as the transition amplitude on a pants diagram with topological lines $x$ and $y$ inserted on the two legs.
The complex number $J^{xy}_{(i, j), (z; k, \mu)}$ represents a three-point function of defect operators, which is related to the transition amplitude on a pair of pants via the state-operator correspondence.
See Fig.~\ref{fig: pants} for diagrammatic representations of these data.
\begin{figure}[t]
\begin{minipage}{0.3\linewidth}
$F(x) = \adjincludegraphics[valign = c, width = 0.35\linewidth, trim = 8pt 0pt 8pt 0pt]{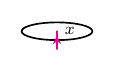}$\\
(a)
\end{minipage}%
\begin{minipage}{0.35\linewidth}
$J_{x, y} = \adjincludegraphics[valign = c, width = 0.5\linewidth, trim = 8pt 0pt 8pt 0pt]{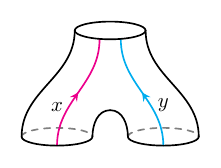}$\\
(b)
\end{minipage}%
\begin{minipage}{0.35\linewidth}
$J^{xy}_{(i, j), (z; k, \mu)} = \adjincludegraphics[valign = c, width = 0.5\linewidth, trim = 8pt 0pt 8pt 0pt]{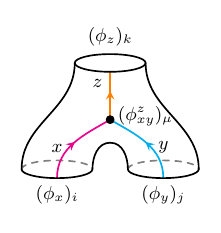}$\\
(c)
\end{minipage}%
\caption{(a) The vector space $F(x)$ represents the state space on a circle twisted by $x \in \cC$. (b) The isomorphism $J_{x, y}$ represents the transition amplitude on a pair of pants. (c) The matrix element of $J_{x, y}$ represents the three-point function.}
\label{fig: pants}
\end{figure}

\vspace{10pt}
\noindent{\bf Classification of general gapped phases.}
The classification of more general gapped phases with fusion category symmetries was also proposed in \cite{Thorngren:2019iar}.
Specifically, general 1+1d bosonic gapped phases with symmetry $\cC$ are in one-to-one correspondence with module categories over $\cC$ \cite{Thorngren:2019iar, Komargodski:2020mxz}.
In particular, SPT phases with symmetry $\cC$ correspond to $\cC$-module categories with a single simple object, which can be identified with fiber functors of $\cC$.
As in the case of SPT phases, the data of a module category can be interpreted in terms of topological field theory that describes the low energy limit of a gapped phase with symmetry \cite{Huang:2021zvu}.
In this paper, we will focus only on 1+1d SPT phases except for Sec.~\ref{sec: Generalizations}.

\subsection{Isomorphisms of fiber functors}
\label{sec: Isomorphisms of fiber functors}
In this subsection, we recall the definition of isomorphisms between fiber functors, which play a crucial role in studying interfaces and parameterized families of SPT phases in later sections.

Let $(F, J)$ and $(F^{\prime}, J^{\prime})$ be fiber functors of a fusion category $\cC$.\footnote{Here and in the subsequent sections, we take the isomorphism $\varphi$ in Eq.~\eqref{eq: fiber functor isom} to be the identity without loss of generality.}
A monoidal natural transformation $\eta: F \Rightarrow F^{\prime}$ between fiber functors $F$ and $F^{\prime}$ is a set of linear maps
\begin{equation}
\eta_x: F(x) \rightarrow F^{\prime}(x), \quad \forall x \in \cC
\end{equation}
that are natural in $x \in \cC$ and satisfy the following commutative diagram:\footnote{The linear map $\eta_x$ is said to be natural in $x$ if it satisfies $\eta_y \circ F(f) = F^{\prime}(f) \circ \eta_x$ for any morphism $f \in \Hom_{\cC}(x, y)$. We note that the structure isomorphisms such as $\alpha_{x, y, z}$ in Eq.~\eqref{eq: associator} and $J_{x, y}$ in Eq.~\eqref{eq: fiber functor isom} are required to be natural.}
\begin{equation}
\adjincludegraphics[valign = c, width = 0.3\linewidth]{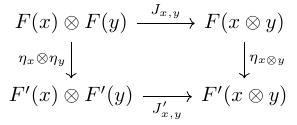}
\label{eq: natural isom cd}
\end{equation}
The monoidal natural transformation $\eta$ is called a monoidal natural isomorphism if the linear maps $\eta_x$ are isomorphisms for all $x \in \cC$.
In what follows, a monoidal natural isomorphism between fiber functors is simply referred to as an isomorphism.
In particular, an isomorphism from fiber functor $F$ to itself is called an automorphism of $F$.
Once we choose bases of $F(x)$ and $F^{\prime}(x)$ to be $\{(\phi_x)_i\}$ and $\{(\phi^{\prime}_x)_i\}$ respectively, we can explicitly write down a monoidal natural transformation $\eta: F \Rightarrow F^{\prime}$ as
\begin{equation}
\eta_x((\phi_x)_i) = \sum_{1 \leq j \leq \dim(x)} (\eta_x)_{ij} (\phi^{\prime}_x)_j,
\label{eq: eta}
\end{equation}
where $(\eta_x)_{ij}$ is the $(i, j)$-component of a $\dim(x) \times \dim(x)$ matrix $\eta_x$.
When $\eta_x: F \Rightarrow F^{\prime}$ is an isomorphism, the matrix $\eta_x$ is invertible for all $x \in \cC$.

The composition of isomorphisms $\eta: F \Rightarrow F^{\prime}$ and $\eta^{\prime}: F^{\prime} \Rightarrow F^{\prime \prime}$ defines an isomorphism $\eta^{\prime} \circ \eta: F \Rightarrow F^{\prime \prime}$.
In particular, automorphisms of a fiber functor $F$ form a finite group $\Aut^{\otimes}(F)$ with the multiplication given by the composition.

For later convenience, let us explicitly write down the commutative diagram \eqref{eq: natural isom cd} in terms of matrix representations of $J_{x, y}$ and $\eta_x$.
To this end, we notice that Eq.~\eqref{eq: natural isom cd} is equivalent to
\begin{equation}
\eta_z \circ F((\phi_{xy}^z)_{\mu}) \circ J_{x, y} = F^{\prime}((\phi_{xy}^z)_{\mu}) \circ J^{\prime}_{x, y} \circ (\eta_x \otimes \eta_y), \quad \forall (\phi_{xy}^z)_{\mu} \in \Hom_{\cC}(x \otimes y, z).
\end{equation}
By taking the matrix elements of both sides of the above equation, we obtain
\begin{equation}
\sum_k J^{xy}_{(i, j), (z; k, \mu)} (\eta_z)_{k k^{\prime}} = \sum_{i^{\prime}, j^{\prime}} (\eta_x)_{i i^{\prime}} (\eta_y)_{j j^{\prime}} (J^{\prime})^{xy}_{(i^{\prime}, j^{\prime}), (z; k^{\prime}, \mu)}.
\label{eq: natural isom elements}
\end{equation}

\vspace{10pt}
\noindent{\bf Classification of parameterized families of SPT states.}
We will argue in Sec.~\ref{sec: Parameterized family and Thouless pump} that $\Aut^{\otimes}(F)$ classifies $S^1$-paramererized families of invertible states in a $\cC$-SPT phase labeled by $F$.
The composition of automorphisms of $F$ corresponds to the concatenation of two $S^1$-families.
In contrast to the classification of SPT phases, the classification of $S^1$-parameterized families of SPT states has a natural group structure.

When $\cC$ is an ordinary group symmetry, the set of SPT phases also has a group structure.
Accordingly, the classification of parameterized families does not depend on a phase.
This is because stacking an SPT state with the whole family gives another family of SPT states in a different phase.
Indeed, as we will see below, the group $\Aut^{\otimes}(F)$ does not depend on $F$ when $\cC$ is an ordinary group symmetry.
On the other hand, when $\cC$ is a non-invertible symmetry, the set of SPT phases does not necessarily have a group structure.
Consequently, the classification of parameterized families may depend on a phase.
Indeed, as we will see below, the group $\Aut^{\otimes}(F)$ generally depends on $F$ when $\cC$ is non-invertible.

\vspace{10pt}
\noindent{\bf Tannaka-Krein duality.}
Let us consider the set $\End(F)$ of all (not necessarily monoidal) natural transformations from $F$ to itself.
The set $\End(F)$ is a vector space because for any natural transformations $\eta, \eta^{\prime} \in \End(F)$ and complex numbers $\lambda, \lambda^{\prime} \in \C$, one can define a natural transformation $\lambda \eta + \lambda^{\prime} \eta^{\prime} \in \End(F)$ by setting $(\lambda \eta + \lambda^{\prime} \eta^{\prime})_x := \lambda \eta_x + \lambda^{\prime} \eta^{\prime}_x$ for all $x \in \cC$.
The vector space $\End(F)$ has the structure of an algebra, whose multiplication is given by the composition of natural transformations.
Furthermore, the monoidal structure of $F$ allows us to define a comultiplication of natural transformations, which, together with the antipode, makes $\End(F)$ into a Hopf algebra.
The representation category $\Rep(\End(F))$ of this Hopf algebra turns out to be equivalent to the original fusion category $\cC$, see, e.g., \cite[Chapter 5]{EGNO2015}.
In particular, when $\cC = \Rep(H)$ and $F: \cC \rightarrow \Vect$ is the forgetful functor, the algebra $\End(F)$ is isomorphic to $H$ as a Hopf algebra, meaning that one can reconstruct a Hopf algebra from its representation category.
This phenomenon is known as Tannaka-Krein duality \cite{Tannaka1939, Krein1949, Ulbrich1990, JS1991, Schauenburg1992}.

We note that monoidal natural automorphisms of a fiber functor $F: \cC \rightarrow \Vect$ are group-like elements of Hopf algebra $\End(F)$.
In particular, when $\cC$ is $\Rep(H)$ and $F$ is the forgetful functor, the group $\Aut^{\otimes}(F)$ of automorphisms of $F$ is isomorphic to the group $G(H)$ of group-like elements of $H$.
For example, when $F$ is the forgetful functor of $\Rep(G)$ with $G$ a finite group, we have a group isomorphism $\Aut^{\otimes}(F) \cong G$.
This fact will naturally appear in later sections where we study parameterized families of SPT states with fusion category symmetries.

\vspace{10pt}
\noindent{\bf Examples.}
Let us see a few simple examples of fiber functors and their automorphisms.
\begin{itemize}
\item \textbf{$\Vect_G$.}
We first consider fiber functors of $\Vect_G$, the category of $G$-graded vector spaces.
Simple objects of $\Vect_G$ are labeled by elements of $G$, and the fusion rules are given by the group multiplication law:
\begin{equation}
g \otimes h = gh, \quad \forall g, h \in G.
\end{equation}
The associator and $F$-symbols of $\Vect_G$ are all trivial.
A fiber functor $F: \Vect_G \rightarrow \Vect$ maps every simple object $g$ to a one-dimensional vector space $F(g) \cong \C$ because the quantum dimension of $g$ is one.
Therefore, the fiber functor $F$ is completely characterized by a set of isomorphisms $J_{g, h}: F(g) \otimes F(h) \rightarrow F(gh)$.
Since the source and target of $J_{g, h}$ are both one-dimensional, it can be identified with a complex number $J(g, h) \in \mathrm{U}(1)$.\footnote{We recall that fusion categories are supposed to be unitary in this paper. Thus, $J(g, h)$ takes values in $\mathrm{U}(1)$ rather than $\C^{\times}$.}
With this identification, the commutative diagram \eqref{eq: fiber functor cd} reduces to the cocycle condition on $J: G \times G \rightarrow \mathrm{U}(1)$, i.e., 
\begin{equation}
J(h, k) J (g, hk) = J(gh, k) J(g, h), \quad \forall g, h, k \in G.
\end{equation}
Similarly, an isomorphism $\eta: F \Rightarrow F^{\prime}$ between two fiber functors $F$ and $F^{\prime}$ is characterized by the set of isomorphisms $\eta_g: F(g)\rightarrow F^{\prime}(g)$, which can be identified with complex numbers $\eta(g) \in \mathrm{U}(1)$.
The commutative diagram \eqref{eq: natural isom cd} reduces to
\begin{equation}
J(g, h) = J^{\prime}(g, h) \delta \eta(g, h), \quad \delta \eta(g, h) = \eta(g) \eta(h) / \eta(gh),
\label{eq: VecG natural isom}
\end{equation}
where $J$ and $J^{\prime}$ are 2-cocycles associated with $F$ and $F^{\prime}$ respectively.
The above equation implies that two fiber functors are isomorphic to each other if and only if the associated 2-cocycles $J$ and $J^{\prime}$ are in the same cohomology class.
Thus, isomorphism classes of fiber functors of $\Vect_G$ are classified by the second group cohomology $\mathrm{H}_{\text{gp}}^2(G, \mathrm{U}(1))$.
This agrees with the well-known classification of 1+1d bosonic SPT phases with symmetry $G$ \cite{Chen:2010zpc, Schuch_2011, Pollmann_2010}.

Equation \eqref{eq: VecG natural isom} also implies that an automorphism of a fiber functor $F$ is given by a function $\eta: G \rightarrow \mathrm{U}(1)$ that satisfies the cocycle condition
\begin{equation}
\delta \eta (g, h) = 1, \quad \forall g, h \in G.
\end{equation}
Therefore, the group $\Aut^{\otimes}(F)$ of automorphisms of $F$ is isomorphic to the first group cohomology $\mathrm{H}_{\text{gp}}^1(G, \mathrm{U}(1))$ regardless of a fiber functor $F$.
We note that $\Aut^{\otimes}(F) = \mathrm{H}_{\text{gp}}^1(G, \mathrm{U}(1))$ classifies $S^1$-parameterized families of $G$-symmetric invertible states in a 1+1d bosonic SPT phase~\cite{Thorngren:1612.00846, Hermele2021CMSA, Thorngren2021YITP, Shiozaki:2021weu, Bachmann:2022bhx}.

\item \textbf{$\Rep(G)$.}
As another example, we consider fiber functors of the representation category $\Rep(G)$ of a finite group $G$.
It is known that fiber functors of $\Rep(G)$ are classified by (equivalence classes of) pairs $(K, \omega)$, where $K$ is a subgroup of $G$ and $\omega \in \mathrm{Z}_{\text{gp}}^2(K, \mathrm{U}(1))$ is a 2-cocycle such that the twisted group algebra $\C[K]^{\omega}$ is simple \cite[Corollary 7.12.20]{EGNO2015}.
In particular, the fiber functor corresponding to the trivial subgroup $K=\{e\}$ is the forgetful functor, which exists for all $G$.

Due to Tannaka-Krein duality, the group of automorphisms of the forgetful functor of $\Rep(G)$ is isomorphic to $G$ itself.
This shows that the group of automorphisms of a fiber functor can be non-abelian when the fusion category has non-invertible objects.
This is in contrast to the fact that the group of automorphisms $\Aut^{\otimes}(F) = \mathrm{H}_{\text{gp}}^1(G, \mathrm{U}(1))$ is always abelian for an invertible (i.e., pointed) fusion category $\Vect_G$.

When the fiber functor $F: \Rep(G) \rightarrow \Vect$ corresponds to a non-trivial pair $(K, \omega)$, the group of automorphisms of $F$ is not necessarily isomorphic to $G$.
One can indeed construct such a fiber functor if there exists a group $G^{\prime}$ that is not isomorphic to $G$ but has the equivalent representation category as $G$.
Specifically, the composition of an equivalence $\Phi: \Rep(G) \rightarrow \Rep(G^{\prime})$ and the forgetful functor of $\Rep(G^{\prime})$ gives us a fiber functor $F: \Rep(G) \rightarrow \Vect$, whose group of automorphisms is isomorphic to $G^{\prime}$ due to Tannaka-Krein duality.
An example of such a pair $(G, G^{\prime})$ was constructed in \cite{EG2001}, where such $G$ and $G^{\prime}$ are called isocategorical.
The fact that $\Aut^{\otimes}(F)$ depends on a fiber functor $F$ is also in contrast to the case of invertible symmetry $\Vect_G$.

\end{itemize}

\subsection{Matrix product states with fusion category symmetries}
\label{sec: Matrix product states with fusion category symmetries}
The ground states of a gapped phase in 1+1 dimensions can be efficiently represented by matrix product states (MPSs) \cite{FNW92, PhysRevB.73.094423, Hastings:2007iok, Cirac:2020obd}.
Furthermore, the symmetry operators acting on an MPS can often be represented by matrix product operators (MPOs) \cite{Molnar:2022nmh,Garre-Rubio:2022uum,Lootens:2021tet,Lootens:2022avn}.
Therefore, MPSs symmetric under the action of MPOs are useful tools to study 1+1d gapped phases with fusion category symmetries.
In this subsection, following \cite{Garre-Rubio:2022uum}, we introduce MPSs symmetric under MPOs and discuss the invariants associated with such MPSs.
See, e.g., \cite{Cirac:2020obd} for a recent review of MPSs and MPOs.

\vspace{10pt}
\noindent{\bf Matrix product states.}
An MPS is a tensor network state constructed from a three-leg tensor $A$ as
\begin{equation}
\ket{A} = \adjincludegraphics[scale=1.2,trim={10pt 10pt 10pt 10pt}, valign = c]{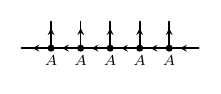},
\label{eq: MPS}
\end{equation}
where each vertical leg represents a physical Hilbert space $\mathcal{H} \cong \C^d$ and each horizontal leg represents a bond Hilbert space $V \cong \C^D$.
The arrows on the edges specify the source and target of each tensor $A: V \rightarrow V \otimes \mathcal{H}$.
These arrows will be omitted when no confusion can arise.\footnote{The choice of the orientation of an arrow is just a matter of convention.}
For simplicity, we only consider translationally invariant MPSs in this paper, i.e., the MPS tensor $A$ does not depend on sites.
Furthermore, the MPSs discussed in this paper are supposed to obey a periodic boundary condition unless otherwise stated.

Given a basis $\{b_i \mid i = 1, 2, \cdots, d\}$ of the physical Hilbert space $\mathcal{H}$, one can express the MPS \eqref{eq: MPS} on a periodic chain as
\begin{equation}
\ket{A} = \sum_{i_1, i_2, \cdots, i_L} \mathop{\text{tr}}(A^{i_1} A^{i_2} \cdots A^{i_n}) \ket{b_{i_1}, b_{i_2}, \cdots, b_{i_L}},
\end{equation}
where $A^i$ is a $D \times D$ matrix and $L$ denotes the number of lattice sites.
An MPS tensor $A$ is said to be injective if the set $\{A^i \mid i = 1, 2, \cdots, d\}$ generates the full matrix algebra $M_{D}(\C)$.
An injective MPS can be realized as the unique ground state of a gapped Hamiltonian belonging to an SPT phase \cite{Perez-Garcia:2006nqo}.
Conversely, the ground state of an SPT phase can always be approximated by an injective MPS \cite{FNW92, PhysRevB.73.094423, Hastings:2007iok, Cirac:2020obd}.
As such, we can restrict our attention to injective MPSs when studying SPT phases.

\vspace{10pt}
\noindent{\bf Matrix product operators.}
An MPO is a tensor network operator built from a four-leg tensor $\cO$ as
\begin{equation}
\widehat{\cO} = \adjincludegraphics[scale=1.2,trim={10pt 10pt 10pt 10pt},valign = c]{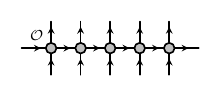}.
\label{eq: MPO}
\end{equation}
We denote the bond Hilbert space of an MPO $\widehat{\cO}$ by $V_{\cO} \cong \C^{\chi}$, where $\chi$ is the bond dimension.
The arrows on the edges show that the MPO tensor $\cO$ is a linear map from $V_{\cO} \otimes \cH$ to $\cH \otimes V_{\cO}$.
For a given basis $\{b_i \mid i = 1, 2, \cdots, d\}$ of the physical Hilbert space $\cH$, the MPO \eqref{eq: MPO} on a periodic chain can be expressed in terms of $\chi \times \chi$ matrices $\{\cO^{ij} \mid i, j = 1, 2, \cdots, d\}$ as
\begin{equation}
\widehat{\cO} = \sum_{i_1 \cdots, i_L} \sum_{j_1, \cdots, j_L} \mathop{\text{tr}}(\cO^{i_1 j_1} \cdots \cO^{i_L j_L}) \ket{b_{i_1}, \cdots, b_{i_L}} \bra{b_{j_1}, \cdots, b_{j_L}}.
\end{equation}
An MPO is said to be injective if the set $\{\cO^{ij} \mid i, j = 1, 2, \cdots, d\}$ generates the full matrix algebra $M_{\chi}(\C)$.

MPOs can be used to represent fusion category symmetries on the lattice \cite{Bultinck:2015bot, Lootens:2021tet, Lootens:2022avn, Molnar:2022nmh, Garre-Rubio:2022uum}.
Indeed, for any fusion category $\cC$, there exists an MPO representation of $\cC$, i.e., a set of injective MPOs $\{\widehat{\cO}_x \mid x \in \Simp(\cC)\}$ together with three-leg tensors $\{(\phi_{xy}^z)_{\mu} \mid \mu = 1, 2, \cdots, N_{xy}^z\}$ and $\{(\overline{\phi}_{xy}^z)_{\mu} \mid \mu = 1, 2, \cdots, N_{xy}^z\}$ that allow us to locally fuse two MPOs $\widehat{\cO}_x$ and $\widehat{\cO}_y$ as follows:
\begin{equation}
\cO_{x \otimes y} := \cO_x \cO_y = \adjincludegraphics[scale=1.2,trim={10pt 10pt 10pt 10pt},valign = c]{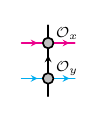}
= \sum_{z \in \Simp(\cC)} \sum_{1 \leq \mu \leq N_{xy}^z} \adjincludegraphics[scale=1.2,trim={10pt 10pt 10pt 10pt},valign = c]{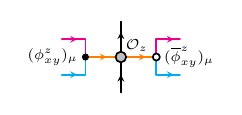}.
\label{eq: fusion tensors}
\end{equation}
The tensors $\{(\phi_{xy}^z)_{\mu}\}$ and $\{(\overline{\phi}_{xy}^z)_{\mu}\}$ are called fusion and splitting tensors respectively.
These tensors satisfy the following orthogonality relation:
\begin{equation}
\adjincludegraphics[valign=c, scale=1.2, trim={10pt 10pt 10pt 10pt}]{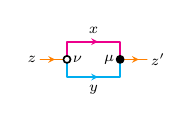}
= \delta_{\mu, \nu} \delta_{z, z^{\prime}} \adjincludegraphics[valign=c, scale=1.2, trim={10pt 10pt 10pt 10pt}]{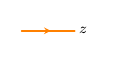}\;.
\label{eq: fusion orthogonality}
\end{equation}
On the left-hand side, the fusion and splitting tensors $(\phi_{xy}^z)_{\mu}$ and $(\overline{\phi}_{xy}^z)_{\nu}$ are abbreviated as $\mu$ and $\nu$, respectively.
Equations~\eqref{eq: fusion tensors} and \eqref{eq: fusion orthogonality} imply that an MPO representation of $\cC$ on a periodic chain obeys the same fusion rules as $\cC$, i.e.,
\begin{equation}
\widehat{\cO}_x \widehat{\cO}_y = \sum_{z} N_{xy}^z \widehat{\cO}_z.
\label{eq: MPO fusion rules}
\end{equation}
Using the fusion and splitting tensors in Eq.~\eqref{eq: fusion tensors}, one can extract the $F$-symbols of $\cC$ as follows:
\begin{equation}
\adjincludegraphics[valign=c, scale=1.2, trim={10pt 10pt 10pt 10pt}]{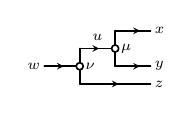}
= \sum_{v} \sum_{\rho, \sigma} (F^{xyz}_w)_{(u; \mu, \nu), (v, \rho, \sigma)}
\adjincludegraphics[valign=c, scale=1.2, trim={10pt 10pt 10pt 10pt}]{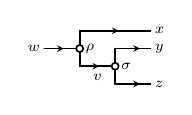}.
\label{eq: MPO F}
\end{equation}
The pentagon equation for the $F$-symbols follows from the associativity of MPOs.
Explicit construction of MPOs $\{\widehat{O}_x \mid x \in \cC\}$ that implement general fusion category symmetry $\cC$ is discussed in, e.g., \cite{Molnar:2022nmh, Lootens:2021tet, Lootens:2022avn}.

As a simple example, let us consider an MPO representation of $\Rep(G)$.
For concreteness, we choose the physical Hilbert space $\mathcal{H}$ to be a $G$-dimensional vector space faithfully graded by $G$, i.e., 
\begin{equation}
\mathcal{H} = \bigoplus_{g \in G} \C \ket{g},
\end{equation}
where $\{\ket{g} \mid g \in G\}$ is a basis of $\mathcal{H}$.
The bond Hilbert space of an MPO $\widehat{\cO}_{\rho}$ is given by the representation space $V_{\rho}$ of representation $\rho \in \Rep(G)$.
One can define the MPO tensor $\cO_{\rho}$ as
\begin{equation}
\cO_{\rho} = \sum_{g \in G} \rho(g) \otimes \ket{g} \bra{g},
\label{eq: MPO Rep(G)}
\end{equation}
where $\rho(g)$ is the representation matrix acting on the bond Hilbert space $V_{\rho}$.
We note that the MPO $\widehat{\cO}_{\rho}$ labeled by an irreducible representation $\rho$ is injective due to Schur's lemma. 
For the above MPO representation, the fusion tensors are simply given by projections from the tensor product representation $\rho_1 \otimes \rho_2 \cong \bigoplus_{\rho_3} N_{\rho_1 \rho_2}^{\rho_3} \rho_3$ to its fusion channel $\rho_3$.
Similarly, the splitting tensors are given by inclusion maps from $\rho_3$ to $\rho_1 \otimes \rho_2$.
By construction, the $F$-symbols defined by Eq. \eqref{eq: MPO F} agree with those of fusion category $\Rep(G)$.
The MPO representation \eqref{eq: MPO Rep(G)} can be generalized to that of $\Rep(H)$ where $H$ is a Hopf algebra \cite{Inamura:2021szw} or a weak Hopf algebra \cite{Molnar:2022nmh}.
Here, we recall that a general fusion category is equivalent to the representation category of a weak Hopf algebra, while a general non-anomalous fusion category is equivalent to the representation category of a Hopf algebra.

\vspace{10pt}
\noindent{\bf MPSs symmetric under MPOs.}
Let $\{\widehat{\cO}_x \mid x \in \Simp(\cC)\}$ be an MPO representation of a fusion category $\cC$.
An injective MPS $A$ is said to be $\cC$-symmetric if there exist three-leg tensors $\{(\phi_x)_i\}$ and $\{(\overline{\phi}_x)_i\}$ called action tensors, which allow us to apply an MPO $\hat{\cO}_x$ locally to an injective MPS $\ket{A}$ as follows \cite{Garre-Rubio:2022uum}:\footnote{In \cite{Garre-Rubio:2022uum}, a block-injective MPS symmetric under MPO symmetry is defined by the condition that the set of MPSs with arbitrary boundary conditions is closed under the action of the MPO algebra. The existence of the action tensors follows from this condition.}
\begin{equation}
\cO_x A = \;\adjincludegraphics[scale=1.2,trim={10pt 10pt 10pt 10pt},valign = c]{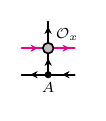}\;
=\sum_i\;\adjincludegraphics[scale=1.2,trim={10pt 10pt 10pt 10pt},valign = c]{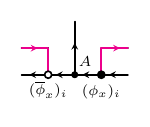}\;.
\label{eq: action tensor}
\end{equation}
Here, the index $i$ runs from 1 to some positive integer $n_x$.
The action tensors are required to satisfy the following orthogonality relation:
\begin{equation}
\;\adjincludegraphics[scale=1.2,trim={10pt 10pt 10pt 10pt},valign = c]{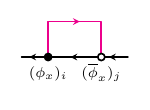}\;
= \delta_{ij} \;\adjincludegraphics[scale=1.2,trim={10pt 10pt 10pt 10pt},valign = c]{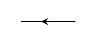}\;.
\label{eq: action orthogonality}
\end{equation}
The sets of action tensors $\{(\phi_x)_i \mid i = 1, 2, \cdots, n_x\}$ and $\{(\overline{\phi}_x)_i \mid i = 1, 2, \cdots, n_x\}$ are simply written as $\phi$ and $\overline{\phi}$ when no confusion can arise.
Equations~\eqref{eq: action tensor} and \eqref{eq: action orthogonality} imply that a $\cC$-symmetric injective MPS $\ket{A}$ on a periodic chain is invariant under the action of $\cC$ up to scalar multiplication:
\begin{equation}
\widehat{\cO}_x \ket{A} = n_x \ket{A}.
\label{eq: invariance}
\end{equation}
This equation forces $n_x = \dim(x)$ because the quantum dimension is the unique one-dimensional representation of the fusion rules such that all objects are represented by positive real numbers.

Using the action tensors in Eq.~\eqref{eq: action tensor}, one can define a set of complex numbers $L^{xy}_{(i, j), (z; k, \mu)}$ called $L$-symbols:
\begin{equation}
\adjincludegraphics[scale=1.2,trim={10pt 10pt 10pt 10pt},valign = c]{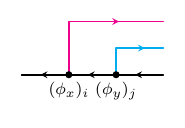}\;
=\sum_{z} \sum_{k, \mu} L^{xy}_{(i, j), (z; k, \mu)} \;\adjincludegraphics[scale=1.2,trim={10pt 10pt 10pt 10pt},valign = c]{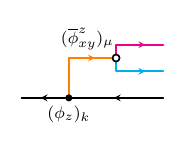}\;.
\label{eq: L-symbols}
\end{equation}
The associativity of MPO actions implies that the $L$-symbols satisfy consistency conditions analogous to the pentagon equation \cite{Garre-Rubio:2022uum}.
The inverse of the $L$-symbol $L^{xy}$ is denoted by $\overline{L}^{xy}$, which obeys
\begin{equation}
\adjincludegraphics[scale=1.2,trim={10pt 10pt 10pt 10pt},valign = c]{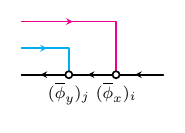}\;
= \sum_z \sum_{k, \mu} \overline{L}^{xy}_{(z;, k, \mu), (i, j)} \;\adjincludegraphics[scale=1.2,trim={10pt 10pt 10pt 10pt},valign = c]{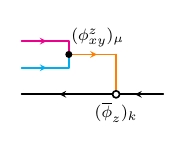}\;.
\label{eq: L-symbols inverse}
\end{equation}
Due to the orthogonality relation~\eqref{eq: action orthogonality}, the $L$-symbols $L^{xy}_{(i, j), (z; k, \mu)}$ and $\overline{L}^{xy}_{(z; k, \mu), (i, j)}$ can be expressed in terms of the following closed diagrams:
\begin{equation}
L^{xy}_{(i, j), (z; k, \mu)} = \frac{1}{D} \adjincludegraphics[valign=c, scale=1, trim={10pt 10pt 10pt 10pt}]{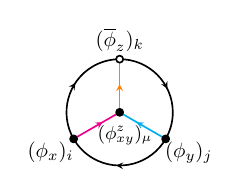}, \quad
\overline{L}^{xy}_{(z; k, \mu), (i, j)} = \frac{1}{D} \adjincludegraphics[valign=c, scale=1, trim={10pt 10pt 10pt 10pt}]{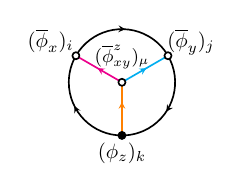},
\label{eq: L benz}
\end{equation}
where $D$ is the bond dimension of the MPS.

The $L$-symbols have gauge ambiguities because they depend on the choice of action tensors.
A gauge transformation of the action tensors is defined by
\begin{equation}
(\phi_x)_i \rightarrow {}^U(\phi_x)_i = \sum_{j} (U_x)_{ij} (\phi_x)_j, \quad 
(\overline{\phi}_x)_i \rightarrow {}^U(\overline{\phi}_x)_i = \sum_{j} (\overline{\phi}_x)_j (U_x)^{-1}_{ji},
\end{equation}
where $U_x$ is a $\dim(x) \times \dim(x)$ invertible matrix.
Correspondingly, a gauge transformation of the $L$-symbols is given by
\begin{equation}
L^{xy}_{(i, j), (z; k, \mu)} \rightarrow {}^UL^{xy}_{(i, j), (z; k, \mu)} = \sum_{i^{\prime}, j^{\prime}, k^{\prime}} (U_x)_{i i^{\prime}} (U_y)_{j j^{\prime}} (U_z)^{-1}_{k^{\prime} k} L^{xy}_{(i^{\prime}, j^{\prime}), (z; k^{\prime}, \mu)}.
\label{eq: L gauge}
\end{equation}
The $L$-symbols that are related by the above gauge transformation are said to be equivalent.

Equivalence classes of $L$-symbols classify $\cC$-symmetric SPT phases whose ground states are represented by injective MPSs \cite{Garre-Rubio:2022uum}.\footnote{Ref.~\cite{Garre-Rubio:2022uum} also defined $L$-symbols for block-injective MPSs symmetric under MPO symmetries. Such $L$-symbols can be identified with the data of a module category and characterize more general gapped phases.}
This is consistent with the classification of SPT phases via fiber functors.
Indeed, the $L$-symbols associated with a $\cC$-symmetric injective MPS can be identified with the data of a fiber functor as
\begin{equation}
L^{xy}_{(i, j), (z; k, \mu)} = J^{xy}_{(i, j), (z; k, \mu)},
\end{equation}
where the right-hand side is the matrix element of the isomorphism $J_{x, y}$ associated with a fiber functor, see Eq.~\eqref{eq: J}.
The consistency conditions on the $L$-symbols agree with the coherence conditions~\eqref{eq: fiber functor cd} on the isomorphism $J$.
Moreover, the gauge transformation \eqref{eq: L gauge} of the $L$-symbols defines an isomorphism \eqref{eq: natural isom elements} of the corresponding fiber functor, meaning that equivalence classes of $L$-symbols correspond to isomorphism classes of fiber functors.
In Sec.~\ref{sec: Fiber functors from injective MPS}, we will show that one can extract the data of a fiber functor/$L$-symbols from an injective MPS by using the triple inner product introduced in \cite{OR23}.

\vspace{10pt}
\noindent{\bf Completeness conditions.}
Before proceeding, we mention that the action tensors $\{(\phi_x)_i, (\overline{\phi}_x)_i\}$ and the fusion tensors $\{(\phi_{xy}^z)_{\mu}, (\overline{\phi}_{xy}^z)_{\mu}\}$ generally do not satisfy the following ``completeness" conditions:
\begin{equation}
\sum_i \;\adjincludegraphics[scale=1.3,trim={10pt 10pt 10pt 0pt},valign = c]{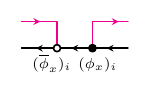}\;
= \;\adjincludegraphics[scale=1.2,trim={10pt 10pt 10pt 10pt},valign = c]{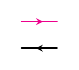}\;, \qquad
\sum_{z, \mu} \;\adjincludegraphics[valign=c, scale=1.3, trim={10pt 10pt 10pt 10pt}]{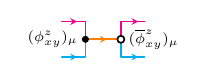}\;
= \;\adjincludegraphics[valign=c, scale=1.2, trim={10pt 10pt 10pt 10pt}]{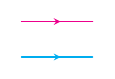}\;.
\label{eq: completeness}
\end{equation}
Given the orthogonality relations~\eqref{eq: fusion orthogonality} and \eqref{eq: action orthogonality}, one can show that the above completeness conditions are equivalent to the condition that the bond dimension $\chi_x$ of the MPO $\cO_x$ is equal to $\dim(x)$ for all $x \in \Simp(\cC)$, i.e.,
\begin{equation}
\text{Eq.~\eqref{eq: completeness}} \quad \Leftrightarrow \quad \chi_x = \dim(x) \text{ for all $x \in \Simp(\cC)$}.
\label{eq: completeness dim}
\end{equation}
The sketch of a proof goes as follows. 
First, we notice that the left-hand side of the first equality of Eq.~\eqref{eq: completeness} is an idempotent (as a linear map from left to right) due to the orthogonality relation~\eqref{eq: action orthogonality}.
The orthogonality further implies that this idempotent is bijective if and only if the bond dimension of $\cO_x$ agrees with the quantum dimension of $x$.
Then, the statement follows from the fact that a bijective idempotent is the identity map.
A similar argument shows that the second equality of Eq.~\eqref{eq: completeness} is also satisfied if and only if $\chi_x = \dim(x)$ for all $x \in \Simp(\cC)$.

Even though the completeness conditions do not hold in general, one can show the following weaker equalities:
\begin{equation}
\begin{aligned}
\sum_i \;\adjincludegraphics[valign=c, scale=1.2, trim={10pt 10pt 10pt 10pt}]{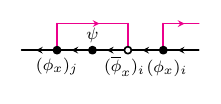}\;
& =\;\adjincludegraphics[valign=c, scale=1.2, trim={10pt 10pt 10pt 10pt}]{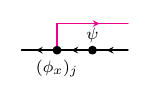}\;,
\\
\sum_{z, \mu}\;\adjincludegraphics[valign=c, scale=1.2, trim={10pt 10pt 10pt 0pt}]{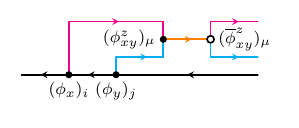}\;
& =\;\adjincludegraphics[scale=1.2,trim={10pt 10pt 10pt 0pt},valign = c]{tikz/out/L_symbols1.pdf}\;,
\end{aligned}
\label{eq: weak completeness}
\end{equation}
where $\psi \in \Aut(V)$ is an arbitrary unitary operator.
The first equality for $\psi = \text{id}$ follows from the orthogonality relation~\eqref{eq: action orthogonality} for the action tensors.
The case of a more general $\psi$ will be shown in Appendix~\ref{sec: Weak completeness relation} under a mild assumption on the duality structure.
The second equality follows from the orthogonality relation~\eqref{eq: fusion orthogonality} for the fusion tensors and the existence of the $L$-symbols~\eqref{eq: L-symbols}.
Equation~\eqref{eq: weak completeness} will be referred to as the weak completeness relation.

We note that an MPO representation that satisfies the completeness condition~\eqref{eq: completeness} exists for any non-anomalous fusion category $\Rep(H)$, where $H$ is a semisimple Hopf algebra.
The explicit construction of such MPOs is discussed in \cite{Inamura:2021szw, Molnar:2022nmh}.
For example, the MPO representation~\eqref{eq: MPO Rep(G)} of $\Rep(G)$ symmetry satisfies the completeness condition because the bond Hilbert space $V_{\rho}$ of an MPO $\cO_{\rho}$ is the representation space of $\rho \in \Rep(G)$ and hence we have $\chi_{\rho} = \dim(V_{\rho}) = \dim(\rho)$.
For a finite group symmetry $\Vect_G$, the completeness condition~\eqref{eq: completeness dim} means that the symmetry operators are on-site.
We do not assume the completeness condition in this paper.

\section{Fiber functors from injective MPS}
\label{sec: Fiber functors from injective MPS}

As we reviewed in Sec.~\ref{sec: Fiber functors and SPT phases}, SPT phases with fusion category symmetry $\cC$ are classified by fiber functors of $\cC$.
In this section, we discuss how to extract these data from the ground state of a lattice system.
The key idea is to exploit the correspondence between a TQFT partition function with defects and a generalized overlap of MPSs. 
In TQFT, a fiber functor $(F, J): \cC \rightarrow \Vect$ is represented diagrammatically as shown in Fig.~\ref{fig: pants}, where $F(x)$ is given by the $x$-twisted sector on a circle and $J_{x, y}$ is given the transition amplitude on a pants diagram. 
Correspondingly, in the context of MPS, it is expected that $F(x)$ is given by a vector space spanned by the fixed points of a mixed transfer matrix,
and $J_{x,y}$ is given by a triple inner product with three symmetry defects. In the following, we will explain this in detail.

\subsection{Non-abelian factor system and fiber functor}

\subsubsection{Non-abelian factor system}
\label{sec: Nonabelian factor system}
Let us first define a non-abelian generalization of the factor system from an injective MPS with fusion category symmetry.
This quantity turns out to be essentially the same as a fiber functor.

Let $\{\cO_x\left.\right|x\in \Simp(\mathcal{C})\}$ be an MPO representation of a fusion category symmetry $\mathcal{C}$, 
and let $A$ be a $\mathcal{C}$-symmetric injective MPS tensor.
The corresponding MPO matrix and MPS matrix are denoted by $\cO_{x}^{ij}$ and $A^{i}$, respectively.
The bond dimensions of $\cO_x$ and $A$ are given by $\chi_x$ and $D$.
Since $A$ is $\cC$-symmetric, it satisfies the symmetricity condition~\eqref{eq: action tensor}, which can be rewritten as
\begin{equation}
       \sum_j \cO^{ij}_x \otimes A^j = \hat{\phi}_x^{-1} (1_x \otimes A^i) \hat{\phi}_x
       \Longleftrightarrow
       \adjincludegraphics[scale=1,trim={10pt 10pt 10pt 10pt},valign = c]{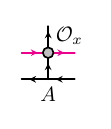}
       =\adjincludegraphics[scale=1,trim={10pt 10pt 10pt 10pt},valign = c]{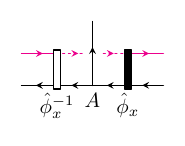}\;.
\label{eq: fractionalized symmetry}
\end{equation}
Here, $1_x$ denotes the $\dim(x) \times \dim(x)$ identity matrix, and the dotted line represents an auxiliary bond Hilbert space $F_x=\C^{\dim(x)}$.\footnote{When the bond dimension is $\chi_x = \dim(x)$, we can choose $F_x$ to be the bond Hilbert space itself.}
On the right-hand side, the crossing of two lines represents the trivial braiding $v \otimes w \mapsto w \otimes v$, and the tensors $\hat{\phi}_x$ and $\hat{\phi}_x^{-1}$ are defined by
\begin{equation}
\adjincludegraphics[valign=c, scale=1, trim={10pt 10pt 10pt 10pt}]{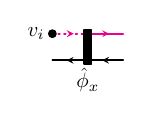}\;
=\;\adjincludegraphics[valign=c, scale=1, trim={10pt 10pt 10pt 10pt}]{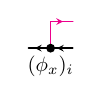}\;, \qquad
\adjincludegraphics[valign=c, scale=1, trim={10pt 10pt 10pt 10pt}]{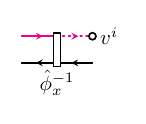}\;
=\;\adjincludegraphics[valign=c, scale=1, trim={10pt 10pt 10pt 10pt}]{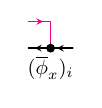},
\label{eq: basis-indep phi}
\end{equation}
where $\{v_i \mid i = 1, 2, \cdots, \dim(x)\}$ and $\{v^i \mid i = 1, 2, \cdots, \dim(x)\}$ are dual bases of $F_x$ and $F_x^{*}$.
We note that $\hat{\phi}_x^{-1}$ is the right inverse of $\hat{\phi}_x$ due to the orthogonality relation~\eqref{eq: action orthogonality}.\footnote{The tensor $\hat{\phi}_x^{-1}$ is also the left inverse of $\hat{\phi}_x$ if and only if the completeness condition~\eqref{eq: completeness} holds.}
By abuse of terminology, we will refer to $\hat{\phi}_x$ as the action tensor.
The action tensor $\hat{\phi}_x$ is uniquely determined up to the following gauge redundancy:
\begin{equation}
\hat{\phi}_x \mapsto \hat{\phi}^{\prime}_x = \;\adjincludegraphics[valign=c, scale=1, trim={10pt 10pt 10pt 0pt}]{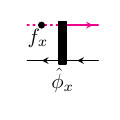}\;,
\qquad f_x \in GL_{\dim(x)}(\C).
\label{eq: gauge redundancy}
\end{equation}
Indeed, if there are two action tensors $\hat{\phi}_x^1$ and $\hat{\phi}_x^2$ that satisfy Eq.~\eqref{eq: fractionalized symmetry}, their composition $\hat{\phi}_x^1 (\hat{\phi}_x^2)^{-1}$ commutes with $1_x \otimes A^i$ for all $i$, and since $A$ is injective, $\hat{\phi}_x^1 (\hat{\phi}_x^2)^{-1}$ has to be of the form $f_x \otimes \text{id}_D$.

For a pair of two simple objects $x, y \in \Simp(\cC)$, we can define the action tensor $\hat{\phi}_{x \otimes y}$ for the tensor product object $x \otimes y \in \cC$ as follows:
\begin{equation}
\hat{\phi}_{x \otimes y} = \sum_{z, \mu} \;\adjincludegraphics[valign=c, scale=1, trim={10pt 10pt 10pt 10pt}]{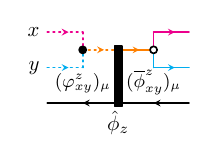}\;, \qquad
\hat{\phi}_{x \otimes y}^{-1} = \sum_{z, \mu} \;\adjincludegraphics[valign=c, scale=1, trim={10pt 10pt 10pt 10pt}]{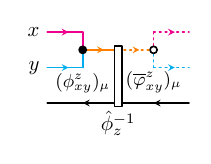}\;.
\label{eq: phi xy}
\end{equation}
Here, $(\varphi_{xy}^z)_{\mu}$ and $(\overline{\varphi}_{xy}^z)_{\mu}$ are fusion and splitting tensors for the auxiliary bond Hilbert spaces, which satisfy the orthogonality relation analogous to Eq.~\eqref{eq: fusion orthogonality}.
A straightforward computation shows that the above action tensors satisfy Eq.~\eqref{eq: fractionalized symmetry} where the MPO $\cO_x$ is replaced by $\cO_{x \otimes y}$ defined by Eq.~\eqref{eq: fusion tensors}.
One can also define the action tensors for the tensor product of more than two simple objects.
For example, for the tensor product of three simple objects $x, y, z \in \Simp(\cC)$, the action tensors $\hat{\phi}_{(x \otimes y) \otimes z}$ and $\hat{\phi}_{x \otimes (y \otimes z)}$ are given by
\begin{equation}
\hat{\phi}_{(x \otimes y) \otimes z} = \sum_{u, w, \mu, \nu} \;\adjincludegraphics[valign=c, scale=1, trim={10pt 10pt 10pt 10pt}]{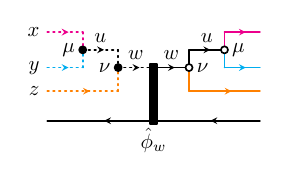}\;, \qquad
\hat{\phi}_{x \otimes (y \otimes z)} = \sum_{v, w, \rho, \sigma} \;\adjincludegraphics[valign=c, scale=1, trim={10pt 10pt 10pt 10pt}]{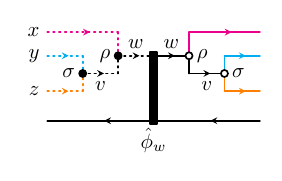}\;,
\label{eq: phi xyz}
\end{equation}
where the fusion and splitting tensors are abbreviated as $\mu$, $\nu$, $\rho$, and $\sigma$.
The above action tensors may differ from each other by a gauge transformation $f_{x, y, z} \in GL_{\dim(x)\dim(y)\dim(z)}$.
Nevertheless, one can always choose $\varphi$'s so that $f_{x, y, z}$ becomes the identity because the MPO tensors for $(x \otimes y) \otimes z$ and $x \otimes (y \otimes z)$ are the same, i.e., $\cO_{(x \otimes y) \otimes z} = \cO_{x \otimes (y \otimes z)}$.
More concretely, if we choose $\varphi$'s so that they have the same $F$-symbols as $\phi$'s, the two action tensors in Eq.~\eqref{eq: phi xyz} agree with each other.\footnote{For instance, when $\chi_x = \dim(x)$ for all $x \in \Simp(\cC)$, we can choose $(\varphi_{xy}^z)_{\mu} = (\phi_{xy}^z)_{\mu}$.}
In what follows, we will suppose $\hat{\phi}_{(x \otimes y) \otimes z} = \hat{\phi}_{x \otimes (y \otimes z)}$ without loss of generality.

Since the action tensors $\{\hat{\phi}_x \mid x \in \cC\}$ have the gauge ambiguity~\eqref{eq: gauge redundancy}, $\hat{\phi}_{x \otimes y}$ generally differs from the multiplication of $\hat{\phi}_x$ and $\hat{\phi}_y$ by a gauge transformation.
Namely, we have
\begin{equation}
    \adjincludegraphics[scale=1,trim={10pt 0pt 10pt 0pt},valign = c]{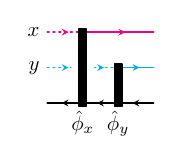}
    =\adjincludegraphics[scale=1,trim={10pt 0pt 10pt 0pt},valign = c]{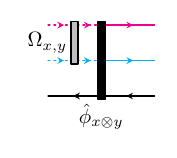}.
\label{eq: def Omega}
\end{equation}
The linear map $\Omega_{x, y}: F_x \otimes F_y \rightarrow F_x \otimes F_y$ will be referred to as a non-abelian factor system. 
To derive the consistency condition on $\Omega_{x, y}$, let us consider the multiplication of three action tensors $\hat{\phi}_x$, $\hat{\phi}_y$, and $\hat{\phi}_z$ in two different ways: 
\begin{equation}
  \adjincludegraphics[scale=1,trim={10pt 10pt 10pt 10pt},valign = c]{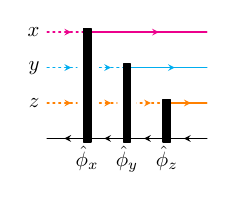}\;
  = \;\adjincludegraphics[scale=1,trim={10pt 10pt 10pt 10pt},valign = c]{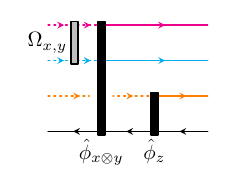}\;
  = \;\adjincludegraphics[scale=1,trim={10pt 10pt 10pt 10pt},valign = c]{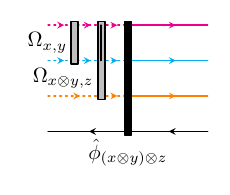}\;,
\label{eq: Omega associativity 1}
\end{equation}
\begin{equation}
  \adjincludegraphics[scale=1,trim={10pt 10pt 10pt 10pt},valign = c]{tikz/out/associativity1_new.pdf}\;
  = \;\adjincludegraphics[scale=1,trim={10pt 10pt 10pt 10pt},valign = c]{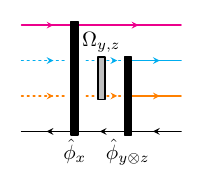}\;
  = \;\adjincludegraphics[scale=1,trim={10pt 10pt 10pt 10pt},valign = c]{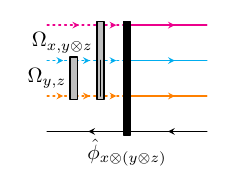}\;.
\label{eq: Omega associativity 2}
\end{equation}
By comparing Eqs.~\eqref{eq: Omega associativity 1} and \eqref{eq: Omega associativity 2}, we find that $\Omega_{x, y}$ satisfies a non-abelian cocycle condition:
\begin{equation}
  \Omega_{x \otimes y, z} (\Omega_{x, y} \otimes1_z) = \Omega_{x , y \otimes z} (1_x\otimes\Omega_{y, z})
  \Leftrightarrow
  \adjincludegraphics[scale=1,trim={10pt 10pt 10pt 10pt},valign = c]{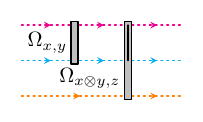}
  =\adjincludegraphics[scale=1,trim={10pt 10pt 10pt 10pt},valign = c]{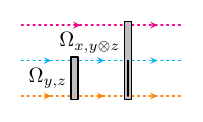}.
  \label{eq: nonabelian cocycle}
\end{equation}
Here, the diagrams in the above equation are read from left to right, and hence, e.g., the left-hand side represents $\Omega_{x \otimes y, z} (\Omega_{x, y} \otimes1_z)$ rather than $(\Omega_{x, y} \otimes1_z) \Omega_{x \otimes y, z}$.
We note that Eq.~\eqref{eq: nonabelian cocycle} agrees with the coherence condition~\eqref{eq: fiber functor cd} on a fiber functor, provided that the fusion category $\cC$ is strictified, i.e., $\alpha^{\cC}_{x, y, z} = \text{id}$.
Thus, the isomorphisms $\{\Omega_{x, y} \mid x, y \in \cC\}$ together with the set of vector spaces $\{F_x \mid x \in \cC\}$ can be identified with the data of a fiber functor $(F, \Omega): \cC \rightarrow \Vect$.

As we mentioned above, the action tensors $\{\hat{\phi}_x \mid x \in \Simp(\cC)\}$ have a gauge redundancy.
Under the gauge transformation~\eqref{eq: gauge redundancy}, the non-abelian factor system transforms as
\begin{equation}
  \Omega_{x,y} \mapsto (\bigoplus_{z \in x \otimes y} f_z)^{-1} \Omega_{x,y} (f_x \otimes f_y).
\label{eq: Omega gauge transformation}
\end{equation}
We note that the gauge transformation of $\hat{\phi}_{x \otimes y}$ is restricted to the one induced by the gauge transformations for simple objects $z \in x \otimes y$.
The non-abelian cocycle condition~\eqref{eq: nonabelian cocycle} is invariant under this gauge transformation.

\subsubsection{Mixed transfer matrix}
\label{sec: Mixed transfer matrix}

In the previous subsection, we argued that the data of a fiber functor is encoded in a $\cC$-symmetric injective MPS as a non-abelian factor system.
In the following, we will discuss how to extract this data.
To this end, we will use the mixed transfer matrix for a $\cC$-symmetric injective MPS and its fixed points. 

Given an MPS tensor $A$, a transfer matrix $T_{A}$ is defined by
\begin{eqnarray}
    T_{A}:=\sum_i A^{i\ast}\otimes A^{i}
    =\adjincludegraphics[scale=1,trim={10pt 10pt 10pt 10pt},valign = c]{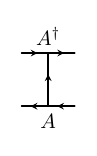},
\end{eqnarray}
where $A^{i\ast}$ denotes the complex conjugate of $A^i$.
The left and right eigenvectors of $T_A$ with the largest eigenvalues are called left and right fixed points.
When $A$ is injective, the left and right fixed points of $T_A$ are unique and have the same positive eigenvalue $\lambda \in \R$ \cite{SP-GWC10,Perez-Garcia:2006nqo,Cirac:2020obd}.
Without loss of generality, one can choose $\lambda=1$ by normalizing the MPS tensor $A$.
Thus, for an injective MPS, the left fixed point $\Lambda_A^L$ and the right fixed point $\Lambda_A^R$ are the unique (up to scalar) solutions to the following equation:
\begin{equation}
\adjincludegraphics[valign=c, scale=1, trim={10pt 10pt 10pt 10pt}]{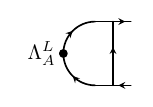}\; =
\;\adjincludegraphics[valign=c, scale=1, trim={10pt 10pt 10pt 10pt}]{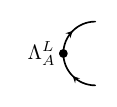}\;, \qquad
\adjincludegraphics[valign=c, scale=1, trim={10pt 10pt 10pt 10pt}]{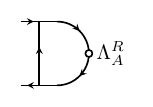}\; =
\;\adjincludegraphics[valign=c, scale=1, trim={10pt 10pt 10pt 10pt}]{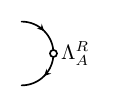}\;.
\end{equation}

Similarly, for a $\cC$-symmetric injective MPS tensor $A$, we define the mixed transfer matrix $T_x$ as 
\begin{equation}
T_{x} = \sum_{ij} A^{i\ast} \otimes \cO_{x}^{ij} \otimes  A^{j}
= \;\adjincludegraphics[scale=1,trim={10pt 10pt 10pt 10pt},valign = c]{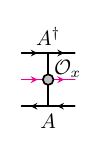}\;
= \sum_i \;\adjincludegraphics[scale=1,trim={10pt 10pt 10pt 10pt},valign = c]{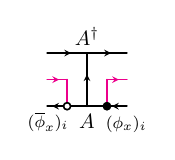}\;,
\end{equation}
where the last equality follows from Eq.~\eqref{eq: action tensor}.
The fixed point equation for the mixed transfer matrix $T_x$ is given by
\begin{equation}
\adjincludegraphics[valign=c, scale=1, trim={10pt 10pt 10pt 10pt}]{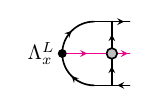}\;
= \;\adjincludegraphics[valign=c, scale=1, trim={10pt 10pt 10pt 10pt}]{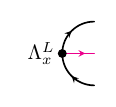}\;, \qquad
\adjincludegraphics[valign=c, scale=1, trim={10pt 10pt 10pt 10pt}]{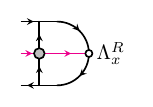}\;
= \;\adjincludegraphics[valign=c, scale=1, trim={10pt 10pt 10pt 10pt}]{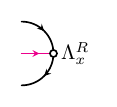}\;
\label{eq: fixed point eq}
\end{equation}
Combined with the orthogonality relation~\eqref{eq: action orthogonality}, the above equation implies that fixed points $\Lambda_x^L$ and $\Lambda_x^R$ have to satisfy
\begin{equation}
\adjincludegraphics[valign=c, scale=1, trim={10pt 10pt 10pt 10pt}]{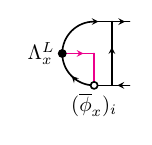}\;
= \;\adjincludegraphics[valign=c, scale=1, trim={10pt 10pt 10pt 10pt}]{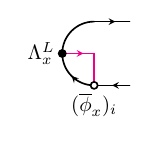}\;, \qquad
\adjincludegraphics[valign=c, scale=1, trim={10pt 10pt 10pt 10pt}]{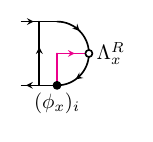}\;
= \;\adjincludegraphics[valign=c, scale=1, trim={10pt 10pt 10pt 10pt}]{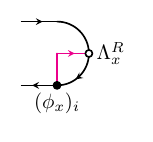}\;
\end{equation}
for all $i = 1, 2, \cdots, \dim(x)$.
This shows that $\Lambda_x^L (\overline{\phi}_x)_i$ and $(\phi_x)_i \Lambda_x^R$ are left and right fixed points of the transfer matrix $T_A$.
Since the fixed point of $T_A$ is unique, they are proportional to $\Lambda_A^L$ and $\Lambda_A^R$, i.e.,
\begin{equation}
\adjincludegraphics[valign=c, scale=1, trim={10pt 10pt 10pt 10pt}]{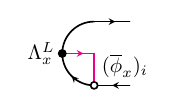}\; = 
a_{i}^L \Lambda_A^L, \qquad
\adjincludegraphics[valign=c, scale=1, trim={10pt 10pt 10pt 10pt}]{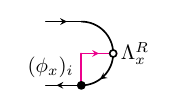}\; = 
a_{i}^R \Lambda_A^R,
\end{equation}
where $a_i^L$ and $a_i^R$ are arbitrary complex numbers.
Substituting this back into the fixed point equation~\eqref{eq: fixed point eq} leads us to
\begin{equation}
\Lambda_x^L = \sum_i a_i^L \Lambda_x^{L(i)}, \qquad
\Lambda_x^R = \sum_i a_i^R \Lambda_x^{R(i)}, 
\label{eq: general fixed point}
\end{equation}
where $\Lambda_x^{L(i)}$ and $\Lambda_x^{R(i)}$ are given by
\begin{equation}
\Lambda_x^{L(i)} = \;\adjincludegraphics[valign=c, scale=1, trim={10pt 10pt 10pt 0pt}]{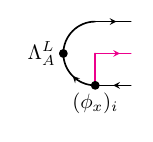}\; = 
\;\adjincludegraphics[valign=c, scale=1, trim={10pt 10pt 10pt 0pt}]{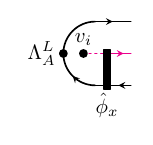}\;, \qquad
\Lambda_x^{R(i)} = \;\adjincludegraphics[valign=c, scale=1, trim={10pt 10pt 10pt 0pt}]{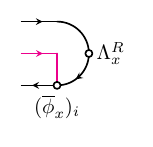}\; = 
\;\adjincludegraphics[valign=c, scale=1, trim={10pt 10pt 10pt 0pt}]{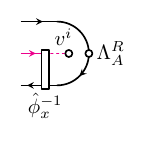}\;.
\label{eq: LR fixed points}
\end{equation}
Here, $\{v_i\}$ and $\{v^i\}$ are dual bases of the auxiliary bond Hilbert space $F_x = \C^{\dim(x)}$ and its dual $F_x^*$ as in Eq.~\eqref{eq: basis-indep phi}.
Equation~\eqref{eq: general fixed point} shows that any fixed point of $T_x$ is given by a linear combination of the fixed points in Eq.~\eqref{eq: LR fixed points}.

In the following, we often write the left fixed point $\Lambda_x^{L(i)}$ as $\ket{\Lambda_x^{L(i)}}$, and define a vector space $F^{\rm MPS}(x)$ assigned to each object $x\in \mathcal{C}$ as 
\begin{equation}
F^{\rm MPS}(x) = \bigoplus_{i=1}^{\dim(x)}\mathbb{C}\ket{\Lambda_{x}^{L(i)}}.
\label{eq: FMPS}
\end{equation}
By definition, the dimension of $F^{\rm MPS}(x)$ is $\dim(x)$.
Similarly, the right fixed point $\Lambda_x^{R(i)}$ will be often denoted by $\bra{\Lambda_x^{R(i)}}$.
This notation is justified because $\{\Lambda_x^{R(i)}\}$ is the dual basis of $\{\Lambda_x^{L(i)}\}$ due to the orthogonality relation~\eqref{eq: action orthogonality}.

\subsubsection{Triple inner product}
\label{sec: Triple inner product}
The assignment of the vector space~\eqref{eq: FMPS} defines a functor $F^{\rm MPS}: \cC \rightarrow \Vect$, which turns out to be a fiber functor.
To see this, let us consider the triple inner product of infinite MPSs \cite{OR23}, which is defined by the following diagram:
\begin{equation}
J_{x,y}^{\rm MPS}
:=\adjincludegraphics[trim={10pt 10pt 10pt 10pt},valign = c]{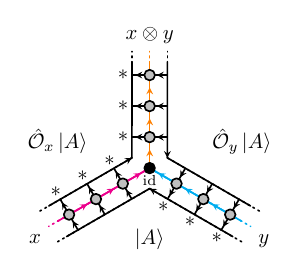}
:F^{\rm MPS}(x)\otimes F^{\rm MPS}(y)\to F^{\rm MPS}(x\otimes y).
\label{eq: MPS J}
\end{equation}
The middle dot in the diagram represents the identity map on the bond Hilbert space $V_x \otimes V_y$.
We emphasize that the triple inner product $J_{x,y}^{\rm MPS}$ is a linear map rather than a single number: terminating each wing of the diagram computes its matrix element.
Since each wing is infinitely long, the matrix element of $J_{x,y}^{\rm MPS}$ vanishes unless the endpoint of each wing is a fixed point of the mixed transfer matrix.\footnote{We recall that a fixed point of a (mixed) transfer matrix is an eigenvector with the largest eigenvalue $\lambda = 1$.}
Therefore, $J_{x, y}^{\rm MPS}$ can be regarded as a linear map from $F^{\rm MPS}(x) \otimes F^{\rm MPS}(y)$ to $F^{\rm MPS}(x \otimes y)$.

Let us explicitly compute the matrix element of $J_{x, y}^{\rm MPS}$.
To this end, we terminate the wings of $J_{x, y}^{\rm MPS}$ by fixed points $\ket{\Lambda_x^{L(i)}}$, $\ket{\Lambda_y^{L(j)}}$, and $\bra{\Lambda_{x \otimes y}^{R(k)}}$.
Here, $i$ runs from $1$ to $\dim(x)$, $j$ runs from $1$ to $\dim(y)$, and $k$ runs from $1$ to $\dim(x) \dim(y)$.
Using the fixed point equation~\eqref{eq: fixed point eq} successively annihilates the MPO tensors in Eq.~\eqref{eq: MPS J}.
As a result, the matrix element of $J_{x, y}^{\rm MPS}$ can be recast into
\begin{equation}
(J_{x, y}^{\rm MPS})_{(i, j), k}
:= \bra{\Lambda_{x\otimes y}^{R(k)}} J_{x,y}^{\rm MPS} \ket{\Lambda_{x}^{L(i)}}\otimes\ket{ \Lambda_{y}^{L(j)}}
=\adjincludegraphics[scale=1,trim={10pt 10pt 10pt 10pt},valign = c]{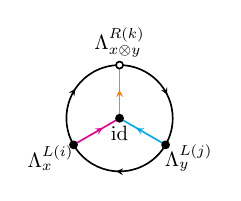}.
\label{eq: benz id}
\end{equation}
To evaluate the right-hand side, we suppose that $A$ is left canonical, meaning that the left fixed point of $T_A$ is the identity, i.e., $\Lambda_A^L = 1_D$.
We also suppose that the right fixed point $\Lambda_A^R$ is normalized so that $\tr{\Lambda_A^R} = 1$.\footnote{This is equivalent to the normalization of the infinite MPS.}
In this case, by substituting Eq.~\eqref{eq: LR fixed points} into the right-hand side of Eq.~\eqref{eq: benz id}, we obtain 
\begin{equation}
(J_{x, y}^{\rm MPS})_{(i, j), k}
= \;\adjincludegraphics[scale=1,trim={10pt 10pt 10pt 0pt},valign = c]{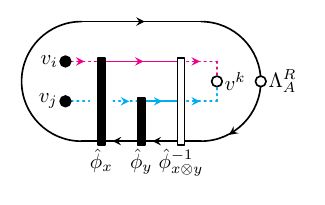}\;
= \;\adjincludegraphics[scale=1,trim={10pt 10pt 10pt 0pt},valign = c]{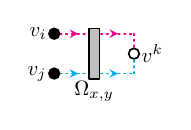}\;.
\label{eq: JMPS=Omega}
\end{equation}
This shows that the matrix element of the triple inner product $J_{x, y}^{\rm MPS}$ agrees with that of the non-abelian factor system $\Omega_{x, y}$.
As discussed in Sec.~\ref{sec: Nonabelian factor system}, $\Omega_{x, y}$ satisfies the coherence condition of a fiber functor.
Thus, $J_{x, y}^{\rm MPS}$ also satisfies the same coherence condition, meaning that $(F^{\rm MPS}, J^{\rm MPS}): \cC \rightarrow \Vect$ is a fiber functor.
It is now clear that one can also compute the $L$-symbols similarly: by inserting the fusion tensor $(\phi_{xy}^z)_{\mu}$ in the middle of the diagram, we obtain
\begin{equation}
L^{xy}_{(i,j),(z;k,\mu)}=\frac{1}{D}\adjincludegraphics[scale=1,trim={10pt 10pt 10pt 10pt},valign = c]{tikz/out/L_benz.pdf}
=\adjincludegraphics[trim={10pt 10pt 10pt 10pt},valign = c]{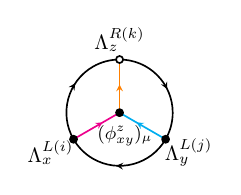}
=\;\adjincludegraphics[trim={10pt 10pt 10pt 10pt},valign = c]{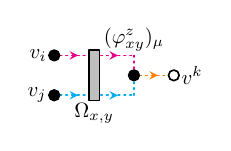}.
\label{eq: L MPS}
\end{equation}
We recall that $(\varphi_{xy}^z)_{\mu}$ on the right-hand side is an auxiliary fusion tensor, which has the same $F$-symbols as $(\phi_{xy}^z)_{\mu}$.

\subsubsection{Torus partition functions}
\label{sec: Torus partition functions}
In the preceding subsections, 
we explained how to compute the fiber functor from the triple inner product of infinite MPSs.
From the perspective of TQFT, 
the data of a fiber functor corresponds to the partition function on a pair of pants with defects as shown in Fig.~\ref{fig: pants}. 
Based on the ideas explained in Sec.~\ref{sec: Introduction and Summary}, 
it is also possible to evaluate more general partition functions in the framework of MPS. 
In this subsection, we evaluate the partition function on a torus as a specific example.

Let us consider the torus partition function in the presence of symmetry defects of the following configuration:
\begin{equation}
  Z\left[\;\;\adjincludegraphics[scale=0.8,trim={10pt 10pt 10pt 10pt},valign = c]{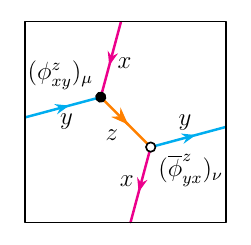}\;\;\right]\;\;.
\end{equation}
As in the computation of a fiber functor, we glue MPSs along these symmetry lines to evaluate the above partition function:
\begin{equation}
  \adjincludegraphics[scale=0.75,trim={10pt 10pt 10pt 10pt},valign = c]{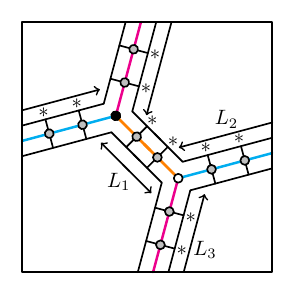}\;\;.
\end{equation}
Here, we put a conjugation $*$ on the left side of the defects. 
The complex number evaluated by the above diagram depends on the number of sites $L_1$, $L_2$, and $L_3$ in each interval.
However, by taking the large volume limit where $L_i$ goes to infinity for all $i$,
this dependence disappears, 
and the diagram reduces to the TQFT partition function:
\begin{equation}
  Z\left[\;\;\adjincludegraphics[scale=0.8,trim={10pt 10pt 10pt 10pt},valign = c]{tikz/out/torus_partition_function_new.pdf}\;\;\right]\;\;
  =\;\;\lim_{L_i\to\infty}\;\;\adjincludegraphics[scale=0.75,trim={10pt 10pt 10pt 10pt},valign = c]{tikz/out/torus_partition_function_MPS_new.pdf}\;\;.
\label{eq: torus_partition_function_MPS}
\end{equation}
In fact, in this limit, each interval produces a projector to $F^{\rm MPS}(x)$, and hence the right-hand side of the above equation can be computed as
\begin{equation}
  \lim_{L_i\to\infty}\;\;\adjincludegraphics[scale=0.75,trim={10pt 10pt 10pt 10pt},valign = c]{tikz/out/torus_partition_function_MPS_new.pdf}\;\;
  =\sum_{ijk}L_{(i,j),(z;k,\mu)}^{xy}\overline{L}_{(z;k,\nu), (j,i)}^{yx}.
\end{equation}
This is nothing but the torus partition function of an SPT phase with fusion category symmetry \cite{Thorngren:2019iar, Huang:2021zvu, Inamura:2021wuo}.

\subsection{Abelianization of the non-abelian factor system}
\label{sec: Abelianization of the nonabelian factor system}

We have seen that the fiber functor or the non-abelian factor system can be used as an invariant that characterizes SPT phases.
However, these are classified by non-abelian cohomology theories, and it may not be easy to compute the invariants and classifications in general. 
In this section, we will abelianize the non-abelian factor system to construct an invariant that takes values in (abelian) cohomology theories. 
Abelianization results in a loss of information, so these invariants do not provide a complete classification. 
Nevertheless, if these invariants take different values, it is clear that the two phases are different.
Additionally, as it only involves ordinary cohomology theories, 
one can invoke various techniques for calculating cohomology groups. 
This makes it relatively easy to identify the structure of the classification.

Before we define the abelianization, let us first define a cohomology theory. 
For a non-anomalous fusion category $\mathcal{C}$, we consider the following cochain $C^k(\mathcal{C;\mathbb{Z}})$:
\begin{eqnarray}
    C^k(\mathcal{C};\mathbb{Z}):=\{n:\Simp(\mathcal{C})^{k}\to\mathbb{Z}\}.
\end{eqnarray}
A cochain $n$ is extended additively to non-simple objects as $n(\cdots, \rho_i \oplus \rho_i^{\prime}, \cdots) = n(\cdots, \rho_i, \cdots) + n(\cdots, \rho_i^{\prime}, \cdots)$.
For this cochain, we introduce a differential
\begin{eqnarray}
    \delta_k:C^k(\mathcal{C};\mathbb{Z})\to C^{k+1}(\mathcal{C};\mathbb{Z})
\end{eqnarray}
defined by
\begin{equation}
\begin{aligned}
  (\delta_k n)(\rho_1,...,\rho_{k+1}) &:=\dim(\rho_1)n(\rho_2,...,\rho_{k+1})+\sum_{i=1}^{k}(-1)^{i}n(\rho_1,...,\rho_i\otimes\rho_{i+1},...,\rho_{k+1}) \\
  & \quad ~ +(-1)^{k+1}n(\rho_1,...,\rho_{k})\dim(\rho_{k+1}),
\end{aligned}
\end{equation}
for $n\in C^k(\mathcal{C};\mathbb{Z})$. We remark that $\dim(\rho)$ is an integer for all $\rho\in \mathcal{C}$ since $\mathcal{C}$ is non-anomalous.
We will drop the subscript $k$ of $\delta_k$ when it is clear from the context.
Then, we can show that
\begin{equation}
  \delta^{2} := \delta_{k+1} \delta_k =0:C^{k}(\mathcal{C};\mathbb{Z})\to C^{k+2}(\mathcal{C};\mathbb{Z}).
\end{equation}
See Appendix~\ref{sec: cohomology} for a proof.
Therefore, we can define the cohomology theory associated to the complex $(C^{k}(\mathcal{C};\mathbb{Z}),\delta)$:
\begin{equation}
  \cohoZ{k}{\mathcal{C}}:={\rm Ker}(\delta_k)/{\rm Im}(\delta_{k-1}).
\end{equation}
Despite the notation, $\cohoZ{k}{\cC}$ depends only on the fusion ring of $\cC$.
We note that $\cohoZ{k}{\mathcal{C}}$ is a special case of the Hochschild cohomology.\footnote{We are grateful to Yosuke Kubota for pointing this out. As a Hochschild cohomology, $\cohoZ{k}{\mathcal{C}}$ should be written as $\mathrm{H}^k(\text{Gr}(\cC), \Z)$, where $\text{Gr}(\cC)$ is the fusion ring of $\cC$ and $\Z$ is equipped with a $\text{Gr}(\cC)$-bimodule structure defined by $x \cdot m = m \cdot x = \dim(x)m$ for all $x \in \cC$ and $m \in \Z$.}

One can define an element of $\cohoZ{3}{\mathcal{C}}$ from the non-abelian factor system $\{\Omega_{x, y}\}$:
we take an $\mathbb{R}$-lift $r_{x,y}\in \mathbb{R}$ of $\det(\Omega_{x,y})$, i.e.,
\begin{equation}
  r_{x,y}:=\frac{1}{2\pi i}\log(\det(\Omega_{x,y})).
\end{equation}
Note that $r_{x,y}$ has an ambiguity because one can shift it by an integer.
Then we define $n_{x,y,z}$ as 
\begin{equation}
  n_{x,y,z}:=r_{y,z}-\sum_{w_{i}\in x\otimes y}r_{w_i,z} +\sum_{w_{j}\in y\otimes z}r_{x,w_{j}}-r_{x,y}. 
\end{equation}
Due to the non-abelian cocycle condition~\eqref{eq: nonabelian cocycle}, $n_{x,y,z}$ takes values in $\mathbb{Z}$.
Therefore, $\{n_{x,y,z}\}$ can be regarded as an element of $C^{3}(\mathcal{C};\mathbb{Z})$.
Moreover, we can check that $(\delta n)_{x,y,z,w}=0$.
On the other hand, under the shift of $r_{x,y}$ 
\begin{equation}
  r_{x,y}\mapsto r_{x,y} + m_{x,y}, \quad m_{x,y}\in \mathbb{Z},
\end{equation}
the cochain $n_{x,y,z}$ changes by the coboundary of $m_{x,y}$:
\begin{equation}
  n_{x,y,z}\mapsto n_{x,y,z} + (\delta m)_{x,y,z}.
\end{equation}
Thus, $\{n_{x,y,z}\}$ can be regarded as an element of the cohomology:
\begin{equation}
  \left[n_{x,y,z}\right]\in \cohoZ{3}{\mathcal{C}}.
\end{equation}
In general, the abelianization $\Omega_{x, y} \mapsto [n_{x, y, z}]$ is neither injective nor surjective.

\subsection{Example: $\Rep(G)$ symmetry}
\label{sec: Example Rep(G) symmetry I}
As a simple example, let us compute the $L$-symbols for the cluster state with $G \times \Rep(G)$ symmetry.
We start by recalling an MPS representation of the $G \times \Rep(G)$ cluster state following \cite{Fechisin:2023dkj}, see also \cite{Brell_2015} for an earlier work on this state.
The physical and bond Hilbert spaces of the $G \times \Rep(G)$ cluster state are both given by the regular representation of $G$.
In what follows, the basis of the physical Hilbert space is denoted by $\{\ket{g}_{\phys} \mid g \in G\}$, while the basis of the bond Hilbert space is denoted by $\{\ket{g} \mid g \in G\}$.
The MPS tensors on odd sites (black dots) and even sites (white dots) are given by\footnote{Our convention for the $G \times \Rep(G)$ cluster state is slightly different from that in \cite{Fechisin:2023dkj}. Specifically, in \cite{Fechisin:2023dkj}, the MPS tensor on the even sites is given by $\sum_{g \in G} L_g \otimes \ket{g}_{\text{phys}}$. Correspondingly, the virtual bonds of the MPO tensor $Z_{\rho}^{\text{phys}}$ in Eq.~\eqref{eq: cluster symmetry} have the opposite orientation in \cite{Fechisin:2023dkj}.}
\begin{equation}
\adjincludegraphics[scale=1,trim={10pt 10pt 10pt 10pt},valign = c]{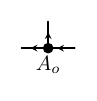}= \sum_{g \in G} P_g \otimes \ket{g}_{\phys}, \quad 
\adjincludegraphics[scale=1,trim={10pt 10pt 10pt 10pt},valign = c]{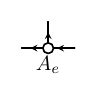}= \sum_{g \in G} L_g^{\dagger} \otimes \ket{g}_{\phys},
\label{eq: GxRep(G) MPS}
\end{equation}
where $P_g$ and $L_g$ are the projection and the left multiplication acting on the bond Hilbert space, i.e.,
\begin{equation}
P_g = \ket{g}\bra{g}, \quad L_g = \sum_{h \in G} \ket{gh} \bra{h}.
\label{eq: Pg Lg}
\end{equation}
We note that the unit cell of the cluster state consists of two sites.
The cluster state given by the above MPS tensors has both $G$ and $\Rep(G)$ symmetries.
The $G$ symmetry is implemented by the right multiplication of $G$ on odd sites, while the $\Rep(G)$ symmetry is implemented by the action of Wilson line operators~\eqref{eq: MPO Rep(G)} on even sites.
More specifically, the actions of $G$ and $\Rep(G)$ symmetries are given by
\begin{equation}
\mathcal{O}_{g}\cdot A_{\rm cluster}
=\adjincludegraphics[scale=1,trim={10pt 10pt 10pt 10pt},valign = c]{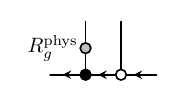}
=\adjincludegraphics[scale=1,trim={10pt 10pt 10pt 10pt},valign = c]{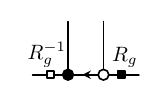}\;, \quad
\mathcal{O}_{\rho}\cdot A_{\rm cluster}
=\adjincludegraphics[scale=1,trim={10pt 10pt 10pt 10pt},valign = c]{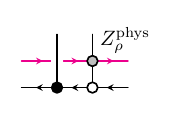}
=\adjincludegraphics[scale=1,trim={10pt 10pt 10pt 10pt},valign = c]{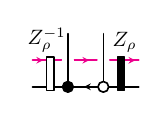}\;,
\label{eq: cluster symmetry}
\end{equation}
where $R_g$ and $Z_{\rho}$ are defined by
\begin{equation}
R_g = \sum_{h \in G} \ket{hg^{-1}} \bra{h}, \quad Z_{\rho} = \sum_{g \in G} \rho(g) \otimes \ket{g} \bra{g}.
\label{eq: Rg Zrho}
\end{equation}
We note that $\rho(g)$ in Eq.~\eqref{eq: Rg Zrho} is acting on the MPO's bond Hilbert space $V_{\rho}$, which is the representation space of $\rho$.
The operators $R_g^{\phys}$ and $Z_{\rho}^{\phys}$ in Eq.~\eqref{eq: cluster symmetry} are defined by the same equations except that the bras and kets now have the subscript ``phys."
Equation \eqref{eq: cluster symmetry} clearly shows that the cluster state defined above has $G \times \Rep(G)$ symmetry.
This cluster state is connected to the trivial product state by a $\Rep(G)$-symmetric finite depth quantum circuit, although it is not by a $G \times \Rep(G)$ symmetric one \cite{Fechisin:2023dkj}.
As such, we say that the $G \times \Rep(G)$ cluster state is in the trivial phase with $\Rep(G)$ symmetry.\footnote{This is just a matter of terminology. In general, defining the notion of a trivial phase with non-invertible symmetry is a subtle problem due to the lack of group structure on the set of SPT phases \cite{Thorngren:2019iar, Fechisin:2023dkj}.}
We note that the cluster state \eqref{eq: GxRep(G) MPS} is the ground state of the $\Rep(G)$ symmetric model discussed in \cite[Section 4.4]{Inamura:2021szw}, see Appendix \ref{sec: GxRep(G) cluster state} for more details.

Now, let us compute the $L$-symbols.
For simplicity, we will focus on the $L$-symbols associated with $\Rep(G)$ symmetry.
Computing the $L$-symbols for the full $G \times \Rep(G)$ symmetry is also straightforward.
The MPO tensor $\cO_{\rho}$ defined in Eq.~\eqref{eq: cluster symmetry} implies that the action tensors for $\Rep(G)$ symmetry are given by
\begin{equation}
(\phi_{\rho})_i = \adjincludegraphics[scale=1,trim={10pt 10pt 10pt 10pt},valign = c]{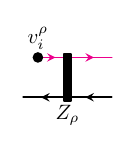}, \qquad 
(\overline{\phi}_{\rho})_i = \adjincludegraphics[scale=1,trim={10pt 10pt 10pt 10pt},valign = c]{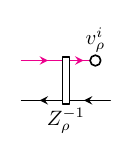},
\label{eq: cluster action tensors}
\end{equation}
where $\{v^{\rho}_i \mid i = 1, 2, \cdots, \dim(\rho)\}$ and $\{v_{\rho}^i = (v_i^{\rho})^{\dagger} \mid i = 1, 2, \cdots, \dim(\rho)\}$ are dual bases of $V_{\rho}$ and its dual $V_{\rho}^*$.
By using these action tensors, we can compute the $L$-symbols~\eqref{eq: L benz} as follows:
\begin{equation}
\begin{aligned}
L^{\rho_1 \rho_2}_{(i, j), (\rho_3; k, \mu)} & = \frac{1}{|G|} \sum_{g \in G} \sum_{i^{\prime}, j^{\prime}, k^{\prime}} \rho_1(g)_{i i^{\prime}} \rho_2(g)_{j j^{\prime}} ((\phi_{\rho_1 \rho_2}^{\rho_3})_{\mu})_{i^{\prime} j^{\prime}}^{k^{\prime}} \rho_3(g)^{\dagger}_{k^{\prime} k}, \\
\overline{L}^{\rho_1 \rho_2}_{(\rho_3; k, \mu), (i, j)} & = \frac{1}{|G|} \sum_{g \in G} \sum_{i^{\prime}, j^{\prime}, k^{\prime}} \rho_3(g)_{k k^{\prime}} ((\overline{\phi}_{\rho_1 \rho_2}^{\rho_3})_{\mu})_{i^{\prime} j^{\prime}}^{k^{\prime}} \rho_1(g)^{\dagger}_{i^{\prime} i} \rho_2(g)^{\dagger}_{j^{\prime} j}. 
\end{aligned}
\label{eq: L for Rep(G)}
\end{equation}
Here, $\rho(g)$ denotes the unitary representation matrix whose $(i, j)$-component $\rho(g)_{ij}$ is defined by
\begin{equation}
\rho(g) v^{\rho}_i = \sum_{j} \rho(g)_{ij} v^{\rho}_j.
\end{equation}
Similarly, $((\phi_{\rho_1 \rho_2}^{\rho_3})_{\mu})_{ij}^k$ in Eq.~\eqref{eq: L for Rep(G)} denotes the matrix element of the fusion tensor $(\phi_{\rho_1 \rho_2}^{\rho_3})_{\mu}: V_{\rho_1} \otimes V_{\rho_2} \rightarrow V_{\rho_3}$, namely,
\begin{equation}
(\phi_{\rho_1 \rho_2}^{\rho_3})_{\mu} (v^{\rho_1}_i \otimes v^{\rho_2}_j) = \sum_{k} ((\phi_{\rho_1 \rho_2}^{\rho_3})_{\mu})_{ij}^k v^{\rho_3}_k.
\end{equation}
The corresponding splitting tensor is given by $(\overline{\phi}_{\rho_1 \rho_2}^{\rho_3})_{\mu} = (\phi_{\rho_1, \rho_2}^{\rho_3})_{\mu}^{\dagger}$.
By using the identities\footnote{The first equality of eq.~\eqref{eq: Schur orthogonality} follows from the fusion rules. The second equality is known as the Schur orthogonality relation for irreducible representations $\rho, \lambda \in \Rep(G)$.}
\begin{equation}
\begin{aligned}
\rho_1(g)_{i i^{\prime}} \rho_2(g)_{j j^{\prime}} & = \sum_{\rho_3, \mu} \sum_{k, k^{\prime}} ((\phi_{\rho_1 \rho_2}^{\rho_3})_{\mu})_{i j}^{k} ((\overline{\phi}_{\rho_1 \rho_2}^{\rho_3})_{\mu})_{i^{\prime} j^{\prime}}^{k^{\prime}} \rho_3(g)_{k k^{\prime}}, \\
\sum_{g \in G} \rho(g)^{\dagger}_{ii^{\prime}} \lambda(g)_{jj^{\prime}} & = \delta_{\rho, \lambda} \delta_{ij^{\prime}} \delta_{i^{\prime} j} |G|/\dim(\rho) \quad \text{for $\rho, \lambda \in \Simp(\Rep(G))$},
\end{aligned}
\label{eq: Schur orthogonality}
\end{equation}
one can rewrite the $L$-symbols in eq.~\eqref{eq: L for Rep(G)} as
\begin{equation}
L^{\rho_1 \rho_2}_{(i, j), (\rho_3; k, \mu)} = ((\phi_{\rho_1 \rho_2}^{\rho_3})_{\mu})_{ij}^k, \quad 
\overline{L}^{\rho_1 \rho_2}_{(\rho_3; k, \mu), (i, j)} = ((\overline{\phi}_{\rho_1 \rho_2}^{\rho_3})_{\mu})_{ij}^k.
\label{eq: L for Rep(G) simplified}
\end{equation}
The above $L$-symbols correspond to the forgetful functor of $\Rep(G)$.

We note that the $L$-symbols in eq.~\eqref{eq: L for Rep(G) simplified} are the same as those for the trivial product state with $\Rep(G)$ symmetry.
This is consistent with the fact the $G \times \Rep(G)$ cluster state is in the trivial $\Rep(G)$-symmetric phase.

\section{Interface modes of SPT phases with fusion category symmetries}
\label{sec: Interface modes of SPT phases with fusion category symmetries}
In this section, we study interfaces of 1+1d SPT phases with fusion category symmetry $\cC$.
We first derive the symmetry algebra, which we call the interface algebra, acting on the interface of two SPT phases with the same symmetry.
The irreducible representations of this algebra can be regarded as localized modes at the interface, i.e. interface modes.\footnote{This is a similar idea to the correspondence between the irreducible representations of the tube algebra and the anyons in a topologically ordered system.}
It turns out that the interface algebra does not have one-dimensional representations if the adjacent SPT phases are different.
This implies that two different SPT phases must have degenerate interface modes that form a higher dimensional representation of the interface algebra.
In other words, the symmetry at the interface of different SPT phases is anomalous.
On the other hand, the interface algebra between the same SPT phase has one-dimensional representations, which describe non-degenerate interface modes.
These one-dimensional representations are in one-to-one correspondence with automorphisms of a fiber functor $F$ and form a group $\Aut^{\otimes}(F)$.
As an illustrative example, we will study the interface algebras of SPT phases with $\Rep(G)$ symmetry in detail.
Our analysis in this section is similar to that in~\cite{Rubio:2024aiw}, which studies domain wall excitations in 1+1d gapped phases with anomalous finite group symmetries.

\subsection{Symmetry algebra at the interface}
\label{sec: Symmetry algebra at the interface}
Let us consider the interface between two SPT phases $\SPT_1$ and $\SPT_2$ with fusion category symmetry $\cC$.
The ground states of these SPT phases are represented by $\cC$-symmetric injective MPSs $A_1$ and $A_2$.
The physical and bond Hilbert spaces of $A_i$ are denoted by $\cH_i$ and $V_i$ respectively.
Here, $\cH_1$ and $V_1$ are not necessarily isomorphic to $\cH_2$ and $V_2$.
The action tensors associated with $A_1$ and $A_2$ are written as $\{(\phi_x^1)_i \mid i = 1, 2, \cdots, \dim(x)\}$ and $\{(\phi_x^2)_i \mid i = 1, 2, \cdots, \dim(x)\}$.
Similarly, the $L$-symbols for $A_1$ and $A_2$ are $\{(L_1)^{xy}_{(i, j), (z; k, \mu)}\}$ and $\{(L_2)^{xy}_{(i, j), (z; k, \mu)}\}$.

To derive the symmetry algebra at the interface, we consider the symmetry action on a periodic MPS in the presence of interfaces.
Specifically, we consider the following MPS on a periodic chain
\begin{equation}
\ket{\psi_{12}, \psi_{21}} = \;\adjincludegraphics[scale=1,trim={10pt 10pt 10pt 10pt},valign = c]{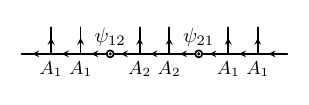}\;,
\label{eq: interface MPS}
\end{equation}
where a half of the chain is occupied by $\mathsf{SPT}_1$ and the other half is occupied by $\mathsf{SPT}_2$.
The interface degrees of freedom in Eq.~\eqref{eq: interface MPS} are represented by arbitrary two-leg tensors $\psi_{12}: V_2 \rightarrow V_1$ and $\psi_{21}: V_1 \rightarrow V_2$.
The action of a symmetry MPO $\widehat{\cO}_x$ on MPS \eqref{eq: interface MPS} is computed as
\begin{equation}
\widehat{\cO}_x \ket{\psi_{12}, \psi_{21}}
= \;\adjincludegraphics[scale=1,trim={10pt 10pt 10pt 10pt},valign = c]{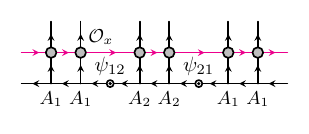}\;
= \sum_{1 \leq i, j \leq \dim(x)} \Ket{(\widehat{\cO}_x^{12})_{ij} \psi_{12}, (\widehat{\cO}_x^{21})_{ji} \psi_{21}},
\label{eq: periodic chain with interface}
\end{equation}
where the operators $(\widehat{\cO}_x^{12})_{ij}$ and $(\widehat{\cO}_x^{21})_{ji}$ are defined by
\begin{equation}
(\widehat{\cO}_x^{12})_{ij} \psi_{12} = \;\adjincludegraphics[scale=1,trim={10pt 10pt 10pt 10pt},valign = c]{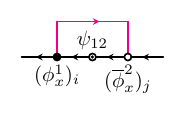}\;, \qquad 
(\widehat{\cO}_x^{21})_{ji} \psi_{21} = \;\adjincludegraphics[scale=1,trim={10pt 10pt 10pt 10pt},valign = c]{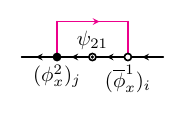}\;.
\end{equation}
These operators can be regarded as symmetry operators acting on the interfaces.
The symmetry algebra $\cA_{12}$ acting on the left interface is spanned by operators $\{(\widehat{\cO}_x^{12})_{ij} \mid x \in \Simp(\cC), ~ i, j = 1, 2, \cdots, \dim(x)\}$.
Similarly, the symmetry algebra $\cA_{21}$ acting on the right interface is spanned by operators $\{(\widehat{\cO}_x^{21})_{ij} \mid x \in \Simp(\cC), ~ i, j = 1, 2, \cdots, \dim(x)\}$.
In what follows, we will focus on $\cA_{12}$ without loss of generality.
The multiplication of symmetry operators in $\cA_{12}$ can be computed explicitly as
\begin{equation}
\begin{aligned}
(\widehat{\cO}_x^{12})_{ii^{\prime}} (\widehat{\cO}_y^{12})_{jj^{\prime}} \psi_{12}
& = \;\adjincludegraphics[scale=1,trim={10pt 10pt 10pt 10pt},valign = c]{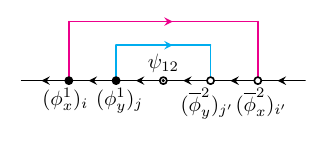}\;
 = \sum_{z, \mu} \;\adjincludegraphics[scale=1,trim={10pt 10pt 10pt 10pt},valign = c]{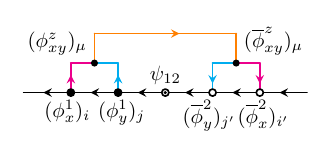}\\
& = \sum_{z, \mu} \sum_{k, k^{\prime}} \adjincludegraphics[scale=1,trim={10pt 10pt 10pt 10pt},valign = c]{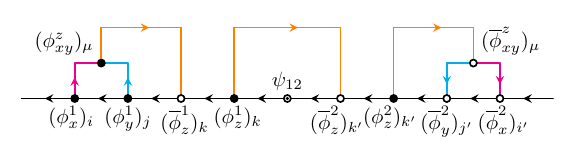}\\
& = \sum_{z, \mu} \sum_{k, k^{\prime}} (L_1)^{xy}_{(i, j), (z; k, \mu)} (\overline{L_2})^{xy}_{(z; k^{\prime}, \mu), (i^{\prime}, j^{\prime})} (\widehat{\cO}_z^{12})_{k k^{\prime}} \psi_{12}.
\end{aligned}
\label{eq: interface algebra derivation}
\end{equation}
The second and third equalities follow from the weak completeness relation~\eqref{eq: weak completeness}.\footnote{The last line of Eq.~\eqref{eq: interface algebra derivation} can also be obtained directly from the left diagram on the first line by using the orthogonality relation~\eqref{eq: fusion orthogonality} and the $L$-move~\eqref{eq: L-symbols} \eqref{eq: L-symbols inverse}.}
The above multiplication law is associative due to the consistency condition on the $L$-symbols.
The unit element of the algebra $\cA_{12}$ is the identity operator $\widehat{\cO}_{\1}^{12}: \psi_{12} \mapsto \psi_{12}$.
Here and in what follows, we omit the indices $i$ and $j$ of an operator $(\widehat{\cO}_x^{12})_{ij}$ when $\dim(x) = 1$ because the indices are unique in this case.

Summarizing, we find that the symmetry algebra acting on the interface of $\cC$-symmetric injective MPSs $\ket{A_1}$ and $\ket{A_2}$ is given by
\begin{equation}
\cA_{12} = \text{Span}\{ (\widehat{\cO}_x^{12})_{ij} \mid x \in \Simp(\cC), ~ i, j = 1, 2, \cdots, \dim(x) \}
\end{equation}
that is equipped with the following associative multiplication:
\begin{equation}
(\widehat{\cO}_x^{12})_{ii^{\prime}} (\widehat{\cO}_y^{12})_{jj^{\prime}} = \sum_{z, \mu} \sum_{k, k^{\prime}} (L_1)^{xy}_{(i, j), (z; k, \mu)} (\overline{L_2})^{xy}_{(z; k^{\prime}, \mu), (i^{\prime}, j^{\prime})} (\widehat{\cO}_z^{12})_{k k^{\prime}}.
\label{eq: interface algebra}
\end{equation}
We emphasize that the symmetry algebra $\cA_{12}$ depends only on the $L$-symbols (or fiber functors), which characterize the SPT phases represented by $\cC$-symmetric injective MPSs.
In particular, $\cA_{12}$ does not depend on microscopic details of MPS tensors.

We note that each symmetry operator $\widehat{\cO}_x$ acting on the bulk SPT phase is factorized into multiple operators $\{ (\widehat{\cO}_x^{12})_{i j} \mid i, j = 1, 2, \cdots, \dim(x)\}$ at the interface.
Accordingly, the dimension of the symmetry algebra $\cA_{12}$ is not equal to the number of simple objects of $\cC$ but coincide with the total dimension of $\cC$, i.e., 
\begin{equation}
\dim(\cA_{12}) = \sum_{x \in \Simp(\cC)} \dim(x)^2 = \dim(\cC).
\end{equation}

The interface algebra $\cA_{12}$ is a special case of the algebra $\cA_{\cM, \cN}$ studied in \cite{KK2011}, where $\cM$ and $\cN$ are general $\cC$-module categories.
Specifically, $\cA_{\cM, \cN}$ reduces to $\cA_{12}$ when both $\cM$ and $\cN$ are module categories associated with fiber functors.
When $\cM$ and $\cN$ are the same, the algebra $\cA_{\cM, \cN}$ has the structure of a weak Hopf algebra \cite{KK2011}, whose representation category is equivalent to the dual fusion category $\Fun_{\cC}(\cM, \cM)$, the category of $\cC$-module endofunctors of $\cM$ \cite{KK2011, Barter_2022}.
See also \cite{Choi:2023xjw, Cordova:2024vsq, Cordova:2024iti,Copetti:2024onh,Copetti:2024dcz, Zheng2024Oxford} for recent applications of this algebra.

\vspace{10pt}
\noindent{\bf Example.}
As the simplest example, let us consider the interface algebra of SPT phases with ordinary finite group symmetry $G$.
When the symmetry $\cC = \Vect_G$ is a finite group, the $L$-symbols are nothing but group 2-cocycles $\omega \in \mathrm{Z}_{\text{gp}}^2(G, \mathrm{U}(1))$.
More specifically, the relation between $L$-symbols and 2-cocycles is given by
\begin{equation}
L^{g,h} = \omega(g, h),
\end{equation}
where the indices of the $L$-symbols are omitted because they are unique.
The symmetry algebra $\cA_{12}$ at the interface of two SPT phases $\mathsf{SPT}_{\omega_1}$ and $\mathsf{SPT}_{\omega_2}$ is spanned by operators $\{\widehat{\cO}_g^{12} \mid g \in G\}$, which obey the following multiplication law:
\begin{equation}
\widehat{\cO}_g^{12} \widehat{\cO}_h^{12} = \omega_1(g, h) \omega_2(g, h)^{-1} \widehat{\cO}_{gh}^{12}.
\end{equation}
Namely, $\cA_{12}$ is isomorphic to the twisted group algebra $\C[G]^{\omega_1 \omega_2^{-1}}$ with the twist $\omega_1 \omega_2^{-1} \in \mathrm{Z}_{\text{gp}}^2(G, \mathrm{U(1)})$.
In particular, if $\omega_1$ and $\omega_2$ belong to different cohomology classes $[\omega_1] \neq [\omega_2] \in \mathrm{H}_{\text{gp}}^2(G, \mathrm{U}(1))$, the symmetry $G$ is realized projectively at the interface.
This phenomenon is known as an anomaly inflow \cite{Callan:1984sa}.

As a more interesting example, we will study the interface algebras for the non-invertible $\Rep(D_8)$ symmetry in Sec.~\ref{sec: Example: Rep(D8) symmetry II}.
The interface algebras for the other non-anomalous $\Z_2 \times \Z_2$ Tambara-Yamagami symmetries $\Rep(Q_8)$ and $\Rep(H_8)$ will also be discussed in Appendix \ref{sec: Interface algebra for Tambara-Yamagami categories}.

\subsection{Interface modes}
\label{sec: Interface modes}
In this subsection, we show that the interface algebra $\cA_{12}$ has one-dimensional representations if and only if $\SPT_1$ and $\SPT_2$ are the same SPT phase.
Consequently, different SPT phases must have degenerate interface modes between them.
In other words, the interface algebra between two different SPT phases is anomalous.
This can be regarded as an anomaly inflow for fusion category symmetry.

We first show that $\cA_{12}$ has a one-dimensional representation if $\SPT_1$ and $\SPT_2$ are the same phase.
When $\SPT_1$ and $\SPT_2$ are the same, the associated $L$-symbols $L_1$ and $L_2$ are equivalent to each other.
Therefore, there exists a set of invertible matrices $\{U_x \mid x \in \Simp(\cC)\}$ such that
\begin{equation}
(L_1)^{xy}_{(i, j), (z; k, \mu)} = \sum_{i^{\prime}, j^{\prime}, k^{\prime}} (U_x)_{i i^{\prime}} (U_y)_{j j^{\prime}} (U_z)^{-1}_{k^{\prime} k} (L_2)^{xy}_{(i^{\prime}, j^{\prime}), (z; k^{\prime}, \mu)}.
\end{equation}
Since $U_x$, $U_y$, and $U_z$ are invertible, one can rewrite the above equation as
\begin{equation}
(U_x)_{i i^{\prime}} (U_y)_{j j^{\prime}} = \sum_{z, \mu} \sum_{k, k^{\prime}} (L_1)^{xy}_{(i, j), (z; k, \mu)} (\overline{L_2})^{xy}_{(z; k^{\prime}, \mu), (i^{\prime}, j^{\prime})} (U_z)_{k k^{\prime}}.
\label{eq: one-dim rep}
\end{equation}
Comparing Eq.~\eqref{eq: one-dim rep} with Eq.~\eqref{eq: interface algebra}, we find that $(\widehat{\cO}_x^{12})_{ij} = (U_x)_{ij}$ defines a one-dimensional representation of $\cA_{12}$.
Here, a one-dimensional representation of an algebra $\cA$ is a representation such that all elements of $\cA$ are represented by complex numbers.

Conversely, if the interface algebra $\cA_{12}$ has a one-dimensional representation, one can show that $\SPT_1$ and $\SPT_2$ are the same SPT phase.
To see this, we first show that if $\cA_{12}$ has a (non-zero) one-dimensional representation, so does $\cA_{21}$.
To this end, we recall that a fusion category $\cC$ is unitary as we assume throughout the paper.
Since $\cC$ is unitary, there exists a basis such that the $L$-symbols $L_1$ and $L_2$ are unitary matrices.\footnote{We recall that the $L$-symbols are matrix elements of the structure isomorphism $J$ of a fiber functor, which is unitary.}
In this basis, given a one-dimensional representation $(U_x^{12})_{ij}$ of $\cA_{12}$, we can obtain a one-dimensional representation $(U_x^{21})_{ij}$ of $\cA_{21}$ by taking the Hermitian conjugate of $U_x^{12}$, i.e., 
\begin{equation}
(U_x^{21})_{ij} := (U_x^{12})^{\dagger}_{ij} = (U_x^{12})^*_{ji}.
\end{equation}
Indeed, it immediately follows from the unitarity of the $L$-symbols that $(U_x)^{21}_{ij}$ defined above is a one-dimesional representation of $\mathcal{A}_{21}$, i.e., it satisfies
\begin{equation}
(U_x^{21})_{i i^{\prime}} (U_y^{21})_{j j^{\prime}} = \sum_{z, \mu} \sum_{k, k^{\prime}} (L_2)^{xy}_{(i, j), (z; k, \mu)} (\overline{L_1})^{xy}_{(z; k^{\prime}, \mu), (i^{\prime}, j^{\prime})} (U_z^{21})_{k k^{\prime}},
\end{equation}
as long as $(U_x)^{12}_{ij}$ is a one-dimensional representation of $\mathcal{A}_{12}$:
\begin{equation}
(U_x^{12})_{i i^{\prime}} (U_y^{12})_{j j^{\prime}} = \sum_{z, \mu} \sum_{k, k^{\prime}} (L_1)^{xy}_{(i, j), (z; k, \mu)} (\overline{L_2})^{xy}_{(z; k^{\prime}, \mu), (i^{\prime}, j^{\prime})} (U_z^{12})_{k k^{\prime}}.
\label{eq: one-dim rep 12 21}
\end{equation}
Now, we notice that the product $(U_x^{12} U_x^{21})_{ij} := \sum_k (U_x^{12})_{ik} (U_x^{21})_{kj}$ defines a (non-zero) one-dimensional representation of the self-interface algebra $\cA_{11}$ of $\SPT_1$.
Similarly, the product $(U_x^{21} U_x^{12})_{ij} := \sum_k (U_x^{21})_{ik} (U_x^{12})_{kj}$ defines a (non-zero) one-dimensional representation of the self-interface algebra $\cA_{22}$ of $\SPT_2$.
The $n$th power $(U_x^{12} U_x^{21})^n$ is also a (non-zero) one-dimensional representation of $\cA_{11}$ for any positive integer $n$.
Since $\cA_{11}$ has only finitely many one-dimensional representations, there exists a positive integer $n$ such that the one-dimensional representation given by $(\widehat{\cO}_x^{11})_{ij} = (U_x^{12} U_x^{21})^n_{ij}$ is isomorphic to the obvious one-dimensional representation $(\widehat{\cO}_x^{11})_{ij} = \delta_{ij}$.
In particular, the determinant of $U_x^{12} U_x^{21}$ is non-zero, which implies that $U_x^{12}$ and $U_x^{21}$ are invertible.
Thus, Eq.~\eqref{eq: one-dim rep 12 21} gives an equivalence between the $L$-symbols $L_1$ and $L_2$, meaning that $\SPT_1$ and $\SPT_2$ are the same SPT phase.

In summary, we find that $\SPT_1$ and $\SPT_2$ are the same SPT phase if and only if the interface algebra $\cA_{12}$ between them has a one-dimensional representation.
In particular, the interface algebra between different SPT phases does not have one-dimensional representations.
This implies that interface degrees of freedom between different SPT phases always form a higher dimensional representation of $\cA_{12}$, which leads to degenerate ground states on a periodic chain in the presence of interfaces.
Put differently, the symmetry algebra at the interface of different SPT phases is anomalous.\footnote{Symmetry is called anomalous if it does not admit a gapped phase with a unique ground state. This definition of an anomaly was proposed in 1+1d in \cite{Thorngren:2019iar} and can be employed in any spacetime dimension, including 0+1d as in the current situation.}
This result generalizes the anomaly inflow on the boundaries of SPT phases with ordinary group symmetries.

From the point of view of symmetry topological field theory, representations of the interface algebra $\cA_{12}$ (or more generally, $\cA_{\cM, \cN}$) can be understood as boundary excitations of 2+1d topological orders \cite{KK2011}.
We will discuss this perspective in more detail in Sec.~\ref{sec: Parameterized family and Thouless pump}.

\vspace{10pt}
\noindent{\bf Non-degenerate interface modes.}
One-dimensional representations of the interface algebra $\cA_{12}$ describe non-degenerate interface modes between the same SPT phase $\SPT_1 = \SPT_2$.
It is clear from the above discussion that one-dimensional representations of the interface algebra $\cA_{12}$ are in one-to-one correspondence with gauge transformations between $L$-symbols $L_1$ and $L_2$.
As we discussed at the end of Sec.~\ref{sec: Matrix product states with fusion category symmetries}, gauge transformations between $L_1$ and $L_2$ are also in one-to-one correspondence with isomorphisms between the corresponding fiber functors $F_1$ and $F_2$.
Therefore, one-dimensional representations of $\cA_{12}$ are in one-to-one correspondence with isomorphisms of fiber functors $F_1$ and $F_2$.
Specifically, given a one-dimensional representation $\{(\widehat{\cO}_x^{12})_{ij} = (U_x)_{ij}\}$ of $\cA_{12}$, the matrix elements of the corresponding isomorphism $\eta: F_1 \Rightarrow F_2$ are given by 
\begin{equation}
(\eta_x)_{ij} = (U_x)_{ij}.
\label{eq: one-dim rep = isom}
\end{equation}

When the $L$-symbols $L_1$ and $L_2$ are equal to each other, one-dimensional representations of $\cA_{12} =: \cA$ can be canonically identified with automorphisms of the corresponding fiber functor $F$.\footnote{General representations of $\cA$ including higher dimensional ones can be identified with monoidal natural transformations of $F$~\cite{KK2011, Barter_2022}. These describe local excitations within a single SPT phase.}
These one-dimensional representations form a group under the composition of the interface modes as follows:
\begin{equation}
\begin{aligned}
& (\widehat{\cO}_x)_{ij} \psi = (U_x)_{ij} \psi, \quad (\widehat{\cO}_x)_{ij} \psi^{\prime} = (U_x^{\prime})_{ij} \psi^{\prime}\\
& \Rightarrow (\widehat{\cO}_x)_{ij} (\psi \circ \psi^{\prime}) = \sum_{k} (U_x)_{ik} (U_x^{\prime})_{kj} \psi \circ \psi^{\prime} = (U_x U_x^{\prime})_{ij} \psi \circ \psi^{\prime}.
\end{aligned}
\end{equation}
The unit element of this group is given by the obvious one-dimensional representation $\{(\widehat{\cO}_x)_{ij} = \delta_{ij}\}$.
Equation \eqref{eq: one-dim rep = isom} tells us that the group of one-dimensional representations of $\cA$ is isomorphic to $\Aut^{\otimes}(F)$.

\subsection{Examples}
\label{sec: Examples}

\subsubsection{Non-degenerate interfaces of $\Rep(G)$ SPT phases}
\label{sec: Non-degenerate interfaces of a Rep(G) SPT phase}
Let us consider representations of the interface algebra $\cA$ between two copies of the $G \times \Rep(G)$ cluster state.
As in Sec.~\ref{sec: Example Rep(G) symmetry I}, we will focus on the $\Rep(G)$ symmetry of the cluster state rather than the full $G \times \Rep(G)$ symmetry.

The interface algebra $\cA$ is spanned by the operators $\{(\widehat{\cO}_\rho)_{ij}\}$ that are labeled by irreducible representations $\rho \in \Rep(G)$ and indices $i, j = 1, 2, \cdots, \dim(\rho)$.
We note that the dimension of $\cA$ is equal to the order of $G$.
Since the $L$-symbols for the $\Rep(G)$ symmetry of the cluster state are given by Eq.~\eqref{eq: L for Rep(G) simplified}, the multiplication of operators $\{(\widehat{\cO}_\rho)_{ij}\}$ can be written explicitly as
\begin{equation}
(\widehat{\cO}_{\rho_1})_{i i^{\prime}} (\widehat{\cO}_{\rho_2})_{j j^{\prime}} = \sum_{\rho_3, \mu} \sum_{k, k^{\prime}} ((\phi_{\rho_1 \rho_2}^{\rho_3})_{\mu})_{ij}^k ((\overline{\phi}_{\rho_1 \rho_2}^{\rho_3})_{\mu})_{i^{\prime} j^{\prime}}^{k^{\prime}} (\widehat{O}_{\rho_3})_{k k^{\prime}}.
\label{eq: Rep(G) interface}
\end{equation}
From the above equation, we find that $(\widehat{\cO}_{\rho})_{ij} = \rho(g)_{ij}$ gives a one-dimensional representation of $\cA$ for any $g \in G$, cf. the first equality of eq.~\eqref{eq: Schur orthogonality}.
These $|G|$ one-dimensional representations exhaust all the irreducible representations of $\cA$ because a $|G|$-dimensional semisimple algebra can have at most $|G|$ irreducible representations.
The product of one-dimensional representations $\{(\widehat{\cO}_{\rho})_{ij} = \rho(g)_{ij}\}$ and $\{(\widehat{\cO}_{\rho})_{ij} = \rho(h)_{ij}\}$ is given by
\begin{equation}
(\widehat{\cO}_{\rho})_{ij} = \sum_{k} \rho(g)_{ik} \rho(h)_{kj} = \rho(gh)_{ij}.
\end{equation}
Thus, we find that the group of one-dimensional representations of $\cA$ is isomorphic to $G$.
This result can also be derived from the fact that the $G \times \Rep(G)$ cluster state corresponds to the forgetful functor of $\Rep(G)$ as an SPT phase with $\Rep(G)$ symmetry.
Specifically, the discussions at the end of Sec.~\ref{sec: Interface modes} imply that the group of one-dimensional representations of $\cA$ is isomorphic to the group of automorphisms of the forgetful functor.
This group is isomorphic to $G$ due to Tannaka-Krein duality.

\subsubsection{Degenerate interfaces of $\Rep(D_8)$ SPT phases}
\label{sec: Example: Rep(D8) symmetry II}
As another example, we consider interfaces of different SPT phases with $\Rep(D_8)$ symmetry.
Here, $\Rep(D_8)$ is a fusion category consisting of five simple objects $\{(a, b), D \mid a, b = 0, 1\}$, which obey the following non-invertible fusion rules:
\begin{equation}
(a, b) \otimes (c, d) = (a+c, b+d), \quad (a, b) \otimes D = D \otimes (a, b) = D, \quad D \otimes D = \bigoplus_{a, b =0, 1} (a, b),
\end{equation}
where $a+c$ and $b+d$ are defined modulo 2.
The $\Rep(D_8)$ symmetry admits three SPT phases \cite{Thorngren:2019iar}, which we denote by $\SPT_{\nu_1}$, $\SPT_{\nu_2}$, and $\SPT_{\nu_3}$.
As we derive in Appendix~\ref{sec: Interface algebra for Tambara-Yamagami categories}, the interface algebra between any two of them has two two-dimensional irreducible representations.
This means that there are two different interface modes, each of which is two-fold degenerate.
The symmetry transformations of these interface modes are determined by the irreducible representations of the interface algebra.
See Appendix~\ref{sec: Interface algebra for Tambara-Yamagami categories} for a detailed analysis of the interface algebras and their representations.

In this subsection, using the results in Appendix~\ref{sec: Interface algebra for Tambara-Yamagami categories}, we discuss the action of the $\Rep(D_8)$ symmetry on a periodic chain consisting of two regions occupied by different SPT phases.
For concreteness, we suppose that half of the chain is occupied by $\SPT_{\nu_1}$ and the other half is occupied by $\SPT_{\nu_2}$.
The other cases can also be studied similarly.
We denote the left and right interfaces by $\cI_{\nu_1, \nu_2}$ and $\cI_{\nu_2, \nu_1}$ respectively.
In the presence of these interfaces, the symmetry action on a periodic chain is given by Eq.~\eqref{eq: periodic chain with interface}, that is,
\begin{equation}
\widehat{\cO}_x \ket{\psi_L, \psi_R} = \sum_{1 \leq i, j \leq \dim(x)} \ket{(\widehat{\cO}^{\cI_{\nu_1, \nu_2}}_x)_{ij} \psi_L, (\widehat{\cO}^{\cI_{\nu_2, \nu_1}}_x)_{ji} \psi_R},
\label{eq: Rep(D8) action on periodic chain}
\end{equation}
where $\psi_L$ and $\psi_R$ are the left and right interface modes.
We suppose that $\psi_L$ and $\psi_R$ are in irreducible representations $\rho_L$ and $\rho_R$ of the interface algebras $\cA_{\cI_{\nu_1, \nu_2}}$ and $\cA_{\cI_{\nu_2, \nu_1}}$.
Specifically, $\rho_L$ and $\rho_R$ are either of the two irreducible representations $R_e$ and $R_d$ listed in Table~\ref{tab: reps nu1nu2}.
Since $\rho_L$ and $\rho_R$ are two-dimensional, the interface modes $\psi_L$ and $\psi_R$ take values in $\C^2$.
We denote the Pauli operators acting on them by $X$, $Y$, and $Z$.
Based on the data in Table~\ref{tab: reps nu1nu2}, one can compute the symmetry action~\eqref{eq: Rep(D8) action on periodic chain} for all choices of $(\rho_L, \rho_R)$.

\begin{itemize}
\item When $(\rho_L, \rho_R) = (R_e, R_e) \text{ or } (R_d, R_d)$, the symmetry action~\eqref{eq: Rep(D8) action on periodic chain} reduces to
\begin{equation}
\begin{aligned}
\widehat{\cO}_{(0, 0)} \ket{\psi_L, \psi_R} & = \widehat{\cO}_{(1, 1)} \ket{\psi_L, \psi_R} = \ket{\psi_L, \psi_R}, \\
\widehat{\cO}_{(0, 1)} \ket{\psi_L, \psi_R} & = \widehat{\cO}_{(1, 0)} \ket{\psi_L, \psi_R} = Z^L Z^R \ket{\psi_L, \psi_R}, \\
\widehat{\cO}_D \ket{\psi_L, \psi_R} & = \frac{1}{2}(S_+^L S_+^R + S_-^L S_-^R) \ket{\psi_L, \psi_R},
\end{aligned}
\end{equation}
where the operators with superscript $L/R$ act on the left/right interface and $S_{\pm} := X \pm iY$.
Remarkably, the local factor $Z^{L/R}$ of the invertible operators $\cO_{(0, 1)}$ and $\cO_{(1, 0)}$ anti-commutes with the local factors $S_{\pm}^{L/R}$ of the non-invertible operator $\cO_D$.
This anti-commutation relation of local factors was originally observed in \cite{Seifnashri:2024dsd} and can be regarded as a manifestation of an anomaly at the interface.
We note that the global symmetry operators still obey the correct fusion rules
\begin{equation}
\widehat{\cO}_{(a, b)} \widehat{\cO}_{(c, d)} = \widehat{\cO}_{(a+c, b+d)}, \quad \widehat{\cO}_{(a, b)} \widehat{\cO}_D = \widehat{\cO}_D \widehat{\cO}_{(a, b)} = \widehat{\cO}_D, \quad \widehat{\cO}_D \widehat{\cO}_D = \sum_{a, b = 0, 1} \widehat{\cO}_{(a, b)}.
\label{eq: global Rep(D8) fusion rules}
\end{equation}
In particular, $\widehat{\cO}_{(a, b)}$ and $\widehat{\cO}_D$ commute with each other.
\item When $(\rho_L, \rho_R) = (R_e, R_d) \text{ or } (R_d, R_e)$, the symmetry action~\eqref{eq: Rep(D8) action on periodic chain} reduces to
\begin{equation}
\begin{aligned}
\widehat{\cO}_{(0, 0)} \ket{\psi_L, \psi_R} & = - \widehat{\cO}_{(1, 1)} \ket{\psi_L, \psi_R} = \ket{\psi_L, \psi_R}, \\
\widehat{\cO}_{(0, 1)} \ket{\psi_L, \psi_R} & = - \widehat{\cO}_{(1, 0)} \ket{\psi_L, \psi_R} = Z^L Z^R \ket{\psi_L, \psi_R}, \\
\widehat{\cO}_D \ket{\psi_L, \psi_R} & = 0.
\end{aligned}
\end{equation}
This result is consistent with the fusion rules \eqref{eq: global Rep(D8) fusion rules}.
\end{itemize}

\section{Parameterized family and non-abelian Thouless pump}
\label{sec: Parameterized family and Thouless pump}
In this section, we study families of SPT states parameterized by a circle $S^1$.
We first define an invariant of an $S^1$-parameterized family of $\cC$-symmetric injective MPSs, where $\cC$ is a fusion category.
The invariant turns out to be an automorphism of a fiber functor $F: \cC \rightarrow \Vect$.
Physically, this invariant can be understood as a generalized Thouless pump.
Indeed, as we will see later, an adiabatic evolution of a $\cC$-symmetric injective MPS along a non-trivial family gives rise to a pump of a non-degenerate interface mode.\footnote{We recall that there is a one-to-one correspondence between automorphisms of a fiber functor and non-degenerate interface modes, see Sec.~\ref{sec: Interface modes}.}
As an illustrative example, we provide concrete lattice models of $S^1$-parameterized families of SPT states with $\Rep(G)$ symmetry.
In what follows, a family of SPT states parameterized by $X$ will be sometimes abbreviated as an $X$-family.

\subsection{Invariants of one-parameter families of SPT phases}
\label{sec: Invariants of one-parameter families of SPT phases}
Let $\{A(\theta) \mid \theta \in [0, 2\pi]\}$ be a family of $\cC$-symmetric injective MPS tensors parameterized by an interval $[0, 2\pi]$.
The MPS tensor $A(\theta)$ is supposed to be continuous in $\theta$.
In particular, the physical Hilbert space $\cH$ and the bond Hilbert space $V=\C^D$ do not depend on $\theta$.
We also suppose that the state $\ket{A(\theta)}$ associated with the MPS tensor $A(\theta)$ is $2\pi$-periodic, i.e.,
\begin{equation}
\ket{A(0)} = \ket{A(2\pi)},
\label{eq: circle condition}
\end{equation}
which means that the parameter space is a circle $S^1$.
We emphasize that $A(\theta)$ is not necessarily $2\pi$-periodic due to the gauge redundancy that we will discuss shortly.
Since $A(\theta)$ is $\cC$-symmetric, there exists a set of action tensors $\{(\phi(\theta)_x)_i \mid x \in \Simp(\cC), ~ i = 1, 2, \cdots, \dim(x)\}$ that satisfies Eq.~\eqref{eq: action tensor} for each $\theta \in [0, 2\pi]$, that is, we have
\begin{equation}
\adjincludegraphics[scale=1,trim={10pt 10pt 10pt 10pt},valign = c]{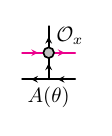}\;
= \sum_{i} \;\adjincludegraphics[scale=1,trim={10pt 10pt 10pt 10pt},valign = c]{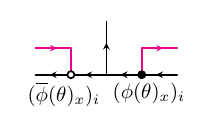}\;.
\label{eq: action tensor theta}
\end{equation}
The $L$-symbols computed from the action tensors $\{(\phi(\theta)_x)_i \}$ are denoted by $\{L(\theta)^{xy}_{(i, j), (z; k, \mu)}\}$.
Since $A(\theta)$ is continuous, one can take $\phi(\theta)$ and $L(\theta)$ to be continuous in $\theta \in [0, 2\pi]$ without loss of generality.

The MPS tensor $A(\theta)$ has a gauge ambiguity because the state $\ket{A(\theta)}$ is invariant under the following gauge transformation:
\begin{equation}
A(\theta) \rightarrow {}_{\psi(\theta)} A(\theta) = \adjincludegraphics[scale=1,trim={10pt 10pt 10pt 10pt},valign = c]{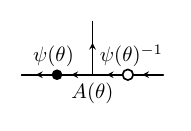},
\label{eq: psi A}
\end{equation}
where $\psi(\theta)$ is an arbitrary automorphism of $V$, i.e., a $D \times D$ invertible matrix.
This gauge transformation induces the following transformation of $\{(\phi(\theta)_x)_i\}$ so that the symmetricity condition~\eqref{eq: action tensor theta} holds for any $\theta \in [0, 2\pi]$:
\begin{equation}
(\phi(\theta)_x)_i \rightarrow {}_{\psi(\theta)}(\phi(\theta)_x)_i = \;\adjincludegraphics[scale=1,trim={10pt 10pt 10pt 10pt},valign = c]{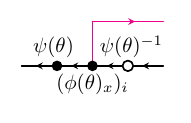}\;, \qquad 
(\overline{\phi}(\theta)_x)_i \rightarrow {}_{\psi(\theta)}(\overline{\phi}(\theta)_x)_i = \;\adjincludegraphics[scale=1,trim={10pt 10pt 10pt 10pt},valign = c]{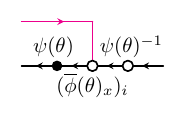}
\label{eq: psi phi}
\end{equation}
The action tensors themselves also have a gauge ambiguity because Eq.~\eqref{eq: action tensor theta} is invariant under the change of a basis
\begin{equation}
(\phi(\theta)_x)_i \rightarrow {}^{U(\theta)}(\phi(\theta)_x)_i = \sum_j (U(\theta)_x)_{ij} (\phi(\theta)_x)_j,
\label{eq: U phi}
\end{equation}
where $U(\theta)_x$ is a $\dim(x) \times \dim(x)$ invertible matrix.
Accordingly, the $L$-symbols $L(\theta)^{xy}_{(i, j), (z; k, \mu)}$ are also transformed as in Eq.~\eqref{eq: L gauge}, which we reproduce here:
\begin{equation}
L(\theta)^{xy}_{(i, j), (z; k, \mu)} \rightarrow {}^{U(\theta)}L(\theta)^{xy}_{(i, j), (z; k, \mu)} = \sum_{i^{\prime}, j^{\prime}, k^{\prime}} (U(\theta)_x)_{i i^{\prime}} (U(\theta)_y)_{j j^{\prime}} (U(\theta)_z)^{-1}_{k^{\prime} k} L(\theta)^{xy}_{(i^{\prime}, j^{\prime}), (z; k^{\prime}, \mu)}.
\end{equation}
We note that the transformation \eqref{eq: psi phi} does not affect the $L$-symbols because $\psi(\theta)$ and $\psi(\theta)^{-1}$ are canceled out in the definition \eqref{eq: L benz} of the $L$-symbols.

Now, let us define an invariant of an $S^1$-parameterized family of SPT states with symmetry $\cC$.
We provide three different ways to define the invariant for a given $S^1$-parameterized family $\{A(\theta), \phi(\theta), L(\theta)\}$.

\begin{itemize}
\item {\bf Non-periodicity of $\phi(\theta)$.}
We first define the invariant of a one-parameter family using the non-periodicity of the action tensors $\phi(\theta)$ in a specific gauge.
To this end, we first perform a gauge transformation $\psi(\theta) \in \Aut(V)$ so that the MPS tensor $A(\theta)$ becomes $2\pi$-periodic.
Such a gauge transformation always exists because $A(0)$ and $A(2\pi)$ are gauge equivalent due to Eq.~\eqref{eq: circle condition}.
Furthermore, we choose the gauge transformation $U(\theta)$ so that the $L$-symbols $L(\theta)$ become constant in $\theta$.
This choice of a gauge is always possible because the equivalence class of $L(\theta)$ does not depend on $\theta$.\footnote{In other words, we are in the same SPT phase throughout the parameter space $S^1$.}
In the above gauge, the action tensors $\phi(\theta)$ are not necessarily $2\pi$-periodic.
Nevertheless, since we have $A(0) = A(2\pi)$, the non-periodicity of $\phi(\theta)$ is at most a gauge transformation
\begin{equation}
(\phi(2\pi)_x)_i = \sum_j (\eta_x)_{ij} (\phi(0)_x)_j,
\label{eq: pump inv 1}
\end{equation}
where $\eta_x$ is a $\dim(x) \times \dim(x)$ invertible matrix.
Moreover, since the $L$-symbols are constant in $\theta$, the gauge transformation $\eta$ preserves the $L$-symbols, meaning that $\eta$ defines an automorphism of the corresponding fiber functor $F$, cf. equations \eqref{eq: eta} and \eqref{eq: natural isom elements}.
This automorphism $\eta$ can be regarded as an invariant of an $S^1$-parameterized family of injective MPSs.

\item {\bf Non-periodicity of $A(\theta)$.}
We can also define an invariant of an $S^1$-parameterized family by using the non-periodicity of the MPS tensor $A(\theta)$ in another gauge.
To see this, we first choose a gauge so that $A(\theta)$ is $2\pi$-periodic and $L(\theta)$ is constant as above.
We then perform a gauge transformation $\psi(\theta) \in \Aut(V)$ so that the action tensors $\phi(\theta)$ become $2\pi$-periodic.
A straightforward computation shows that the $2\pi$-periodicity of the action tensor after the gauge transformation is equivalent to\footnote{The action tensors in Eq.~\eqref{eq: periodicity of phi} are those before the gauge transformation.}
\begin{equation}
\adjincludegraphics[valign=c, scale=1, trim={10pt 10pt 10pt 10pt}]{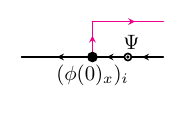}\;
= \sum_{j} (\eta_x)_{ij} \;\adjincludegraphics[valign=c, scale=1, trim={10pt 10pt 10pt 10pt}]{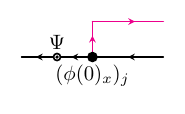}\;,
\label{eq: periodicity of phi}
\end{equation}
where we defined $\Psi := \psi(0)^{-1} \psi(2\pi)$, and $\eta_x$ is a $\dim(x) \times \dim(x)$ invertible matrix defined in Eq.~\eqref{eq: pump inv 1}.
The above equation, together with the orthogonality relation~\eqref{eq: action orthogonality}, implies that $\Psi$ has to satisfy
\begin{equation}
\adjincludegraphics[scale=1,trim={10pt 10pt 10pt 10pt},valign = c]{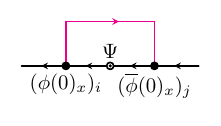}\; = (\eta_x)_{ij} \Psi.
\label{eq: pump inv 2}
\end{equation}
That is, $\Psi$ is a one-dimensional representation of the self-interface algebra.
Conversely, if $\Psi$ satisfies Eq.~\eqref{eq: pump inv 2}, it also satisfies Eq.~\eqref{eq: periodicity of phi} due to the weak completeness relation~\eqref{eq: weak completeness}.
Therefore, the action tensors become $2\pi$-periodic after the gauge transformation $\psi(\theta)$ if and only if $\Psi$ satisfies Eq.~\eqref{eq: pump inv 2}.
We denote the solution to this equation by $\Psi = \psi_{\eta}$, which always exists and is unique up to scalar multiplication due to the one-to-one correspondence between one-dimensional representations of the self-interface algebra and automorphisms of a fiber functor.
After the above gauge transformation, the MPS tensor $A(\theta)$ is no longer $2\pi$-periodic in general, while the $L$-symbols remain constant in $\theta$.
The non-periodicity of $A(\theta)$ after the gauge transformation is given by $\psi_{\eta}^{\prime} = \psi(0) \psi_{\eta} \psi(0)^{-1}$:
\begin{equation}
A(2\pi) = \adjincludegraphics[scale=1,trim={10pt 10pt 10pt 10pt},valign = c]{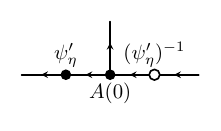}\;.
\label{eq: non-periodic A}
\end{equation}
We note that $\psi_{\eta}$ and $\psi_{\eta}^{\prime}$ are equivalent one-dimensional representations, both of which are associated with the same automorphism $\eta$.
Any $2\pi$-periodic gauge transformation cannot change this equivalence class because such a gauge transformation $\psi^{\prime}$ acts on $\psi_{\eta}^{\prime}$ by conjugation: $\psi_{\eta}^{\prime} \mapsto \psi^{\prime}(0) \psi_{\eta}^{\prime} \psi^{\prime}(0)^{-1}$.
Thus, the equivalence class of the one-dimensional representation $\psi_{\eta}^{\prime}$ can be regarded as an invariant of an $S^1$-parameterized family of injective MPSs.
The invariant defined here is related to the one defined by the non-periodicity of $\phi(\theta)$ (cf. Eq.~\eqref{eq: pump inv 1}) via the one-to-one correspondence between one-dimensional representations of the interface algebra and automorphisms of a fiber functor.

\item {\bf Automorphism of $L(\theta)$.}
A yet another way to define the invariant of an $S^1$-parameterized family is to look into $L(\theta)$ in a gauge such that both $A(\theta)$ and $\phi(\theta)$ are $2\pi$-periodic.
In such a gauge, the $L$-symbols are not necessarily constant in $\theta$.
More specifically, $L(\theta)$ is given by a gauge transformation of $L(0)$ as
\begin{equation}
L(\theta)^{xy}_{(i, j), (z; k, \mu)} = \sum_{i^{\prime}, j^{\prime}, k^{\prime}} (U(\theta)_x)_{i i^{\prime}} (U(\theta)_y)_{j j^{\prime}} (U(\theta)_z)^{-1}_{k^{\prime} k} L(0)^{xy}_{(i^{\prime}, j^{\prime}), (z; k^{\prime}, \mu)},
\end{equation}
where $U(\theta)_x$ is a $\dim(x) \times \dim(x)$ invertible matrix.
Without loss of generality, we choose $(U(0)_x)_{ij} = \delta_{ij}$ and take $U(\theta)_x$ to be continuous in $\theta$.
Since $L(\theta)$ is $2\pi$-periodic due to the periodicity of $A(\theta)$ and $\phi(\theta)$, the gauge transformation $U(\theta)$ at $\theta = 2\pi$ defines an automorphism of a fiber functor corresponding to $L(0)$.
In particular, at $\theta = 2\pi$, we have
\begin{equation}
U(2\pi)_x = \eta_x^{-1},
\label{eq: pump inv 3}
\end{equation}
where $\eta_x$ is the same as the one in Eqs.~\eqref{eq: pump inv 1} and \eqref{eq: pump inv 2}.
Thus, the automorphism induced at $\theta = 2\pi$ can be identified with the invariant of an $S^1$-parameterized family.

\end{itemize}

From the above discussions, we find that one can associate an automorphism of a fiber functor to a given $S^1$-parameterized family of $\cC$-symmetric injective MPSs.
We call this automorphism a pump invariant due to its relation to a generalized Thouless pump, which will be detailed in Sec.~\ref{sec: Thouless pump of non-degenerate interface modes}.
We expect that the pump invariant is a complete invariant of a family, meaning that $S^1$-parameterized families of SPT states with symmetry $\cC$ are classified by automorphisms of a fiber functor.

\vspace{10pt}
\noindent{\bf Group of one-parameter families.}
The set of $S^1$-parameterized families that share the same MPS at $\theta = 0$ is equipped with a group structure.
Specifically, given a pair of families $\{A_1(\theta)\}$ and $\{A_2(\theta)\}$ such that $\ket{A_1(0)} = \ket{A_2(0)}$, their product $\{A(\theta)\}$ is defined by
\begin{equation}
A(\theta) = 
\begin{cases}
A_1(2\theta) \quad & 0 \leq \theta \leq \pi, \\
A_2(2\theta - 2\pi) \quad & \pi \leq \theta \leq 2\pi.
\end{cases}
\label{eq: product family}
\end{equation}
Here, we chose a gauge such that $A_1(2\pi) = A_2(0)$, which is possible because $\ket{A_1(2\pi)}$ and $\ket{A_2(0)}$ are the same state.
The automorphism (i.e., the pump invariant) associated with the product $\{A(\theta)\}$ is given by the composition of the automorphisms associated with $\{A_1(\theta)\}$ and $\{A_2(\theta)\}$.
This implies that the pump invariant gives rise to a group homomorphism from the group of $S^1$-parameterized families of injective MPSs to the group $\Aut^{\otimes}(F)$ of automorphisms of a fiber functor $F: \cC \rightarrow \Vect$.
We expect that this homomorphism is an isomorphism, meaning that $S^1$-parameterized families of SPT states with symmetry $\cC$ are classified by $\Aut^{\otimes}(F)$ as a group.
We note that $\Aut^{\otimes}(F)$ can be non-abelian in general.
In the case of finite group symmetry $\cC = \Vect_G$, we have $\Aut^{\otimes}(F) \cong \mathrm{H}^1_{\text{gp}}(G, \mathrm{U}(1))$ irrespective of $F$ (cf. Sec.~\ref{sec: Isomorphisms of fiber functors}).
This agrees with the known classification of $S^1$-parameterized families of SPT phases with $G$ symmetry~\cite{Thorngren:1612.00846, Hermele2021CMSA, Thorngren2021YITP, Shiozaki:2021weu, Bachmann:2022bhx}.

\vspace{10pt}
\noindent{\bf Abelianization of the pump invariant.}
As we saw in Sec.~\ref{sec: Abelianization of the nonabelian factor system}, 
we can abelianize the non-abelian factor systems and define a topological invariant that takes values in the third Hochschild cohomology $\cohoZ{3}{\cC}$.
Similarly, we can abelianize the pump invariant.
To define the abelianization, we first choose a gauge where $A(\theta)$ and $\phi(\theta)$ are $2\pi$-periodic.
Note that in this gauge, $L(\theta)$ is a $2\pi$-periodic function that is not necessarily constant.
As we saw around Eqs.~\eqref{eq: JMPS=Omega} and \eqref{eq: L MPS},
the $L$-symbols can be identified with the matrix elements of the non-abelian factor system.
Since the $L$-symbols are now parameterized by $\theta \in S^1$, the corresponding non-abelian factor system $\Omega_{x, y}(\theta)$ is also parameterized by $S^1$.
Then, we define the abelianization of the pump invariant by the winding number of $\det(\Omega_{x, y}(\theta))$:
\begin{equation}
    n_{x,y}:=\frac{1}{2\pi i}\int d\log\det(\Omega_{x,y}(\theta)).
\label{eq: winding}
\end{equation}

This invariant takes values in $\cohoZ{2}{\cC}$. 
Let us confirm this property.
First of all, the winding number $n_{x,y}$ is quantized to integers due to the periodicity of $\Omega_{x, y}(\theta)$. 
In addition, $n_{x,y}$ satisfies the cocycle condition $(\delta n)_{x,y,z}=0$, which follows from the cocycle condition for $\Omega_{x,y}$ described in Eq.~\eqref{eq: nonabelian cocycle}. 
Furthermore, $n_{x, y}$ is defined up to coboundary because gauge transformations can shift it by a coboundary.
More concretely, a gauge transformation that preserves the $2\pi$-periodicity of $A(\theta)$ and $\phi(\theta)$ changes the non-abelian factor system as
\begin{equation}
\Omega_{x,y}(\theta) \mapsto \left(\bigoplus_{z \in x \otimes y} f_z(\theta)\right)^{-1} \Omega_{x,y}(\theta) \left(f_x(\theta) \otimes f_y(\theta)\right),
\end{equation}
where $f_x(\theta)$ is a $2\pi$-periodic $\dim(x) \times \dim(x)$ invertible matrix, cf. Eq.~\eqref{eq: Omega gauge transformation}.
Under this gauge transformation, the winding number $n_{x, y}$ changes as
\begin{equation}
n_{x,y}\mapsto n_{x,y} + \dim(y)\cdot m_{x}-\sum_{z\in x\otimes y}m_{z} +\dim(x)\cdot m_{y} = n_{x, y} + (\delta m)_{x, y},
\end{equation}
where $m_{x}\in\mathbb{Z}$ is the winding number of $\det(f(\theta)_{x})$.
Thus, the abelianization $[n_{x, y}]$ of the pump invariant takes values in $\cohoZ{2}{\cC}$.
In general, the map $\Aut^{\otimes}(F) \rightarrow \cohoZ{2}{\cC}$ defined by $\{\Omega_{x, y}(\theta)\} \mapsto [n_{x, y}]$ is neither surjective nor injective.

The term ``abelianization" is justified because the composition law of $[n_{x, y}]$ is abelian.
Indeed, for the product~\eqref{eq: product family} of two families $\{A_1(\theta)\}$ and $\{A_2(\theta)\}$, the winding number $n_{x, y}$ can be computed as
\begin{equation}
\begin{aligned}
n_{x, y}
& = \frac{1}{2\pi i} \int_0^{\pi} d\theta \frac{d}{d\theta} \log\det(\Omega^1_{x,y}(2\theta)) + \frac{1}{2\pi i} \int_{\pi}^{2\pi} d\theta \frac{d}{d\theta} \log\det(\Omega^2_{x,y}(2\theta - 2\pi)) \\
& = \frac{1}{2\pi i} \int d\log\det(\Omega^1_{x,y}(\theta)) + \frac{1}{2\pi i} \int d\log\det(\Omega^2_{x,y}(\theta)) \\
&= n^1_{x, y} + n^2_{x, y}.
\end{aligned}
\end{equation}
Here, $\Omega_{x, y}^i(\theta)$ and $n_{x, y}^i$ are the non-abelian factor system for $\{A^i(\theta)\}$ and its winding number.
Put differently, taking the winding number of the non-abelian factor system gives a homomorphism from $\Aut^{\otimes}(F)$ to its abelianization $\Aut^{\otimes}(F)/[\Aut^{\otimes}(F),  \Aut^{\otimes}(F)]$.

\subsection{Thouless pump of non-degenerate interface modes}
\label{sec: Thouless pump of non-degenerate interface modes}
\subsubsection{Adiabatic evolution of $\cC$-symmetric injective MPS}
The pump invariant defined in the previous subsection can be interpreted as a generalized Thouless pump.
To see this, let us consider the following injective MPS on a periodic chain:
\begin{equation}
\ket{A(0; \theta; 0)} = \adjincludegraphics[scale=1,trim={10pt 10pt 10pt 10pt},valign = c]{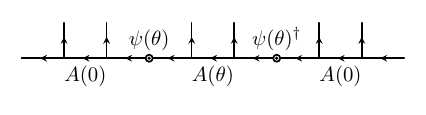}.
\label{eq: adiabatic evolution}
\end{equation}
The state in the middle of the chain is adiabatically evolved along an $S^1$-parameterized family $\{A(\theta)\}$.
The interface tensor $\psi(\theta) \in \Aut(V)$, which we choose to be unitary, is inserted so that the state $\ket{A(0; \theta; 0)}$ is $\cC$-symmetric for all $\theta \in [0, 2\pi]$, i.e.,
\begin{equation}
\widehat{\cO}_x \ket{A(0; \theta; 0)} = \dim(x) \ket{A(0; \theta; 0)}, \quad \forall x \in \cC.
\label{eq: symmetric 1}
\end{equation}
We choose $\psi(0) = \text{id}_V$ so that the initial state is $\ket{A(0)}$.
The condition \eqref{eq: symmetric 1} can be expressed in terms of diagrams as
\begin{equation}
\sum_{i,j} \adjincludegraphics[scale=1,trim={10pt 10pt 10pt 10pt},valign = c]{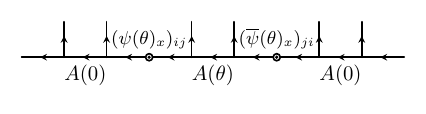} = \dim(x) \;\adjincludegraphics[scale=1,trim={10pt 10pt 10pt 10pt},valign = c]{tikz/out/pump_interpretation.pdf}\;,
\label{eq: symmetric 2}
\end{equation}
where $(\psi(\theta)_x)_{ij}$ and $(\overline{\psi}(\theta)_x)_{ji}$ are defined by
\begin{equation}
(\psi(\theta)_x)_{ij} = \adjincludegraphics[scale=1,trim={10pt 10pt 10pt 10pt},valign = c]{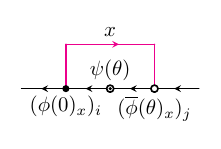}, \quad 
(\overline{\psi}(\theta)_x)_{ji} = \adjincludegraphics[scale=1,trim={10pt 10pt 10pt 10pt},valign = c]{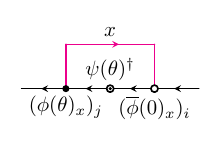}. 
\end{equation}
In the following, we will show that the tensor $\psi(\theta)$ that satisfies Eq.~\eqref{eq: symmetric 2} always exists.

To begin with, let us show that the symmetricity condition \eqref{eq: symmetric 2} is satisfied if and only if
\begin{equation}
(\psi(\theta)_x)_{ij} = (U(\theta)_x)_{ij} \psi(\theta), \quad (\overline{\psi}(\theta)_x)_{ij} = (U(\theta)_x)^{-1}_{ij} \psi(\theta)^{\dagger},
\label{eq: symmetric 3}
\end{equation}
where $U(\theta)_x$ is a $\dim(x) \times \dim(x)$ invertible matrix.\footnote{The two equalities in Eq.~\eqref{eq: symmetric 3} are equivalent to each other due to the weak completeness relation~\eqref{eq: weak completeness}.}
The ``if" part can be readily seen by substituting Eq.~\eqref{eq: symmetric 3} into the left-hand side of Eq.~\eqref{eq: symmetric 2}.
The ``only if" part follows from the injectivity of the MPS tensors $A(0)$ and $A(\theta)$.
More concretely, since $A(0)$ and $A(\theta)$ are injective, equation~\eqref{eq: symmetric 2} is satisfied only if
\begin{equation}
\sum_{i, j} (\psi(\theta)_x)_{ij} A (\overline{\psi}(\theta)_x)_{ji} = \dim(x) \psi(\theta) A \psi(\theta)^{\dagger}
\label{eq: injectivity}
\end{equation}
for any $D \times D$ matrix $A$ with $D$ being the bond dimension of the MPS.
In particular, in a gauge where $(\overline{\phi}(\theta)_x)_i = (\phi(\theta)_x)^{\dagger}_i$, the above equation can be written as
\begin{equation}
\sum_{i, j} \psi(\theta)^{\dagger}(\psi(\theta)_x)_{ij} A [\psi(\theta)^{\dagger} (\psi(\theta)_x)_{ij}]^{\dagger} = \dim(x) A,
\label{eq: id Kraus}
\end{equation}
where we used $(\overline{\psi}(\theta)_x)_{ji} = (\psi(\theta)_x)_{ij}^{\dagger}$, which follows from our gauge choice.
On the other hand, in the same gauge, the operators of the form $\psi(\theta)^{\dagger} (\psi(\theta)_x)_{ij}$ satisfy the following Kraus condition:
\begin{equation}
\sum_{i, j} [\psi(\theta)^{\dagger} (\psi(\theta)_x)_{ij}]^{\dagger} \psi(\theta)^{\dagger} (\psi(\theta)_x)_{ij}
= \sum_{i, j} (\overline{\psi}(\theta)_x)_{ji} \psi(\theta) \psi(\theta)^{\dagger} (\psi(\theta)_x)_{ij}
= \dim(x) \text{id}_V.
\label{eq: Kraus}
\end{equation}
Here, the second equality follows from the orthogonality~\eqref{eq: action orthogonality} and weak completeness~\eqref{eq: weak completeness} of the action tensors.
Equations \eqref{eq: id Kraus} and \eqref{eq: Kraus} show that $\left\{\psi(\theta)^{\dagger} (\psi(\theta)_x)_{ij} / \sqrt{\dim(x)} ~ \middle| ~ i, j = 1, 2, \cdots, \dim(x)\right\}$ in this gauge are Kraus operators that implement the identity operation on any $D \times D$ matrix $A$.
Such Kraus operators are known to be proportional to the identity.\footnote{See, e.g., \url{https://quantumcomputing.stackexchange.com/questions/25917/what-are-the-possible-kraus-operators-of-the-identity-channel}}
Therefore, there exists a set of complex numbers $\{(U(\theta)_x)_{ij}\}$ that satisfies
\begin{equation}
(\psi(\theta)_x)_{ij} = (U(\theta)_x)_{ij} \psi(\theta).
\end{equation}
The $\dim(x) \times \dim(x)$ matrix $U(\theta)_x$ has to be unitary (in particular, invertible) so that Eq.~\eqref{eq: symmetric 2} holds.
In a general gauge, $U(\theta)_x$ may not be unitary but is invertible.
Thus, we find that the state $\ket{A(0; \theta; 0)}$ is $\cC$-symmetric if and only if $\psi(\theta)$ satisfies Eq.~\eqref{eq: symmetric 3}.

Due to the orthogonality relation~\eqref{eq: action orthogonality} and the weak completeness relation~\eqref{eq: weak completeness}, equation~\eqref{eq: symmetric 3} can also be written as
\begin{equation}
(\phi(\theta)_x)_{i}= \sum_j (U(\theta)_x)_{ij}^{-1} \;\adjincludegraphics[scale=1,trim={10pt 10pt 10pt 10pt},valign = c]{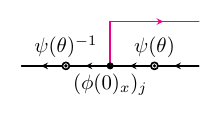}.
\label{eq: symmetric 4}
\end{equation}
Let us argue that $\psi(\theta)$ that satisfies the above equation always exists.
To this end, we first define operators $(\phi(\theta)_x)_i^{\alpha}$ by
\begin{equation}
(\phi(\theta)_x)_i^{\alpha} = \;\adjincludegraphics[scale=1,trim={10pt 10pt 10pt 10pt},valign = c]{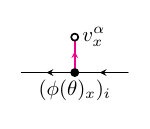}\;,
\end{equation}
where $\{v_x^{\alpha} \mid \alpha = 1, 2, \cdots, \dim V_x\}$ is a basis of the dual of the bond Hilbert space $V_x$.
The operators $\{(\phi(\theta)_x)_i^{\alpha} \in \End(V) \mid x \in \Simp(\cC), ~ i = 1, 2, \cdots, \dim(x),~ \alpha = 1, 2, \cdots, \dim V_x\}$ form a representation of a finite dimensional algebra whose multiplication is given by the composition $(\phi(\theta)_x)_i^{\alpha} (\phi(\theta)_y)_j^{\beta} \in \End(V)$. 
The structure constant of this algebra is determined by the $L$-symbols $L(\theta)$.
In particular, the algebra structure is independent of $\theta$ if we choose a gauge where the $L$-symbols are constant in $\theta$.
Since everything is supposed to be continuous in $\theta$, the equivalence class of the representation $\{(\phi(\theta)_x)_i^{\alpha}\}$ is also continuous in $\theta$.
Thus, if we suppose that the above algebra has only finitely many inequivalent finite dimensional representations, the representation $\{(\phi(\theta)_x)_i^{\alpha}\}$ must be in the same equivalence class for all $\theta$.
Therefore, there exists an automorphism $\psi(\theta) \in \Aut(V)$ such that
\begin{equation}
(\phi(\theta)_x)_i = \psi(\theta)^{-1} (\phi(0)_x)_i \psi(\theta).
\end{equation}
The above equation leads to Eq.~\eqref{eq: symmetric 4} if we perform a gauge transformation of the action tensor $(\phi(\theta)_x)_i \rightarrow \sum_j (U(\theta)_x)_{ij} (\phi(\theta)_x)_j$, where $U(\theta)_x$ is an invertible $\dim(x) \times \dim(x)$ matrix such that $U(0)_x = 1$.
Thus, there always exists $\psi(\theta)$ that fits into Eq.~\eqref{eq: symmetric 4}, or equivalently, Eq.~\eqref{eq: symmetric 2}.
In other words, one can always find the interface tensor $\psi(\theta)$ in Eq.~\eqref{eq: adiabatic evolution} that makes the state $\ket{A(0; \theta; 0)}$ symmetric under $\cC$.

\subsubsection{Non-abelian Thouless pump}
Now, let us show that changing $\theta$ from $0$ to $2\pi$ pumps a non-degenerate interface mode.
To this end, we choose a gauge in which the action tensors $\{(\phi(\theta)_x)_i\}$ do not depend on $\theta$.
This choice of a gauge is possible due to Eq.~\eqref{eq: symmetric 4}.
We note that $\psi(2\pi)$ is the identity in this gauge.
Since the action tensors $(\phi(\theta)_x)_i$ are $2\pi$-periodic and the $L$-symbols $L(\theta)$ are constant in $\theta$, the MPS tensor $A(\theta)$ is not necessarily $2\pi$-periodic.
Specifically, the MPS tensors at $\theta = 0$ and $\theta = 2\pi$ are related by a gauge transformation
\begin{equation}
A(2\pi) = \adjincludegraphics[scale=1,trim={10pt 10pt 10pt 10pt},valign = c]{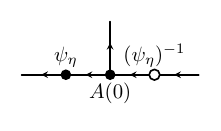},
\label{eq: Thouless pump A}
\end{equation}
where $\psi_{\eta}$ is the pump invariant associated with the $S^1$-parameterized family $\{A(\theta)\}$, see Eq. \eqref{eq: non-periodic A}.\footnote{The tensor $\psi_{\eta}^{\prime}$ in Eq.~\eqref{eq: non-periodic A} is now written as $\psi_{\eta}$.}
Therefore, going around the parameter space $S^1$ changes the initial state $\ket{A(0)} = \ket{A(0; 0; 0)}$ into
\begin{equation}
\ket{A(0; 2\pi; 0)} = \adjincludegraphics[scale=1,trim={10pt 10pt 10pt 10pt},valign = c]{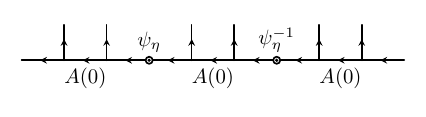}.
\label{eq: Thouless pump}
\end{equation}
This shows that $\ket{A(0; \theta; 0)}$ returns to the initial state $\ket{A(0)}$ at $\theta = 2\pi$ except that non-degenerate interface modes $\psi_{\eta}$ and $\psi_{\eta}^{-1}$ are localized at the interfaces of two regions.
This non-degenerate interface mode is regarded as a generalized Thouless pump.
We recall that $\psi_{\eta}$ is a one-dimensional representation of the self-interface algebra.
Since the one-dimensional representations generally form a non-abelian group, they will be referred to as non-abelian Thouless pumps.

One can also see the generalized Thouless pump in a different gauge where $A(\theta)$ is $2\pi$-periodic and $L(\theta)$ is constant in $\theta$.
In this gauge, the action tensors $\{(\phi(\theta)_x)_i\}$ are not necessarily $2\pi$-periodic.
Specifically, Eq.~\eqref{eq: symmetric 4} implies that the action tensors at $\theta = 0$ and $\theta = 2\pi$ are related by
\begin{equation}
(\phi(2\pi)_x)_i = \adjincludegraphics[scale=1,trim={10pt 10pt 10pt 10pt},valign = c]{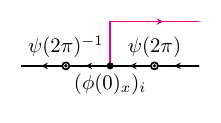}.
\label{eq: pump in another gauge 1}
\end{equation}
Here, we used the fact that $(U(\theta)_x)_{ij} = \delta_{ij}$ in this gauge because $L(\theta)$ is constant in $\theta$.
On the other hand, as we discussed in Sec.~\ref{sec: Invariants of one-parameter families of SPT phases}, the action tensors at $\theta = 0$ and $\theta = 2\pi$ are also related by a gauge transformation
\begin{equation}
(\phi(2\pi)_x)_i = \sum_j (\eta_x)_{ij} (\phi(0)_x)_j,
\label{eq: pump in another gauge 2}
\end{equation}
where $\eta$ is the pump invariant associated with the $S^1$-parameterized family $\{A(\theta)\}$, see Eq.~\eqref{eq: pump inv 1}.
Combining Eqs.~\eqref{eq: pump in another gauge 1} and \eqref{eq: pump in another gauge 2}, we find that $\psi(2\pi) = \psi_{\eta}$ in this gauge.
Therefore, the state $\ket{A(0; 2\pi; 0)}$ is again given by Eq.~\eqref{eq: Thouless pump}.
This is consistent with the fact the Thouless pump does not depend on a gauge.

When $\cC$ is an ordinary finite group symmetry $\Vect_G$, the pumped interface mode $\psi_{\eta}$ is a one-dimensional representation of the group algebra $\C[G]$, i.e., $\psi_{\eta}$ is a $G$-charge.
Thus, in this case, the generalized Thouless pump reduces to the ordinary one.

\subsection{Example: $\Rep(G)$ symmetry}
\label{sec: Example: Rep(G) symmetry II}
As an example, let us consider concrete models of $S^1$-parameterized families of SPT states with $\Rep(G)$ symmetry.
We construct these models using the $G \times \Rep(G)$ cluster state that we reviewed in Sec.~\ref{sec: Example Rep(G) symmetry I}.
To obtain an $S^1$-parameterized family from the cluster state, we first fix an arbitrary group element $g \in G$ and choose any continuous path $\{R_g(\theta) \mid \theta \in [0, 2\pi]\}$ of unitary operators, acting on the physical Hilbert space $\cH = \C^{|G|}$, such that
\begin{equation}
R_g(0) = \text{id}_{\cH}, \quad R_g(2\pi) = R_g.
\end{equation}
Namely, the path $\{R_g(\theta)\}$ interpolates between the identity operator and the right multiplication of $g^{-1}$.\footnote{We recall that $R_g$ denotes the right multiplication of $g^{-1}$ as defined in Eq.~\eqref{eq: Rg Zrho}.}
Such a path always exists because the space of unitary operators on $\cH$ (i.e., the space of $|G| \times |G|$ unitary matrices) is path-connected.
Given a path $\{R_g(\theta)\}$ of unitary operators, one can construct a family $\{A_g(\theta) \mid \theta \in [0, 2\pi]\}$ of MPS tensors by applying $R(\theta)$ to odd sites of the cluster state:
\begin{equation}
A_g(\theta) = \;\adjincludegraphics[scale=1,trim={10pt 10pt 10pt 10pt},valign = c]{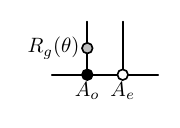}.
\label{eq: cluster pump}
\end{equation}
See Eq.~\eqref{eq: GxRep(G) MPS} for the definition of the MPS tensor of the $G \times \Rep(G)$ cluster state.
We note that the MPS tensors $A_g(0)$ and $A_g(2\pi)$ are gauge equivalent because $G$ is a symmetry of the cluster state, cf. Eq.~\eqref{eq: cluster symmetry}.
Therefore, the parameter space is a circle $S^1$.
The path of Hamiltonians $\{H_g(\theta)\}$ is explicitly given by
\begin{equation}
H_g(\theta) = \left(\prod_{i: \text{odd}} R_g(\theta)_i \right) H_{\text{cluster}} \left(\prod_{i: \text{odd}} R_g(\theta)_i \right)^{\dagger},
\label{eq: cluster_family}
\end{equation}
where $R_g(\theta)_i$ denotes $R_g(\theta)$ acting on site $i$ and $H_{\text{cluster}}$ is the Hamiltonian of the $G \times \Rep(G)$ cluster model given in \cite{Fechisin:2023dkj}.

The above family $\{A_g(\theta)\}$ preserves the $\Rep(G)$ symmetry at each $\theta \in [0, 2\pi]$, meaning that $\{A_g(\theta)\}$ gives rise to an $S^1$-parameterized family of SPT states with $\Rep(G)$ symmetry.
To compute the pump invariant of this family, we notice that the action tensors $(\phi(\theta)_{\rho})_i$ for the $\Rep(G)$ symmetry does not depend on $\theta$ by construction:
\begin{equation}
\adjincludegraphics[scale=1,trim={10pt 10pt 10pt 10pt},valign = c]{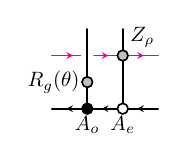}\;
=\;\adjincludegraphics[scale=1,trim={10pt 10pt 10pt 10pt},valign = c]{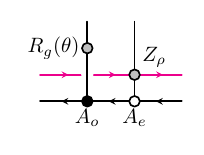}\;
=\;\sum_{i} \adjincludegraphics[scale=1,trim={10pt 10pt 10pt 10pt},valign = c]{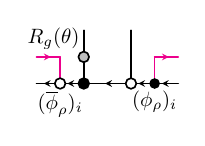}\;,
\end{equation}
where $Z_{\rho}$, $\phi_{\rho}$, and $\overline{\phi}_{\rho}$ are defined by Eqs.~\eqref{eq: Rg Zrho} and \eqref{eq: cluster action tensors}.
Since the action tensors do not depend on $\theta$, the $L$-symbols are constant in $\theta$.
Therefore, the pump invariant of the family appears as the non-periodicity of the MPS tensor $A_g(\theta)$.
As we can see from Eq.~\eqref{eq: cluster symmetry}, the MPS tensor $A_g(\theta)$ is not $2\pi$-periodic when $g \in G$ is not the unit element.
Concretely, $A_g(0)$ and $A_g(2\pi)$ differ by a gauge transformation $R_g^{-1}$ acting on the bond Hilbert space.
Accordingly, the family $\{A_g(\theta)\}$ pumps a non-degenerate interface mode $R_g^{-1}$ when $\theta$ increases from $0$ to $2\pi$.
Here, we note that $R_g^{-1}$ is indeed a one-dimensional representation of the interface algebra, i.e.,
\begin{equation}
(\widehat{\cO}_{\rho})_{ij} R_g^{-1} = \rho(g)_{ij} R_g^{-1}.
\end{equation}
Thus, we find that the pump invariant of the $S^1$-parameterized family $\{A_g(\theta)\}$ constructed above is a one-dimensional representation $\{(\widehat{\cO}_{\rho})_{ij} = \rho(g)_{ij}\}$.
Since $g \in G$ is arbitrary, one can construct as many families as $|G|$ in this way.
Remarkably, the above construction exhausts all $S^1$-parameterized families of $\Rep(G)$ SPT states in the trivial phase because one-dimensional representations of the interface algebra of the cluster state are classified by $G$ as we discussed in Sec.~\ref{sec: Non-degenerate interfaces of a Rep(G) SPT phase}.
We note that the generalized Thouless pump in this example can be non-abelian.
This is in contrast to the case of an ordinary group symmetry, where the Thouless pump is always abelian.

\section{Classification of parameterized families}
\label{sec: Generalizations}
In this section, we propose conjectures on the classification of various parameterized families of gapped systems, generalizing the classification of $S^1$-parameterized families of 1+1d SPT states discussed in the previous section.
Our conjectures are supported by general arguments and concrete examples.

\subsection{General gapped phases in 1+1d}
\label{sec: General gapped phases}
\subsubsection{Conjecture}
Let us consider the classification of $S^1$-parameterized families of general 1+1d gapped systems with fusion category symmetry $\cC$.
In general, $\cC$-symmetric gapped phases in 1+1d are in one-to-one correspondence with module categories over $\cC$ \cite{Thorngren:2019iar}.
In particular, SPT phases with symmetry $\cC$ correspond to $\cC$-module categories with a single simple object, which can be identified with fiber functors of $\cC$ \cite{EGNO2015}.
In what follows, the gapped phase labeled by a $\cC$-module category $\cM$ is denoted by $\cT_{\cM}^{\cC}$.

Our conjecture on the classification of $S^1$-parameterized families of general $\cC$-symmetric gapped systems is as follows:
\begin{conjecture}
$S^1$-parameterized families of 1+1d $\cC$-symmetric gapped systems in gapped phase $\cT_{\cM}^{\cC}$ are classified by the group $\Fun_{\cC}(\cM, \cM)^{\text{inv}}$ of invertible objects of the category $\Fun_{\cC}(\cM, \cM)$ of $\cC$-module endofunctors of $\cM$.
\label{con: general gapped phases}
\end{conjecture}
Let us give a motivation for the above conjecture.
First of all, since $\cC$-symmetric gapped phases are classified by $\cC$-module categories, an $S^1$-parameterized family of gapped systems should give rise to an $S^1$-parameterized family of $\cC$-module categories $\cM(\theta)$.
When the gapped systems in the family belong to the same phase $\cT^{\cC}_{\cM}$, the $\cC$-module category $\cM(\theta)$ is equivalent to $\cM$ for all $\theta \in S^1 = [0, 2\pi]$.
In particular, $\cM(0)$ and $\cM(2\pi)$ are the same module category, which we choose to be $\cM$.
Nevertheless, going around $S^1$ can induce a non-trivial automorphism of $\cM$, as we saw in the case of SPT phases (cf. Sec.~\ref{sec: Parameterized family and Thouless pump}).
Thus, an $S^1$-parameterized family of gapped systems is associated with an automorphism of $\cM$, i.e., an invertible $\cC$-module endofunctor of $\cM$.
We expect that this automorphism classifies $S^1$-parameterized families of gapped systems.
In the context of MPS, it should be straightforward to construct an invariant of the classification by applying the analysis in Sec.~\ref{sec: Invariants of one-parameter families of SPT phases} to $\cC$-symmetric block-injective MPSs.

When $\cM$ has only one simple object, the group of invertible objects of $\Fun_{\cC}(\cM, \cM)$ is isomorphic to the group $\Aut^{\otimes}(F)$ of automorphisms of the fiber functor $F$ corresponding to $\cM$.
Thus, as a special case, our conjecture implies that $S^1$-parameterized families of SPT states labeled by a fiber functor $F$ are classified by $\Aut^{\otimes}(F)$.
This is supported by the discussions on pump invariants in Sec.~\ref{sec: Invariants of one-parameter families of SPT phases}.

\vspace{10 pt}
\noindent{\bf Relation to symmetry TFT.}
Conjecture~\ref{con: general gapped phases} is also natural from the perspective of symmetry TFT. 
In the symmetry TFT construction, a $\mathcal{C}$-symmetric system in 1+1d is realized as a three-dimensional TFT on a slab as shown in Fig.~\ref{fig: SymTFT}.
\begin{figure}[t]
\adjincludegraphics[valign=c, scale=1, trim={10pt 10pt 10pt 10pt}]{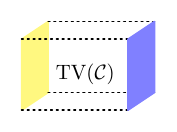}
\caption{The symmetry TFT construction of 1+1d systems with symmetry $\cC$. The left boundary of the slab is topological and supports topological lines described by $\cC$. On the other hand, the right boundary is arbitrary and dictates the dynamics of the 1+1d system. When the right boundary is also topological, the whole 1+1d system becomes topological.}
\label{fig: SymTFT}
\end{figure}
The 3d TFT in the bulk is the Turaev-Viro-Barrett-Westbury TFT $\text{TV}(\cC)$ \cite{Turaev:1992hq, Barrett:1993ab}, which is the low-energy limit of the 2+1d string-net model \cite{Levin:2004mi}.
In this construction, the dynamics of the 1+1d system is determined by the choice of a boundary condition of $\text{TV}(\cC)$.
In particular, 1+1d gapped systems are obtained by choosing this boundary to be topological.
The $\cC$-symmetric gapped phase $\cT^{\cC}_{\cM}$ corresponds to the topological boundary labeled by the $\cC$-module category $\cM$.
Thus, in this framework, an $S^1$-parameterized family of $\cC$-symmetric gapped states corresponds to an $S^1$-parameterized family of topological boundaries of $\text{TV}(\cC)$.

Given a family $\{\cM(\theta) \mid \theta \in S^1\}$ of topological boundaries of $\text{TV}(\cC)$, one can consider a spatially modulated boundary of the symmetry TFT as illustrated in Fig.~\ref{fig: modulated boundary}.
\begin{figure}[t]
\adjincludegraphics[valign=c, scale=1, trim={10pt 10pt 10pt 10pt}]{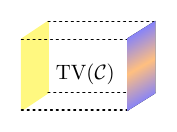}
$\quad \xrsquigarrow{~coarse graining~} \quad$
\adjincludegraphics[valign=c, scale=1, trim={10pt 10pt 10pt 10pt}]{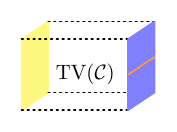}
\caption{The middle region of the right boundary is modulated gradually from $\cM(0)$ to $\cM(2\pi)$. This modulated region becomes an invertible topological line when viewed from far away. After squashing the bulk, the invertible line on the boundary becomes an invertible interface in 1+1d. The correspondence between invertible interfaces and $S^1$-parameterized families was already observed in Sec.~\ref{sec: Thouless pump of non-degenerate interface modes} in the case of SPT phases.}
\label{fig: modulated boundary}
\end{figure}
If the whole system is viewed from far away, the modulation looks like a topological line on the boundary.
This topological line is invertible because the modulation can be undone by the reversed modulation.
Therefore, from the perspective of symmetry TFT, $S^1$-parameterized families of $\mathcal{C}$-symmetric gapped states should be classified by invertible topological lines on the topological boundary $\cM$.
Since topological lines on $\cM$ are classified by $\Fun_\mathcal{C}(\mathcal{M}, \mathcal{M})$ \cite{KK2011, Fuchs:2012dt, Kapustin:2010hk}, we are led to Conjecture~\ref{con: general gapped phases}.

\vspace{10pt}
\noindent{\bf Relation to gauging.}
Provided that the conjecture holds for some gapped phase $\cT$, we can show that the conjecture also holds for other gapped phases obtained by the gauging of $\cT$.
To see this, we first recall the gauging of 1+1d gapped phases with fusion category symmetries.
Let $\cT^{\cC}_{\cM}$ be a gapped phase labeled by a $\cC$-module category $\cM$ and let $\cT^{\cC}_{\cM}/A$ denote the gauging of $\cT^{\cC}_{\cM}$ by a separable algebra object $A \in \cC$.
The symmetry of $\cT^{\cC}_{\cM}/A$ is described by the category ${}_A \cC_A$ of $(A, A)$-bimodules in $\cC$ \cite{Bhardwaj:2017xup}.
Due to Ostrik's theorem~\cite{Ostrik2003}, any $\cC$-module category $\cM$ can be written as the category $\cC_B$ of right $B$-modules in $\cC$, where $B$ is also a separable algebra object in $\cC$.
When $\cM = \cC_B$, the ${}_A \cC_A$-module category corresponding to the gauged theory $\cT^{\cC}_{\cM}/A$ is the category ${}_A \cC_B$ of $(A, B)$-bimodules in $\cC$ \cite[Appendix A]{Komargodski:2020mxz}.

Now, we assume that the conjecture holds for $\cT^{\cC}_{\cM}$, that is, $S^1$-parameterized families of $\cT^{\cC}_{\cM}$ are classified by the group of invertible objects of $\Fun_{\cC}(\cM, \cM) = \Fun_{\cC}(\cC_B, \cC_B)$.
Since the gauging is an invertible operation~\cite{Bhardwaj:2017xup},\footnote{Precisely, there exists a separable algebra $A^{\prime} \in {}_A \cC_A$ such that $T/A/A^{\prime}$ goes back to the original theory $\cT$.} $S^1$-parameterized families of the gauged theory $\cT^{\cC}_{\cM}/A$ are also classified by $\Fun_{\cC}(\cC_B, \cC_B)^{\text{inv}}$.
On the other hand, Theorem 7.12.16 of \cite{EGNO2015} implies that $\Fun_{\cC}(\cC_B, \cC_B)$ is equivalent to $\Fun_{{}_A \cC_A}({}_A \cC_B, {}_A \cC_B)$ as a fusion category:
\begin{equation}
\Fun_{\cC}(\cC_B, \cC_B) \cong \Fun_{{}_A \cC_A} (\Fun_{\cC}(\cC_A, \cC_B), \Fun_{\cC}(\cC_A, \cC_B)) \cong \Fun_{{}_A \cC_A} ({}_A \cC_B, {}_A \cC_B).
\end{equation}
Therefore, $S^1$-parameterized families of the gauged theory $\cT^{\cC}_{\cM}/A = \cT^{{}_A \cC_A}_{{}_A \cC_B}$ are classified by the group of invertible objects of $\Fun_{{}_A \cC_A}({}_A \cC_B, {}_A \cC_B)$.
This shows that the conjecture holds for the gauged theory $\cT^{\cC}_{\cM}/A$ as long as it holds for the original theory $\cT^{\cC}_{\cM}$.

\subsubsection{Examples}
\noindent{\bf Example 1: $G$-SSB phase.}
As a simple example, we apply our conjecture to the gapped phase that spontaneously breaks a finite group symmetry $G$ down to the trivial group.
The fusion category that describes a finite group symmetry $G$ is the category $\Vect_G$ of $G$-graded vector spaces.
The $\Vect_G$-module category corresponding to the $G$-SSB phase is the regular module $\Vect_G$.\footnote{For a fusion category $\cC$, a regular $\cC$-module is $\cC$ itself, on which $\cC$ acts by its tensor product $\otimes: \cC \times \cC \rightarrow \cC$.} 
Thus, our conjecture implies that $S^1$-parameterized families of gapped systems in the $G$-SSB phase are classified by the group of invertible objects of $\Fun_{\Vect_G}(\Vect_G, \Vect_G) \cong \Vect_G$, which is $G$.
Physically, this group $G$ can be interpreted as the group of invertible domain walls pumped by $S^1$-parameterized families.

One can also understand this classification from the point of view of gauging.
If we gauge the symmetry $G$ of the $G$-SSB phase, we obtain a $\Rep(G)$ SPT phase that corresponds to the forgetful functor of $\Rep(G)$.
Our conjecture implies that $S^1$-parameterized families of invertible states in this $\Rep(G)$ SPT phase are classified by the group of automorphisms of the forgetful functor, which is isomorphic to $G$ due to Tannaka-Krein duality.
Thus, $S^1$-parameterized families of gapped systems in the $G$-SSB phase should also be classified by $G$ because the gauging is invertible.

\vspace{10pt}
\noindent{\bf Example 2: $\cC$-SSB phase.}
More generally, we can consider the SSB phase of a general fusion category symmetry $\cC$.
The $\cC$-module category corresponding to the $\cC$-SSB phase is the regular module category $\cC$.
The associated functor category is $\Fun_{\cC}(\cC, \cC) \cong \cC$, and hence Conjecture~\ref{con: general gapped phases} suggests that $S^1$-parameterized families of gapped systems in the $\cC$-SSB phase are classified by $\cC^{\text{inv}}$, the group of invertible objects of $\cC$.

\vspace{10pt}
\noindent{\bf Example 3: $\Z_4/\Z_2$-SSB phase.}
One can also consider examples where the symmetry is partially broken.
A simple example is the spontaneous symmetry breaking of $\Z_4$ down to $\Z_2$.
The $\Vect_{\Z_4}$-module category corresponding to this gapped phase is $\Vect_{\Z_2}$.
The action of $\Vect_{\Z_4}$ on $\Vect_{\Z_2}$ is defined via the group homomorphism that maps $a \in \Z_4$ to $a \pmod 2 \in \Z_2$.
The corresponding functor category is $\Fun_{\Vect_{\Z_4}}(\Vect_{\Z_2}, \Vect_{\Z_2}) \cong \Vect_{\Z_2 \times \Z_2}^{\omega}$ \cite{Bhardwaj:2017xup, Tachikawa:2017gyf, naidu2007categorical, naidu2008lagrangian, uribe2017classification}, where $[\omega] \neq 0 \in \mathrm{H}_{\text{gp}}^3(\Z_2 \times \Z_2, \mathrm{U}(1))$ is given by
\begin{equation}
\omega((a_1, b_1), (a_2, b_2), (a_3, b_3)) = (-1)^{a_1 a_2 b_3}
\end{equation}
for $a_i, b_i \in \Z_2 = \{0, 1\}$.
Thus, Conjecture~\ref{con: general gapped phases} suggests that $S^1$-parameterized families of gapped systems in the $\Z_4/\Z_2$-SSB phase are classified by $\Z_2 \times \Z_2$.
We note that $[\omega] \in \mathrm{H}_{\text{gp}}^3(\Z_2 \times \Z_2, \mathrm{U}(1))$ does not affect the classification of $S^1$-parameterized families.
However, it would affect the classification for more general parameter spaces, see Sec.~\ref{sec: General parameter space} for more details.

\subsection{General parameter space in 1+1d}
\label{sec: General parameter space}
\subsubsection{Conjecture}
In the discussions so far, the parameter space was supposed to be a circle $S^1$.
More generally, one can consider families of 1+1d gapped systems parameterized by a general parameter space $X$, which we assume to be sufficiently nice (e.g., a CW complex).
Our conjecture on the classification of such families is as follows\footnote{ In the following, we consider the classification of families for a fixed module category $\cM$. Therefore, the classification for $X=pt$ is always trivial.}:
\begin{conjecture}
$X$-parameterized families of 1+1d gapped systems in a gapped phase $\cT^{\cC}_{\cM}$ are classified by the set $[X, B\underline{\Fun}_{\cC}(\cM, \cM)^{\text{inv}}]$ of homotopy classes of maps from $X$ to the classifying space of a 2-group $\underline{\Fun}_{\cC}(\cM, \cM)^{\text{inv}}$.
See below for the definition of $\underline{\Fun}_{\cC}(\cM, \cM)^{\text{inv}}$.
The set $[X, B\underline{\Fun}_{\cC}(\cM, \cM)^{\text{inv}}]$ is also known as the non-abelian \v{C}ech cohomology $\check{\mathrm{H}}^1(X, \underline{\Fun}_{\cC}(\cM, \cM)^{\text{inv}})$.
\label{con: general parameter space}
\end{conjecture}
The above conjecture is based on the hypothesis that the space of 1+1d systems in a $\cC$-symmetric gapped phase labeled by $\cM$ is homotopy equivalent to $B\underline{\Fun}_{\cC}(\cM, \cM)^{\text{inv}}$.
See \cite{Hsin:2020cgg, Aasen:2022cdu, Hsin:2022iug} for a similar conjecture for 2+1d topologically ordered phases.\footnote{In \cite{Hsin:2020cgg, Aasen:2022cdu, Hsin:2022iug}, it is conjectured that the space of 2+1d gapped Hamiltonians that realize a topological order described by a modular tensor category $\cB$ (up to invertible topological order) is homotopy equivalent to the classifying space $B\underline{\underline{\text{Pic}}}(\cB)$ of the categorical Picard 2-group $\underline{\underline{\text{Pic}}}(\cB)$.
Here, the categorical Picard 2-group $\underline{\underline{\text{Pic}}}(\cB)$ is a monoidal 2-category consisting of invertible $\cB$-modules, invertible $\cB$-module functors, and invertible $\cB$-module natural transformations \cite{ENO2010}. See also \cite[Appendix F]{Kitaev:2005hzj} for an earlier discussion. We will return to this conjecture at the end of Sec.~\ref{sec: General classification in higher dimensions}.}
Conjecture~\ref{con: general parameter space} reduces to Conjecture~\ref{con: general gapped phases} when the parameter space $X$ is a circle because we have
\begin{equation}
[S^1, B\underline{\Fun}_{\cC}(\cM, \cM)^{\text{inv}}] = \pi_1(B\underline{\Fun}_{\cC}(\cM, \cM)^{\text{inv}}) \cong \Fun_{\cC}(\cM, \cM)^{\text{inv}}.
\end{equation}
See Eq.~\eqref{eq: homotopy groups} for more general homotopy groups of $B\underline{\Fun}_{\cC}(\cM, \cM)^{\text{inv}}$.
As in the case of Conjecture~\ref{con: general gapped phases}, if Conjecture \ref{con: general parameter space} holds for a $\cC$-symmetric gapped phase $\cT^{\cC}_{\cM}$, it also holds for any other gapped phases obtained by the gauging of $\cT^{\cC}_{\cM}$ because $\Fun_{\cC}(\cM, \cM)$ is invariant under gauging.

Let us unpack the above conjecture in detail.
In Conjecture~\ref{con: general parameter space},  $\underline{\Fun}_{\cC}(\cM, \cM)^{\text{inv}}$ denotes a categorical group consisting of invertible objects and invertible morphisms of $\Fun_{\cC}(\cM, \cM)$.
More concretely, $\underline{\Fun}_{\cC}(\cM, \cM)^{\text{inv}}$ is a 2-group $(G, A, \rho, \beta)$,\footnote{There are various descriptions of a 2-group \cite{Baez:2003yaq}. In this paper, a 2-group $\mathbb{G}$ refers to a quadruple $(G, A, \rho, \beta)$ with $G$ a group, $A$ an abelian group, $\rho$ a group homomorphism from $G$ to $\Aut(A)$, and $\beta$ a group 3-cocycle of $G$ with coefficients in $A$. We refer the reader to, e.g., \cite{Kapustin:2013uxa, Benini:2018reh} for more detailed expositions of 2-groups for physicists.} where the 0-form group $G$ is $\Fun_{\cC}(\cM, \cM)^{\text{inv}}$, the 1-form group $A$ is $\mathrm{U}(1)$, the action $\rho: G \rightarrow \Aut(A)$ is trivial, and the Postnikov class $[\beta] \in \mathrm{H}_{\text{gp}}^3(\Fun_{\cC}(\cM, \cM)^{\text{inv}}, \mathrm{U}(1))$ is given by the $F$-symbols of the invertible subcategory of $\Fun_{\cC}(\cM, \cM)$.
We suppose that the 1-form group $\mathrm{U}(1)$ is equipped with the standard (i.e., continuous) topology of the Lie group $\mathrm{U}(1)$, not the discrete topology.
On the other hand, the 0-form group $\Fun_{\cC}(\cM, \cM)^{\text{inv}}$ has the discrete topology because it is a finite group.
As we will see below, the choice of a topology of $\mathrm{U}(1)$ plays an important role in the classification of parameterized families.
In what follows, $\mathrm{U}(1)$ equipped with the standard topology will be simply written as $\mathrm{U}(1)$, whereas the same underlying group equipped with the discrete topology will be denoted by $\mathrm{U}(1)_{\delta}$.

The classifying space of $\underline{\Fun}_{\cC}(\cM, \cM)^{\text{inv}}$ is denoted by $B\underline{\Fun}_{\cC}(\cM, \cM)^{\text{inv}}$.
If we use $\mathrm{U}(1)$ with the standard topology, the realization of $B\underline{\Fun}_{\cC}(\cM, \cM)^{\text{inv}}$ as a space is given by a Postnikov tower
\begin{equation}
  \adjincludegraphics[scale=1,trim={10pt 10pt 10pt 10pt},valign = c]{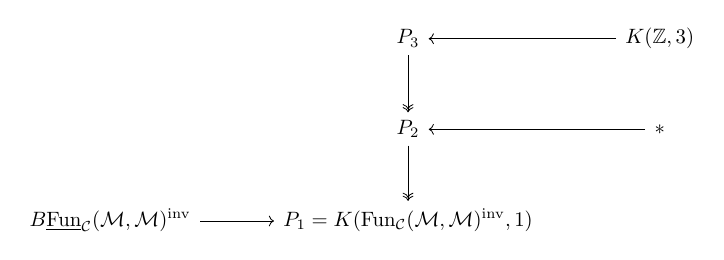}
\end{equation}
with the Postnikov invariant $\beta_{1}([\beta])\in \cohoZ{4}{B\Fun_{\mathcal{C}}(\mathcal{M}, \mathcal{M})^{\text{inv}}}$. 
Here, $\beta_{1}:\cohoU{3}{\bullet}\to\cohoZ{4}{\bullet}$ is the Bochstein homomorphism associated with $0\to \mathbb{Z} \to \mathbb{R} \to \mathrm{U}(1)$,
and surjective arrows represent Serre fibrations.
If we use $\mathrm{U}(1)_{\delta}$ with the discrete topology, the realization as a space is given by a Postnikov tower
\begin{equation}
  \adjincludegraphics[scale=1,trim={10pt 10pt 10pt 10pt},valign = c]{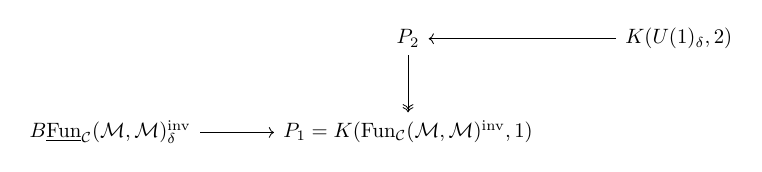}
\end{equation}
with the Postnikov invariant $[\beta] \in \coho{3}{B\Fun_{\mathcal{C}}(\mathcal{M}, \mathcal{M})^{\text{inv}}}{\mathrm{U}(1)_{\delta}} \simeq \mathrm{H}^3_{\text{gp}}(\Fun_{\mathcal{C}}(\mathcal{M}, \mathcal{M})^{\text{inv}}, \mathrm{U}(1))$.\footnote{We note $\mathrm{H}^3_{\text{gp}}(G, A) \simeq \coho{3}{BG}{A}$ if $A$ is discrete~\cite{wigner1970algebraic}.}
See, e.g.,~\cite{BS2009} for details of the classifying space of a general topological 2-group.

For our purposes, it suffices to know the following fact about the classifying space: for a topological 2-group $\mathbb{G}$, the homotopy classes of maps from $X$ to $B \mathbb{G}$ are in one-to-one correspondence with elements of the non-abelian \v{C}ech cohomology $\check{\mathrm{H}}^1(X, \mathbb{G})$ \cite{BS2009, Wockel2011, Bartels2006, Jurco2011, BBK2012}, i.e., 
\begin{equation}
[X, B \mathbb{G}] = \check{\mathrm{H}}^1(X, \mathbb{G}).
\end{equation}
The explicit description of the \v{C}ech cohomology $\check{\mathrm{H}}^1(X, \mathbb{G})$ for a general topological 2-group $\mathbb{G}$ is given in \cite[Section 4]{BS2009}.
The description in \cite{BS2009} is based on the crossed module corresponding to $\mathbb{G} = (G, A, \rho, \beta)$.
When $G$ is finite and $\rho$ is trivial as in the case of our interest $\mathbb{G} = \underline{\Fun}_{\cC}(\cM, \cM)^{\text{inv}}$, one can directly describe the \v{C}ech cohomology $\check{\mathrm{H}}^1(X, \mathbb{G})$ in terms of $(G, A, \beta)$ \cite{Berwickevans2021}.
We will review this description below.
See also \cite{Kapustin:2013uxa} for the direct description of the \v{C}ech cohomology with coefficients in a general discrete 2-group.

\vspace{10pt}
\noindent{\bf Non-abelian \v{C}ech cohomology.}
For concreteness, let us give an explicit description of the \v{C}ech cohomology $\check{\mathrm{H}}^1(X, \underline{\Fun}_{\cC}(\cM, \cM)^{\text{inv}})$ following \cite{Berwickevans2021}.
We first fix an open cover $\{U_i\}_{i \in I}$ of $X$, where $I$ is an index set.
The overlaps of open subsets are denoted by $U_{ij \cdots k} := U_i \cap U_j \cap \cdots \cap U_k$.
Here, any finite overlaps $U_{ij \cdots k}$ are supposed to be contractible, i.e., $\{U_i\}$ is a good open cover.
An element of $\check{\mathrm{H}}^1(X, \underline{\Fun}_{\cC}(\cM, \cM)^{\text{inv}})$ is an equivalence class of a pair $(\{g_{ij}\}, \{a_{ijk}\})$ of continuous functions
\begin{equation}
g_{ij}: U_{ij} \rightarrow \Fun_{\cC}(\cM, \cM)^{\text{inv}}, \qquad a_{ijk}: U_{ijk} \rightarrow \mathrm{U}(1)
\end{equation}
such that 
\begin{equation}
\begin{aligned}
(dg)_{ijk}(x) & := g_{jk}(x) g_{ik}(x)^{-1} g_{ij}(x) = 1,& & \forall x \in U_{ijk}, \\
(da)_{ijkl}(y) & := a_{jkl}(y) a_{ikl}(y)^{-1} a_{ijl}(y) a_{ijk}(y)^{-1} = \beta(g_{ij}(y), g_{jk}(y), g_{kl}(y)), & & \forall y \in U_{ijkl}.
\end{aligned}
\label{eq: Cech object}
\end{equation}
Here, the 3-cocycle $\beta$ is normalized so that $\beta(g, h, k) = 1$ if any of $g$, $h$, and $k$ is the unit element of $\underline{\Fun}_{\cC}(\cM, \cM)^{\text{inv}}$.
This normalization is always possible without changing the cohomology class of $\beta$.
Two pairs $(\{g_{ij}\}, \{a_{ijk}\})$ and $(\{g^{\prime}_{ij}\}, \{a^{\prime}_{ijk}\})$ are said to be equivalent if and only if there exists a pair $(\{f_i\}, \{\gamma_{ij}\})$ of continuous functions
\begin{equation}
f_i: U_i \rightarrow \Fun_{\cC}(\cM, \cM)^{\text{inv}}, \qquad \gamma_{ij}: U_{ij} \rightarrow \mathrm{U}(1)
\end{equation}
such that 
\begin{equation}
\begin{aligned}
g^{\prime}_{ij}(x) & = f_i(x) g_{ij}(x) f_j(x)^{-1}, & & \forall x \in U_{ij}, \\
a^{\prime}_{ijk}(y) & = a_{ijk}(y) (d \gamma)_{ijk}(y) \zeta_{ijk}(y), & & \forall y \in U_{ijkl},
\end{aligned}
\label{eq: Cech equivalence}
\end{equation}
where 
\begin{equation}
\zeta_{ijk}(y) := \frac{\beta(g^{\prime}_{ij}(y), f_j(y), g_{jk}(y))}{\beta(f_i(y), g_{ij}(y), g_{jk}(y)) \beta(g^{\prime}_{ij}(y), g^{\prime}_{jk}(y), f_k(y))}.
\end{equation}
The function $\zeta$ is called the first descendant of $\beta$ \cite{Kapustin:2013uxa, Thorngren:2018ziu} and satisfies
\begin{equation}
(\delta \zeta)_{ijkl}(y) = \beta(g^{\prime}_{ij}(y), g^{\prime}_{jk}(y), g^{\prime}_{kl}(y)) \beta(g_{ij}(y), g_{jk}(y), g_{kl}(y))^{-1},
\end{equation}
which follows from the cocycle condition on $\beta$.
Summarizing, the \v{C}ech cohomology is given by the quotient
\begin{equation}
\check{\mathrm{H}}^1(X, \underline{\Fun}_{\cC}(\cM, \cM)^{\text{inv}}) = \{(\{g_{ij}\}, \{a_{ijk}\}) \text{ satisfying Eq.~\eqref{eq: Cech object}}\} / \sim,
\label{eq: Cech cohomology}
\end{equation}
where the equivalence relation $(\{g_{ij}\}, \{a_{ijk}\}) \sim (\{g^{\prime}_{ij}\}, \{a^{\prime}_{ijk}\})$ is defined by Eq.~\eqref{eq: Cech equivalence}.
We note that $g_{ij}(x)$ and $f_i(x)$ in Eqs.~\eqref{eq: Cech object} and \eqref{eq: Cech equivalence} are constant on $U_{ij}$ and $U_i$ because $\Fun_{\cC}(\cM, \cM)^{\text{inv}}$ is finite and hence has the discrete topology.

It is worth mentioning that the \v{C}ech cohomology $\check{\mathrm{H}}^1(X, \underline{\Fun}_{\cC}(\cM, \cM)^{\text{inv}})$ depends on the topology of the 1-form group $\mathrm{U}(1)$.
This is because continuous functions involved in the definition~\eqref{eq: Cech cohomology} depend on the topology of $\mathrm{U}(1)$.
Specifically, since $\mathrm{U}(1)$ has the standard topology, the values $\a_{ijk}(y)$ and $\gamma_{ij}(y)$ in Eqs.~\eqref{eq: Cech object} and \eqref{eq: Cech equivalence} can vary continuously with respect to $y$.
On the other hand, if we replace $\mathrm{U}(1)$ with $\mathrm{U}(1)_{\delta}$, the functions $a_{ijk}(y)$ and $\gamma_{ij}(y)$ have to be constant on $U_{ijk}$ and $U_{ij}$.

When the Postnikov class $[\beta] \in \mathrm{H}_{\text{gp}}^3(\Fun_{\cC}(\cM, \cM)^{\text{inv}}, \mathrm{U}(1))$ is trivial, the functions $g_{ij}$ and $a_{ijk}$ in Eq.~\eqref{eq: Cech object} are independent of each other.
As a result, the \v{C}ech cohomology~\eqref{eq: Cech cohomology} is decomposed into the direct sum 
\begin{equation}
\begin{aligned}
\check{\mathrm{H}}^1(X, \underline{\Fun}_{\cC}(\cM, \cM)^{\text{inv}}) & = \check{\mathrm{H}}^1(X, \Fun_{\cC}(\cM, \cM)^{\text{inv}}) \oplus \check{\mathrm{H}}^2(X, \mathrm{U}(1)) \\
& = \check{\mathrm{H}}^1(X, \Fun_{\cC}(\cM, \cM)^{\text{inv}}) \oplus \check{\mathrm{H}}^3(X, \Z).
\end{aligned}
\label{eq: trivial Postnikov}
\end{equation}
Here, the second equality follows from the identity
\begin{equation}
\check{\mathrm{H}}^2(X, \mathrm{U}(1)) = [X, B^2\mathrm{U}(1)] = [X, K(\Z, 3)] = \check{\mathrm{H}}^3(X, \Z),
\end{equation}
where $K(\Z, 3) = B^3 \Z = B^2 \mathrm{U}(1)$ is the third Eilenberg-MacLane space of $\mathbb{Z}$.
We note that the second line of Eq.~\eqref{eq: trivial Postnikov} can also be regarded as the singular cohomology because the \v{C}ech cohomology with coefficients in a discrete group agrees with the singular cohomology with the same coefficient.

The \v{C}ech cohomology $\check{\mathrm{H}}^1(X, \underline{\Fun}_{\cC}(\cM, \cM)^{\text{inv}})$ does not necessarily have a group structure when $X$ is a general parameter space.
This is because there is no natural way to define the product of maps from $X$ to $B\underline{\Fun}_{\cC}(\cM, \cM)^{\text{inv}}$ in general.

\vspace{10pt}
\noindent{\bf $S^k$-parameterized families.}
Let us apply Conjecture~\ref{con: general parameter space} to the case where the parameter space $X$ is a $k$-dimensional sphere $S^k$.
When $k=1$, the above explicit description of the \v{C}ech cohomology shows that $\check{\mathrm{H}}^1(S^1, \underline{\Fun}_{\cC}(\cM, \cM)^{\text{inv}})$ is given by $\Fun_{\cC}(\cM, \cM)^{\text{inv}}$.
To compute the \v{C}ech cohomology of $S^{k \geq 2}$, we note that one can always take $\{g_{ij}\}$ to be trivial when $k \geq 2$ because, for finite group $G$, any principal $G$-bundle over $S^{k \geq 2}$ can be trivialized.
Therefore, the \v{C}eck cohomology $\check{\mathrm{H}}^1(S^k, \underline{\Fun}_{\cC}(\cM, \cM)^{\text{inv}})$ for $k \geq 2$ reduces to the second \v{C}ech cohomology of $S^k$ with coefficients in $\mathrm{U}(1)$.
Since $\mathrm{U}(1)$ is equipped with the standard topology, the second \v{C}ech cohomology is given by
\begin{equation}
\check{\mathrm{H}}^2(S^k, \mathrm{U}(1)) = \check{\mathrm{H}}^3(S^k, \Z) = \delta_{k, 3} \Z.
\label{eq: homotopy groups}
\end{equation}
Thus, we find
\begin{equation}
\check{\mathrm{H}}^1(S^k, \underline{\Fun}_{\cC}(\cM, \cM)^{\text{inv}}) = \pi_k(B \underline{\Fun}_{\cC}(\cM, \cM)^{\text{inv}}) = 
\begin{cases}
\Fun_{\cC}(\cM, \cM)^{\text{inv}} \quad & k = 1,\\
\Z \quad & k = 3,\\
0 \quad & k \neq 1, 3.
\end{cases}
\end{equation}
We note that the \v{C}ech cohomology of $S^3$ is given by $\Z$ regardless of $\cC$ and $\cM$, meaning that $S^3$-parameterized families of any $\cC$-symmetric gapped phase $\cM$ are classified by $\Z$.

We mention that replacing $\mathrm{U}(1)$ with $\mathrm{U}(1)_{\delta}$ leads us to a different result.
Specifically, the \v{C}ech cohomology of $S^k$ with coefficients in $\underline{\Fun}_{\cC}(\cM, \cM)^{\text{inv}}_{\delta}$ equipped with the discrete topology is given by
\begin{equation}
\check{\mathrm{H}}^1(S^k, \underline{\Fun}_{\cC}(\cM, \cM)^{\text{inv}}_{\delta}) = \pi_k(B \underline{\Fun}_{\cC}(\cM, \cM)^{\text{inv}}_{\delta}) = 
\begin{cases}
\Fun_{\cC}(\cM, \cM)^{\text{inv}} \quad & k = 1,\\
\mathrm{U}(1) \quad & k = 2,\\
0 \quad & k \geq 3.
\end{cases}
\end{equation}
This result is inconsistent with the known classification of invertible states in 1+1d \cite{Kitaev2011SCGP, Kitaev2013SCGP, Kitaev2015IPAM}.
Thus, one should employ the standard topology of $\mathrm{U}(1)$ rather than the discrete one.
Choosing a suitable topology of $\mathrm{U}(1)$ is analogous to choosing a suitable dual (i.e., the Anderson dual) of the bordism group in the classification of SPT phases \cite{Kapustin:2014tfa, Kapustin:2014dxa, Freed:2016rqq}.\footnote{
  The topology of $\mathrm{U}(1)$ needs to be set according to the specific objectives. 
  For example, in classifications using effective field theory, 
  the non-triviality of $\mathrm{U}(1)$ is detected through the partition function. 
  If one considers the values of the partition function to have physical significance, 
  it is natural to use the discrete topology 
  \cite{Hsin:2022iug}.
  On the other hand, in analyses using lattice systems, $\mathrm{U}(1)$ appears as an unphysical phase. 
  Therefore, it is natural to endow $\mathrm{U}(1)$ with the standard topology.
}

\subsubsection{Examples}
\noindent{\bf Example 1: Trivial phase without symmetry.}
As the simplest example, let us consider parameterized families of trivial states without symmetry.
The symmetry category in this case is given by $\cC = \Vect$, and the trivial phase corresponds to the $\Vect$-module category $\cM = \Vect$.
Thus, the functor category $\underline{\Fun}_{\cC}(\cM, \cM)$ is equivalent to $\Vect$.
In particular, the 2-group $\underline{\Fun}_{\cC}(\cM, \cM)^{\text{inv}}$ consists only of a 1-form group $\mathrm{U}(1)$.
Consequently, the classification of $X$-parameterized families of trivial states is given by
\begin{equation}
\check{\mathrm{H}}^1(X, \underline{\Vect}^{\text{inv}}) = \check{\mathrm{H}}^3(X, \Z).
\end{equation}
This agrees with Kitaev's proposal for the classification of invertible states in 1+1d \cite{Kitaev2011SCGP, Kitaev2013SCGP, Kitaev2015IPAM}.
The invariant of the classification is known as a higher Berry phase \cite{Kitaev2019UTA, Kapustin:2020eby, Hsin:2020cgg, Artymowicz:2023erv, Cordova:2019jnf, Cordova:2019uob, Choi:2022odr, Wen:2021gwc, Ohyama:2023suc, Qi:2023ysw, Sommer:2024dtb, Ohyama:2024jsg}.
From our point of view, this classification originates from the fact that symmetry category $\Vect$ of the trivial states contains the $\mathrm{U}(1)$ 1-form symmetry, which exists in any quantum systems in 1+1d.\footnote{In general, any quantum systems in $(n+1)$d have a $\mathrm{U}(1)$ $n$-form symmetry, which is generated by scalar multiples of the identity point operator. Usually, a system only with this symmetry is said to have no symmetry.}

\vspace{10pt}
\noindent{\bf Example 2: SPT phases with $G$ symmetry.}
One can also incorporate a finite group symmetry $G$ in the above example.
A 1+1d SPT phase with symmetry $G$ is labeled by a fiber functor of $\Vect_G$, or equivalently, a $\Vect_G$-module category with a single simple object.
When the simple object of a $\Vect_G$-module category $\cM$ is unique, the functor category $\Fun_{\Vect_G}(\cM, \cM)$ is equivalent to $\Rep(G)$ as a fusion category \cite{EGNO2015}.
Therefore, the group of invertible objects of $\Fun_{\Vect_G}(\cM, \cM)$ is isomorphic to an abelian group $\mathrm{H}_{\text{gp}}^1(G, \mathrm{U}(1))$, i.e., the group of one-dimensional representations of $G$.
Furthermore, since $\Rep(G)$ is non-anomalous, its invertible subcategory is also non-anomalous, meaning that the $F$-symbols of this subcategory are trivial.
Thus, the Postnikov class $\beta$ of the 2-group $\underline{\Fun}_{\cC}(\cM, \cM)^{\text{inv}}$ is trivial.
Consequently, the \v{C}ech cohomology splits into the direct sum
\begin{equation}
\check{\mathrm{H}}^1(X, \underline{\Rep(G)}^{\text{inv}}) = \check{\mathrm{H}}^1(X, \mathrm{H}_{\text{gp}}^1(G, \mathrm{U}(1))) \oplus \check{\mathrm{H}}^3(X, \Z).
\label{eq: X-family of G-SPT}
\end{equation}
This agrees with the known classification of $X$-parameterized families of 1+1d SPT states with symmetry $G$~\cite{Thorngren:1612.00846, Hermele2021CMSA, Thorngren2021YITP, Hsin:2020cgg,Kapustin:2020mkl}.

\vspace{10pt}
\noindent{\bf Example 3: SPT phases with $\Rep(G)$ symmetry.}
As another example, let us consider parameterized families of $\Rep(G)$ SPT states labeled by the forgetful functor of $\Rep(G)$.
The $\Rep(G)$-module category corresponding to this phase is $\Vect$, on which $\Rep(G)$ acts via the forgetful functor.
The functor category $\Fun_{\Rep(G)}(\Vect, \Vect)$ is equivalent to $\Vect_G$ \cite{EGNO2015}, and hence the group of its invertible objects is isomorphic to $G$.
Furthermore, since $\Vect_G$ is non-anomalous, its $F$-symbols are trivial, meaning that the Postnikov class of $\underline{\Fun}_{\Rep(G)}(\Vect, \Vect)^{\text{inv}}$ is trivial.
Therefore, Conjecture~\ref{con: general parameter space} implies that $X$-parameterized families are classified by
\begin{equation}
\check{\mathrm{H}}^1(X, \underline{\Vect}_G^{\text{inv}}) = \check{\mathrm{H}}^1(X, G) \oplus \check{\mathrm{H}}^3(X, \Z).
\label{eq: X-family of Rep(G) SPT}
\end{equation}
We note that $\check{\mathrm{H}}^1(X, G)$ does not admit a group structure for generic $X$ when $G$ is non-abelian.
Equation~\eqref{eq: X-family of Rep(G) SPT} implies that $S^1$-parameterized families are classified by $G$.
Concrete models of such families were constructed in Sec.~\ref{sec: Example: Rep(G) symmetry II} by using the $G \times \Rep(G)$ cluster state.

\vspace{10pt}
\noindent{\bf Example 4: SPT phases with $\Rep(H)$ symmetry.}
More generally, for any finite-dimensional semisimple Hopf algebra $H$, its representation category $\Rep(H)$ has an SPT phase labeled by the forgetful functor of $\Rep(H)$.
The corresponding $\Rep(H)$-module category is $\Vect$, on which $\Rep(H)$ acts via the forgetful functor.
The functor category $\Fun_{\Rep(H)}(\Vect, \Vect)$ is equivalent to $\Rep(H^*)$, where $H^*$ is the dual of $H$ \cite[Example 7.12.26]{EGNO2015}.
In particular, the group of invertible objects of $\Fun_{\Rep(H)}(\Vect, \Vect)$ is isomorphic to $G(H)$, the group of group-like elements of $H$.
Since $\Rep(H^*)$ is also non-anomalous, its group-like subcategory has the trivial $F$-symbols, meaning that the 2-group $\underline{\Fun}_{\Rep(H)}(\Vect, \Vect)$ has the trivial Postnikov class.
Therefore, Conjecture \ref{con: general parameter space} suggests that $X$-parameterized families are classified by
\begin{equation}
\check{\mathrm{H}}^1(X, \underline{\Rep(H^*)}^{\text{inv}}) = \check{\mathrm{H}}^1(X, G(H)) \oplus \check{\mathrm{H}}^3(X, \Z).
\end{equation}
In particular, $S^1$-parameterized families are classified by $G(H)$.
As in the case of $\Rep(G)$, it should be possible to construct concrete models of these families by using the Hopf algebra analogue of the $G \times \Rep(G)$ cluster state \cite{Jia:2024bng}.
We leave this construction for the future.

\subsection{Non-chiral SPT phases in 2+1d}
\label{sec: Higher dimensions}
\subsubsection{Conjecture}
One can also consider parameterized families of gapped systems in higher dimensions.
In this subsection, we focus on $S^1$- and $S^2$-parameterized families of 2+1d non-chiral SPT states with finite non-invertible symmetries.\footnote{A gapped phase is said to be non-chiral if and only if it admits gapped boundaries.}
More general (non-chiral) gaped systems parameterized by more general parameter spaces will be discussed in Sec.~\ref{sec: General classification in higher dimensions}.

Finite non-invertible symmetries in 2+1 dimensions are expected to be described by fusion 2-categories.
Physically, a fusion 2-category consists of topological surfaces, topological lines, and topological point operators.
The associativity of their fusion product is captured by a higher analogue of the $F$-symbols, known as the 10-j symbols.
See \cite{Douglas:2018qfz, barrett2024gray} for the precise definition of fusion 2-categories.

It is natural to expect that 2+1d non-chiral SPT phases with fusion 2-category symmetry $\cC$ are classified by fiber 2-functors of $\cC$.
Here, a fiber 2-functor is a tensor 2-functor from $\cC$ to the 2-category $2\Vect$ of finite semisimple 1-categories. 
Our conjecture for the classification of $S^1$- and $S^2$-parameterized families of invertible states in these phases is as follows:
\begin{conjecture}
$S^1$-parameterized families of non-chiral 2+1d SPT states labeled by a fiber 2-functor $F: \cC \rightarrow 2\Vect$ are classified by the group of isomorphism classes of monoidal natural automorphisms of $F$. Similarly, $S^2$-parameterized families of these SPT states are classified by the group of invertible monoidal modifications of the identity natural transformation of $F$.
\label{con: 2+1d SPT}
\end{conjecture}
We refer the reader to \cite{gordon1995coherence} \cite[Chapter 2]{schommerpries2014classification} for the definitions of monoidal 2-functors, monoidal natural transformations of monoidal 2-functors, and monoidal modifications of monoidal natural transformations.

Our conjecture on $S^1$-parameterized families is motivated by a similar argument to the case of 1+1d.
Since an SPT phase is labeled by a fiber 2-functor $F$, an $S^1$-parameterized family of SPT states in the same phase gives rise to an $S^1$-parameterized family of fiber 2-functors, all of which are isomorphic to $F$.
In particular, going around the parameter space $S^1$ induces a natural automorphism of $F$.
We expect that this automorphism completely characterizes the $S^1$-parameterized family of these SPT states.

To motivate the conjecture on $S^2$-parameterized families, we point out that an $S^2$-family can be viewed as an $S^1$-family of $S^1$-families as illustrated in Fig.~\ref{fig: S2 family}.
\begin{figure}[t]
\includegraphics[scale=1.5, trim={10pt 10pt 10pt 10pt}]{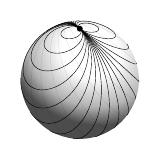}
\caption{An $S^2$-family is viewed as an $S^1$-family of $S^1$-families. Each loop represents an $S^1$-family embedded in the $S^2$-family. The black dot at the top represents a constant family, which we choose to be the initial and final $S^1$-families.}
\label{fig: S2 family}
\end{figure}
The initial and final $S^1$-families are represented by a single dot in the figure, which indicates that these families are constant in $\theta \in S^1$.
Since any loops in $S^2$ are contractible, any $S^1$-familes embedded in the $S^2$-family are trivial.
Consequently, the monoidal natural transformations induced by these $S^1$-families are isomorphic to the identity natural transformation.
In particular, the monoidal natural transformations induced by the initial and final $S^1$-families are the identity because these families are constant.
Therefore, an $S^2$-family, viewed as an $S^1$-family of $S^1$-families, induces a monoidal modification of the identity natural transformation.
We expect that this monoidal modification is an invariant that classifies the $S^2$-parameterized family of 2+1d SPT states.
See \cite{Aasen:2022cdu} for a similar argument in the context of $S^2$-families of 2+1d topological orders.

\subsubsection{Example}
As a simple example, let us consider the case of finite group symmetry $G$.
We start by describing in detail the relevant fusion 2-categorical notions such as fiber 2-functors, monoidal natural transformations, and monoidal modifications.
We will then apply Conjecture~\ref{con: 2+1d SPT} and find that our conjecture reproduces the correct classification of parameterized families of 2+1d SPT states.

As a fusion 2-category, a finite group symmetry $G$ is described by $2\Vect_G$, the 2-category of $G$-graded 2-vector spaces.
Simple objects of $2\Vect_G$ are labeled by elements of $G$ and form a group $G$ under the tensor product.
The Hom categories consisting of 1-morphisms and 2-morphisms are trivial, that is, $\Hom_{2\Vect_G}(g, h) = \delta_{g, h} \Vect$ as a 1-category.
The 10-j symbols of $2\Vect_G$ are also trivial, meaning that the symmetry group $G$ is non-anomalous.\footnote{An anomalous finite group symmetry $G$ is described by $2\Vect_G^{\alpha}$ with the 10-j symbols given by a 4-cocycle $[\alpha] \in \mathrm{H}_{\text{gp}}^4(G, \mathrm{U}(1))$.}

A fiber 2-functor $(F, J, \omega): 2\Vect_G \rightarrow 2\Vect$ consists of 1-isomorphisms $J_{g, h}: F(g) \otimes F(h) \rightarrow F(gh)$ and 2-isomorphisms $\omega_{g, h, k}$ fitting into the following diagram:
\begin{equation}
\begin{tikzcd}
\adjincludegraphics[valign=c, scale=1]{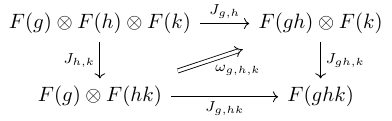}
\end{tikzcd}
\label{eq: fiber 2-functor}
\end{equation}
The identity 1-morphisms are omitted in the above diagram.
The coherence condition on $\omega_{g, h, k}$ is given by the equality of the following two diagrams:
\begin{equation}
\adjincludegraphics[valign=c, scale=1]{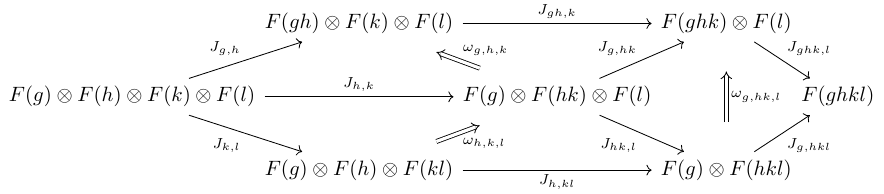}
\label{eq: fiber 2-functor coherence1}
\end{equation}
\begin{equation}
\adjincludegraphics[valign=c, scale=1]{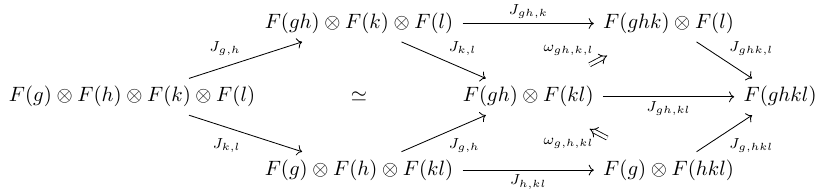}
\label{eq: fiber 2-functor coherence2}
\end{equation}
The 2-isomorphism $\omega_{g, h, k}$ in Eq.~\eqref{eq: fiber 2-functor} can be identified with a complex number $\omega(g, h, k) \in \mathrm{U}(1)$ because the source and target 1-morphisms of $\omega_{g, h, k}$ are invertible and hence simple.\footnote{The space of 2-morphisms between simple 1-morphisms is isomorphic to $\C$ \cite{Douglas:2018qfz}.}
With this identification, the equality of two diagrams in Eqs.~\eqref{eq: fiber 2-functor coherence1} and \eqref{eq: fiber 2-functor coherence2} leads to
\begin{equation}
\omega(h, k, l) \omega(g, hk, l) \omega(g, h, k) = \omega(gh, k, l) \omega(g, h, kl) \Leftrightarrow \delta \omega = 1.
\label{eq: 3-cocycle omega}
\end{equation}
Thus, the 2-isomorphism $\omega$ associated with a fiber 2-functor defines a group 3-cocycle $\omega \in \mathrm{Z}_{\text{gp}}^3(G, \mathrm{U}(1))$.

Given two fiber 2-functors $(F, J, \omega)$ and $(F^{\prime}, J^{\prime}, \omega^{\prime})$, a monoidal natural isomorphism $(\eta, \nu): (F, J, \omega) \Rightarrow (F^{\prime}, J^{\prime}, \omega^{\prime})$ consists of 1-isomorphisms $\eta_g: F(g) \rightarrow F^{\prime}(g)$ and 2-isomorphisms $\nu_{g, h}$ fitting into the following diagram:
\begin{equation}
\adjincludegraphics[valign=c, scale=1]{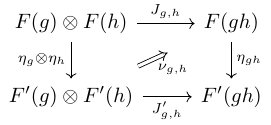}
\label{eq: natural transformation 2}
\end{equation}
The coherence condition on $\nu$ is given by the equality of the following two diagrams:
\begin{equation}
\adjincludegraphics[valign=c, scale=1]{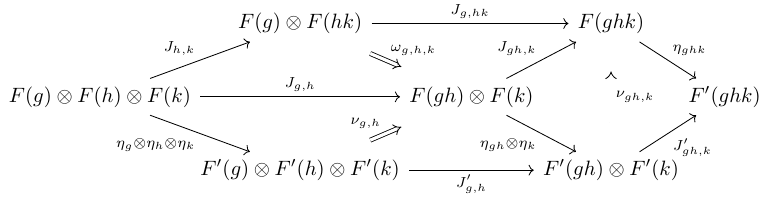}
\label{eq: nu coherence 1}
\end{equation}
\begin{equation}
\adjincludegraphics[valign=c, scale=1]{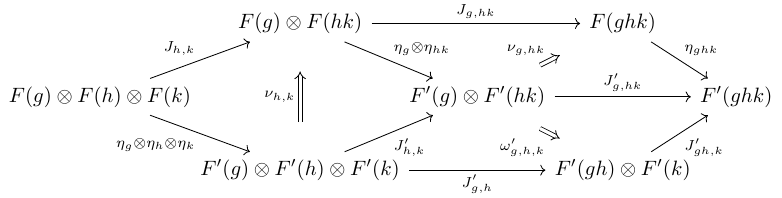}
\label{eq: nu coherence 2}
\end{equation}
The 2-isomorphism $\nu_{g, h}$ in Eq.~\eqref{eq: natural transformation 2} can be identified with a complex number $\nu(g, h) \in \mathrm{U}(1)$ just as in the case of $\omega_{g, h, k}$.
With this identification, the equality of the two diagrams in Eqs.~\eqref{eq: nu coherence 1} and \eqref{eq: nu coherence 2} reduces to
\begin{equation}
\nu(g, h) \nu(gh, k) \omega(g, h, k)^{-1} = \nu(h, k) \nu(g, hk) \omega^{\prime}(g, h, k)^{-1} \Leftrightarrow \omega^{\prime} = \omega \delta \nu.
\label{eq: equivalence of omega}
\end{equation}
Thus, fiber 2-functors $(F, J, \omega)$ and $(F^{\prime}, J^{\prime}, \omega^{\prime})$ are isomorphic to each other if and only if $\omega$ and $\omega^{\prime}$ are in the same cohomology class, i.e., $[\omega] = [\omega^{\prime}] \in \mathrm{H}_{\text{gp}}^3(G, \mathrm{U}(1))$.

Given two monoidal natural isomorphisms $(\eta, \nu), (\eta^{\prime}, \nu^{\prime}): (F, J, \omega) \Rightarrow (F^{\prime}, J^{\prime}, \omega^{\prime})$, an invertible monoidal modification $\lambda: (\eta, \nu) \Rrightarrow (\eta^{\prime}, \nu^{\prime})$ consists of 2-isomorphisms $\lambda_g: \eta_g \Rightarrow \eta^{\prime}_g$ such that
\begin{equation}
\adjincludegraphics[valign=c, scale=1]{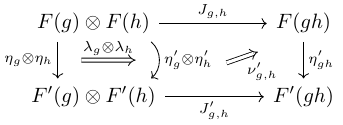}
~ = ~
\adjincludegraphics[valign=c, scale=1]{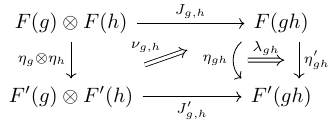}
\end{equation}
By identifying the 2-isomorphism $\lambda_g$ with a complex number $\lambda(g) \in \mathrm{U}(1)$, the above equation reduces to
\begin{equation}
\nu^{\prime}(g, h) \lambda(g) \lambda(h) = \nu(g, h) \lambda(gh) \Leftrightarrow \nu^{\prime} = \nu / \delta \lambda.
\label{eq: equivalence of nu}
\end{equation}

Now, let us apply Conjecture~\ref{con: 2+1d SPT} to 2+1d SPT phases with symmetry $G$.
First of all, Eqs.~\eqref{eq: 3-cocycle omega} and \eqref{eq: equivalence of omega} show that isomorphism classes of fiber 2-functors of $2\Vect_G$ are classified by the third group cohomology $\mathrm{H}_{\text{gp}}^3(G, \mathrm{U}(1))$.
This agrees with the well-known classification of 2+1d SPT phases \cite{Chen:2011bcp, Chen:2011pg}.

When $\omega = \omega^{\prime}$, Eq.~\eqref{eq: equivalence of omega} reduces to $\delta \nu = 1$, which implies that a monoidal natural automorphism of a fiber 2-functor $(F, J, \omega)$ is associated with a group 2-cocycle $\nu$.
Furthermore, Eq.~\eqref{eq: equivalence of nu} implies that natural automorphisms $(\eta, \nu)$ and $(\eta^{\prime}, \nu^{\prime})$ are isomorphic to each other if and only if the associated 2-cocycles are equivalent.
Therefore, isomorphism classes of monoidal natural automorphisms of a fiber 2-functor $(F, J, \omega)$ are classified by the second group cohomology $\mathrm{H}_{\text{gp}}^2(G, \mathrm{U}(1))$.
This agrees with the known classification of $S^1$-parameterized families of $G$-SPT states in 2+1d~\cite{Thorngren:1612.00846, Hermele2021CMSA, Thorngren2021YITP, Shiozaki:2021weu, Hsin:2020cgg}.

When $\nu = \nu^{\prime}$, Eq.~\eqref{eq: equivalence of nu} reduces to $\delta \lambda = 1$, meaning that an invertible modification of a natural isomorphism $(\eta, \nu)$ is associated with a 1-cocycle $\lambda$, i.e., a group homomorphism $\lambda: G \rightarrow \mathrm{U}(1)$.
In particular, the group of invertible modifications of the identity natural transformation of a fiber 2-functor $(F, J, \omega)$ is the first group cohomology $\mathrm{H}_{\text{gp}}^1(G, \mathrm{U}(1))$.
Thus, Conjecture~\ref{con: 2+1d SPT} suggests that $S^2$-parameterized families of $G$-SPT states in 2+1d are classified by $\mathrm{H}_{\text{gp}}^1(G, \mathrm{U}(1))$.
This also agrees with the known classification~\cite{Thorngren:1612.00846, Hermele2021CMSA, Thorngren2021YITP, Hsin:2020cgg}.

More generally, $X$-parameterized families of $G$-SPT states in 2+1d are known to be classified by~\cite{Thorngren:1612.00846, Hermele2021CMSA, Thorngren2021YITP, Hsin:2020cgg}
\begin{equation}
\check{\mathrm{H}}^1(X, \mathrm{H}_{\text{gp}}^2(G, \mathrm{U}(1))) \oplus \check{\mathrm{H}}^2(X, \mathrm{H}_{\text{gp}}^1(G, \mathrm{U}(1))) \oplus \check{\mathrm{H}}^4(X, \Z).
\end{equation}
We will reproduce this classification as a special case of our general conjecture in the next subsection.

\subsection{General non-chiral gapped phases in 2+1d}
\label{sec: General classification in higher dimensions}
\subsubsection{Conjecture}
Although we only discussed $S^1$- and $S^2$-parameterized families of 2+1d non-chiral SPT states above, it is natural to expect that a similar conjecture holds for more general gapped systems parameterized by more general parameter spaces.
In this subsection, we consider general parameterized families of general non-chiral gapped systems in 2+1d.

As in the 1+1d case, 2+1d non-chiral gapped phases with fusion 2-category symmetry $\cC$ are expected to be classified by module 2-categories over $\cC$.
For these phases, we propose the following conjecture on the classification of parameterized families:
\begin{conjecture}
$X$-parameterized families of 2+1d $\cC$-symmetric gapped systems in a non-chiral gapped phase labeled by a $\cC$-module 2-category $\cM$ are classified by the set $[X, B \underline{\underline{\Fun}}_{\cC}(\cM, \cM)^{\text{inv}}]$ of the homotopy classes of maps from $X$ to the classifying space of $\underline{\underline{\Fun}}_{\cC}(\cM, \cM)^{\text{inv}}$.
Here, $\underline{\underline{\Fun}}_{\cC}(\cM, \cM)^{\text{inv}}$ is a 3-group consisting of invertible objects, invertible 1-morphisms, and invertible 2-morphisms of the fusion 2-category $\Fun_{\cC}(\cM, \cM)$ of $\cC$-module 2-endofunctors of $\cM$. Equivalently, $X$-parameterized families of these gapped systems are classified by the non-abelian \v{C}ech cohomology $\check{\mathrm{H}}^1(X, \underline{\underline{\Fun}}_{\cC}(\cM, \cM)^{\text{inv}})$.
\label{con: general parameter space 2}
\end{conjecture}
We note that the above conjecture is a natural generalization of the conjecture for the 1+1d case, cf. Conjecture~\ref{con: general parameter space}.
A formal generalization to higher dimensions should be clear.
However, we do not state the conjecture in higher dimensions here, partly because the theory of fusion $n$-categories and related notions for $n \geq 3$ are still under development despite significant advances in recent years \cite{Gaiotto:2019xmp, Johnson_Freyd_2022, Kong:2020iek}.

The conjectural one-to-one correspondence between 2+1d non-chiral gapped phases with symmetry $\cC$ and $\cC$-module 2-categories is supported by the explicit construction of lattice models in \cite{Inamura:2023qzl}.
As discussed in~\cite{Inamura:2023qzl}, one can construct a $\cC$-symmetric commuting projector model from each $\cC$-module 2-category.
The gapped phase realized in this model should be non-chiral because the Hamiltonian is given by the sum of local commuting projectors.
Thus, the lattice models in~\cite{Inamura:2023qzl} relate $\cC$-module 2-categories to $\cC$-symmetric non-chiral gapped phases in 2+1d.
For example, when the symmetry $\cC$ is trivial, i.e., when $\cC = 2\Vect$, an indecomposable module 2-category over $\cC$ is given by the 2-category $\Mod(\cD)$ of module categories over a fusion 1-category $\cD$~\cite{Decoppet2023Morita}.
The commuting projector model of \cite{Inamura:2023qzl} constructed from $\Mod(\cD)$ is the Levin-Wen model~\cite{Levin:2004mi} constructed from $\cD$.
This model exhibits a topological order whose anyon contents are described by the Drinfeld center $Z(\cD)$ of $\cD$.
Since the Drinfeld center of a fusion category describes the most general non-chiral topological order in 2+1d \cite{Fuchs:2012dt, Freed:2020qfy}, the aforementioned one-to-one correspondence is established at least in the case of the trivial symmetry.\footnote{General 2+1d topological orders up to invertible ones are believed to be classified by modular tensor categories, which are in one-to-one correspondence with Lagrangian algebras in the Drinfeld center $Z(2\Vect)$ of $2\Vect$ \cite{Decoppet:2023rlx}.}

In what follows, we provide several examples that verify the classification of parameterized families stated in Conjecture~\ref{con: general parameter space 2}

\subsubsection{Examples}
\noindent{\bf Example 1: trivial phase without symmetry.}
The simplest example is the parameterized families of trivial states without symmetry.
From the point of view of module 2-categories, the trivial phase corresponds to the regular module $\cM = 2\Vect$ over $\cC = 2\Vect$.
The associated functor 2-category $\Fun_{2\Vect}(2\Vect, 2\Vect)$ is also equivalent to $2\Vect$.
The 3-group $\underline{\underline{\Fun}}_{2\Vect}(2\Vect, 2\Vect)^{\text{inv}} = \underline{\underline{2\Vect}}^{\text{inv}}$ consists only of the 2-form group $\mathrm{U}(1)$.
We denote this 3-group by $\mathrm{U}(1)^{[2]}$, where the superscript specifies the form degree.
Applying Conjecture~\ref{con: general parameter space 2} to this example suggests that $X$-parameterized families of the trivial states without symmetry are classified by
\begin{equation}
\check{\mathrm{H}}^1(X, \underline{\underline{2\Vect}}^{\text{inv}}) = \check{\mathrm{H}}^{1}(X, \mathrm{U}(1)^{[2]}) = [X, K(\Z, 4)] = \check{\mathrm{H}}^4(X, \mathbb{Z}),
\end{equation}
where $K(\Z, 4) = B^3 \mathrm{U}(1)$ is the fourth Eilenberg-MacLance space of $\Z$.
Here, $\mathrm{U}(1)$ is supposed to have the standard topology rather than the discrete one.
The above result agrees with the known classification of parameterized families of trivial invertible states \cite{Kitaev2011SCGP, Kitaev2013SCGP, Kitaev2015IPAM} via higher Berry phases \cite{Qi:2023ysw,Kapustin:2020eby,Hsin:2020cgg,Ohyama:2024ytt, Sommer:2024lzp}.
As in 1+1d, this classification originates from the fact that the ``trivial" fusion 2-category $2\Vect$ captures a $\mathrm{U}(1)$ 2-form symmetry, which exists in any quantum systems in 2+1d.

\vspace{10pt}
\noindent{\bf Example 2: SPT phases with $G$ symmetry.}
Another simple example is the parameterized families of SPT states with finite group symmetry $G$.
For simplicity, here we only consider $G$-symmetric trivial states.\footnote{The classification of parameterized families should not depend on the choice of an SPT phase because stacking a $G$-SPT phase is an invertible operation.}
From the point of view of module 2-categories, the trivial $G$-symmetric phase corresponds to a $2\Vect_G$-module 2-category $2\Vect$.
The associated functor 2-category is $\Fun_{2\Vect_G}(2\Vect, 2\Vect) \cong 2\Rep(G)$, the category of 2-representations of $G$ \cite{Decoppet2023Morita}.
Invertible objects and invertible 1-morphisms of $2\Rep(G)$ form $\mathrm{H}_{\text{gp}}^2(G, \mathrm{U}(1))$ and $\mathrm{H}_{\text{gp}}^1(G, \mathrm{U}(1))$ respectively, see e.g. \cite{Bhardwaj:2022lsg, Bartsch:2022mpm, Delcamp:2023kew} for a detailed description of $2\Rep(G)$ for physicists.
Thus, the 3-group $\underline{\underline{2\Rep(G)}}^{\text{inv}}$ consists of a 0-form group $\mathrm{H}_{\text{gp}}^2(G, \mathrm{U}(1))$, a 1-form group $\mathrm{H}_{\text{gp}}^1(G, \mathrm{U}(1))$, and a 2-form group $\mathrm{U}(1)$.
The 3-group structure on $\underline{\underline{2\Rep(G)}}^{\text{inv}}$ is trivial, meaning that the actions of the 0-form group on the 1-form and 2-form groups are trivial, and the Postnikov classes are also trivial.\footnote{The invertible objects and invertible 1-morphisms of $2\Rep(G)$ are known as gauged SPT defects in the context of finite group gauge theories \cite{Barkeshli:2022wuz, Barkeshli:2022edm}. These gauged SPT defects from an anomaly-free 2-group symmetry $\mathrm{H}_{\text{gp}}^2(G, \mathrm{U}(1))^{[0]} \times \mathrm{H}_{\text{gp}}^1(G, \mathrm{U}(1))^{[1]}$ with the trivial action of the 0-form group and the trivial Postnikov class \cite{Barkeshli:2022wuz}. Put differently, $\underline{\underline{2\Rep(G)}}^{\text{inv}}$ has the trivial $3$-group structure. The authors thank Ryohei Kobayashi for discussions on this point.}
Therefore, Conjecture~\ref{con: general parameter space 2} suggests that $X$-parameterized families of trivial states with $G$ symmetry are classified by
\begin{equation}
\begin{aligned}
\check{\mathrm{H}}^1(X, \underline{\underline{2\Rep(G)}}^{\text{inv}}) & = \check{\mathrm{H}}^1(X, \mathrm{H}_{\text{gp}}^2(G, \mathrm{U}(1))^{[0]} \times \mathrm{H}_{\text{gp}}^1(G, \mathrm{U}(1))^{[1]} \times \mathrm{U}(1)^{[2]})\\
& = \check{\mathrm{H}}^1(X, \mathrm{H}_{\text{gp}}^2(G, \mathrm{U}(1))) \oplus \check{\mathrm{H}}^2(X, \mathrm{H}_{\text{gp}}^1(G, \mathrm{U}(1))) \oplus \check{\mathrm{H}}^4(X, \mathrm{U}(1)).
\end{aligned}
\end{equation}
This agrees with the known classification of parameterized families of 2+1d $G$-SPT states~\cite{Thorngren:1612.00846, Hermele2021CMSA, Thorngren2021YITP, Hsin:2020cgg, Kapustin:2020mkl}.

\vspace{10pt}
\noindent{\bf Example 3: SPT phases with $2\Rep(G)$ symmetry.}
Yet another example is the parameterized families of SPT states with $2\Rep(G)$ symmetry.
As opposed to the case of $G$ symmetry, the set of SPT phases with $2\Rep(G)$ symmetry does not have a natural group structure in general.
As a result, the classification of parameterized families depends on the phases.
Here, for simplicity, we consider families of $2\Rep(G)$-symmetric trivial states, which are obtained by gauging $G$-symmetry-breaking states.

From the point of view of module 2-categories, the trivial phase with $2\Rep(G)$ symmetry corresponds to a $2\Rep(G)$-module 2-category $2\Vect$.
The associated functor 2-category is $\Fun_{2\Rep(G)}(2\Vect, 2\Vect) \cong 2\Vect_G$ \cite{Decoppet2023Morita}.
Thus, the 3-group $\underline{\underline{\Fun}}_{2\Rep(G)}(2\Vect, 2\Vect)^{\text{inv}}$ consists only of the 0-form group $G$ and the 2-form group $\mathrm{U}(1)$.
The 3-group structure is again trivial.
Therefore, Conjecture~\ref{con: general parameter space 2} suggests that $X$-parameterized families of trivial states with $2\Rep(G)$ symmetry are classified by
\begin{equation}
\check{\mathrm{H}}^1(X, \underline{\underline{2\Vect_G}}^{\text{inv}}) = \check{\mathrm{H}}^1(X, G^{[0]} \times \mathrm{U}(1)^{[2]}) = \check{\mathrm{H}}^1(X, G) \oplus \check{\mathrm{H}}^4(X, \Z).
\end{equation}
We note that parameterized families of $G$-SSB states in 2+1d are also classified by the same cohomology group because $G$-SSB states and $2\Rep(G)$-symmetric trivial states are related by gauging.

\vspace{10pt}
\noindent{\bf Example 4: topological orders without symmetry.}
Finally, we consider the parameterized families of topologically ordered states with the trivial symmetry $\cC = 2\Vect$.
As already mentioned, an indecomposable module 2-category over $2\Vect$ is given by the 2-category $\Mod(\cD)$ of module categories over a fusion 1-category $\cD$ \cite{Decoppet2023Morita}.
The associated functor 2-category is $\Fun_{2\Vect}(\Mod(\cD), \Mod(\cD)) \cong \Mod(Z(\cD))$ \cite[Example 5.3.7]{Decoppet2023Morita}, where $Z(\cD)$ denotes the Drinfeld center of $\cD$.
The 3-group consisting of invertible objects, invertible 1-morphisms, and invertible 2-morphisms of $\Mod(Z(\cD))$ is known as the categorical Picard 2-group $\underline{\underline{\text{Pic}}}(Z(\cD))$ of $Z(\cD)$ \cite{ENO2010}.
Thus, we have
\begin{equation}
\underline{\underline{\Fun}}_{2\Vect}(\Mod(\cD), \Mod(\cD))^{\text{inv}} = \underline{\underline{\text{Pic}}}(Z(\cD)).
\end{equation}
Therefore, Conjecture~\ref{con: general parameter space 2} suggests that $X$-parameterized families of gapped systems in a topologically ordered phase $Z(\cD)$ are classified by the set $[X, B \underline{\underline{\text{Pic}}}(Z(\cD))] = \check{\mathrm{H}}^1(X, \underline{\underline{\text{Pic}}}(Z(\cD)))$.
This agrees with the conjecture in \cite{Aasen:2022cdu}.

\section{Future Directions}

In this paper, we developed the analysis of interface modes and Thouless pumps in $\mathcal{C}$-SPT phases within the framework of MPS representations.
Specifically, by investigating the properties of interface modes, we generalized the bulk-boundary correspondence for usual SPT phases, 
and by investigating Thouless pumps, we revealed pumping phenomena with respect to non-abelian charges. 
In addition, through categorical considerations, we proposed several conjectures for the classification of parameterized families of general gapped phases in both 1+1d and higher dimensions. 
The next task to complete is verifying these conjectures using tensor networks and constructing models that realize their classifications.
We leave them for future studies.

\begin{acknowledgments}
We thank Ryohei Kobayashi and Yosuke Kubota for useful discussions.
We are grateful to the long-term workshop YITP-T-23-01 held at Yukawa Institute for Theoretical Physics, Kyoto University, where a part of this work was done.
S.O. is supported by JSPS KAKENHI Grant Number 23KJ1252 and 24K00522.
\end{acknowledgments}

\appendix

\section{Tambara-Yamagami categories}
\label{sec: Tambara-Yamagami categories}
In this appendix, we review the Tambara-Yamagami categories \cite{TY1998}.
A Tambara-Yamagami category $\TY(A, \chi, \epsilon)$ is specified by a triple $(A, \chi, \epsilon)$, where $A$ is a finite abelian group, $\chi: A \times A \rightarrow \mathrm{U}(1)$ is a symmetric non-degenerate bicharacter of $A$, and $\epsilon \in \{+1, -1\}$ is a sign.
The set of simple objects consists of invertible objects labeled by elements of $A$ and a single non-invertible object $D$.
The fusion rules of these simple objects are given by
\begin{equation}
a \otimes b = ab, \quad a \otimes D = D \otimes a = D, \quad D \otimes D = \bigoplus_{a \in A} a,
\end{equation}
where $a$ and $b$ represent invertible objects labeled by group elements $a, b \in A$.
The above fusion rules imply that the quantum dimension of the non-invertible object $D$ is given by $\dim(D) = \sqrt{|A|}$, while the quantum dimensions of invertible objects are all $1$.
The $F$-symbols in a specific gauge can be written as
\begin{equation}
(F^{aDb}_D)_{DD} = (F^{DaD}_b)_{DD} = \chi(a, b), \quad (F^{DDD}_D)_{a, b} = \frac{\epsilon}{\sqrt{|A|}} \chi(a, b)^{-1}.
\label{eq: TY F}
\end{equation}
The other $F$-symbols are all trivial.

Simple examples of our interest are $\Z_2 \times \Z_2$ Tambara-Yamagami categories $\TY(\Z_2 \times \Z_2, \chi, \epsilon)$.
When $A = \Z_2 \times \Z_2$, there are two choices for a symmetric non-degenerate bicharacter $\chi$, which we denote by $\chi_+$ and $\chi_-$.
These bicharacters are explicitly given by \cite{TY1998}
\begin{equation}
\chi_{\pm}((1, 0), (1, 0)) = \chi_{\pm}((0,1), (0, 1)) = \pm 1, \quad \chi_{\pm}((1, 0), (0, 1)) = \mp 1,
\label{eq: chi pm}
\end{equation}
where $(1, 0)$ and $(0, 1)$ are generators of $\Z_2 \times \Z_2$.
The above equation uniquely determines $\chi_{\pm}(a, b)$ for general $a, b \in \Z_2 \times \Z_2$ because $\chi_{\pm}$ is a symmetric bicharacter.
Corresponding to different choices for $\chi$ and $\epsilon$, there are four inequivalent $\Z_2 \times \Z_2$ Tambara-Yamagami categories.
It is known that $\TY(\Z_2 \times \Z_2, \chi_+, +1)$ is equivalent to the representation category $\Rep(D_8)$ of the dihedral group $D_8$ of order 8.
Specifically, invertible and non-invertible simple objects of $\TY(\Z_2 \times \Z_2, \chi_+, +1)$ correspond to one-dimensional and two-dimensional irreducible representations of $D_8$.
Similarly, $\TY(\Z_2 \times \Z_2, \chi_+, -1)$ is equivalent to the representation category $\Rep(Q_8)$ of the quaternion group $Q_8$, and $\TY(\Z_2 \times \Z_2, \chi_-, +1)$ is equivalent to the representation category $\Rep(H_8)$ of an eight-dimensional semisimple Hopf algebra $H_8$ known as the Kac-Paljutkin algebra \cite{KP1966}.
The above three $\Z_2 \times \Z_2$ Tambara-Yamagami categories are non-anomalous because they are equivalent to representation categories of semisimple Hopf algebras.
Thus, these fusion categories admit SPT phases.
On the other hand, the remaining one $\TY(\Z_2 \times \Z_2, \chi_-, -1)$ is anomalous and hence does not admit SPT phases.
See Table~\ref{tab: TY} for a summary of four $\Z_2 \times \Z_2$ Tambara-Yamagami categories.
\begin{table}[t]
\begin{tabular}{c|c|c}
 & $\epsilon = 1$ & $\epsilon = -1$\\ \hline
$\chi = \chi_+$ & $\Rep(D_8)$ & $\Rep(Q_8)$ \\
$\chi = \chi_-$ & $\Rep(H_8)$ & anomalous
\end{tabular}
\caption{There are four $\Z_2 \times \Z_2$ Tambara-Yamagami categories depending on the choice of $\chi$ and $\epsilon$. Three of them are non-anomalous, while the other one is anomalous. Ref.~\cite{Tam2000} showed that $\Rep(D_8)$ admits three fiber functors, while $\Rep(Q_8)$ and $\Rep(H_8)$ admit only one fiber functor.}
\label{tab: TY}
\end{table}

\section{Fiber functors of Tambara-Yamagami categories}
\label{sec: Fiber functors of Tambara-Yamagami categories}
In this appendix, we review fiber functors of Tambara-Yamagami categories following \cite{Tam2000}.
The fiber functors and the corresponding $L$-symbols that we write down below are derived in a specific gauge where the $F$-symbols of the Tambara-Yamagami categories are given by Eq.~\eqref{eq: TY F}.

We first recall the classification of fiber functors of $\TY(A, \chi, \epsilon)$ \cite{Tam2000}, see also \cite{Thorngren:2019iar}.
Proposition 3.2 of \cite{Tam2000} claims that isomorphism classes of fiber functors of $\TY(A, \chi, \epsilon)$ are in one-to-one correspondence with equivalence classes of triples $(\sigma, \xi, \nu)$, where $\sigma: A \rightarrow A$ is an involutive automorphism of $A$, $\xi: A \times A \rightarrow \mathrm{U}(1)$ is a 2-cocycle on $A$, and $\nu: A \rightarrow \mathrm{U}(1)$ is a $\mathrm{U}(1)$-valued function on $A$, such that 
\begin{equation}
\begin{aligned}
\chi(a, b) & = \xi(a, \sigma(b)) / \xi(\sigma(b), a), \\
\nu(a) \nu(b) / \nu(ab) & = \xi(a, b) / \xi(\sigma(b), \sigma(a)), \\
\nu(a) \nu(\sigma(a)) & = 1, \\
\sum_{\substack{a \in A \\ \text{s.t. } \sigma(a) = a}} \nu(a) & = \epsilon \sqrt{|A|}.
\end{aligned}
\label{eq: fiber functor classification}
\end{equation}
Here, two triples $(\sigma, \xi, \nu)$ and $(\sigma^{\prime}, \xi^{\prime}, \nu^{\prime})$ are said to be equivalent if and only if
\begin{equation}
\sigma^{\prime} = \sigma, \quad \xi^{\prime}(a, b) = \xi(a, b) \zeta(a) \zeta(b) / \zeta(ab), \quad \nu^{\prime}(a) = \nu(a) \zeta(a) / \zeta(\sigma(a))
\end{equation}
for some $\mathrm{U}(1)$-valued function $\zeta: A \rightarrow \mathrm{U}(1)$.
The 2-cocycle $\xi$ can always be normalized so that $\xi(g, 1) = \xi(1, h) = 1$.
In what follows, we suppose that $\xi$ is a normalized 2-cocycle.
We note that $\xi$ has to be non-trivial due to the first equality of Eq.~\eqref{eq: fiber functor classification}.
In particular, the non-degeneracy of $\chi$ implies that the twisted group algebra $\C[A]^{\xi}$ is simple, i.e., it is isomorphic to a matrix algebra $M_{\sqrt{|A|}}(\C)$.

Let us write down a fiber functor in terms of the triple $(\sigma, \xi, \nu)$.
A fiber functor $F$ of a Tambara-Yamagami category $\TY(A, \chi, \epsilon)$ is characterized by the natural isomorphism
\begin{equation}
\begin{aligned}
J_{a, b}&: F(a) \otimes F(b) \rightarrow F(ab), &
J_{a, D}&: F(a) \otimes F(D) \rightarrow F(D), \\
J_{D, a}&: F(D) \otimes F(a) \rightarrow F(D), & 
J_{D, D}&: F(D) \otimes F(D) \rightarrow \bigoplus_{a \in A} F(a).
\end{aligned}
\label{eq: TY fiber functor}
\end{equation}
Here, $F(a) \cong \C$ is a one-dimensional vector space and $F(D) \cong \C^{\sqrt{|A|}}$ is a $\sqrt{|A|}$-dimensional vector space.
Once we choose bases $\{v_a\}$ and $\{v_D^i \mid i = 1, 2, \cdots, \sqrt{|A|}\}$ of $F(a)$ and $F(D)$, the isomorphisms in Eq.~\eqref{eq: TY fiber functor} can be represented by $L$-symbols $\{L^{a, b}, L^{a, D}_{i, j}, L^{D, a}_{i, j}, L^{D, D}_{(i, j), a} \mid a, b \in A, ~ i, j = 1, 2, \cdots, \sqrt{|A|}\}$, that is,
\begin{equation}
\begin{aligned}
J_{a, b} (v_a \otimes v_b) & = L^{a, b} v_{a \otimes b}, &
J_{a, D} (v_a \otimes v_D^i) & = \sum_j L^{a, D}_{i, j} v_D^j, \\
J_{D, a} (v_D^i \otimes v_a) & = \sum_j L^{D, a}_{i, j} v_D^j, &
J_{D, D} (v_D^i \otimes v_D^j) & = \bigoplus_{a \in A} L^{D, D}_{(i, j), a} v_a.
\end{aligned}
\end{equation}
The $L$-symbols in the above equation are related to the triple $(\sigma, \xi, \nu)$ as follows \cite{Tam2000}:
\begin{equation}
\begin{aligned}
L^{a, b} & = \xi(a, b), & 
\sum_{j} L^{a, D}_{i, j} L^{b, D}_{j, k} & = \xi(a, b) L^{ab, D}_{i, k}, \\
L^{D, a}_{i, j} & = \nu(a) L^{\sigma(a), D}_{i, j}, &
L^{D, D}_{(i, j), a} & = \xi(a^{-1}, a)^{-1} \sum_{k} L^{a^{-1}, D}_{i, k} \gamma_{kj}.
\end{aligned}
\label{eq: TY L-symbols}
\end{equation}
Here, $\gamma_{kj}$ in the last equation is the $(k, j)$-component of a non-degenerate $\sqrt{|A|} \times \sqrt{|A|}$ matrix that satisfies
\begin{equation}
L^{a, D} \gamma = \nu(a) \gamma (L^{\sigma(a), D})^T,
\label{eq: gamma}
\end{equation}
where $L^{a, D}$ is the $\sqrt{|A|} \times \sqrt{|A|}$ matrix whose $(i, j)$-component is the $L$-symbol $L^{a, D}_{i, j}$, and $(L^{a, D})^T$ is its transpose.
The matrix $\gamma$ that satisfies the above equation is unique up to scalar multiplication \cite{Tam2000}.

The second equation on the first line of Eq.~\eqref{eq: TY L-symbols} implies that $L^{a, D}$ is the representation matrix of a $\sqrt{|A|}$-dimensional representation of the twisted group algebra $\C[A]^{\xi}$.
Such a representation is unique up to unitary equivalence because $\C[A]^{\xi}$ is a simple algebra as already mentioned.
Thus, $L^{a, D}_{i, j}$ is uniquely determined up to the change of the basis of $F(D)$.

\section{$L$-symbols for fiber functors of $\Z_2 \times \Z_2$ Tambara-Yamagami categories}
\label{sec: L-symbols for fiber functors of Tambara-Yamagami categories}
In this appendix, we write down the $L$-symbols for fiber functors of non-anomalous $\Z_2 \times \Z_2$ Tambara-Yamagami categories, i.e., $\Rep(D_8)$, $\Rep(Q_8)$, and $\Rep(H_8)$.

\subsection{$\Rep(D_8)$ symmetry}
\label{sec: Fiber functors of Rep(D8)}
We first consider the case of $\TY(\Z_2 \times \Z_2, \chi_+, +1) \cong \Rep(D_8)$.
As we reviewed in Appendix~\ref{sec: Fiber functors of Tambara-Yamagami categories}, fiber functors of $\TY(A, \chi, \epsilon)$ are in one-to-one correspondence with triples $(\sigma, \xi, \nu)$ that satisfies Eq.~\eqref{eq: fiber functor classification}.
When $(A, \chi, \epsilon) = (\Z_2 \times \Z_2, \chi_+, +1)$, there are three inequivalent solutions to Eq.~\eqref{eq: fiber functor classification} \cite{Tam2000, Thorngren:2019iar}.
These solutions share the same $\xi$ and $\sigma$, which are given by\footnote{One can equally choose a different representative $\xi$ of a non-trivial cohomology class $[\xi] \neq 0 \in \mathrm{H}_{\text{gp}}^2(\Z_2 \times \Z_2, \mathrm{U}(1)) \cong \Z_2$.}
\begin{equation}
\xi(a, b) = (-1)^{a_2 b_1}, \quad \sigma(a) = a,
\end{equation}
where $a = (a_1, a_2), b = (b_1, b_2) \in \Z_2 \times \Z_2 = \{(0, 0), (0, 1), (1, 0), (1, 1)\}$.
The three solutions are distinguished by different $\nu$'s, which we denote by $\nu_1$, $\nu_2$, and $\nu_3$.
Concretely, these $\nu$'s are given by
\begin{equation}
\nu_1(0, 1) = \nu_2(1, 0) = \nu_3(1, 1) = -1
\label{eq: nu}
\end{equation}
and $\nu_i(a) = +1$ otherwise.

Since all of the three fiber functors have the same $\xi$ and $\sigma$, some of the $L$-symbols are independent of fiber functors.
Specifically, $L^{a, b}$, $L^{a, D}$, and their inverses are given regardless of $\nu$ as
\begin{equation}
L^{a, b} = \overline{L}^{a, b} = (-1)^{a_2 b_1}, 
\qquad 
L^{a, D}_{i, j} = (-1)^{a_1 a_2} \overline{L}^{a, D}_{i, j} =
\begin{cases}
\delta_{ij} \quad & a = (0, 0),\\
Z_{ij} \quad & a = (0, 1), \\
X_{ij} \quad & a = (1, 0), \\
-iY_{ij} \quad & a = (1, 1),
\end{cases}
\label{eq: Lab LaD}
\end{equation}
where $X$, $Y$, and $Z$ denote Pauli matrices.
On the other hand, $L^{D, a}$, $L^{D, D}$ and their inverses depend on the choice of $\nu$.

\begin{itemize}
\item When $\nu = \nu_1$, Eqs.~\eqref{eq: gamma} and \eqref{eq: Lab LaD} imply $\gamma=X$, and hence we have
\begin{equation}
L^{D, a}_{i, j} = 
\begin{cases}
\delta_{ij} \quad & a = (0, 0), \\
-Z_{ij} \quad & a = (0, 1), \\
X_{ij} \quad & a = (1, 0), \\
-iY_{ij} \quad & a = (1, 1),
\end{cases}
\qquad
L^{D, D}_{(i, j), a} = 
\begin{cases}
X_{ij} \quad & a = (0, 0), \\
iY_{ij} \quad & a = (0, 1), \\
\delta_{ij} \quad & a = (1, 0), \\
Z_{ij} \quad & a = (1, 1).
\end{cases}
\end{equation}

\item When $\nu = \nu_2$, Eqs.~\eqref{eq: gamma} and \eqref{eq: Lab LaD} imply $\gamma=Z$, and hence we have
\begin{equation}
L^{D, a}_{i, j} = 
\begin{cases}
\delta_{ij} \quad & a = (0, 0), \\
Z_{ij} \quad & a = (0, 1), \\
-X_{ij} \quad & a = (1, 0), \\
-iY_{ij} \quad & a = (1, 1),
\end{cases}
\qquad
L^{D, D}_{(i, j), a} = 
\begin{cases}
Z_{ij} \quad & a = (0, 0), \\
\delta_{ij} \quad & a = (0, 1), \\
-iY_{ij} \quad & a = (1, 0), \\
-X_{ij} \quad & a = (1, 1).
\end{cases}
\end{equation}

\item When $\nu = \nu_3$, Eqs.~\eqref{eq: gamma} and \eqref{eq: Lab LaD} imply $\gamma=I$, and hence we have
\begin{equation}
L^{D, a}_{i, j} = 
\begin{cases}
\delta_{ij} \quad & a = (0, 0), \\
Z_{ij} \quad & a = (0, 1), \\
X_{ij} \quad & a = (1, 0), \\
iY_{ij} \quad & a = (1, 1),
\end{cases}
\qquad
L^{D, D}_{(i, j), a} = 
\begin{cases}
\delta_{ij} \quad & a = (0, 0), \\
Z_{ij} \quad & a = (0, 1), \\
X_{ij} \quad & a = (1, 0), \\
iY_{ij} \quad & a = (1, 1).
\end{cases}
\end{equation}
\end{itemize}
In all cases, the inverses of $L^{D, a}$ and $L^{D, D}$ are given by
\begin{equation}
\overline{L}^{D, a}_{i, j} = (-1)^{a_1 a_2} L^{D, a}_{i, j}, \qquad \overline{L}^{D, D}_{a, (i, j)} = \frac{1}{2} L^{D, D}_{(i, j), a}.
\end{equation}

\subsection{$\Rep(Q_8)$ symmetry}
\label{sec: Fiber functors of Rep(Q8)}
The $\Z_2 \times \Z_2$ Tambara-Yamagami category $\TY(\Z_2 \times \Z_2, \chi_+, -1) \cong \Rep(Q_8)$ has only one fiber functor, which is the forgetful functor.
The corresponding solution to Eq.~\eqref{eq: fiber functor classification} is given by
\begin{equation}
\xi(a, b) = (-1)^{a_2 b_1}, \quad \sigma(a) = a, \quad \nu(a) = (-1)^{a_1 + a_2 + a_1a_2},
\end{equation}
where $a = (a_1, a_2), b = (b_1, b_2) \in \Z_2 \times \Z_2$.
The $L$-symbols $L^{a, b}$, $L^{a, D}_{i, j}$, and their inverses are given by Eq.~\eqref{eq: Lab LaD}.
The other $L$-symbols $L^{D, a}$ and $L^{D, D}$ are given by
\begin{equation}
L^{D, a}_{i, j} = 
\begin{cases}
\delta_{ij} \quad & a = (0, 0),\\
-Z_{ij} \quad & a = (0, 1), \\
-X_{ij} \quad & a = (1, 0), \\
iY_{ij} \quad & a = (1, 1),
\end{cases}
\qquad
L^{D, D}_{(i, j), a} = 
\begin{cases}
-iY_{ij} \quad & a = (0, 0), \\
-X_{ij} \quad & a = (0, 1), \\
Z_{ij} \quad & a = (1, 0), \\
\delta_{ij} \quad & a = (1, 1).
\end{cases}
\label{eq: LDa LDD for Q8}
\end{equation}
Here, we used $\gamma=-iY$, which is the unique solution to Eq.~\eqref{eq: gamma} up to scalar.
The inverses of $L^{D, a}$ and $L^{D, D}$ are given by
\begin{equation}
\overline{L}^{D, a}_{i, j} = (-1)^{a_1 a_2} L^{D, a}_{i, j}, \qquad 
\overline{L}^{D, D}_{a, (i, j)} = \frac{1}{2} L^{D, D}_{(i, j), a}.
\label{eq: inverse of LDa LDD Q8}
\end{equation}

\subsection{$\Rep(H_8)$ symmetry}
\label{sec: Fiber functors of Rep(H8)}
The $\Z_2 \times \Z_2$ Tambara-Yamagami category $\TY(\Z_2 \times \Z_2, \chi_-, +1) \cong \Rep(H_8)$ has only one fiber functor, which is the forgetful functor.
The corresponding solution to Eq.~\eqref{eq: fiber functor classification} is given by
\begin{equation}
\xi(a, b) = (-1)^{a_2 b_1}, \quad \sigma(a) = (a_2, a_1), \quad \nu(a) = 1,
\end{equation}
where $a = (a_1, a_2), b = (b_1, b_2) \in \Z_2 \times \Z_2$.
The $L$-symbols $L^{a, b}$, $L^{a, D}_{i, j}$, and their inverses are given by Eq.~\eqref{eq: Lab LaD}.
The other $L$-symbols $L^{D, a}$ and $L^{D, D}$ are given by
\begin{equation}
L^{D, a}_{i, j} = 
\begin{cases}
\delta_{ij} \quad & a = (0, 0), \\
X_{ij} \quad & a = (0, 1), \\
Z_{ij} \quad & a = (1, 0), \\
-iY_{ij} \quad & a = (1, 1),
\end{cases}
\qquad
L^{D, D}_{(i, j), a} = 
\begin{cases}
\frac{1}{\sqrt{2}} (X + Z)_{ij} \quad & a = (0, 0), \\
\frac{1}{\sqrt{2}} (1 + iY)_{ij} \quad & a = (0, 1), \\
\frac{1}{\sqrt{2}} (1 - iY)_{ij} \quad & a = (1, 0), \\
- \frac{1}{\sqrt{2}} (X - Z)_{ij} \quad & a = (1, 1).
\end{cases}
\label{eq: LDa LDD for H8}
\end{equation}
Here, we used $\gamma = \frac{1}{\sqrt{2}} (X + Z)$, which is the unique solution to Eq.~\eqref{eq: gamma} up to scalar.
The inverses of $L^{D, a}$ and $L^{D, D}$ are given by
\begin{equation}
\overline{L}^{D, a}_{i, j} = (-1)^{a_1 a_2} L^{D, a}_{i, j}, \qquad 
\overline{L}^{D, D}_{a, (i, j)} = \frac{1}{2} L^{D, D}_{(i, j), a}.
\label{eq: inverse of LDa LDD H8}
\end{equation}

\section{Interface algebras for $\Z_2 \times \Z_2$ Tambara-Yamagami categories}
\label{sec: Interface algebra for Tambara-Yamagami categories}
In this appendix, we study representations of symmetry algebras acting on the interface of SPT phases with non-anomalous $\Z_2 \times \Z_2$ Tambara-Yamagami symmetries, i.e., $\Rep(D_8)$, $\Rep(Q_8)$, and $\Rep(H_8)$ symmetries.
For any of these symmetries, the symmetry algebra $\cA$ acting on the interface $\cI$ of two SPT phases $\SPT_1$ and $\SPT_2$ is spanned by\footnote{In Sec.~\ref{sec: Interface modes of SPT phases with fusion category symmetries}, a symmetry operator acting on the interface $\cI$ was written as $(\widehat{\cO}_x^{\cI})_{ij}$. In this appendix, we omit the hat.}
\begin{equation}
\{\cO^{\cI}_a, (\cO^{\cI}_D)_{ij} \mid a \in \Z_2 \times\Z_2, ~ i, j = 1, 2\}.
\end{equation}
In particular, $\cA$ is an eight-dimensional algebra.
The multiplication of symmetry operators is given by Eq.~\eqref{eq: interface algebra}, i.e., 
\begin{equation}
\begin{aligned}
\cO_a^{\cI} \cO_b^{\cI} & = (L_1)^{a, b} (\overline{L_2})^{a, b} \cO_{ab}^{\cI}, &
\cO_a^{\cI} (\cO_D^{\cI})_{ij} & = \sum_{k, k^{\prime}} (L_1)^{a, D}_{i, k} (\overline{L_2})^{a, D}_{k^{\prime}, j} (\cO_D^{\cI})_{kk^{\prime}}, \\
(\cO_D^{\cI})_{ij} \cO_a^{\cI} & = \sum_{k, k^{\prime}} (L_1)^{D, a}_{i, k} (\overline{L_2})^{D, a}_{k^{\prime}, j} (\cO_D^{\cI})_{kk^{\prime}}, &
(\cO_D^{\cI})_{ij} (\cO_D^{\cI})_{kl} & = \sum_{a \in \Z_2 \times \Z_2} (L_1)^{D, D}_{(i, k), a} (\overline{L_2})^{D, D}_{a, (j, l)} \cO_a^{\cI},
\end{aligned}
\label{eq: TY self-interface}
\end{equation}
where $L_1$ and $L_2$ are the $L$-symbols associated with $\SPT_1$ and $\SPT_2$ respectively.
Since $L^{a, b}$ and $L^{a, D}$ are given by Eq.~\eqref{eq: Lab LaD} regardless of SPT phases, the first line of the above equation reduces to
\begin{equation}
\cO_a^{\cI} \cO_b^{\cI} = \cO_{ab}^{\cI}, \qquad \cO_a^{\cI} (\cO_D^{\cI})_{ij} = 
\begin{cases}
(\cO_D^{\cI})_{ij} \quad & a = (0, 0), \\
(Z \cO_D^{\cI} Z)_{ij} \quad & a = (0, 1), \\
(X \cO_D^{\cI} X)_{ij} \quad & a = (1, 0), \\
(Y \cO_D^{\cI} Y)_{ij} \quad & a = (1, 1).
\end{cases}
\label{eq: OaOb OaOD}
\end{equation}
On the other hand, the second line of Eq.~\eqref{eq: TY self-interface} depends on SPT phases.
Since the interface algebra $\cA$ should be semisimple, all the irreducible representations of $\cA$ are obtained by the direct sum decomposition of $\cA$.

\subsection{$\mathrm{Rep}(D_8)$ symmetry}
\label{sec: Interface algebra for Rep(D8) symmetry}
\subsubsection{Self-interface}
The $\Rep(D_8)$ symmetry has three inequivalent SPT phases, which are distinguished by different values of $\nu$, see Eq.~\eqref{eq: nu}.
We denote the $\Rep(D_8)$ SPT phase associated with triple $(\sigma, \xi, \nu)$ by $\SPT_{\nu}$.
The interface between $\SPT_{\nu_i}$ and $\SPT_{\nu_j}$ is denoted by $\cI_{\nu_i, \nu_j}$, and the associated interface algebra is denoted by $\cA_{\nu_i, \nu_j}$.
In what follows, we write down the symmetry algebra $\cA_{\nu, \nu}$ acting on the self-interface $\cI_{\nu, \nu}$ and enumerate all irreducible representations of $\cA_{\nu, \nu}$.
When no confusion can arise, we will omit the subscripts of $\cI_{\nu_i, \nu_j}$ and $\cA_{\nu_i, \nu_j}$.

\begin{itemize}
\item When $\SPT_1 = \SPT_2 = \SPT_{\nu_1}$, the multiplication of symmetry operators on the second line of Eq.~\eqref{eq: TY self-interface} is given by
\begin{equation}
\begin{aligned}
(\cO_D^{\cI})_{ij} \cO_a^{\cI} & = \cO_a^{\cI} (\cO_D^{\cI})_{ij}, \\
(\cO_D^{\cI})_{ij} (\cO_D^{\cI})_{kl} & = \frac{1}{2} X_{ik} X_{jl} \cO_{(0, 0)}^{\cI} - \frac{1}{2} Y_{ik} Y_{jl} \cO_{(0, 1)}^{\cI} + \frac{1}{2} \delta_{ik} \delta_{jl} \cO_{(1, 0)}^{\cI} + \frac{1}{2} Z_{ik} Z_{jl} \cO_{(1, 1)}^{\cI}.
\end{aligned}
\label{eq: ODOa ODOD}
\end{equation}
The multiplications \eqref{eq: OaOb OaOD} and \eqref{eq: ODOa ODOD} imply that the interface algebra $\cA$ is commutative.
Thus, due to the Artin-Wedderburn theorem, the eight-dimensional algebra $\cA$ can be decomposed into a direct sum of eight one-dimensional subalgebras as
\begin{equation}
\cA = \bigoplus_{s = \pm} \bigoplus_{1 \leq i \leq 4} \C e_i^s,
\label{eq: direct sum nu1}
\end{equation}
where the direct summands $e_i^s$ are idempotents of $\cA$, i.e., they are elements of $\cA$ such that $e_i^s e_j^t = \delta_{ij} \delta_{st} e_i^s$.
The explicit forms of these idempotents are given by
\begin{equation}
e_1^{\pm} = \frac{1}{2} (e_1 \pm e_1^{\prime}), \quad e_2^{\pm} = \frac{1}{2} (e_2 \pm ie_2^{\prime}), \quad
e_3^{\pm} = \frac{1}{2} (e_3 \pm e_3^{\prime}), \quad e_4^{\pm} = \frac{1}{2} (e_4 \pm ie_4^{\prime}),
\end{equation}
where $e_i$ and $e_i^{\prime}$ for $i = 1, 2, 3, 4$ are defined as follows:
\begin{equation}
\begin{aligned}
e_1 & = \frac{1}{4} \left(\cO_{(0, 0)}^{\cI} + \cO_{(0, 1)}^{\cI} + \cO_{(1, 0)}^{\cI} + \cO_{(1, 1)}^{\cI}\right), & 
e_2 & = \frac{1}{4} \left(\cO_{(0, 0)}^{\cI} + \cO_{(0, 1)}^{\cI} - \cO_{(1, 0)}^{\cI} - \cO_{(1, 1)}^{\cI}\right), \\
e_3 & = \frac{1}{4} \left(\cO_{(0, 0)}^{\cI} - \cO_{(0, 1)}^{\cI} + \cO_{(1, 0)}^{\cI} - \cO_{(1, 1)}^{\cI}\right), & 
e_4 & = \frac{1}{4} \left(\cO_{(0, 0)}^{\cI} - \cO_{(0, 1)}^{\cI} - \cO_{(1, 0)}^{\cI} + \cO_{(1, 1)}^{\cI}\right), \\
e_1^{\prime} & = \frac{1}{2} \left((\cO_D^{\cI})_{11} + (\cO_D^{\cI})_{22}\right), &
e_2^{\prime} & = \frac{1}{2} \left((\cO_D^{\cI})_{11} - (\cO_D^{\cI})_{22}\right), \\
e_3^{\prime} & = \frac{1}{2} \left((\cO_D^{\cI})_{12} + (\cO_D^{\cI})_{21}\right), &
e_4^{\prime} & = \frac{1}{2} \left((\cO_D^{\cI})_{12} - (\cO_D^{\cI})_{21}\right).
\end{aligned}
\label{eq: pre-idempotent}
\end{equation}
The direct sum decomposition \eqref{eq: direct sum nu1} implies that the interface algebra $\cA$ has eight one-dimensional representations, which we denote by $R_i^s$ where $s = \pm$ and $i = 1, 2, 3, 4$.
The (one by one) representation matrix of each idempotent $e_j^t$ in representation $R_i^s$ is given by
\begin{equation}
R_i^s(e^t_j) = \delta_{ij} \delta_{st}.
\end{equation}
Namely, $e_i^s$ acts as one in $R_i^s$ and zero in the other representations.
Going back to the original basis $\{\cO_a^{\cI}, (\cO_D^{\cI})_{ij} \mid a \in \Z_2 \times \Z_2, ~ i, j = 1, 2\}$ of the interface algebra $\cA$, one can summarize the action of symmetry operators in each representation $R_i^s$ as in Table~\ref{tab: reps nu1}.
\begin{table}[t]
\begin{tabular}{c|cccccccc}
$\cI = \cI_{\nu_1, \nu_1}$ & $R_1^+$ & $R_1^-$ & $R_2^+$ & $R_2^-$ & $R_3^+$ & $R_3^-$ & $R_4^+$ & $R_4^-$ \\ \hline
$\cO_{(0, 0)}^{\cI}$ & 1 & 1 & 1 & 1 & 1 & 1 & 1 & 1 \\
$\cO_{(0, 1)}^{\cI}$ & 1 & 1 & 1 & 1 & -1 & -1 & -1 & -1 \\
$\cO_{(1, 0)}^{\cI}$ & 1 & 1 & -1 & -1 & 1 & 1 & -1 & -1 \\ 
$\cO_{(1, 1)}^{\cI}$ & 1 & 1 & -1 & -1 & -1 & -1 & 1 & 1 \\
$(\cO_D^{\cI})_{ij}$ & $\delta_{ij}$ & $-\delta_{ij}$ & $-iZ_{ij}$ & $iZ_{ij}$ & $X_{ij}$ & $-X_{ij}$ & $Y_{ij}$ & $-Y_{ij}$
\end{tabular}
\caption{Irreducible representations of the self-interface algebra $\cA$ for the $\Rep(D_8)$ SPT phase $\SPT_{\nu_1}$. All irreducible representations are one-dimensional and form a group $D_8$ under the multiplication.}
\label{tab: reps nu1}
\end{table}
The group of one-dimensional representations $\{R_i^s \mid s = \pm, ~ i = 1, 2, 3, 4\}$ is isomorphic to $D_8$.

\item When $\SPT_1 = \SPT_2 = \SPT_{\nu_2}$, the multiplication of symmetry operators on the second line of Eq.~\eqref{eq: TY self-interface} is given by
\begin{equation}
\begin{aligned}
(\cO_D^{\cI})_{ij} \cO_a^{\cI} & = \cO_a^{\cI} (\cO_D^{\cI})_{ij}, \\
(\cO_D^{\cI})_{ij} (\cO_D^{\cI})_{kl} & = \frac{1}{2} Z_{ik} Z_{jl} \cO_{(0, 0)}^{\cI} + \frac{1}{2} \delta_{ik} \delta_{jl} \cO_{(0, 1)}^{\cI} - \frac{1}{2} Y_{ik} Y_{jl} \cO_{(1, 0)}^{\cI} + \frac{1}{2} X_{ik} X_{jl} \cO_{(1, 1)}^{\cI}.
\end{aligned}
\label{eq: ODOa ODOD nu2}
\end{equation}
The multiplications \eqref{eq: OaOb OaOD} and \eqref{eq: ODOa ODOD nu2} imply that the interface algebra $\cA$ is commutative.
Thus, $\cA$ can be decomposed into a direct sum of eight one-dimensional subalgebras as
\begin{equation}
\cA = \bigoplus_{s = \pm} \bigoplus_{1 \leq i \leq 4} \C e_i^s,
\label{eq: direct sum nu2}
\end{equation}
where the idempotents $e_i^s$ are given by
\begin{equation}
e_1^{\pm} = \frac{1}{2} (e_1 \pm e_1^{\prime}), \quad e_2^{\pm} = \frac{1}{2} (e_2 \pm e_2^{\prime}), \quad
e_3^{\pm} = \frac{1}{2} (e_3 \pm ie_3^{\prime}), \quad e_4^{\pm} = \frac{1}{2} (e_4 \pm ie_4^{\prime}).
\end{equation}
Here, $e_i$ and $e_i^{\prime}$ for $i = 1, 2, 3, 4$ are defined by Eq.~\eqref{eq: pre-idempotent}.
The direct sum decomposition \eqref{eq: direct sum nu2} implies that the interface algebra $\cA$ has eight one-dimensional representations, which we denote by $R_i^s$ where $s = \pm$ and $i = 1, 2, 3, 4$.
The (one by one) representation matrix of each idempotent $e_j^t$ in representation $R_i^s$ is given by
\begin{equation}
R_i^s(e^t_j) = \delta_{ij} \delta_{st}. 
\end{equation}
Going back to the original basis $\{\cO_a^{\cI}, (\cO_D^{\cI})_{ij} \mid a \in \Z_2 \times \Z_2, ~ i, j = 1, 2\}$ of the interface algebra $\cA$, one can summarize the action of symmetry operators in each representation $R_i^s$ as in Table~\ref{tab: reps nu2}.
\begin{table}[t]
\begin{tabular}{c|cccccccc}
$\cI = \cI_{\nu_2, \nu_2}$ & $R_1^+$ & $R_1^-$ & $R_2^+$ & $R_2^-$ & $R_3^+$ & $R_3^-$ & $R_4^+$ & $R_4^-$ \\ \hline
$\cO_{(0, 0)}^{\cI}$ & 1 & 1 & 1 & 1 & 1 & 1 & 1 & 1 \\
$\cO_{(0, 1)}^{\cI}$ & 1 & 1 & 1 & 1 & -1 & -1 & -1 & -1 \\
$\cO_{(1, 0)}^{\cI}$ & 1 & 1 & -1 & -1 & 1 & 1 & -1 & -1 \\ 
$\cO_{(1, 1)}^{\cI}$ & 1 & 1 & -1 & -1 & -1 & -1 & 1 & 1 \\
$(\cO_D^{\cI})_{ij}$ & $\delta_{ij}$ & $-\delta_{ij}$ & $Z_{ij}$ & $-Z_{ij}$ & $-iX_{ij}$ & $iX_{ij}$ & $Y_{ij}$ & $-Y_{ij}$
\end{tabular}
\caption{Irreducible representations of the self-interface algebra $\cA$ for the $\Rep(D_8)$ SPT phase $\SPT_{\nu_2}$. All irreducible representations are one-dimensional and form a group $D_8$ under the multiplication.}
\label{tab: reps nu2}
\end{table}
The group of one-dimensional representations $\{R_i^s \mid s = \pm, ~ i = 1, 2, 3, 4\}$ is isomorphic to $D_8$.

\item When $\SPT_1 = \SPT_2 = \SPT_{\nu_3}$, the multiplication of symmetry operators on the second line of Eq.~\eqref{eq: TY self-interface} is given by
\begin{equation}
\begin{aligned}
(\cO_D^{\cI})_{ij} \cO_a^{\cI} & = \cO_a^{\cI} (\cO_D^{\cI})_{ij}, \\
(\cO_D^{\cI})_{ij} (\cO_D^{\cI})_{kl} & = \frac{1}{2} \delta_{ik} \delta_{jl} \cO_{(0, 0)}^{\cI} + \frac{1}{2} Z_{ik} Z_{jl} \cO_{(0, 1)}^{\cI} + \frac{1}{2} X_{ik} X_{jl} \cO_{(1, 0)}^{\cI} - \frac{1}{2} Y_{ik} Y_{jl} \cO_{(1, 1)}^{\cI}.
\end{aligned}
\label{eq: ODOa ODOD nu3}
\end{equation}
The multiplications \eqref{eq: OaOb OaOD} and \eqref{eq: ODOa ODOD nu3} imply that the interface algebra $\cA$ is commutative.
Thus, $\cA$ can be decomposed into a direct sum of eight one-dimensional subalgebras as
\begin{equation}
\cA = \bigoplus_{s = \pm} \bigoplus_{1 \leq i \leq 4} \C e_i^s,
\label{eq: direct sum nu3}
\end{equation}
where the idempotents $e_i^s$ are given by
\begin{equation}
e_1^{\pm} = \frac{1}{2} (e_1 \pm e_1^{\prime}), \quad e_2^{\pm} = \frac{1}{2} (e_2 \pm e_2^{\prime}), \quad
e_3^{\pm} = \frac{1}{2} (e_3 \pm e_3^{\prime}), \quad e_4^{\pm} = \frac{1}{2} (e_4 \pm e_4^{\prime}).
\end{equation}
Here, $e_i$ and $e_i^{\prime}$ for $i = 1, 2, 3, 4$ are defined by Eq.~\eqref{eq: pre-idempotent}.
The direct sum decomposition \eqref{eq: direct sum nu3} implies that the interface algebra $\cA$ has eight one-dimensional representations, which we denote by $R_i^s$ where $s = \pm$ and $i = 1, 2, 3, 4$.
The (one by one) representation matrix of each idempotent $e_j^t$ in representation $R_i^s$ is given by
\begin{equation}
R_i^s(e^t_j) = \delta_{ij} \delta_{st}. 
\end{equation}
Going back to the original basis $\{\cO_a^{\cI}, (\cO_D^{\cI})_{ij} \mid a \in \Z_2 \times \Z_2, ~ i, j = 1, 2\}$ of the interface algebra $\cA$, one can summarize the action of symmetry operators in each representation $R_i^s$ as in Table~\ref{tab: reps nu3}.
\begin{table}[t]
\begin{tabular}{c|cccccccc}
$\cI = \cI_{\nu_3, \nu_3}$ & $R_1^+$ & $R_1^-$ & $R_2^+$ & $R_2^-$ & $R_3^+$ & $R_3^-$ & $R_4^+$ & $R_4^-$ \\ \hline
$\cO_{(0, 0)}^{\cI}$ & 1 & 1 & 1 & 1 & 1 & 1 & 1 & 1 \\
$\cO_{(0, 1)}^{\cI}$ & 1 & 1 & 1 & 1 & -1 & -1 & -1 & -1 \\
$\cO_{(1, 0)}^{\cI}$ & 1 & 1 & -1 & -1 & 1 & 1 & -1 & -1 \\ 
$\cO_{(1, 1)}^{\cI}$ & 1 & 1 & -1 & -1 & -1 & -1 & 1 & 1 \\
$(\cO_D^{\cI})_{ij}$ & $\delta_{ij}$ & $-\delta_{ij}$ & $Z_{ij}$ & $-Z_{ij}$ & $X_{ij}$ & $-X_{ij}$ & $iY_{ij}$ & $-iY_{ij}$
\end{tabular}
\caption{Irreducible representations of the self-interface algebra $\cA$ for the $\Rep(D_8)$ SPT phase $\SPT_{\nu_3}$. All irreducible representations are one-dimensional and form a group $D_8$ under the multiplication.}
\label{tab: reps nu3}
\end{table}
The group of one-dimensional representations $\{R_i^s \mid s = \pm, ~ i = 1, 2, 3, 4\}$ is isomorphic to $D_8$.
\end{itemize}

\subsubsection{Interface between different SPT phases}
The interface algebra $\cA$ between different SPT phases $\SPT_1$ and $\SPT_2$ does not have one-dimensional representations as shown in Sec.~\ref{sec: Interface modes}.
Thus, the direct sum decomposition of $\cA$ does not contain one-dimensional factors.
In addition, since $\cA$ is an eight-dimensional algebra, the only possible decomposition is
\begin{equation}
\cA = M_2(\C) \oplus M_2(\C).
\label{eq: degenerate interface D8}
\end{equation}
The above direct sum decomposition implies that $\cA$ has two two-dimensional irreducible representations.
In what follows, we will explicitly decompose $\cA$ into a direct sum~\eqref{eq: degenerate interface D8} and write down its irreducible representations.
Given a representation $R$ of the interface algebra $\cA_{\nu_i, \nu_j}$, one can obtain the corresponding representation of $\cA_{\nu_j, \nu_i}$, which we denote by $R$ by abuse of notation, as follows:
\begin{equation}
R(\cO_a^{\cI_{\nu_j, \nu_i}}) = R(\cO_a^{\cI_{\nu_i, \nu_j}})^*, \quad R((\cO_D^{\cI_{\nu_j, \nu_i}})_{kl}) = R(\cO_D^{\cI_{\nu_i, \nu_j}})_{lk})^*.
\end{equation}
Therefore, we only need to consider three interface algebras $\cA_{\nu_1, \nu_2}$, $\cA_{\nu_2, \nu_3}$, and $\cA_{\nu_3, \nu_1}$.

\begin{itemize}
\item When $\SPT_1 = \SPT_{\nu_1}$ and $\SPT_2 = \SPT_{\nu_2}$, the multiplication of symmetry operators on the second line of Eq.~\eqref{eq: TY self-interface} reduces to
\begin{equation}
\begin{aligned}
(\cO_D^{\cI})_{ij} \cO_a^{\cI} & = 
\begin{cases}
(\cO_D^{\cI})_{ij} \quad & a = (0, 0), \\
-(Z \cO_D^{\cI} Z)_{ij} \quad & a = (0, 1), \\
-(X \cO_D^{\cI} X)_{ij} \quad & a = (1, 0), \\
(Y \cO_D^{\cI} Y)_{ij} \quad & a = (1, 1),
\end{cases}\\
(\cO_D^{\cI})_{ij} (\cO_D^{\cI})_{kl} & = \frac{1}{2} X_{ik} Z_{jl} \cO_{(0, 0)}^{\cI} + \frac{1}{2} iY_{ik} \delta_{jl} \cO_{(0, 1)}^{\cI} - \frac{1}{2} \delta_{ik} iY_{jl} \cO_{(1, 0)}^{\cI} - \frac{1}{2} Z_{ik} X_{jl} \cO_{(1, 1)}^{\cI}.
\end{aligned}
\label{eq: ODOa ODOD nu1nu2}
\end{equation}
The multiplications \eqref{eq: OaOb OaOD} and \eqref{eq: ODOa ODOD nu1nu2} imply that the interface algebra $\cA$ is decomposed as Eq.~\eqref{eq: degenerate interface D8}, where the two factors of the direct sum decomposition~\eqref{eq: degenerate interface D8} are generated by $\{e_i, e_i^{\prime} \mid i = 1, 4\}$ and $\{e_i, e_i^{\prime} \mid i = 2, 3\}$ respectively.
Here, $e_i$ and $e_i^{\prime}$ are defined by Eq.~\eqref{eq: pre-idempotent}.
Indeed, in a new basis given by
\begin{equation}
\begin{aligned}
e_{11} & = e_1, & e_{12} & = ie_1^{\prime}, & e_{21} & = ie_4^{\prime}, & e_{22} & = e_4, \\
d_{11} & = e_2, & d_{12} & = e_2^{\prime}, & d_{21} & = e_3^{\prime}, & d_{22} & = e_3,
\end{aligned}
\end{equation}
the multiplication of $\cA$ can be written as
\begin{equation}
e_{ij} e_{kl} = \delta_{jk} e_{il}, \quad d_{ij} d_{kl} = \delta_{jk} d_{il}, \quad e_{ij} d_{kl} = d_{ij} e_{kl} = 0,
\end{equation}
which shows that each of $\{e_{ij} \mid i, j = 1, 2\}$ and $\{d_{ij} \mid i, j = 1, 2\}$ spans the full matrix algebra $M_2(\C)$.
The above direct sum decomposition implies that $\cA$ has two two-dimensional irreducible representations, which we denote by $R_e$ and $R_d$.
The (two by two) representation matrices of $e_{ij}$ and $d_{ij}$ in these representations are given by
\begin{equation}
\begin{aligned}
R_e(e_{11}) & = 
\begin{pmatrix}
1 & 0 \\
0 & 0
\end{pmatrix}, &
R_e(e_{12}) & = 
\begin{pmatrix}
0 & 1 \\
0 & 0
\end{pmatrix}, &
R_e(e_{21}) & = 
\begin{pmatrix}
0 & 0 \\
1 & 0
\end{pmatrix}, &
R_e(e_{22}) & = 
\begin{pmatrix}
0 & 0 \\
0 & 1
\end{pmatrix}, & 
R_e(d_{ij}) & = 0, \\
R_d(d_{11}) & = 
\begin{pmatrix}
1 & 0 \\
0 & 0
\end{pmatrix}, &
R_d(d_{12}) & = 
\begin{pmatrix}
0 & 1 \\
0 & 0
\end{pmatrix}, &
R_d(d_{21}) & = 
\begin{pmatrix}
0 & 0 \\
1 & 0
\end{pmatrix}, &
R_d(d_{22}) & = 
\begin{pmatrix}
0 & 0 \\
0 & 1
\end{pmatrix}, & 
R_d(e_{ij}) & = 0.
\end{aligned}
\label{eq: Re Rd}
\end{equation}
Going back to the original basis $\{\cO_a^{\cI}, (\cO_D^{\cI})_{ij} \mid a \in \Z_2 \times \Z_2, ~ i, j = 1, 2\}$ of the interface algebra $\cA$, one can summarize the action of symmetry operators in each representation as in Table~\ref{tab: reps nu1nu2}.
\begin{table}[t]
\begin{minipage}{0.5 \linewidth}
\begin{tabular}{c|cc}
$\cI = \cI_{\nu_1, \nu_2}$ & $R_e$ & $R_d$ \\ \hline
$\cO_{(0, 0)}^{\cI}$ & $I_2$ & $I_2$ \\
$\cO_{(0, 1)}^{\cI}$ & $Z$ & $Z$ \\
$\cO_{(1, 0)}^{\cI}$ & $Z$ & $-Z$ \\
$\cO_{(1, 1)}^{\cI}$ & $I_2$ & $-I_2$ \\
$(\cO_D^{\cI})_{ij}$ & 
$\begin{pmatrix}
-\frac{i}{2}S_+ & -\frac{i}{2}S_- \\
\frac{i}{2}S_- & -\frac{i}{2}S_+
\end{pmatrix}_{ij}
$ &
$\begin{pmatrix}
\frac{1}{2}S_+ & \frac{1}{2}S_- \\
\frac{1}{2}S_- & -\frac{1}{2}S_+
\end{pmatrix}_{ij}
$
\end{tabular}
\end{minipage}%
\begin{minipage}{0.5 \linewidth}
\begin{tabular}{c|cc}
$\cI = \cI_{\nu_2, \nu_1}$ & $R_e$ & $R_d$ \\ \hline
$\cO_{(0, 0)}^{\cI}$ & $I_2$ & $I_2$ \\
$\cO_{(0, 1)}^{\cI}$ & $Z$ & $Z$ \\
$\cO_{(1, 0)}^{\cI}$ & $Z$ & $-Z$ \\
$\cO_{(1, 1)}^{\cI}$ & $I_2$ & $-I_2$ \\
$(\cO_D^{\cI})_{ij}$ & 
$\begin{pmatrix}
\frac{i}{2}S_+ & -\frac{i}{2}S_- \\
\frac{i}{2}S_- & \frac{i}{2}S_+
\end{pmatrix}_{ij}
$ &
$\begin{pmatrix}
\frac{1}{2}S_+ & \frac{1}{2}S_- \\
\frac{1}{2}S_- & -\frac{1}{2}S_+
\end{pmatrix}_{ij}
$
\end{tabular}
\end{minipage}
\caption{Irreducible representations of the interface algebra $\cA$ between $\Rep(D_8)$ SPT phases $\SPT_{\nu_1}$ and $\SPT_{\nu_2}$. Here, $I_2$ denotes the $2 \times 2$ identity matrix and $S_{\pm} := X \pm iY$. The bottom row reads $R_e((\cO_D^{\cI})_{11}) = -\frac{i}{2}S_+$, etc.}
\label{tab: reps nu1nu2}
\end{table}

\item When $\SPT_1 = \SPT_{\nu_2}$ and $\SPT_2 = \SPT_{\nu_3}$, the multiplication of symmetry operators on the second line of Eq.~\eqref{eq: TY self-interface} reduces to
\begin{equation}
\begin{aligned}
(\cO_D^{\cI})_{ij} \cO_a^{\cI} & = 
\begin{cases}
(\cO_D^{\cI})_{ij} \quad & a = (0, 0), \\
(Z \cO_D^{\cI} Z)_{ij} \quad & a = (0, 1), \\
-(X \cO_D^{\cI} X)_{ij} \quad & a = (1, 0), \\
-(Y \cO_D^{\cI} Y)_{ij} \quad & a = (1, 1),
\end{cases}\\
(\cO_D^{\cI})_{ij} (\cO_D^{\cI})_{kl} & = \frac{1}{2} Z_{ik} \delta_{jl} \cO_{(0, 0)}^{\cI} + \frac{1}{2} \delta_{ik} Z_{jl} \cO_{(0, 1)}^{\cI} - \frac{1}{2} iY_{ik} X_{jl} \cO_{(1, 0)}^{\cI} - \frac{1}{2} X_{ik} iY_{jl} \cO_{(1, 1)}^{\cI}.
\end{aligned}
\label{eq: ODOa ODOD nu2nu3}
\end{equation}
The multiplications \eqref{eq: OaOb OaOD} and \eqref{eq: ODOa ODOD nu2nu3} imply that the interface algebra $\cA$ is decomposed as Eq.~\eqref{eq: degenerate interface D8}, where the two factors of the direct sum decomposition~\eqref{eq: degenerate interface D8} are generated by $\{e_i, e_i^{\prime} \mid i = 1, 2\}$ and $\{e_i, e_i^{\prime} \mid i = 3, 4\}$ respectively.
Here, $e_i$ and $e_i^{\prime}$ are defined by Eq.~\eqref{eq: pre-idempotent}.
Indeed, in a new basis given by
\begin{equation}
\begin{aligned}
e_{11} & = e_1, & e_{12} & = e_1^{\prime}, & e_{21} & = e_2^{\prime}, & e_{22} & = e_2, \\
d_{11} & = e_3, & d_{12} & = e_3^{\prime}, & d_{21} & = e_4^{\prime}, & d_{22} & = e_4,
\end{aligned}
\end{equation}
the multiplication of $\cA$ can be written as
\begin{equation}
e_{ij} e_{kl} = \delta_{jk} e_{il}, \quad d_{ij} d_{kl} = \delta_{jk} d_{il}, \quad e_{ij} d_{kl} = d_{ij} e_{kl} = 0,
\end{equation}
which shows that each of $\{e_{ij} \mid i, j = 1, 2\}$ and $\{d_{ij} \mid i, j = 1, 2\}$ spans the full matrix algebra $M_2(\C)$.
The above direct sum decomposition implies that $\cA$ has two two-dimensional irreducible representations, which we denote by $R_e$ and $R_d$.
The (two by two) representation matrices of $e_{ij}$ and $d_{ij}$ in these representations are given by Eq.~\eqref{eq: Re Rd}.
Going back to the original basis $\{\cO_a^{\cI}, (\cO_D^{\cI})_{ij} \mid a \in \Z_2 \times \Z_2, ~ i, j = 1, 2\}$ of the interface algebra $\cA$, one can summarize the action of symmetry operators in each representation as in Table~\ref{tab: reps nu2nu3}.
\begin{table}[t]
\begin{minipage}{0.5 \linewidth}
\begin{tabular}{c|cc}
$\cI = \cI_{\nu_2, \nu_3}$ & $R_e$ & $R_d$ \\ \hline
$\cO_{(0, 0)}^{\cI}$ & $I_2$ & $I_2$ \\
$\cO_{(0, 1)}^{\cI}$ & $I_2$ & $-I_2$ \\
$\cO_{(1, 0)}^{\cI}$ & $Z$ & $Z$ \\
$\cO_{(1, 1)}^{\cI}$ & $Z$ & $-Z$ \\
$(\cO_D^{\cI})_{ij}$ & 
$\begin{pmatrix}
X & 0 \\
0 & iY
\end{pmatrix}_{ij}
$ &
$\begin{pmatrix}
0 & X \\
iY & 0
\end{pmatrix}_{ij}
$
\end{tabular}
\end{minipage}%
\begin{minipage}{0.5 \linewidth}
\begin{tabular}{c|cc}
$\cI = \cI_{\nu_3, \nu_2}$ & $R_e$ & $R_d$ \\ \hline
$\cO_{(0, 0)}^{\cI}$ & $I_2$ & $I_2$ \\
$\cO_{(0, 1)}^{\cI}$ & $I_2$ & $-I_2$ \\
$\cO_{(1, 0)}^{\cI}$ & $Z$ & $Z$ \\
$\cO_{(1, 1)}^{\cI}$ & $Z$ & $-Z$ \\
$(\cO_D^{\cI})_{ij}$ & 
$\begin{pmatrix}
X & 0 \\
0 & iY
\end{pmatrix}_{ij}
$ &
$\begin{pmatrix}
0 & iY \\
X & 0
\end{pmatrix}_{ij}
$
\end{tabular}
\end{minipage}
\caption{Irreducible representations of the interface algebra $\cA$ between $\Rep(D_8)$ SPT phases $\SPT_{\nu_2}$ and $\SPT_{\nu_3}$. Here, $I_2$ denotes the $2 \times 2$ identity matrix, and the bottom row reads $R_e((\cO_D^{\cI})_{11}) = X$, etc.}
\label{tab: reps nu2nu3}
\end{table}

\item When $\SPT_1 = \SPT_{\nu_3}$ and $\SPT_2 = \SPT_{\nu_1}$, the multiplication of symmetry operators on the second line of Eq.~\eqref{eq: TY self-interface} reduces to
\begin{equation}
\begin{aligned}
(\cO_D^{\cI})_{ij} \cO_a^{\cI} & = 
\begin{cases}
(\cO_D^{\cI})_{ij} \quad & a = (0, 0), \\
-(Z \cO_D^{\cI} Z)_{ij} \quad & a = (0, 1), \\
(X \cO_D^{\cI} X)_{ij} \quad & a = (1, 0), \\
-(Y \cO_D^{\cI} Y)_{ij} \quad & a = (1, 1),
\end{cases}\\
(\cO_D^{\cI})_{ij} (\cO_D^{\cI})_{kl} & = \frac{1}{2} \delta_{ik} X_{jl} \cO_{(0, 0)}^{\cI} + \frac{1}{2} Z_{ik} iY_{jl} \cO_{(0, 1)}^{\cI} + \frac{1}{2} X_{ik} \delta_{jl} \cO_{(1, 0)}^{\cI} + \frac{1}{2} iY_{ik} Z_{jl} \cO_{(1, 1)}^{\cI}.
\end{aligned}
\label{eq: ODOa ODOD nu3nu1}
\end{equation}
The multiplications \eqref{eq: OaOb OaOD} and \eqref{eq: ODOa ODOD nu3nu1} imply that the interface algebra $\cA$ is decomposed as Eq.~\eqref{eq: degenerate interface D8}, where the two factors of the direct sum decomposition~\eqref{eq: degenerate interface D8} are generated by $\{e_i, e_i^{\prime} \mid i = 1, 3\}$ and $\{e_i, e_i^{\prime} \mid i = 2, 4\}$ respectively.
Here, $e_i$ and $e_i^{\prime}$ are defined by Eq.~\eqref{eq: pre-idempotent}.
Indeed, in a new basis given by
\begin{equation}
\begin{aligned}
e_{11} & = e_1, & e_{12} & = e_1^{\prime}, & e_{21} & = e_3^{\prime}, & e_{22} & = e_3, \\
d_{11} & = e_2, & d_{12} & = e_2^{\prime}, & d_{21} & = e_4^{\prime}, & d_{22} & = e_4,
\end{aligned}
\end{equation}
the multiplication of $\cA$ can be written as
\begin{equation}
e_{ij} e_{kl} = \delta_{jk} e_{il}, \quad d_{ij} d_{kl} = \delta_{jk} d_{il}, \quad e_{ij} d_{kl} = d_{ij} e_{kl} = 0,
\end{equation}
which shows that each of $\{e_{ij} \mid i, j = 1, 2\}$ and $\{d_{ij} \mid i, j = 1, 2\}$ spans the full matrix algebra $M_2(\C)$.
The above direct sum decomposition implies that $\cA$ has two two-dimensional irreducible representations, which we denote by $R_e$ and $R_d$.
The (two by two) representation matrices of $e_{ij}$ and $d_{ij}$ in these representations are given by Eq.~\eqref{eq: Re Rd}.
Going back to the original basis $\{\cO_a^{\cI}, (\cO_D^{\cI})_{ij} \mid a \in \Z_2 \times \Z_2, ~ i, j = 1, 2\}$ of the interface algebra $\cA$, one can summarize the action of symmetry operators in each representation as in Table~\ref{tab: reps nu3nu1}.
\begin{table}[t]
\begin{minipage}{0.5 \linewidth}
\begin{tabular}{c|cc}
$\cI = \cI_{\nu_3, \nu_1}$ & $R_e$ & $R_d$ \\ \hline
$\cO_{(0, 0)}^{\cI}$ & $I_2$ & $I_2$ \\
$\cO_{(0, 1)}^{\cI}$ & $Z$ & $Z$ \\
$\cO_{(1, 0)}^{\cI}$ & $I_2$ & $-I_2$ \\
$\cO_{(1, 1)}^{\cI}$ & $Z$ & $-Z$ \\
$(\cO_D^{\cI})_{ij}$ & 
$\begin{pmatrix}
\frac{1}{2}S_+ & \frac{1}{2}S_- \\
\frac{1}{2}S_- & \frac{1}{2}S_+
\end{pmatrix}_{ij}
$ &
$\begin{pmatrix}
\frac{1}{2}S_+ & \frac{1}{2}S_- \\
-\frac{1}{2}S_- & -\frac{1}{2}S_+
\end{pmatrix}_{ij}
$
\end{tabular}
\end{minipage}%
\begin{minipage}{0.5 \linewidth}
\begin{tabular}{c|cc}
$\cI = \cI_{\nu_1, \nu_3}$ & $R_e$ & $R_d$ \\ \hline
$\cO_{(0, 0)}^{\cI}$ & $I_2$ & $I_2$ \\
$\cO_{(0, 1)}^{\cI}$ & $Z$ & $Z$ \\
$\cO_{(1, 0)}^{\cI}$ & $I_2$ & $-I_2$ \\
$\cO_{(1, 1)}^{\cI}$ & $Z$ & $-Z$ \\
$(\cO_D^{\cI})_{ij}$ & 
$\begin{pmatrix}
\frac{1}{2}S_+ & \frac{1}{2}S_- \\
\frac{1}{2}S_- & \frac{1}{2}S_+
\end{pmatrix}_{ij}
$ &
$\begin{pmatrix}
\frac{1}{2}S_+ & -\frac{1}{2}S_- \\
\frac{1}{2}S_- & -\frac{1}{2}S_+
\end{pmatrix}_{ij}
$
\end{tabular}
\end{minipage}
\caption{Irreducible representations of the interface algebra $\cA$ between $\Rep(D_8)$ SPT phases $\SPT_{\nu_3}$ and $\SPT_{\nu_1}$. Here, $I_2$ denotes the $2 \times 2$ identity matrix and $S_{\pm} := X \pm iY$. The bottom row reads $R_e((\cO_D^{\cI})_{11}) = \frac{1}{2}S_+$, etc.}
\label{tab: reps nu3nu1}
\end{table}
\end{itemize}

\subsection{$\mathrm{Rep}(Q_8)$ symmetry}
\label{sec: Interface algebra for Rep(Q8) symmetry}
The $\Rep(Q_8)$ symmetry has only one SPT phase and hence the only possible interface is the self-interface.
The $L$-symbols for this SPT phase are given by Eqs.~\eqref{eq: Lab LaD}, \eqref{eq: LDa LDD for Q8}, and \eqref{eq: inverse of LDa LDD Q8}.
Substituting these $L$-symbols into the second line of Eq.~\eqref{eq: TY self-interface}, one finds that the multiplication of symmetry operators of the self-interface algebra $\cA$ is given by
\begin{equation}
\begin{aligned}
(\cO_D^{\cI})_{ij} \cO_a^{\cI} & = \cO_a^{\cI} (\cO_D^{\cI})_{ij}, \\
(\cO_D^{\cI})_{ij} (\cO_D^{\cI})_{kl} & = -\frac{1}{2} Y_{ik} Y_{jl} \cO_{(0, 0)}^{\cI} + \frac{1}{2} X_{ik} X_{jl} \cO_{(0, 1)}^{\cI} + \frac{1}{2} Z_{ik} Z_{jl} \cO_{(1, 0)}^{\cI} + \frac{1}{2} \delta_{ik} \delta_{jl} \cO_{(1, 1)}^{\cI}.
\end{aligned}
\label{eq: ODOa ODOD for Q8}
\end{equation}
The multiplications \eqref{eq: OaOb OaOD} and \eqref{eq: ODOa ODOD for Q8} imply that the interface algebra $\cA$ is commutative.
Thus, $\cA$ can be decomposed into a direct sum of eight one-dimensional subalgebras as
\begin{equation}
\cA = \bigoplus_{s = \pm} \bigoplus_{1 \leq i \leq 4} \C e_i^s,
\label{eq: decomposition Q8}
\end{equation}
where the idempotents $e_i^s$ are given by
\begin{equation}
e_1^{\pm} = \frac{1}{2} (e_1 \pm e_1^{\prime}), \quad e_2^{\pm} = \frac{1}{2} (e_2 \pm ie_2^{\prime}), \quad
e_3^{\pm} = \frac{1}{2} (e_3 \pm ie_3^{\prime}), \quad e_4^{\pm} = \frac{1}{2} (e_4 \pm e_4^{\prime}).
\end{equation}
Here, $e_i$ and $e_i^{\prime}$ for $i = 1, 2, 3, 4$ are defined by Eq.~\eqref{eq: pre-idempotent}.
The direct sum decomposition \eqref{eq: direct sum nu3} implies that the interface algebra $\cA$ has eight one-dimensional representations, which we denote by $R_i^s$ where $s = \pm$ and $i = 1, 2, 3, 4$.
The (one by one) representation matrix of each idempotent $e_j^t$ in representation $R_i^s$ is given by
\begin{equation}
R_i^s(e^t_j) = \delta_{ij} \delta_{st}.
\end{equation}
Going back to the original basis $\{\cO_a^{\cI}, (\cO_D^{\cI})_{ij} \mid a \in \Z_2 \times \Z_2, ~ i, j = 1, 2\}$ of the interface algebra $\cA$, one can summarize the action of symmetry operators in each representation $R_i^s$ as in Table~\ref{tab: reps Q8}.
\begin{table}[t]
\begin{tabular}{c|cccccccc}
& $R_1^+$ & $R_1^-$ & $R_2^+$ & $R_2^-$ & $R_3^+$ & $R_3^-$ & $R_4^+$ & $R_4^-$ \\ \hline
$\cO_{(0, 0)}^{\cI}$ & 1 & 1 & 1 & 1 & 1 & 1 & 1 & 1 \\
$\cO_{(0, 1)}^{\cI}$ & 1 & 1 & 1 & 1 & -1 & -1 & -1 & -1 \\
$\cO_{(1, 0)}^{\cI}$ & 1 & 1 & -1 & -1 & 1 & 1 & -1 & -1 \\ 
$\cO_{(1, 1)}^{\cI}$ & 1 & 1 & -1 & -1 & -1 & -1 & 1 & 1 \\
$(\cO_D^{\cI})_{ij}$ & $\delta_{ij}$ & $-\delta_{ij}$ & $-iZ_{ij}$ & $iZ_{ij}$ & $-iX_{ij}$ & $iX_{ij}$ & $iY_{ij}$ & $-iY_{ij}$
\end{tabular}
\caption{Irreducible representations of the self-interface algebra $\cA$ for the $\Rep(Q_8)$ SPT phase. All irreducible representations are one-dimensional and form a group $Q_8$ under the multiplication.}
\label{tab: reps Q8}
\end{table}
The group of one-dimensional representations $\{R_i^s \mid s = \pm, ~ i = 1, 2, 3, 4\}$ is isomorphic to $Q_8$.

\subsection{$\mathrm{Rep}(H_8)$ symmetry}
\label{sec: Interface algebra for Rep(H8) symmetry}
The $\Rep(H_8)$ symmetry has only one SPT phase and hence the only possible interface is the self-interface.
The $L$-symbols for this SPT phase are given by Eqs.~\eqref{eq: Lab LaD}, \eqref{eq: LDa LDD for H8}, and \eqref{eq: inverse of LDa LDD H8}.
Substituting these $L$-symbols into the second line of Eq.~\eqref{eq: TY self-interface}, one finds that the multiplication of symmetry operators of the self-interface algebra $\cA$ is given by
\begin{equation}
\begin{aligned}
(\cO_D^{\cI})_{ij} \cO_a^{\cI} & = 
\begin{cases}
(\cO_D^{\cI})_{ij} \quad & a = (0, 0),\\
(X \cO_D^{\cI} X)_{ij} \quad & a = (0, 1), \\
(Z \cO_D^{\cI} Z)_{ij} \quad & a = (1, 0), \\
(Y \cO_D^{\cI} Y)_{ij} \quad & a = (1, 1),
\end{cases} \\
(\cO_D^{\cI})_{ij} (\cO_D^{\cI})_{kl} & = \frac{1}{4} (X + Z)_{ik} (X + Z)_{jl} \cO_{(0, 0)}^{\cI} + \frac{1}{4} (1+iY)_{ik} (1+iY)_{jl} \cO_{(0, 1)}^{\cI} \\
& \quad + \frac{1}{4} (1-iY)_{ik} (1-iY)_{jl} \cO_{(1, 0)}^{\cI} + \frac{1}{4} (X-Z)_{ik} (X-Z)_{jl} \cO_{(1, 1)}^{\cI}.
\end{aligned}
\label{eq: ODOa ODOD for H8}
\end{equation}
The multiplications \eqref{eq: OaOb OaOD} and \eqref{eq: ODOa ODOD for H8} imply that the interface algebra $\cA$ is non-commutative.
Thus, $\cA$ cannot be a direct sum of eight one-dimensional subalgebras.
On the other hand, since the algebra $\cA$ is acting on the self-interface, it must have one-dimensional representations as shown in Sec.~\ref{sec: Interface modes}.
In other words, the direct sum decomposition of $\cA$ has to contain one-dimensional subalgebras.
Therefore, it follows that $\cA$ is decomposed as
\begin{equation}
\cA = \C \oplus \C \oplus \C \oplus \C \oplus M_2(\C).
\label{eq: decomposition H8}
\end{equation}
The basis of each one-dimensional component $\C$ is given by
\begin{equation}
e_1^+ = \frac{1}{2} (e_1 + e_1^{\prime}), \qquad 
e_1^- = \frac{1}{2} (e_1 - e_1^{\prime}), \qquad 
e_4^+ = \frac{1}{2} (e_4 + ie_4^{\prime}), \qquad 
e_4^- = \frac{1}{2} (e_4 - ie_4^{\prime}),
\label{eq: one-dim basis}
\end{equation}
where $e_1$, $e_1^{\prime}$, $e_4$, and $e_4^{\prime}$ are defined by Eq.~\eqref{eq: pre-idempotent}.
On the other hand, a basis $\{e_{ij} \mid i, j = 1, 2\}$ of the four-dimensional component $M_2(\C)$ can be written as
\begin{equation}
e_{11} = e_2, \qquad e_{12} = e_2^{\prime}, \qquad e_{21} = e_3^{\prime}, \qquad e_{22} = e_3.
\label{eq: four-dim basis}
\end{equation}
The multiplication of the bases in Eqs.~\eqref{eq: one-dim basis} and \eqref{eq: four-dim basis} is 
\begin{equation}
e_i^s e_j^t = \delta_{ij} \delta_{st} e_i^s, \qquad e_{ij} e_{kl} = \delta_{jk} e_{il}, \qquad e_i^s e_{jk} = e_{jk} e_i^s = 0.
\end{equation}
The direct sum decomposition \eqref{eq: decomposition H8} implies that $\cA$ has four one-dimensional representations $\{R_i^s \mid s = \pm, ~i = 1, 4\}$ and a single two-dimensional irreducible representation $R$.
The (one by one) representation matrices of $e_j^t$ and $e_{jk}$ in one-dimensional representation $R_i^s$ are given by
\begin{equation}
R_i^s(e_j^t) = \delta_{ij} \delta_{st}, \qquad R_i^s(e_{jk}) = 0.
\end{equation}
Similarly, the (two by two) representation matrices of $e_i^s$ and $e_{ij}$ in two-dimensional representation $R$ are given by
\begin{equation}
R(e_i^s) = 0, \quad 
R(e_{11}) = 
\begin{pmatrix}
1 & 0 \\
0 & 0
\end{pmatrix}, \quad
R(e_{12}) = 
\begin{pmatrix}
0 & 1 \\
0 & 0
\end{pmatrix}, \quad
R(e_{21}) = 
\begin{pmatrix}
0 & 0 \\
1 & 0
\end{pmatrix}, \quad
R(e_{22}) = 
\begin{pmatrix}
0 & 0 \\
0 & 1
\end{pmatrix}.
\end{equation}
Going back to the original basis $\{\cO_a^{\cI}, (\cO_D^{\cI})_{ij} \mid a \in \Z_2 \times \Z_2, ~i, j = 1, 2\}$ of the interface algebra $\cA$, we can summarize the action of symmetry operators in each representation as in Table~\ref{tab: reps H8}.
The group of one-dimensional representations $\{R_i^s \mid s = \pm, ~i = 1, 4\}$ is isomorphic to $\Z_2 \times \Z_2$, which is the group of group-like elements of $H_8^* \cong H_8$.
\begin{table}[t]
\begin{tabular}{c|ccccc}
& $R_1^+$ & $R_1^-$ & $R_4^+$ & $R_4^-$ & $R$ \\ \hline
$\cO_{(0, 0)}^{\cI}$ & 1 & 1 & 1 & 1 & $I_2$ \\
$\cO_{(0, 1)}^{\cI}$ & 1 & 1 & -1 & -1 & $Z$ \\
$\cO_{(1, 0)}^{\cI}$ & 1 & 1 & -1 & -1 & $-Z$ \\
$\cO_{(1, 1)}^{\cI}$ & 1 & 1 & 1 & 1 & $-I_2$ \\
$(\cO_D^{\cI})_{ij}$ & $\delta_{ij}$ & $-\delta_{ij}$ & $Y_{ij}$ & $-Y_{ij}$ &
$
\begin{pmatrix}
\frac{1}{2}(X+iY) & \frac{1}{2}(X-iY) \\
\frac{1}{2}(X-iY) & -\frac{1}{2}(X+iY) 
\end{pmatrix}_{ij}$
\end{tabular}
\caption{Irreducible representations of the self-interface algebra $\cA$ for the $\Rep(H_8)$ SPT phase. One-dimensional representations form a group $\Z_2 \times \Z_2$ under the multiplication. On the rightmost column, $I_2$ denotes the $2 \times 2$ identity matrix, and the bottom entry means that $R((\cO_D^{\cI})_{11}) = \frac{1}{2}(X+iY)$, etc.}
\label{tab: reps H8}
\end{table}

\section{The $G \times \Rep(G)$ cluster state as a trivial $\Rep(G)$ SPT state}
\label{sec: GxRep(G) cluster state}
In this appendix, we show that the $G \times \Rep(G)$ cluster state in \cite{Fechisin:2023dkj} is the exact ground state of the $\Rep(G)$ symmetric model discussed in \cite[Section 4.4]{Inamura:2021szw}.
This gives us another piece of evidence for the claim that the $G \times \Rep(G)$ cluster state is in the $\Rep(G)$-symmetric trivial phase labeled by the forgetful functor of $\Rep(G)$.

To begin with, we review the general construction of commuting projector models and their ground states with $\Rep(H)$ symmetry following \cite{Inamura:2021szw}, where $H$ is a general finite semisimple Hopf algebra.
The physical Hilbert space of the model is chosen to be a finite semisimple left $H$-comodule algebra $K$, i.e., an algebra equipped with a left $H$-comodule action $\lambda_K: K \rightarrow H \otimes K$ that is compatible with the algebra structure on $K$.
Since $K$ is a semisimple algebra, it admits a unique $\Delta$-separable symmetric Frobenius algebra structure \cite{Fuchs:2002cm, FS2008}.
We denote the multiplication and comultiplication of the Frobenius algebra $K$ by $m_K: K \otimes K \rightarrow K$ and $\Delta_K: K \rightarrow K \otimes K$ respectively.
The Hamiltonian of the model is given by the sum of local commuting projectors as follows:
\begin{equation}
H = - \sum_i h_{i, i+1}, \quad h_{i, i+1} = \Delta_K \circ m_K: K \otimes K \rightarrow K \otimes K.
\label{eq: commuting projector}
\end{equation}
Here, $h_{i, i+1}$ is a local term that acts only on sites $i$ and $i+1$.
The ground states of this Hamiltonian are in one-to-one correspondence with simple left $K$-modules.
Specifically, the ground state labeled by a simple left $K$-module $M$ can be written as an MPS
\begin{equation}
\ket{A_M} = \;\adjincludegraphics[valign=c, scale=1, trim={10pt 10pt 10pt 10pt}]{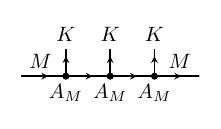},
\end{equation}
where $A_M: M \rightarrow K \otimes M$ denotes the left $K$-comodule action on $M$ defined by
\begin{equation}
A_M = (\text{id}_K \otimes \rho_M) \circ ((\Delta_K \circ u) \otimes \text{id}_M).
\label{eq: Rep(H) MPS}
\end{equation}
Here, $\rho_M: K \otimes M \rightarrow M$ denotes the left $K$-action on $M$ and $u: \C \rightarrow K$ denotes the unit of $K$.
It turns out that the Hamiltonian \eqref{eq: commuting projector} realizes the $\Rep(H)$-symmetric gapped phase that corresponds to a $\Rep(H)$-module category ${}_K \cM$, the category of left $K$-modules \cite{Inamura:2021szw}.
The above construction produces all gapped phases with $\Rep(H)$ symmetry because any $\Rep(H)$-module category is equivalent to ${}_K \cM$ for some left $H$-comodule algebra $K$ \cite{AM2007}.
In particular, the Hamiltonian \eqref{eq: commuting projector} realizes a $\Rep(H)$ SPT phase when $K$ is simple.

Now, let us restrict our attention to the $\Rep(G)$ SPT phase corresponding to the forgetful functor of $\Rep(G)$ \cite[Section 4.4]{Inamura:2021szw}.
To obtain this SPT phase, we choose the physical Hilbert space $K$ to be the smash product $K = \C[G]^* \# \C[G]$ of the group algebra $\C[G]$ and its dual $\C[G]^*$.
Concretely, $K$ is isomorphic to $\C[G]^* \otimes \C[G]$ as a vector space, and the multiplication and comultiplication of $K$ are given by 
\begin{equation}
\begin{aligned}
& (v^{g_1} \# v_{h_1}) \cdot (v^{g_2} \# v_{h_2}) = \delta_{g_1, h_1g_2} v^{g_1} \# v_{h_1h_2},\\
& \Delta_K \circ u(1) = \frac{1}{|G|} \sum_{g, h \in G} (v^g \# v_h) \otimes (v^{h^{-1}g} \# v_{h^{-1}}).
\end{aligned} 
\label{eq: smash}
\end{equation}
Here, $\{v^g \mid g \in G\}$ and $\{v_g \mid g \in G\}$ are dual bases of $\C[G]^*$ and $\C[G]$.
The multiplication of $K$ is simply written as $\cdot$ in Eq.~\eqref{eq: smash}.
The algebra $K$ defined above is simple,\footnote{A straightforward computation shows that the center of $K$ is one-dimensional and hence $K$ is simple.} meaning that an irreducible left $K$-module $M$ is unique up to isomorphism.
Since the dimension of $K$ is $|G|^2$, the dimension of its irreducible module $M$ is $\sqrt{|G|^2} = |G|$.
Thus, $M$ is isomorphic to $\C[G]$ as a vector space.
If we denote the basis of $M$ as $\{w_g \mid g \in G\}$, the action of $K$ on $M$ can be written explicitly as
\begin{equation}
(v^g \# v_h) \cdot w_l = \delta_{l, h^{-1}g} w_{hl},
\label{eq: cluster module}
\end{equation}
where $\cdot$ represents the left $K$-action $\rho_M: K \otimes M \rightarrow M$.
Substituting Eqs.~\eqref{eq: smash} and \eqref{eq: cluster module} into Eq.~\eqref{eq: Rep(H) MPS}, we find that the MPS tensor $A_M$ of the ground state is given by
\begin{equation}
A_M(w_g) = \frac{1}{|G|} \sum_{h \in G} (v^g \# v_h) \otimes w_{h^{-1} g}.
\end{equation}
In terms of diagrams, this MPS tensor can be written as\footnote{The mismatch of the orientations of the arrows in Eq.~\eqref{eq: AM=clsuter} is not a typo: taking the standard basis of $\C[G]$, one can verify that the components of the tensors on both sides agree with each other.}
\begin{equation}
\adjincludegraphics[valign=c, scale=1, trim={10pt 10pt 10pt 10pt}]{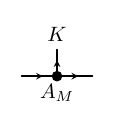}\:
= \;\adjincludegraphics[valign=c, scale=1, trim={10pt 10pt 10pt 10pt}]{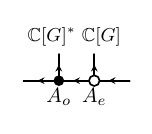}\;,
\label{eq: AM=clsuter}
\end{equation}
where
\begin{equation}
\adjincludegraphics[valign=c, scale=1, trim={10pt 10pt 10pt 10pt}]{tikz/out/MPS_odd.pdf}\; = \sum_{g \in G} P_g \otimes \ket{g}_{\phys}, \quad
\adjincludegraphics[valign=c, scale=1, trim={10pt 10pt 10pt 10pt}]{tikz/out/MPS_even.pdf}\; = \frac{1}{|G|} \sum_{g \in G} L_g \otimes \ket{g}_{\phys}.
\end{equation}
The operators $P_g$ and $L_g$ are defined by $P_g = \ket{g} \bra{g}$ and $L_g = \sum_{h \in G} \ket{gh} \bra{h}$ as in Eq.~\eqref{eq: Pg Lg}.
The above MPS tensors agree with those of the $G \times \Rep(G)$ cluster state in \cite{Fechisin:2023dkj} up to normalization.

\section{Weak completeness relation}
\label{sec: Weak completeness relation}
In this appendix, we show the following weak completeness relation, cf. eq.~\eqref{eq: weak completeness}:
\begin{equation}
\sum_i \;\adjincludegraphics[valign=c, scale=1.2, trim={10pt 10pt 10pt 10pt}]{tikz/out/weak_completeness1.pdf}\;
=\;\adjincludegraphics[valign=c, scale=1.2, trim={10pt 10pt 10pt 10pt}]{tikz/out/weak_completeness2.pdf}\;.
\label{eq: weak completeness appendix}
\end{equation}
Here, $\psi$ is an arbitrary unitary operator acting on the bond Hilbert space $V$ of the MPS.
The bond Hilbert space of the MPO $\widehat{\cO}_x$ will be denoted by $V_x$.
In what follows, every diagram is supposed to represent a linear map from the left side to the right side.
For instance, the diagrams in eq.~\eqref{eq: weak completeness appendix} represent linear maps from $V^*$ to $V_x \otimes V^*$.

To show eq.~\eqref{eq: weak completeness appendix}, we assume the following dual orthogonality relation:
\begin{equation}
\;\adjincludegraphics[scale=1.2,trim={10pt 10pt 10pt 10pt},valign = c]{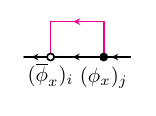}\;
= \delta_{ij} \;\adjincludegraphics[scale=1.2,trim={10pt 10pt 10pt 10pt},valign = c]{tikz/out/orthogonality2.pdf}\;.
\label{eq: dual action orthogonality}
\end{equation}
This equality should be regarded as the orthogonality relation for the dual object $x^* \in \cC$.
Equation~\eqref{eq: dual action orthogonality} implies that the action tensor $(\phi_x)_j$ is surjective as a linear map from $V_x^* \otimes V^*$ to $V^*$.
Correspondingly, the image of $(\phi_x)_j$ viewed as a linear map from $V^*$ to $V_x \otimes V^*$ is of the form
\begin{equation}
\text{Im}((\phi_x)_j) = \text{Span}\{w_{\alpha} \otimes v_{\alpha} \mid \alpha = 1, 2, \cdots, \dim V\} \subset V_x \otimes V^*,
\end{equation}
where $\{v_{\alpha} \mid \alpha = 1, 2, \cdots, \dim V\}$ is a basis of $V^*$ and $w_{\alpha}$ is an element of $V_x$.

Now, since $\psi$ is unitary, it is diagonalizable by another unitary operator $U$ as follows:
\begin{equation}
U^{\dagger} \psi U = \text{diag}(\lambda_1, \lambda_2, \cdots, \lambda_{\dim V}), \quad \abs{\lambda_{i}}=1.
\end{equation}
If we choose a basis $\{v_{\alpha}\}$ to be the set of eigenvectors of $\psi$, the action of $\text{id}_{V_x} \otimes \psi$ on $w_{\alpha} \otimes v_{\alpha}$ can be computed as
$(\text{id}_{V_x} \otimes \psi) (w_{\alpha} \otimes v_{\alpha}) = \lambda_{\alpha} w_{\alpha} \otimes v_{\alpha}$.
In particular, the action of $\text{id}_{V_x} \otimes \psi$ preserves the image of $(\phi_x)_{j}$, i.e.,
\begin{equation}
\text{Im}((\text{id}_{V_x} \otimes \psi) (\phi_x)_j) = \text{Im}((\phi_x)_j).
\end{equation}
Recalling that the idempotent $\sum_i (\phi_x)_i (\overline{\phi}_x)_i$ acts as the identity on $\text{Im}((\phi_x)_j)$ due to the orthogonality relation~\eqref{eq: action orthogonality}, we find that it also acts as the identity on $\text{Im}((\text{id}_{V_x} \otimes \psi) (\phi_x)_j)$.
Thus, the weak completeness relation~\eqref{eq: weak completeness appendix} holds.
We note that the same proof applies to any diagonalizable $\psi$.

\section{Hochschild cohomology}
\label{sec: cohomology}

\subsection{Proof of $\delta^{2}=0$}
In Sec.~\ref{sec: Abelianization of the nonabelian factor system}, we introduced a cohomology theory $\cohoZ{k}{\mathcal{C}}$. 
In that section, we utilized the property $\delta^{2}=0$ without proving it. 
In this section, we show it explicitly:
\begin{proposition}
    $\delta^{2}=0:C^{k}(\mathcal{C}, \mathbb{Z})\to C^{k+2}(\mathcal{C}, \mathbb{Z}).$
\end{proposition}
\begin{proof}\;\\
  Let $n$ be an aribitrary element in $C^{k}(\mathcal{C}, \mathbb{Z})$ 
  and $\{\rho_{i}\}_{i=1}^{k+2}$ be simple objects in $\mathcal{C}$.  
  We compute $(\delta^{2}n)(\rho_{1},...,\rho_{k+2})$ directly:
\begin{align*}
  &\quad(\delta^2 n)(\rho_1,...,\rho_{k+2})\\
  &=\dim(\rho_1)(\delta n)(\rho_2,...,\rho_{k+1})-(\delta n)(\rho_1\otimes\rho_2,\rho_3,...,\rho_{k+1})+\sum_{i=2}^{k}(-1)^{i}(\delta n)(\rho_1,...,\rho_i\otimes\rho_{i+1},...,\rho_{k+1})\\
  &\quad+(-1)^{k+1}(\delta n)(\rho_1,...,\rho_k,\rho_{k+1}\otimes\rho_{k+2})+(-1)^{k+2}(\delta n)(\rho_1,...,\rho_{k+1})\dim(\rho_{k+2}),\\
  &=\dim(\rho_1)\{\dim(\rho_2)n(\rho_3,...,\rho_{k+2})+\sum_{j=2}^{k+1}(-1)^{j-1}n(\rho_2,...,\rho_j\otimes\rho_{j+1},...,\rho_{k+2})\\
  &\quad+(-1)^{k+1}n(\rho_2,...,\rho_{k+1})\dim(\rho_{k+2})\}-\{\dim(\rho_1\otimes\rho_2)n(\rho_3,...,\rho_{k+2})-n(\rho_1\otimes\rho_2\otimes\rho_3,\rho_4,...,\rho_{k+2})\\
  &\quad+\sum_{j=3}^{k+1}(-1)^{j+1}n(\rho_1\otimes\rho_2,...,\rho_{j}\otimes\rho_{j+1},...,\rho_{k+2})
  +(-1)^{k}n(\rho_1\otimes\rho_2,...,\rho_{k+1})\dim(\rho_{k+2})\}\\
  &\quad+\sum_{i=2}^{k}(-1)^i\{\dim(\rho_1)n(\rho_2,...,\rho_i\otimes\rho_{i+1},...,\rho_{k+2})+\sum_{j=1}^{i-2}(-1)^{j}n(\rho_1,...,\rho_j\otimes\rho_{j+1},...,\rho_i\otimes\rho_{i+1},...,\rho_{k+2})\\
  &\quad+(-1)^{i+1}n(\rho_1,...,\rho_{i-1}\otimes\rho_i\otimes\rho_{i+1},...,\rho_{k+2})
  +(-1)^{i}n(\rho_1,...,\rho_{i-1},\rho_i\otimes\rho_{i+1}\otimes\rho_{i+2},...,\rho_{k+2})\\
  &\quad+\sum_{j=i+2}^{k+1}(-1)^{j+1}n(\rho_1,...,\rho_i\otimes\rho_{i+1},...,\rho_{j}\otimes\rho_{j+1},...,\rho_{k+2})\\
  &\quad+(-1)^kn(\rho_1,...,\rho_i\otimes\rho_{i+1},...,\rho_{k+1})\dim(\rho_{k+2})\}\\
  &\quad+(-1)^{k+1}\{\dim(\rho_1)n(\rho_2,...,\rho_k,\rho_{k+1}\otimes\rho_{k+2})
  +\sum_{j=1}^{k-1}(-1)^jn(\rho_1,...,\rho_{j}\otimes\rho_{j+1},...,\rho_k,\rho_{k+1}\otimes\rho_{k+2})\\
  &\quad+(-1)^{k}n(\rho_1,...,\rho_{k-1},\rho_{k}\otimes\rho_{k+1}\otimes\rho_{k+2})
  +(-1)^{k+1}n(\rho_1,...,\rho_k)\dim(\rho_{k+1}\otimes\rho_{k+2})
  \}\\
  &\quad+(-1)^{k+2}\dim(\rho_{k+2})\{\dim(\rho_1)n(\rho_2,...,\rho_{k+1})
  +\sum_{j=1}^k(-1)^jn(\rho_1,...,\rho_{j}\otimes\rho_{j+1},...,\rho_{k+1})\\
  &\quad+(-1)^{k+1}n(\rho_1,...,\rho_{k})\dim(\rho_{k+1})
  \}.
\end{align*}
Since the expression is long, let us calculate it term by term. Note that the first line, the second and third lines, the fourth to sixth lines, the seventh and eighth lines, and the ninth line correspond to the first through fifth terms of the first equation, respectively. Let us denote them as $l_1,l_2,...,l_5$:
\begin{align*}
    l_1=&\dim(\rho_1)\{\dim(\rho_2)n(\rho_3,...,\rho_{k+2})+\sum_{j=2}^{k+1}(-1)^{j-1}n(\rho_2,...,\rho_j\otimes\rho_{j+1},...,\rho_{k+2})\\
    &+(-1)^{k+1}n(\rho_2,...,\rho_{k+1})\dim(\rho_{k+2})\},\\
    l_2=&-\{\dim(\rho_1\otimes\rho_2)n(\rho_3,...,\rho_{k+2})-n(\rho_1\otimes\rho_2\otimes\rho_3,\rho_4,...,\rho_{k+2})\\
    &+\sum_{j=3}^{k+1}(-1)^{j+1}n(\rho_1\otimes\rho_2,...,\rho_{j}\otimes\rho_{j+1},...,\rho_{k+2})+(-1)^{k}n(\rho_1\otimes\rho_2,...,\rho_{k+1})\dim(\rho_{k+2})
    \},\\
    l_3=&+\sum_{i=2}^{k}(-1)^i\{\dim(\rho_1)n(\rho_2,...,\rho_i\otimes\rho_{i+1},...,\rho_{k+2})\\
    &+\sum_{j=1}^{i-2}(-1)^{j}n(\rho_1,...,\rho_j\otimes\rho_{j+1},...,\rho_i\otimes\rho_{i+1},...,\rho_{k+2})\\
    &+(-1)^{i+1}n(\rho_1,...,\rho_{i-1}\otimes\rho_i\otimes\rho_{i+1},...,\rho_{k+2})
    +(-1)^{i}n(\rho_1,...,\rho_{i-1},\rho_i\otimes\rho_{i+1}\otimes\rho_{i+2},...,\rho_{k+2})\\
    &+\sum_{j=i+2}^{k+1}(-1)^{j+1}n(\rho_1,...,\rho_i\otimes\rho_{i+1},...,\rho_{j}\otimes\rho_{j+1},...,\rho_{k+2})\\
    &+(-1)^kn(\rho_1,...,\rho_i\otimes\rho_{i+1},...,\rho_{k+1})\dim(\rho_{k+2})\},\\
    l_4=&+(-1)^{k+1}\{\dim(\rho_1)n(\rho_2,...,\rho_k,\rho_{k+1}\otimes\rho_{k+2})
    +\sum_{j=1}^{k-1}(-1)^jn(\rho_1,...,\rho_{j}\otimes\rho_{j+1},...,\rho_k,\rho_{k+1}\otimes\rho_{k+2})\\
    &+(-1)^{k}n(\rho_1,...,\rho_{k-1},\rho_{k}\otimes\rho_{k+1}\otimes\rho_{k+2})
    +(-1)^{k+1}n(\rho_1,...,\rho_k)\dim(\rho_{k+1}\otimes\rho_{k+2})
    \},\\
    l_5=&+(-1)^{k+2}\dim(\rho_{k+2})\{\dim(\rho_1)n(\rho_2,...,\rho_{k+1})
    +\sum_{j=1}^k(-1)^jn(\rho_1,...,\rho_{j}\otimes\rho_{j+1},...,\rho_{k+1})\\
    &+(-1)^{k+1}n(\rho_1,...,\rho_{k})\dim(\rho_{k+1})
    \}.
\end{align*}
First,  $l_1$ and $l_5$ cancel out the first and last terms of 
$l_2,l_3$, and $l_4$. Let us denote the terms obtained by removing the first and last terms from $l_2, l_3,$ and $ l_4 $ as $ l_2', l_3', $ and $ l_4' $, respectively:
\begin{align*}
    l_2^\prime=&-\{-n(\rho_1\otimes\rho_2\otimes\rho_3,\rho_4,...,\rho_{k+2})
    +\sum_{j=3}^{k+1}(-1)^{j+1}n(\rho_1\otimes\rho_2,...,\rho_{j}\otimes\rho_{j+1},...,\rho_{k+2})
    \},\\
    l_3^\prime=&+\sum_{i=2}^{k}(-1)^i\{
    \sum_{j=1}^{i-2}(-1)^{j}n(\rho_1,...,\rho_j\otimes\rho_{j+1},...,\rho_i\otimes\rho_{i+1},...,\rho_{k+2})\\
    &+(-1)^{i+1}n(\rho_1,...,\rho_{i-1}\otimes\rho_i\otimes\rho_{i+1},...,\rho_{k+2})
    +(-1)^{i}n(\rho_1,...,\rho_{i-1},\rho_i\otimes\rho_{i+1}\otimes\rho_{i+2},...,\rho_{k+2})\\
    &+\sum_{j=i+2}^{k+1}(-1)^{j+1}n(\rho_1,...,\rho_i\otimes\rho_{i+1},...,\rho_{j}\otimes\rho_{j+1},...,\rho_{k+2})
    \},\\
    l_4^\prime=&(-1)^{k+1}\{\sum_{j=1}^{k-1}(-1)^jn(\rho_1,...,\rho_{j}\otimes\rho_{j+1},...,\rho_k,\rho_{k+1}\otimes\rho_{k+2})
    +(-1)^{k}n(\rho_1,...,\rho_{k-1},\rho_{k}\otimes\rho_{k+1}\otimes\rho_{k+2})
    \}.
\end{align*}
Let us focus on the second and third terms of $l_3^\prime$.
\begin{align*}
    &\quad \sum_{i=2}^{k}(-1)^{2i+1}n(\rho_1,...,\rho_{i-1}\otimes\rho_i\otimes\rho_{i+1},...,\rho_{k+2})
    +\sum_{i=2}^{k}(-1)^{2i}n(\rho_1,...,\rho_{i-1},\rho_i\otimes\rho_{i+1}\otimes\rho_{i+2},...,\rho_{k+2})\\
    &=-n(\rho_1\otimes\rho_2\otimes\rho_3,\rho_4...,\rho_{k+2})+n(\rho_1,...,\rho_{k}\otimes\rho_{k+1}\otimes\rho_{k+2}).
\end{align*}
This term cancels out the first term of $l_1'$ and the last term of $l_4' $. Next, let us focus on the first and fourth terms of $l_3^\prime$:
{\small
\begin{align*}
    &\quad\sum_{i=2}^{k}
    \sum_{j=1}^{i-2}(-1)^{i+j}n(\rho_1,...,\rho_j\otimes\rho_{j+1},...,\rho_i\otimes\rho_{i+1},...,\rho_{k+2})+\sum_{i=2}^{k}\sum_{j=i+2}^{k+1}(-1)^{i+j+1}n(\rho_1,...,\rho_i\otimes\rho_{i+1},...,\rho_{j}\otimes\rho_{j+1},...,\rho_{k+2})\\
    &=\sum_{j=1}^{k-1}\sum_{i=j+2}^{k}(-1)^{i+j}n(\rho_1,...,\rho_j\otimes\rho_{j+1},...,\rho_i\otimes\rho_{i+1},...,\rho_{k+2})-\sum_{i=2}^{k}\sum_{j=i+2}^{k+1}(-1)^{i+j}n(\rho_1,...,\rho_i\otimes\rho_{i+1},...,\rho_{j}\otimes\rho_{j+1},...,\rho_{k+2})\\
    &=\sum_{i=1}^{k-1}\sum_{j=i+2}^{k}(-1)^{i+j}n(\rho_1,...,\rho_i\otimes\rho_{i+1},...,\rho_j\otimes\rho_{j+1},...,\rho_{k+2})-\sum_{i=2}^{k}\sum_{j=i+2}^{k+1}(-1)^{i+j}n(\rho_1,...,\rho_i\otimes\rho_{i+1},...,\rho_{j}\otimes\rho_{j+1},...,\rho_{k+2})\\
    &=\sum_{j=3}^{k}(-1)^{1+j}n(\rho_1\otimes\rho_{2},...,\rho_j\otimes\rho_{j+1},...,\rho_{k+2})-
    \sum_{i=2}^{k-1}(-1)^{i+k+1}n(\rho_1,...,\rho_i\otimes\rho_{i+1},...,\rho_{k+1}\otimes\rho_{k+2}).
\end{align*}
}%
Here, we use $\sum_{i=2}^{k}\sum_{j=i+2}^{k+1}=\sum_{j=1}^{k-1}\sum_{i=j+2}^{k}$. On the other hand, the second term of $l_2'$ and the first term of $l_4'$ are
\begin{align*}
    &-\sum_{j=3}^{k+1}(-1)^{j+1}n(\rho_1\otimes\rho_2,...,\rho_{j}\otimes\rho_{j+1},...,\rho_{k+2})
    +\sum_{j=1}^{k-1}(-1)^{j+k+1}n(\rho_1,...,\rho_{j}\otimes\rho_{j+1},...,\rho_k,\rho_{k+1}\otimes\rho_{k+2})\\
    =&-\sum_{j=3}^{k}(-1)^{j+1}n(\rho_1\otimes\rho_2,...,\rho_{j}\otimes\rho_{j+1},...,\rho_{k+2})
    +\sum_{j=2}^{k-1}(-1)^{j+k+1}n(\rho_1,...,\rho_{j}\otimes\rho_{j+1},...,\rho_k,\rho_{k+1}\otimes\rho_{k+2}).
\end{align*}
This term cancels out the contribution from the previous one. Consequently, $\delta^2=0$.
\end{proof}

\subsection{The $2$nd Hochschild cohomology for $\rep{D_{8}}$}
In this section, we show the following proposition:
\begin{proposition}
    \begin{eqnarray}
        \cohoZ{2}{\rep{D_{8}}}\simeq(\zmod{2}\times\zmod{2})\ltimes_{\omega}\zmod{4},
    \end{eqnarray}
    where $\omega$ is a nontrivial element in $\coho{1}{(\zmod{2})^2}{\zmod{4}}^{\otimes2}\subset\coho{2}{(\zmod{2})^2}{\zmod{4}}$
\end{proposition}
\begin{proof}
    We denote the elements of $D_{8}$ as $D_{8}=\{\sigma^i\tau^j\left|\right.\sigma^4=1,\tau^2=1,\sigma\tau=\tau\sigma^3\}$, and denote the simple objects of $\rep{D_{8}}$ as $\Simp(\rep{D_{8}})=\{(0,0),(0,1),(1,0),(1,1),D\}$. Here, $(a,b)$ represents the sign representation with $(\sigma,\tau)=((-1)^{a},(-1)^{b})$ and $D$ is the unique two-dimensional representation of $D_{8}$.
    
    Let $n\in \mathrm{Z}^2(\rep{D_{8}}, \mathbb{Z})$ be a cocycle, i.e.
    \begin{eqnarray}
        (\delta n)(\rho_1,\rho_2,\rho_3)=0,
    \end{eqnarray}
    for any simple objects $\rho_1,\rho_2,\rho_3$. First, we consider the case where $\rho_1,\rho_2,\rho_3$ are one-dimensional representations. Since the one-dimensional representations follow the multiplication rule of the group $\mathbb{Z}_2 \times \mathbb{Z}_2 $, the restrictions on $n$ for one-dimensional representations are the same as that of the cocycles for the group cohomology of $\zmod{2}\times\zmod{2}$. Then, the remaining degrees of freedom are
    \begin{eqnarray}
        n((0,0),(0,0)), \quad n((0,1),(0,1)), \quad n((1,0),(1,0)).
    \end{eqnarray}
    The others are determined as 
    \begin{eqnarray}
        n(\rho_1,\rho_2)=\begin{cases}
            n((0,1),(0,1))+n((1,0),(1,0))& \text{ if }(\rho_1,\rho_2)=((1,1),(1,1)),\\
            n((0,0),(0,0))&\text{ otherwise, }
        \end{cases}
    \end{eqnarray}
    for one-dimensional representations $(\rho_1,\rho_2)$. These are all the constraints for the one-dimensional representations.

    Similarly, for one-dimensional representation $\rho$, $n(\rho,D)$ and $n(D,\rho)$ are completely fixed by the cocycle condition as follows: Since
    \begin{eqnarray}
        (\delta n)((0,0),(0,0),D)&=&n((0,0),D)-n((0,0),D)+n((0,0),D)-2n((0,0),(0,0))=0,\\
        (\delta n)(D,(0,0),(0,0))&=&2n((0,0),(0,0))-n(D,(0,0))+n(D,(0,0))-n(D,(0,0))=0,
    \end{eqnarray}
    we obtain
    \begin{eqnarray}\label{eq:even}
        n((0,0),D)=n(D,(0,0))=2n((0,0),(0,0)).
    \end{eqnarray}
    Also, since
    \begin{eqnarray}
        (\delta n)(\rho,\rho,D)&=&n(\rho,D)-n((0,0),D)+n(\rho,D)-2n(\rho,\rho)=0,\\
        (\delta n)(D,\rho,\rho)&=&2n(\rho,\rho)-n(D,\rho)+n(D,(0,0))-n(D,\rho)=0,
    \end{eqnarray}
    we obtain
    \begin{eqnarray}
        n(\rho,D)&=&n(\rho,\rho)+\frac{1}{2}n((0,0),D)=n(\rho,\rho)+n((0,0),(0,0)),\\
        n(D,\rho)&=&n(\rho,\rho)-\frac{1}{2}n(D,(0,0))=n(\rho,\rho)-n((0,0),(0,0)).
    \end{eqnarray}
    In contrast, there are no restrictions on $n(D,D)$. In summary, $n((0,0),(0,0)), n((0,1),(0,1)), n((1,0),(1,0))$ and $n(D,D)$ remain as undetermined degrees of freedom.

    Next, let us consider the coboundary shift by $l\in C^1(\rep{D_{8}}, \mathbb{Z})$. First, the coboundary action for $n((0,0),(0,0)), n((0,1),(0,1)), n((1,0),(1,0))$ are
    \begin{eqnarray}
        (\delta l)((0,0),(0,0))&=&l((0,0)),\\
        (\delta l)((0,1),(0,1))&=&2l((0,1))-l((0,0)),\\
        (\delta l)((1,0),(1,0))&=&2l((1,0))-l((0,0)).
    \end{eqnarray}
    By using the first equation, we normalize $n((0,0),(0,0))$ to $0$. The second and third equations imply that we can freely shift $n((0,1),(0,1)), n((1,0),(1,0))$ by even integers. Thus, in cohomology theory, $n((0,1),(0,1)), n((1,0),(1,0))$ are regarded as elements of $\zmod{2}$. By using these even integer shift, we normalize $n((0,1),(0,1)), n((1,0),(1,0))$ to $0$ or $1$.

    The coboundary action for $n(D,D)$ is
    \begin{eqnarray}
        (\delta l)(D,D)=4l(D)-l((0,0))-l((0,1))-l((1,0))-l((1,1)).
    \end{eqnarray}
    Since we already fix $l((0,0)),l((0,1)),l((1,0)),l((1,1))$, we cannot shift $n(D,D)$ by $1$. However, by using $(\delta l)(D,D)$, we can shift it by multiples of $4$. Thus, in cohomology theory, $n(D,D)$ is regarded as an element of $\zmod{4}$. Consequently, we show that 
    \begin{eqnarray}
        \cohoZ{2}{\rep{D_{8}}}=\zmod{2}\times\zmod{2}\times\zmod{4},
    \end{eqnarray}
    as a set, where first and second $\zmod{2}$ is a choice of $n((0,1),(0,1))=0,1$ and $n((1,0),(1,0))=0,1$ respectively, and $\zmod{4}$ is a choice of $n(D,D)=0,1,2,3$.

    Finally, we identify the group structure of $\cohoZ{2}{\rep{D_{8}}}$. Let us take $(1,0,0)\in\zmod{2}\times\zmod{2}\times\zmod{4}$. This element is represented as 
    \begin{eqnarray}
        n(\rho,\rho)=\begin{cases}
            1,&\text{ if }(\rho,\rho)=((0,1),(0,1)),((1,1),(1,1)),\\
            0,& \text{ otherwise. }
        \end{cases}
    \end{eqnarray}
    Thus, multiplying this element by two results in 
    \begin{eqnarray}
        n(\rho,\rho)=\begin{cases}
            2,&\text{ if }(\rho,\rho)=((0,1),(0,1)),((1,1),(1,1)),\\
            0,& \text{ otherwise. }
        \end{cases}
    \end{eqnarray}
    We have decided to normalize $n((0,1),(0,1))$ and $n((1,1),(1,1))$ to $0$ or $1$ by using $l((0,1))$ and $l((1,1))$. Under this renormalization, $n(D,D)$ shifts by $2$. In other words, this means that the element of the first $\mathbb{Z}_2$ has carried over to become $2$ in $\mathbb{Z}_4$. Also, the second $\mathbb{Z}_2$ exhibits the same carry-over structure. Therefore, $\cohoZ{2}{\rep{D_{8}}}$ is a nontrivial extension of the $\zmod{2}\times\zmod{2}$ by $\zmod{4}$.
\end{proof}

\section{Computation of pump invariant for $\Rep(D_8)$}
\label{sec: Computation of pump invariant}
Let us compute the pump invariant for $S^1$-parameterized families constructed from the $D_8 \times \Rep(D_8)$ cluster state.
We note that the pump invariant is already computed in Sec.~\ref{sec: Example: Rep(G) symmetry II} in the case of the general $G \times \Rep(G)$ cluster state.
The computations there were performed in a gauge where the action tensors $\hat{\phi}_{\rho}(\theta)$ for $\rho \in \Rep(G)$ are $2\pi$-periodic.
In this appendix, we do a similar computation in a gauge where the MPS tensor $A(\theta)$ is $2\pi$-periodic.
Furthermore, we not only compute the invariant but also write down $A(\theta)$ and $\hat{\phi}_{\rho}(\theta)$ explicitly for all $\theta \in S^1$.

The dihedral group $D_8$ of order 8 is defined by 
\begin{equation}
  D_{8}=\left< \sigma, \tau \left.  \right| \sigma^{4}=\tau^{2}=1, \sigma\tau=\tau\sigma^{3}  \right>. 
\end{equation}
Under the right regular representation of \( D_8 \), 
the generators $\sigma$ and $\tau$ are given as follows:
\begin{align}
  R(\sigma)=\frac{1}{2\sqrt{2}}\left(
\begin{array}{cccccccc}
0 & 0 & 0 & 0 & 0 & 0 & 1 & 0 \\
0 & 0 & 0 & 0 & 0 & 0 & 0 & 1 \\
1 & 0 & 0 & 0 & 0 & 0 & 0 & 0 \\
0 & 1 & 0 & 0 & 0 & 0 & 0 & 0 \\
0 & 0 & 1 & 0 & 0 & 0 & 0 & 0 \\
0 & 0 & 0 & 1 & 0 & 0 & 0 & 0 \\
0 & 0 & 0 & 0 & 1 & 0 & 0 & 0 \\
0 & 0 & 0 & 0 & 0 & 1 & 0 & 0 \\
\end{array}
\right),\;
  R(\tau)=\frac{1}{2\sqrt{2}}\left(
\begin{array}{cccccccc}
0 & 1 & 0 & 0 & 0 & 0 & 0 & 0 \\
1 & 0 & 0 & 0 & 0 & 0 & 0 & 0 \\
0 & 0 & 0 & 0 & 0 & 0 & 0 & 1 \\
0 & 0 & 0 & 0 & 0 & 0 & 1 & 0 \\
0 & 0 & 0 & 0 & 0 & 1 & 0 & 0 \\
0 & 0 & 0 & 0 & 1 & 0 & 0 & 0 \\
0 & 0 & 0 & 1 & 0 & 0 & 0 & 0 \\
0 & 0 & 1 & 0 & 0 & 0 & 0 & 0 \\
\end{array}
\right).
\end{align}
We denote by $e_{i, j}$ the $8 \times 8$ matrix with 1 in the $(i, j)$ entry and 0 elsewhere.
By using these matrices, an MPS representation of the $G\times \rep{G}$ cluster state is given by
\begin{equation}
  A^{g}_{o}=e_{g,g}
\end{equation}
for odd sites and 
\begin{equation}
  A^{g}_{e}=R(g)
\end{equation}
for even sites~\cite{Fechisin:2023dkj} and the translationally invariant MPS is given by 
\begin{equation}
  A^{g,h}_{oe}:=A^{g}_{o}A^{h}_{e}=e_{g,g}R(h)=e_{g,h^{-1}g}.
\end{equation}
We note that this MPS matrix is injective.

Since the parameterization is introduced by a unitary transformation acting on the physical legs (see Eq.~\eqref{eq: cluster pump}),
the ground state of $H_g(\theta)$ is given by
\begin{equation}
  \left(\prod_{i: \text{odd}} R_{g}(\theta)_i \right)\ket{\{A^{g,h}_{oe}\}}_{L},
\end{equation}
where the subscript $L$ denotes the number of unit cells.
Therefore, we can easily get an MPS representation of the ground state as follows:
{\small
\begin{equation*}
\begin{aligned}
  \left(\prod_{i=1}^{L} R_{g}(\theta)_i \right)\ket{\{A^{g,h}_{oe}\}}_{L}
  &=\sum_{\{g_k,h_k\}}\tr{A^{g_1,h_1}_{oe} \cdots A^{g_L,h_L}_{oe}}\left(R_{g}(\theta)_1\ket{g_{1}}\right)\ket{h_{1}} \cdots\left(R_{g}(\theta)_L\ket{g_{L}}\right)\ket{h_{L}} \\
  &=\sum_{\{g_k,h_k\}}\tr{A^{g_1,h_1}_{oe} \cdots A^{g_L,h_L}_{oe}} \left(\sum_{\tilde{g}_1}\ket{\tilde{g}_{1}}[R_{g}(\theta)_1]_{\tilde{g}_{1},g_{1}}\right)\ket{h_1}  \cdots \left(\sum_{\tilde{g}_L}\ket{\tilde{g}_{L}}[R_{g}(\theta)_L]_{\tilde{g}_{L},g_{L}}\right)\ket{h_L} \\
  &=\sum_{\{\tilde{g}_k,h_k\}}\tr{A^{\tilde{g}_1,h_1}_{oe}(\theta;g) \cdots A^{\tilde{g}_L,h_L}_{oe}(\theta;g)} \ket{\tilde{g}_{1},h_{1}, \cdots,\tilde{g}_{L},h_{L}},
\end{aligned}
\end{equation*}
}
where 
\begin{equation}
  A^{g,h}_{oe}(\theta;k)=\sum_{\tilde{g}} [R_{k}(\theta)]_{g,\tilde{g}} A^{\tilde{g},h}_{oe}.
  \label{eq: theta-dep Agh}
\end{equation}

Let us compute the invariant using a gauge where $\phi(\theta)$ is non-periodic.
Although $A^{g,h}_{oe}(\theta; k)$ is not $2\pi$-periodic as a matrix, 
the $2\pi$-periodicity of the MPS guarantees the existence of the transition function $\hat{U}_{k}$:
\begin{equation}
  A^{g,h}_{oe}(2\pi;k)=\hat{U}_{k} A^{g,h}_{oe}(0;k) \hat{U}_{k}^{\dagger}.
  \label{eq: Agh gauge transformation}
\end{equation}
By explicit calculations, we can check that
\begin{equation}
  \hat{U}_{\sigma}=
  \left(
    \begin{array}{cccccccc}
     0 & 0 & 0 & 0 & 0 & 0 & 1 & 0 \\
     0 & 0 & 0 & 1 & 0 & 0 & 0 & 0 \\
     1 & 0 & 0 & 0 & 0 & 0 & 0 & 0 \\
     0 & 0 & 0 & 0 & 0 & 1 & 0 & 0 \\
     0 & 0 & 1 & 0 & 0 & 0 & 0 & 0 \\
     0 & 0 & 0 & 0 & 0 & 0 & 0 & 1 \\
     0 & 0 & 0 & 0 & 1 & 0 & 0 & 0 \\
     0 & 1 & 0 & 0 & 0 & 0 & 0 & 0 \\
    \end{array}
  \right), \quad
  \hat{U}_{\tau}=\left(
    \begin{array}{cccccccc}
     0 & 1 & 0 & 0 & 0 & 0 & 0 & 0 \\
     1 & 0 & 0 & 0 & 0 & 0 & 0 & 0 \\
     0 & 0 & 0 & 1 & 0 & 0 & 0 & 0 \\
     0 & 0 & 1 & 0 & 0 & 0 & 0 & 0 \\
     0 & 0 & 0 & 0 & 0 & 1 & 0 & 0 \\
     0 & 0 & 0 & 0 & 1 & 0 & 0 & 0 \\
     0 & 0 & 0 & 0 & 0 & 0 & 0 & 1 \\
     0 & 0 & 0 & 0 & 0 & 0 & 1 & 0 \\
    \end{array}
  \right),
\label{eq: Usigma Utau}
\end{equation}
and $\hat{U}_{k}$ is given by a multiplication of them for general $k\in D_{8}$.
We interpolate between $\hat{U}_{k}$ and the identity matrix as
\begin{equation}
  \hat{U}_{k}(\theta):=\exp\left(\frac{\theta}{2\pi}\log(\hat{U}_{k})\right),
  \label{eq: def Uk theta}
\end{equation}
and define a $2\pi$-periodic MPS matrix as 
\begin{equation}
  \tilde{A}^{g,h}_{oe}(\theta;k)=\hat{U}_{k}(\theta)^{\dagger} A^{g,h}_{oe}(\theta;k) \hat{U}_{k}(\theta).
  \label{eq: gauge transf of Agh}
\end{equation}
Note that the $L$-symbol calculated from this MPS is constant. 
This is due to the following reasons: 
First, for the MPS before making it $2\pi$-periodic, the transfer matrix is constant. 
Thus, we can choose a gauge so that the action tensor is constant.
Therefore, the $L$-symbol calculated in this gauge is constant. 
Since the $L$-symbol does not change under the gauge transformation Eq.~\eqref{eq: gauge transf of Agh}, the $L$-symbol calculated with $\tilde{A}^{g,h}_{oe}(\theta;k)$ is also constant.

Let us compute the fractionalized symmetry operator $\hat{\phi}_{\rho}^{k}(\theta)$ for $\rho \in \Rep(D_8)$ defined by
\begin{equation}
  \sum_{\tilde{h}} \mathcal{O}^{h\tilde{h}}_{\rho} \tilde{A}^{g,\tilde{h}}_{oe}(\theta;k) = \hat{\phi}^{k}(\theta)_{\rho}^{-1} \tilde{A}^{g,h}_{oe}(\theta;k) \hat{\phi}^{k}(\theta)_{\rho}.
\end{equation}
By explicit calculations, we obtain
\begin{equation}
  \hat{\phi}^{\sigma}(\theta)_{(0,0)}
  =\left(
    \begin{array}{cccccccc}
     1 & 0 & 0 & 0 & 0 & 0 & 0 & 0 \\
     0 & 1 & 0 & 0 & 0 & 0 & 0 & 0 \\
     0 & 0 & 1 & 0 & 0 & 0 & 0 & 0 \\
     0 & 0 & 0 & 1 & 0 & 0 & 0 & 0 \\
     0 & 0 & 0 & 0 & 1 & 0 & 0 & 0 \\
     0 & 0 & 0 & 0 & 0 & 1 & 0 & 0 \\
     0 & 0 & 0 & 0 & 0 & 0 & 1 & 0 \\
     0 & 0 & 0 & 0 & 0 & 0 & 0 & 1 \\
    \end{array}
    \right), \quad
  \hat{\phi}^{\sigma}(\theta)_{(0,1)}
  =\left(
    \begin{array}{cccccccc}
     -1 & 0 & 0 & 0 & 0 & 0 & 0 & 0 \\
     0 & 1 & 0 & 0 & 0 & 0 & 0 & 0 \\
     0 & 0 & -1 & 0 & 0 & 0 & 0 & 0 \\
     0 & 0 & 0 & 1 & 0 & 0 & 0 & 0 \\
     0 & 0 & 0 & 0 & -1 & 0 & 0 & 0 \\
     0 & 0 & 0 & 0 & 0 & 1 & 0 & 0 \\
     0 & 0 & 0 & 0 & 0 & 0 & -1 & 0 \\
     0 & 0 & 0 & 0 & 0 & 0 & 0 & 1 \\
    \end{array}
    \right),
    \end{equation}
  \begin{equation}
  \hat{\phi}^{\sigma}(\theta)_{(1,0)}
  =
  {\small
  \left(
    \begin{array}{cccccccc}
     -c \frac{\theta }{2} & 0 &  \frac{1+i}{2}   s \frac{\theta}{2} & 0 & 0 & 0 & \frac{-1+i}{2}  s \frac{\theta}{2} & 0 \\
     0 & -c \frac{\theta}{2} & 0 & \frac{-1+i}{2}  s \frac{\theta}{2} & 0 & 0 & 0 &  \frac{1+i}{2}   s \frac{\theta}{2} \\
     \frac{1-i}{2} s \frac{\theta}{2} & 0 & c \frac{\theta}{2} & 0 & -\frac{1+i}{2} s \frac{\theta}{2} & 0 & 0 & 0 \\
     0 & -\frac{1+i}{2} s \frac{\theta}{2} & 0 & c \frac{\theta}{2} & 0 & \frac{1-i}{2} s \frac{\theta}{2} & 0 & 0 \\
     0 & 0 & \frac{-1+i}{2}  s \frac{\theta}{2} & 0 & -c \frac{\theta}{2} & 0 &  \frac{1+i}{2}   s \frac{\theta}{2} & 0 \\
     0 & 0 & 0 &  \frac{1+i}{2}   s \frac{\theta}{2} & 0 & -c \frac{\theta}{2} & 0 & \frac{-1+i}{2}  s \frac{\theta}{2} \\
     -\frac{1+i}{2} s \frac{\theta}{2} & 0 & 0 & 0 & \frac{1-i}{2} s \frac{\theta}{2} & 0 & c \frac{\theta}{2} & 0 \\
     0 & \frac{1-i}{2} s \frac{\theta}{2} & 0 & 0 & 0 & -\frac{1+i}{2} s \frac{\theta}{2} & 0 & c \frac{\theta}{2} \\
    \end{array}
    \right),\\
  }
  \end{equation}
  \vspace{11pt}
  \begin{equation}
  \hat{\phi}^{\sigma}(\theta)_{(1,1)}
  =
  {\small
  \left(
    \begin{array}{cccccccc}
     c \frac{\theta}{2} & 0 & -\frac{1+i}{2} s \frac{\theta}{2} & 0 & 0 & 0 & \frac{1-i}{2} s \frac{\theta}{2} & 0 \\
     0 & -c \frac{\theta}{2} & 0 & \frac{-1+i}{2}  s \frac{\theta}{2} & 0 & 0 & 0 &  \frac{1+i}{2}   s \frac{\theta}{2} \\
     \frac{-1+i}{2}  s \frac{\theta}{2} & 0 & -c \frac{\theta}{2} & 0 &  \frac{1+i}{2}   s \frac{\theta}{2} & 0 & 0 & 0 \\
     0 & -\frac{1+i}{2} s \frac{\theta}{2} & 0 & c \frac{\theta}{2} & 0 & \frac{1-i}{2} s \frac{\theta}{2} & 0 & 0 \\
     0 & 0 & \frac{1-i}{2} s \frac{\theta}{2} & 0 & c \frac{\theta}{2} & 0 & -\frac{1+i}{2} s \frac{\theta}{2} & 0 \\
     0 & 0 & 0 &  \frac{1+i}{2}   s \frac{\theta}{2} & 0 & -c \frac{\theta}{2} & 0 & \frac{-1+i}{2}  s \frac{\theta}{2} \\
      \frac{1+i}{2}   s \frac{\theta}{2} & 0 & 0 & 0 & \frac{-1+i}{2}  s \frac{\theta}{2} & 0 & -c \frac{\theta}{2} & 0 \\
     0 & \frac{1-i}{2} s \frac{\theta}{2} & 0 & 0 & 0 & -\frac{1+i}{2} s \frac{\theta}{2} & 0 & c \frac{\theta}{2} \\
    \end{array}
    \right),    
  }
  \end{equation} 
  {
    \begin{equation*}
    \begin{aligned}
      \hspace{-10mm}\hat{\phi}^{\sigma}(\theta)_{D}=&\left(
      \begin{array}{cccccccccccccccc}
       s ^3\frac{\theta}{4} & c ^3\frac{\theta}{4}                                                                                               & 0 & 0                                                                                                                                     & \frac{\omega^{2}_{8}}{2} s \frac{\theta}{2} s\frac{\theta+\pi}{4} & \frac{\omega^{10}_{8}}{2} s \frac{\theta}{2} s \frac{\pi-\theta}{4}   & 0 & 0  \\
       -c ^3\frac{\theta}{4} & s ^3\frac{\theta}{4}                                                                                              & 0 & 0                                                                                                                                     & \frac{\omega^{2}_{8}}{2} s \frac{\theta}{2} s \frac{\pi-\theta}{4} & \frac{\omega^{2}_{8}}{2} s \frac{\theta}{2} s\frac{\theta+\pi}{4}    & 0 & 0  \\
       0 & 0                                                                                                                                     & c ^3\frac{\theta}{4} & s ^3\frac{\theta}{4}                                                                                               & 0 & 0                                                                                                                                     & \frac{\omega^{14}_{8}}{2} s \frac{\theta}{2} s\frac{\theta+\pi}{4} & \frac{\omega^{14}_{8}}{2} s \frac{\theta}{2} s \frac{\pi-\theta}{4}  \\
       0 & 0                                                                                                                                     & s ^3\frac{\theta}{4} & -c ^3\frac{\theta}{4}                                                                                              & 0 & 0                                                                                                                                     & \frac{\omega^{14}_{8}}{2} s \frac{\theta}{2} s \frac{\pi-\theta}{4} & \frac{\omega^{6}_{8}}{2} s \frac{\theta}{2} s\frac{\theta+\pi}{4}  \\
       \frac{\omega^{14}_{8}}{2} s \frac{\theta}{2} s\frac{\theta+\pi}{4} & \frac{\omega^{6}_{8}}{2} s \frac{\theta}{2} s \frac{\pi-\theta}{4}   & 0 & 0                                                                                                                                     & c ^3\frac{\theta}{4} & -s ^3\frac{\theta}{4}                                                                                              & 0 & 0  \\
       \frac{\omega^{14}_{8}}{2} s \frac{\theta}{2} s \frac{\pi-\theta}{4} & \frac{\omega^{14}_{8}}{2} s \frac{\theta}{2} s\frac{\theta+\pi}{4}  & 0 & 0                                                                                                                                     & s ^3\frac{\theta}{4} & c ^3\frac{\theta}{4}                                                                                               & 0 & 0  \\
       0 & 0                                                                                                                                     & \frac{\omega^{2}_{8}}{2} s \frac{\theta}{2} s\frac{\theta+\pi}{4} & \frac{\omega^{2}_{8}}{2} s \frac{\theta}{2} s \frac{\pi-\theta}{4}    & 0 & 0                                                                                                                                     & s ^3\frac{\theta}{4} & -c ^3\frac{\theta}{4}  \\
       0 & 0                                                                                                                                     & \frac{\omega^{2}_{8}}{2} s \frac{\theta}{2} s \frac{\pi-\theta}{4} & \frac{\omega^{10}_{8}}{2} s \frac{\theta}{2} s\frac{\theta+\pi}{4}   & 0 & 0                                                                                                                                     & -c ^3\frac{\theta}{4} & -s ^3\frac{\theta}{4}  \\
       i s ^2\frac{\theta}{4} c \frac{\theta}{4} & -i s \frac{\theta}{4} c ^2\frac{\theta}{4}                                                    & 0 & 0                                                                                                                                     & \frac{\omega^{6}_{8}}{2} s \frac{\theta}{2} s \frac{\pi-\theta}{4} & \frac{\omega^{6}_{8}}{2} s \frac{\theta}{2} s\frac{\theta+\pi}{4}    & 0 & 0  \\
       i s \frac{\theta}{4} c ^2\frac{\theta}{4} & i s ^2\frac{\theta}{4} c \frac{\theta}{4}                                                     & 0 & 0                                                                                                                                     & \frac{\omega^{14}_{8}}{2} s \frac{\theta}{2} s\frac{\theta+\pi}{4} & \frac{\omega^{6}_{8}}{2} s \frac{\theta}{2} s \frac{\pi-\theta}{4}   & 0 & 0  \\
       0 & 0                                                                                                                                     & -i s \frac{\theta}{4} c ^2\frac{\theta}{4} & i s ^2\frac{\theta}{4} c \frac{\theta}{4}                                                    & 0 & 0                                                                                                                                     & \frac{\omega^{2}_{8}}{2} s \frac{\theta}{2} s \frac{\pi-\theta}{4} & \frac{\omega^{10}_{8}}{2} s \frac{\theta}{2} s\frac{\theta+\pi}{4}  \\
       0 & 0                                                                                                                                     & i s ^2\frac{\theta}{4} c \frac{\theta}{4} & i s \frac{\theta}{4} c ^2\frac{\theta}{4}                                                     & 0 & 0                                                                                                                                     & \frac{\omega^{10}_{8}}{2} s \frac{\theta}{2} s\frac{\theta+\pi}{4} & \frac{\omega^{10}_{8}}{2} s \frac{\theta}{2} s \frac{\pi-\theta}{4}  \\
       \frac{\omega^{2}_{8}}{2} s \frac{\theta}{2} s \frac{\pi-\theta}{4} & \frac{\omega^{2}_{8}}{2} s \frac{\theta}{2} s\frac{\theta+\pi}{4}    & 0 & 0                                                                                                                                     & -i s \frac{\theta}{4} c ^2\frac{\theta}{4} & -i s ^2\frac{\theta}{4} c \frac{\theta}{4}                                                   & 0 & 0  \\
       \frac{\omega^{10}_{8}}{2} s \frac{\theta}{2} s\frac{\theta+\pi}{4} & \frac{\omega^{2}_{8}}{2} s \frac{\theta}{2} s \frac{\pi-\theta}{4}   & 0 & 0                                                                                                                                     & i s ^2\frac{\theta}{4} c \frac{\theta}{4} & -i s \frac{\theta}{4} c ^2\frac{\theta}{4}                                                    & 0 & 0  \\
       0 & 0                                                                                                                                     & \frac{\omega^{6}_{8}}{2} s \frac{\theta}{2} s \frac{\pi-\theta}{4} & \frac{\omega^{14}_{8}}{2} s \frac{\theta}{2} s\frac{\theta+\pi}{4}   & 0 & 0                                                                                                                                     & i s ^2\frac{\theta}{4} c \frac{\theta}{4} & i s \frac{\theta}{4} c ^2\frac{\theta}{4}  \\
       0 & 0                                                                                                                                     & \frac{\omega^{14}_{8}}{2} s \frac{\theta}{2} s\frac{\theta+\pi}{4} & \frac{\omega^{14}_{8}}{2} s \frac{\theta}{2} s \frac{\pi-\theta}{4}  & 0 & 0                                                                                                                                     & i s \frac{\theta}{4} c ^2\frac{\theta}{4} & -i s ^2\frac{\theta}{4} c \frac{\theta}{4}  \\
      \end{array}\right.\\
      &\;\\
      &\;\;\;\left.
      \begin{array}{cccccccccccccccc}
       -i s ^2\frac{\theta}{4} c \frac{\theta}{4} & i s \frac{\theta}{4} c ^2\frac{\theta}{4}                                                    & 0 & 0                                                                                                                                     & \frac{\omega^{14}_{8}}{2} s \frac{\theta}{2} s \frac{\pi-\theta}{4} & \frac{\omega^{14}_{8}}{2} s \frac{\theta}{2} s\frac{\theta+\pi}{4}  & 0 & 0 \\
       -i s \frac{\theta}{4} c ^2\frac{\theta}{4} & -i s ^2\frac{\theta}{4} c \frac{\theta}{4}                                                   & 0 & 0                                                                                                                                     & \frac{\omega^{6}_{8}}{2} s \frac{\theta}{2} s\frac{\theta+\pi}{4} & \frac{\omega^{14}_{8}}{2} s \frac{\theta}{2} s \frac{\pi-\theta}{4}   & 0 & 0 \\
       0 & 0                                                                                                                                     & i s \frac{\theta}{4} c ^2\frac{\theta}{4} & -i s ^2\frac{\theta}{4} c \frac{\theta}{4}                                                    & 0 & 0                                                                                                                                     & \frac{\omega^{10}_{8}}{2} s \frac{\theta}{2} s \frac{\pi-\theta}{4} & \frac{\omega^{2}_{8}}{2} s \frac{\theta}{2} s\frac{\theta+\pi}{4} \\
       0 & 0                                                                                                                                     & -i s ^2\frac{\theta}{4} c \frac{\theta}{4} & -i s \frac{\theta}{4} c ^2\frac{\theta}{4}                                                   & 0 & 0                                                                                                                                     & \frac{\omega^{2}_{8}}{2} s \frac{\theta}{2} s\frac{\theta+\pi}{4} & \frac{\omega^{2}_{8}}{2} s \frac{\theta}{2} s \frac{\pi-\theta}{4} \\
       \frac{\omega^{10}_{8}}{2} s \frac{\theta}{2} s \frac{\pi-\theta}{4} & \frac{\omega^{10}_{8}}{2} s \frac{\theta}{2} s\frac{\theta+\pi}{4}  & 0 & 0                                                                                                                                     & i s \frac{\theta}{4} c ^2\frac{\theta}{4} & i s ^2\frac{\theta}{4} c \frac{\theta}{4}                                                     & 0 & 0 \\
       \frac{\omega^{2}_{8}}{2} s \frac{\theta}{2} s\frac{\theta+\pi}{4} & \frac{\omega^{10}_{8}}{2} s \frac{\theta}{2} s \frac{\pi-\theta}{4}   & 0 & 0                                                                                                                                     & -i s ^2\frac{\theta}{4} c \frac{\theta}{4} & i s \frac{\theta}{4} c ^2\frac{\theta}{4}                                                    & 0 & 0 \\
       0 & 0                                                                                                                                     & \frac{\omega^{14}_{8}}{2} s \frac{\theta}{2} s \frac{\pi-\theta}{4} & \frac{\omega^{6}_{8}}{2} s \frac{\theta}{2} s\frac{\theta+\pi}{4}   & 0 & 0                                                                                                                                     & -i s ^2\frac{\theta}{4} c \frac{\theta}{4} & -i s \frac{\theta}{4} c ^2\frac{\theta}{4} \\
       0 & 0                                                                                                                                     & \frac{\omega^{6}_{8}}{2} s \frac{\theta}{2} s\frac{\theta+\pi}{4} & \frac{\omega^{6}_{8}}{2} s \frac{\theta}{2} s \frac{\pi-\theta}{4}    & 0 & 0                                                                                                                                     & -i s \frac{\theta}{4} c ^2\frac{\theta}{4} & i s ^2\frac{\theta}{4} c \frac{\theta}{4} \\
       -s ^3\frac{\theta}{4} & -c ^3\frac{\theta}{4}                                                                                             & 0 & 0                                                                                                                                     & \frac{\omega^{10}_{8}}{2} s \frac{\theta}{2} s\frac{\theta+\pi}{4} & \frac{\omega^{2}_{8}}{2} s \frac{\theta}{2} s \frac{\pi-\theta}{4}   & 0 & 0 \\
       c ^3\frac{\theta}{4} & -s ^3\frac{\theta}{4}                                                                                              & 0 & 0                                                                                                                                     & \frac{\omega^{10}_{8}}{2} s \frac{\theta}{2} s \frac{\pi-\theta}{4} & \frac{\omega^{10}_{8}}{2} s \frac{\theta}{2} s\frac{\theta+\pi}{4}  & 0 & 0 \\
       0 & 0                                                                                                                                     & -c ^3\frac{\theta}{4} & -s ^3\frac{\theta}{4}                                                                                             & 0 & 0                                                                                                                                     & \frac{\omega^{6}_{8}}{2} s \frac{\theta}{2} s\frac{\theta+\pi}{4} & \frac{\omega^{6}_{8}}{2} s \frac{\theta}{2} s \frac{\pi-\theta}{4} \\
       0 & 0                                                                                                                                     & -s ^3\frac{\theta}{4} & c ^3\frac{\theta}{4}                                                                                              & 0 & 0                                                                                                                                     & \frac{\omega^{6}_{8}}{2} s \frac{\theta}{2} s \frac{\pi-\theta}{4} & \frac{\omega^{14}_{8}}{2} s \frac{\theta}{2} s\frac{\theta+\pi}{4} \\
       \frac{\omega^{6}_{8}}{2} s \frac{\theta}{2} s\frac{\theta+\pi}{4} & \frac{\omega^{14}_{8}}{2} s \frac{\theta}{2} s \frac{\pi-\theta}{4}   & 0 & 0                                                                                                                                     & -c ^3\frac{\theta}{4} & s ^3\frac{\theta}{4}                                                                                              & 0 & 0 \\
       \frac{\omega^{6}_{8}}{2} s \frac{\theta}{2} s \frac{\pi-\theta}{4} & \frac{\omega^{6}_{8}}{2} s \frac{\theta}{2} s\frac{\theta+\pi}{4}    & 0 & 0                                                                                                                                     & -s ^3\frac{\theta}{4} & -c ^3\frac{\theta}{4}                                                                                             & 0 & 0 \\
       0 & 0                                                                                                                                     & \frac{\omega^{10}_{8}}{2} s \frac{\theta}{2} s\frac{\theta+\pi}{4} & \frac{\omega^{10}_{8}}{2} s \frac{\theta}{2} s \frac{\pi-\theta}{4}  & 0 & 0                                                                                                                                     & -s ^3\frac{\theta}{4} & c ^3\frac{\theta}{4} \\
       0 & 0                                                                                                                                     & \frac{\omega^{10}_{8}}{2} s \frac{\theta}{2} s \frac{\pi-\theta}{4} & \frac{\omega^{2}_{8}}{2} s \frac{\theta}{2} s\frac{\theta+\pi}{4}   & 0 & 0                                                                                                                                     & c ^3\frac{\theta}{4} & s ^3\frac{\theta}{4} \\
      \end{array}
      \right),
  \end{aligned}
\end{equation*}
  }
and
{
\begin{eqnarray*}
  \hat{\phi}^{\tau}(\theta)_{(0,0)}
  &=&\left(
    \begin{array}{cccccccc}
     1 & 0 & 0 & 0 & 0 & 0 & 0 & 0 \\
     0 & 1 & 0 & 0 & 0 & 0 & 0 & 0 \\
     0 & 0 & 1 & 0 & 0 & 0 & 0 & 0 \\
     0 & 0 & 0 & 1 & 0 & 0 & 0 & 0 \\
     0 & 0 & 0 & 0 & 1 & 0 & 0 & 0 \\
     0 & 0 & 0 & 0 & 0 & 1 & 0 & 0 \\
     0 & 0 & 0 & 0 & 0 & 0 & 1 & 0 \\
     0 & 0 & 0 & 0 & 0 & 0 & 0 & 1 \\
    \end{array}
    \right), \quad
  \hat{\phi}^{\tau}(\theta)_{(0,1)}
  =\left(
    \begin{array}{cccccccc}
     -c \frac{\theta }{2} & i s \frac{\theta }{2} & 0 & 0 & 0 & 0 & 0 & 0 \\
     -i s \frac{\theta }{2} & c \frac{\theta }{2} & 0 & 0 & 0 & 0 & 0 & 0 \\
     0 & 0 & -c \frac{\theta }{2} & i s \frac{\theta }{2} & 0 & 0 & 0 & 0 \\
     0 & 0 & -i s \frac{\theta }{2} & c \frac{\theta }{2} & 0 & 0 & 0 & 0 \\
     0 & 0 & 0 & 0 & -c \frac{\theta }{2} & i s \frac{\theta }{2} & 0 & 0 \\
     0 & 0 & 0 & 0 & -i s \frac{\theta }{2} & c \frac{\theta }{2} & 0 & 0 \\
     0 & 0 & 0 & 0 & 0 & 0 & -c \frac{\theta }{2} & i s \frac{\theta }{2} \\
     0 & 0 & 0 & 0 & 0 & 0 & -i s \frac{\theta }{2} & c \frac{\theta }{2} \\
    \end{array}
    \right),\\
  \hat{\phi}^{\tau}(\theta)_{(1,0)}
  &=&\left(
    \begin{array}{cccccccc}
     -1 & 0 & 0 & 0 & 0 & 0 & 0 & 0 \\
     0 & -1 & 0 & 0 & 0 & 0 & 0 & 0 \\
     0 & 0 & 1 & 0 & 0 & 0 & 0 & 0 \\
     0 & 0 & 0 & 1 & 0 & 0 & 0 & 0 \\
     0 & 0 & 0 & 0 & -1 & 0 & 0 & 0 \\
     0 & 0 & 0 & 0 & 0 & -1 & 0 & 0 \\
     0 & 0 & 0 & 0 & 0 & 0 & 1 & 0 \\
     0 & 0 & 0 & 0 & 0 & 0 & 0 & 1 \\
    \end{array}
    \right), \quad
  \hat{\phi}^{\tau}(\theta)_{(1,1)}
  =\left(
    \begin{array}{cccccccc}
     c \frac{\theta }{2} & -i s \frac{\theta }{2} & 0 & 0 & 0 & 0 & 0 & 0 \\
     i s \frac{\theta }{2} & -c \frac{\theta }{2} & 0 & 0 & 0 & 0 & 0 & 0 \\
     0 & 0 & -c \frac{\theta }{2} & i s \frac{\theta }{2} & 0 & 0 & 0 & 0 \\
     0 & 0 & -i s \frac{\theta }{2} & c \frac{\theta }{2} & 0 & 0 & 0 & 0 \\
     0 & 0 & 0 & 0 & c \frac{\theta }{2} & -i s \frac{\theta }{2} & 0 & 0 \\
     0 & 0 & 0 & 0 & i s \frac{\theta }{2} & -c \frac{\theta }{2} & 0 & 0 \\
     0 & 0 & 0 & 0 & 0 & 0 & -c \frac{\theta }{2} & i s \frac{\theta }{2} \\
     0 & 0 & 0 & 0 & 0 & 0 & -i s \frac{\theta }{2} & c \frac{\theta }{2} \\
    \end{array}
    \right),\\
  \hat{\phi}^{\tau}(\theta)_{D}
  &=&\left(
    \begin{array}{cccc}
    A&0&0&0\\
    0&B&0&0\\
    0&0&-A&0\\
    0&0&0&-B\\
    \end{array}
    \right).
\end{eqnarray*}
}
Here, $s$ and $c$ represent $\sin$ and $\cos$, respectively, and $\omega_8$ is the shorthand notation for $e^{2\pi i/8}$, and $A$ and $B$ are defined by
  \begin{eqnarray*}
  A&=&\left(
    \begin{array}{cccc}
     \sin ^2 \frac{\theta }{4}               &  \cos ^2 \frac{\theta }{4}              & \frac{1}{2} i \sin  \frac{\theta }{2}   & -\frac{1}{2} i \sin  \frac{\theta }{2}  \\
     -\cos ^2 \frac{\theta }{4}               & -\sin ^2 \frac{\theta }{4}             & \frac{1}{2} i \sin  \frac{\theta }{2}   & -\frac{1}{2} i \sin  \frac{\theta }{2}  \\
     -\frac{1}{2} i \sin  \frac{\theta }{2}  &  \frac{1}{2} i \sin  \frac{\theta }{2}  & \cos ^2 \frac{\theta }{4}               & \sin ^2 \frac{\theta }{4}               \\
     -\frac{1}{2} i \sin  \frac{\theta }{2}   & \frac{1}{2} i \sin  \frac{\theta }{2}  & -\sin ^2 \frac{\theta }{4}              & -\cos ^2 \frac{\theta }{4}              \\
    \end{array}
    \right),\;
  B=\left(
    \begin{array}{cccc}
     \cos ^2 \frac{\theta }{4}   & -\sin ^2 \frac{\theta }{4}  & -\frac{1}{2} i \sin  \frac{\theta }{2}  & -\frac{1}{2} i \sin  \frac{\theta }{2}   \\
     -\sin ^2 \frac{\theta }{4}  & \cos ^2 \frac{\theta }{4}  & -\frac{1}{2} i \sin  \frac{\theta }{2}  & -\frac{1}{2} i \sin  \frac{\theta }{2}   \\
     \frac{1}{2} i \sin  \frac{\theta }{2}  & \frac{1}{2} i \sin  \frac{\theta }{2}  & \sin ^2 \frac{\theta }{4}  & -\cos ^2 \frac{\theta }{4}     \\
     \frac{1}{2} i \sin  \frac{\theta }{2}  & \frac{1}{2} i \sin  \frac{\theta }{2}  & -\cos ^2 \frac{\theta }{4}  & \sin ^2 \frac{\theta }{4}     \\
    \end{array}
    \right).
  \end{eqnarray*}
The action tensor $\hat{\phi}_{\rho}^{k}(\theta)$ for general $k \in D_8$ is given by a multiplication of the above matrices.
The pump invariant is the difference of $\hat{\phi}_{\rho}^{k}(\theta)$ at $\theta=0$ and $\theta=2\pi$:
\begin{equation}
  \hat{\phi}^{k}(2\pi)_{\rho}=\eta^{k}_{\rho} \hat{\phi}_{\rho}^{k}(0).
\end{equation}
By using the above matrices, we can explicitly read off the pump invariant:
\begin{equation}
\begin{aligned}
  \eta_{(0,0)}^{\sigma}&=1, \quad 
  \eta_{(0,1)}^{\sigma}=1, \quad 
  \eta_{(1,0)}^{\sigma}=-1, \quad 
  \eta_{(1,1)}^{\sigma}=-1, \quad 
  \eta_{D}^{\sigma}=
    \left(
    \begin{array}{cc}
      0&-1\\
      1&0
    \end{array}
    \right),
\\
  \eta_{(0,0)}^{\tau}&=1, \quad 
  \eta_{(0,1)}^{\tau}=-1, \quad 
  \eta_{(1,0)}^{\tau}=1, \quad 
  \eta_{(1,1)}^{\tau}=-1, \quad 
  \eta_{D}^{\tau}=
    \left(
    \begin{array}{cc}
      0&-1\\
      -1&0
    \end{array}
    \right).
\end{aligned}
\end{equation}
The pump invariant $\eta_{\rho}^{k}$ for general $k\in D_{8}$ is given by a multiplication of the above invariants, i.e.,
\begin{equation}
\eta_{\rho}^{\sigma^i \tau^j} = (\eta_{\rho}^{\sigma})^i (\eta_{\rho}^{\tau})^j.
\label{eq: pump for general k}
\end{equation}

According to the Tannaka-Krein duality, 
it is expected that the pump invariant is classified by $D_8$.  
This is consistent with the fact that the invariants we computed form $D_8$ as shown in Eq.~\eqref{eq: pump for general k}.
Our computation demonstrates that all (equivalence classes of) $S^1$-parameterized families can be constructed in this way.

\bibliography{bibliography}

\begin{thebibliography}{143}%
\makeatletter
\providecommand \@ifxundefined [1]{%
 \@ifx{#1\undefined}
}%
\providecommand \@ifnum [1]{%
 \ifnum #1\expandafter \@firstoftwo
 \else \expandafter \@secondoftwo
 \fi
}%
\providecommand \@ifx [1]{%
 \ifx #1\expandafter \@firstoftwo
 \else \expandafter \@secondoftwo
 \fi
}%
\providecommand \natexlab [1]{#1}%
\providecommand \enquote  [1]{``#1''}%
\providecommand \bibnamefont  [1]{#1}%
\providecommand \bibfnamefont [1]{#1}%
\providecommand \citenamefont [1]{#1}%
\providecommand \href@noop [0]{\@secondoftwo}%
\providecommand \href [0]{\begingroup \@sanitize@url \@href}%
\providecommand \@href[1]{\@@startlink{#1}\@@href}%
\providecommand \@@href[1]{\endgroup#1\@@endlink}%
\providecommand \@sanitize@url [0]{\catcode `\\12\catcode `\$12\catcode `\&12\catcode `\#12\catcode `\^12\catcode `\_12\catcode `\%12\relax}%
\providecommand \@@startlink[1]{}%
\providecommand \@@endlink[0]{}%
\providecommand \url  [0]{\begingroup\@sanitize@url \@url }%
\providecommand \@url [1]{\endgroup\@href {#1}{\urlprefix }}%
\providecommand \urlprefix  [0]{URL }%
\providecommand \Eprint [0]{\href }%
\providecommand \doibase [0]{https://doi.org/}%
\providecommand \selectlanguage [0]{\@gobble}%
\providecommand \bibinfo  [0]{\@secondoftwo}%
\providecommand \bibfield  [0]{\@secondoftwo}%
\providecommand \translation [1]{[#1]}%
\providecommand \BibitemOpen [0]{}%
\providecommand \bibitemStop [0]{}%
\providecommand \bibitemNoStop [0]{.\EOS\space}%
\providecommand \EOS [0]{\spacefactor3000\relax}%
\providecommand \BibitemShut  [1]{\csname bibitem#1\endcsname}%
\let\auto@bib@innerbib\@empty
\bibitem [{\citenamefont {Gu}\ and\ \citenamefont {Wen}(2009)}]{Gu_2009}%
  \BibitemOpen
  \bibfield  {author} {\bibinfo {author} {\bibfnamefont {Z.-C.}\ \bibnamefont {Gu}}\ and\ \bibinfo {author} {\bibfnamefont {X.-G.}\ \bibnamefont {Wen}},\ }\bibfield  {title} {\bibinfo {title} {{Tensor-entanglement-filtering renormalization approach and symmetry-protected topological order}},\ }\bibfield  {journal} {\bibinfo  {journal} {Physical Review B}\ }\textbf {\bibinfo {volume} {80}},\ \href {https://doi.org/10.1103/physrevb.80.155131} {10.1103/physrevb.80.155131} (\bibinfo {year} {2009}),\ \Eprint {https://arxiv.org/abs/0903.1069} {arXiv:0903.1069 [cond-mat.str-el]} \BibitemShut {NoStop}%
\bibitem [{\citenamefont {Pollmann}\ \emph {et~al.}(2012)\citenamefont {Pollmann}, \citenamefont {Berg}, \citenamefont {Turner},\ and\ \citenamefont {Oshikawa}}]{Pollmann:2009mhk}%
  \BibitemOpen
  \bibfield  {author} {\bibinfo {author} {\bibfnamefont {F.}~\bibnamefont {Pollmann}}, \bibinfo {author} {\bibfnamefont {E.}~\bibnamefont {Berg}}, \bibinfo {author} {\bibfnamefont {A.~M.}\ \bibnamefont {Turner}},\ and\ \bibinfo {author} {\bibfnamefont {M.}~\bibnamefont {Oshikawa}},\ }\bibfield  {title} {\bibinfo {title} {{Symmetry protection of topological phases in one-dimensional quantum spin systems}},\ }\href {https://doi.org/10.1103/PhysRevB.85.075125} {\bibfield  {journal} {\bibinfo  {journal} {Phys. Rev. B}\ }\textbf {\bibinfo {volume} {85}},\ \bibinfo {pages} {075125} (\bibinfo {year} {2012})},\ \Eprint {https://arxiv.org/abs/0909.4059} {arXiv:0909.4059 [cond-mat.str-el]} \BibitemShut {NoStop}%
\bibitem [{\citenamefont {Pollmann}\ \emph {et~al.}(2010)\citenamefont {Pollmann}, \citenamefont {Turner}, \citenamefont {Berg},\ and\ \citenamefont {Oshikawa}}]{Pollmann_2010}%
  \BibitemOpen
  \bibfield  {author} {\bibinfo {author} {\bibfnamefont {F.}~\bibnamefont {Pollmann}}, \bibinfo {author} {\bibfnamefont {A.~M.}\ \bibnamefont {Turner}}, \bibinfo {author} {\bibfnamefont {E.}~\bibnamefont {Berg}},\ and\ \bibinfo {author} {\bibfnamefont {M.}~\bibnamefont {Oshikawa}},\ }\bibfield  {title} {\bibinfo {title} {{Entanglement spectrum of a topological phase in one dimension}},\ }\bibfield  {journal} {\bibinfo  {journal} {Physical Review B}\ }\textbf {\bibinfo {volume} {81}},\ \href {https://doi.org/10.1103/physrevb.81.064439} {10.1103/physrevb.81.064439} (\bibinfo {year} {2010}),\ \Eprint {https://arxiv.org/abs/0910.1811} {arXiv:0910.1811 [cond-mat.str-el]} \BibitemShut {NoStop}%
\bibitem [{\citenamefont {Chen}\ \emph {et~al.}(2011{\natexlab{a}})\citenamefont {Chen}, \citenamefont {Gu},\ and\ \citenamefont {Wen}}]{Chen:2010zpc}%
  \BibitemOpen
  \bibfield  {author} {\bibinfo {author} {\bibfnamefont {X.}~\bibnamefont {Chen}}, \bibinfo {author} {\bibfnamefont {Z.-C.}\ \bibnamefont {Gu}},\ and\ \bibinfo {author} {\bibfnamefont {X.-G.}\ \bibnamefont {Wen}},\ }\bibfield  {title} {\bibinfo {title} {{Classification of gapped symmetric phases in one-dimensional spin systems}},\ }\href {https://doi.org/10.1103/PhysRevB.83.035107} {\bibfield  {journal} {\bibinfo  {journal} {Phys. Rev. B}\ }\textbf {\bibinfo {volume} {83}},\ \bibinfo {pages} {035107} (\bibinfo {year} {2011}{\natexlab{a}})},\ \Eprint {https://arxiv.org/abs/1008.3745} {arXiv:1008.3745 [cond-mat.str-el]} \BibitemShut {NoStop}%
\bibitem [{\citenamefont {Chen}\ \emph {et~al.}(2011{\natexlab{b}})\citenamefont {Chen}, \citenamefont {Gu},\ and\ \citenamefont {Wen}}]{Chen_2011_complete}%
  \BibitemOpen
  \bibfield  {author} {\bibinfo {author} {\bibfnamefont {X.}~\bibnamefont {Chen}}, \bibinfo {author} {\bibfnamefont {Z.-C.}\ \bibnamefont {Gu}},\ and\ \bibinfo {author} {\bibfnamefont {X.-G.}\ \bibnamefont {Wen}},\ }\bibfield  {title} {\bibinfo {title} {Complete classification of one-dimensional gapped quantum phases in interacting spin systems},\ }\bibfield  {journal} {\bibinfo  {journal} {Physical Review B}\ }\textbf {\bibinfo {volume} {84}},\ \href {https://doi.org/10.1103/physrevb.84.235128} {10.1103/physrevb.84.235128} (\bibinfo {year} {2011}{\natexlab{b}}),\ \Eprint {https://arxiv.org/abs/1103.3323} {arXiv:1103.3323 [cond-mat.str-el]} \BibitemShut {NoStop}%
\bibitem [{\citenamefont {Chen}\ \emph {et~al.}(2011{\natexlab{c}})\citenamefont {Chen}, \citenamefont {Liu},\ and\ \citenamefont {Wen}}]{Chen:2011bcp}%
  \BibitemOpen
  \bibfield  {author} {\bibinfo {author} {\bibfnamefont {X.}~\bibnamefont {Chen}}, \bibinfo {author} {\bibfnamefont {Z.-X.}\ \bibnamefont {Liu}},\ and\ \bibinfo {author} {\bibfnamefont {X.-G.}\ \bibnamefont {Wen}},\ }\bibfield  {title} {\bibinfo {title} {{Two-dimensional symmetry-protected topological orders and their protected gapless edge excitations}},\ }\href {https://doi.org/10.1103/PhysRevB.84.235141} {\bibfield  {journal} {\bibinfo  {journal} {Phys. Rev. B}\ }\textbf {\bibinfo {volume} {84}},\ \bibinfo {pages} {235141} (\bibinfo {year} {2011}{\natexlab{c}})},\ \Eprint {https://arxiv.org/abs/1106.4752} {arXiv:1106.4752 [cond-mat.str-el]} \BibitemShut {NoStop}%
\bibitem [{\citenamefont {Chen}\ \emph {et~al.}(2013)\citenamefont {Chen}, \citenamefont {Gu}, \citenamefont {Liu},\ and\ \citenamefont {Wen}}]{Chen:2011pg}%
  \BibitemOpen
  \bibfield  {author} {\bibinfo {author} {\bibfnamefont {X.}~\bibnamefont {Chen}}, \bibinfo {author} {\bibfnamefont {Z.-C.}\ \bibnamefont {Gu}}, \bibinfo {author} {\bibfnamefont {Z.-X.}\ \bibnamefont {Liu}},\ and\ \bibinfo {author} {\bibfnamefont {X.-G.}\ \bibnamefont {Wen}},\ }\bibfield  {title} {\bibinfo {title} {{Symmetry protected topological orders and the group cohomology of their symmetry group}},\ }\href {https://doi.org/10.1103/PhysRevB.87.155114} {\bibfield  {journal} {\bibinfo  {journal} {Phys. Rev. B}\ }\textbf {\bibinfo {volume} {87}},\ \bibinfo {pages} {155114} (\bibinfo {year} {2013})},\ \Eprint {https://arxiv.org/abs/1106.4772} {arXiv:1106.4772 [cond-mat.str-el]} \BibitemShut {NoStop}%
\bibitem [{\citenamefont {Schuch}\ \emph {et~al.}(2011)\citenamefont {Schuch}, \citenamefont {P\'{e}rez-Garc\'{i}a},\ and\ \citenamefont {Cirac}}]{Schuch_2011}%
  \BibitemOpen
  \bibfield  {author} {\bibinfo {author} {\bibfnamefont {N.}~\bibnamefont {Schuch}}, \bibinfo {author} {\bibfnamefont {D.}~\bibnamefont {P\'{e}rez-Garc\'{i}a}},\ and\ \bibinfo {author} {\bibfnamefont {I.}~\bibnamefont {Cirac}},\ }\bibfield  {title} {\bibinfo {title} {{Classifying quantum phases using matrix product states and projected entangled pair states}},\ }\bibfield  {journal} {\bibinfo  {journal} {Physical Review B}\ }\textbf {\bibinfo {volume} {84}},\ \href {https://doi.org/10.1103/physrevb.84.165139} {10.1103/physrevb.84.165139} (\bibinfo {year} {2011}),\ \Eprint {https://arxiv.org/abs/1010.3732} {arXiv:1010.3732 [cond-mat.str-el]} \BibitemShut {NoStop}%
\bibitem [{\citenamefont {Thouless}(1983)}]{Thouless83}%
  \BibitemOpen
  \bibfield  {author} {\bibinfo {author} {\bibfnamefont {D.~J.}\ \bibnamefont {Thouless}},\ }\bibfield  {title} {\bibinfo {title} {Quantization of particle transport},\ }\href {https://doi.org/10.1103/PhysRevB.27.6083} {\bibfield  {journal} {\bibinfo  {journal} {Phys. Rev. B}\ }\textbf {\bibinfo {volume} {27}},\ \bibinfo {pages} {6083} (\bibinfo {year} {1983})}\BibitemShut {NoStop}%
\bibitem [{\citenamefont {Teo}\ and\ \citenamefont {Kane}(2010)}]{PhysRevB.82.115120}%
  \BibitemOpen
  \bibfield  {author} {\bibinfo {author} {\bibfnamefont {J.~C.~Y.}\ \bibnamefont {Teo}}\ and\ \bibinfo {author} {\bibfnamefont {C.~L.}\ \bibnamefont {Kane}},\ }\bibfield  {title} {\bibinfo {title} {Topological defects and gapless modes in insulators and superconductors},\ }\href {https://doi.org/10.1103/PhysRevB.82.115120} {\bibfield  {journal} {\bibinfo  {journal} {Phys. Rev. B}\ }\textbf {\bibinfo {volume} {82}},\ \bibinfo {pages} {115120} (\bibinfo {year} {2010})}\BibitemShut {NoStop}%
\bibitem [{\citenamefont {Kitaev}(2011)}]{Kitaev2011SCGP}%
  \BibitemOpen
  \bibfield  {author} {\bibinfo {author} {\bibfnamefont {A.}~\bibnamefont {Kitaev}},\ }\href {http://scgp.stonybrook.edu/video_portal/video.php?id=336} {\bibinfo {title} {Toward a topological classification of many-body quantum states with short-range entanglement}} (\bibinfo {year} {2011}),\ \bibinfo {note} {talk at Simons Center for Geometry and Physics}\BibitemShut {NoStop}%
\bibitem [{\citenamefont {Kitaev}(2013)}]{Kitaev2013SCGP}%
  \BibitemOpen
  \bibfield  {author} {\bibinfo {author} {\bibfnamefont {A.}~\bibnamefont {Kitaev}},\ }\href {http://scgp.stonybrook.edu/video_portal/video.php?id=2010} {\bibinfo {title} {On the classification of short-range entangled states}} (\bibinfo {year} {2013}),\ \bibinfo {note} {talk at Simons Center for Geometry and Physics}\BibitemShut {NoStop}%
\bibitem [{\citenamefont {Kitaev}(2015)}]{Kitaev2015IPAM}%
  \BibitemOpen
  \bibfield  {author} {\bibinfo {author} {\bibfnamefont {A.}~\bibnamefont {Kitaev}},\ }\href {https://www.ipam.ucla.edu/abstract/?tid=12389} {\bibinfo {title} {Homotopy-theoretic approach to spt phases in action: $z_{16}$ classification of three-dimensional superconductors}} (\bibinfo {year} {2015}),\ \bibinfo {note} {talk at Institute for Pure and Applied Mathematics}\BibitemShut {NoStop}%
\bibitem [{\citenamefont {Kapustin}\ and\ \citenamefont {Spodyneiko}(2020{\natexlab{a}})}]{Kapustin:2020mkl}%
  \BibitemOpen
  \bibfield  {author} {\bibinfo {author} {\bibfnamefont {A.}~\bibnamefont {Kapustin}}\ and\ \bibinfo {author} {\bibfnamefont {L.}~\bibnamefont {Spodyneiko}},\ }\href@noop {} {\bibinfo {title} {{Higher-dimensional generalizations of the Thouless charge pump}}} (\bibinfo {year} {2020}{\natexlab{a}}),\ \Eprint {https://arxiv.org/abs/2003.09519} {arXiv:2003.09519 [cond-mat.str-el]} \BibitemShut {NoStop}%
\bibitem [{\citenamefont {Shiozaki}(2022)}]{Shiozaki:2021weu}%
  \BibitemOpen
  \bibfield  {author} {\bibinfo {author} {\bibfnamefont {K.}~\bibnamefont {Shiozaki}},\ }\bibfield  {title} {\bibinfo {title} {{Adiabatic cycles of quantum spin systems}},\ }\href {https://doi.org/10.1103/PhysRevB.106.125108} {\bibfield  {journal} {\bibinfo  {journal} {Phys. Rev. B}\ }\textbf {\bibinfo {volume} {106}},\ \bibinfo {pages} {125108} (\bibinfo {year} {2022})},\ \Eprint {https://arxiv.org/abs/2110.10665} {arXiv:2110.10665 [cond-mat.str-el]} \BibitemShut {NoStop}%
\bibitem [{\citenamefont {Hermele}(2021)}]{Hermele2021CMSA}%
  \BibitemOpen
  \bibfield  {author} {\bibinfo {author} {\bibfnamefont {M.}~\bibnamefont {Hermele}},\ }\href {https://www.youtube.com/watch?v=wtaC0tqXZMU} {\bibinfo {title} {Families of gapped systems and quantum pumps}} (\bibinfo {year} {2021}),\ \bibinfo {note} {talk at Harvard CMSA}\BibitemShut {NoStop}%
\bibitem [{\citenamefont {Wen}\ \emph {et~al.}(2023)\citenamefont {Wen}, \citenamefont {Qi}, \citenamefont {Beaudry}, \citenamefont {Moreno}, \citenamefont {Pflaum}, \citenamefont {Spiegel}, \citenamefont {Vishwanath},\ and\ \citenamefont {Hermele}}]{Wen:2021gwc}%
  \BibitemOpen
  \bibfield  {author} {\bibinfo {author} {\bibfnamefont {X.}~\bibnamefont {Wen}}, \bibinfo {author} {\bibfnamefont {M.}~\bibnamefont {Qi}}, \bibinfo {author} {\bibfnamefont {A.}~\bibnamefont {Beaudry}}, \bibinfo {author} {\bibfnamefont {J.}~\bibnamefont {Moreno}}, \bibinfo {author} {\bibfnamefont {M.~J.}\ \bibnamefont {Pflaum}}, \bibinfo {author} {\bibfnamefont {D.}~\bibnamefont {Spiegel}}, \bibinfo {author} {\bibfnamefont {A.}~\bibnamefont {Vishwanath}},\ and\ \bibinfo {author} {\bibfnamefont {M.}~\bibnamefont {Hermele}},\ }\bibfield  {title} {\bibinfo {title} {{Flow of higher Berry curvature and bulk-boundary correspondence in parametrized quantum systems}},\ }\href {https://doi.org/10.1103/PhysRevB.108.125147} {\bibfield  {journal} {\bibinfo  {journal} {Phys. Rev. B}\ }\textbf {\bibinfo {volume} {108}},\ \bibinfo {pages} {125147} (\bibinfo {year} {2023})},\ \Eprint {https://arxiv.org/abs/2112.07748} {arXiv:2112.07748 [cond-mat.str-el]} \BibitemShut {NoStop}%
\bibitem [{\citenamefont {Spodyneiko}(2023)}]{Spodyneiko:2023vsw}%
  \BibitemOpen
  \bibfield  {author} {\bibinfo {author} {\bibfnamefont {L.}~\bibnamefont {Spodyneiko}},\ }\href@noop {} {\bibinfo {title} {{Hall conductivity pump}}} (\bibinfo {year} {2023}),\ \Eprint {https://arxiv.org/abs/2309.14332} {arXiv:2309.14332 [cond-mat.mes-hall]} \BibitemShut {NoStop}%
\bibitem [{\citenamefont {Ohyama}\ \emph {et~al.}(2022)\citenamefont {Ohyama}, \citenamefont {Shiozaki},\ and\ \citenamefont {Sato}}]{Ohyama:2022cib}%
  \BibitemOpen
  \bibfield  {author} {\bibinfo {author} {\bibfnamefont {S.}~\bibnamefont {Ohyama}}, \bibinfo {author} {\bibfnamefont {K.}~\bibnamefont {Shiozaki}},\ and\ \bibinfo {author} {\bibfnamefont {M.}~\bibnamefont {Sato}},\ }\bibfield  {title} {\bibinfo {title} {{Generalized Thouless pumps in (1+1)-dimensional interacting fermionic systems}},\ }\href {https://doi.org/10.1103/PhysRevB.106.165115} {\bibfield  {journal} {\bibinfo  {journal} {Phys. Rev. B}\ }\textbf {\bibinfo {volume} {106}},\ \bibinfo {pages} {165115} (\bibinfo {year} {2022})},\ \Eprint {https://arxiv.org/abs/2206.01110} {arXiv:2206.01110 [cond-mat.str-el]} \BibitemShut {NoStop}%
\bibitem [{\citenamefont {Gaiotto}\ \emph {et~al.}(2015)\citenamefont {Gaiotto}, \citenamefont {Kapustin}, \citenamefont {Seiberg},\ and\ \citenamefont {Willett}}]{Gaiotto:2014kfa}%
  \BibitemOpen
  \bibfield  {author} {\bibinfo {author} {\bibfnamefont {D.}~\bibnamefont {Gaiotto}}, \bibinfo {author} {\bibfnamefont {A.}~\bibnamefont {Kapustin}}, \bibinfo {author} {\bibfnamefont {N.}~\bibnamefont {Seiberg}},\ and\ \bibinfo {author} {\bibfnamefont {B.}~\bibnamefont {Willett}},\ }\bibfield  {title} {\bibinfo {title} {{Generalized Global Symmetries}},\ }\href {https://doi.org/10.1007/JHEP02(2015)172} {\bibfield  {journal} {\bibinfo  {journal} {JHEP}\ }\textbf {\bibinfo {volume} {02}},\ \bibinfo {pages} {172}},\ \Eprint {https://arxiv.org/abs/1412.5148} {arXiv:1412.5148 [hep-th]} \BibitemShut {NoStop}%
\bibitem [{\citenamefont {C\'{o}rdova}\ \emph {et~al.}(2022)\citenamefont {C\'{o}rdova}, \citenamefont {Dumitrescu}, \citenamefont {Intriligator},\ and\ \citenamefont {Shao}}]{Cordova:2022ruw}%
  \BibitemOpen
  \bibfield  {author} {\bibinfo {author} {\bibfnamefont {C.}~\bibnamefont {C\'{o}rdova}}, \bibinfo {author} {\bibfnamefont {T.~T.}\ \bibnamefont {Dumitrescu}}, \bibinfo {author} {\bibfnamefont {K.}~\bibnamefont {Intriligator}},\ and\ \bibinfo {author} {\bibfnamefont {S.-H.}\ \bibnamefont {Shao}},\ }\bibfield  {title} {\bibinfo {title} {{Snowmass White Paper: Generalized Symmetries in Quantum Field Theory and Beyond}},\ }in\ \href@noop {} {\emph {\bibinfo {booktitle} {{Snowmass 2021}}}}\ (\bibinfo {year} {2022})\ \Eprint {https://arxiv.org/abs/2205.09545} {arXiv:2205.09545 [hep-th]} \BibitemShut {NoStop}%
\bibitem [{\citenamefont {McGreevy}(2023)}]{McGreevy:2022oyu}%
  \BibitemOpen
  \bibfield  {author} {\bibinfo {author} {\bibfnamefont {J.}~\bibnamefont {McGreevy}},\ }\bibfield  {title} {\bibinfo {title} {{Generalized Symmetries in Condensed Matter}},\ }\href {https://doi.org/10.1146/annurev-conmatphys-040721-021029} {\bibfield  {journal} {\bibinfo  {journal} {Ann. Rev. Condensed Matter Phys.}\ }\textbf {\bibinfo {volume} {14}},\ \bibinfo {pages} {57} (\bibinfo {year} {2023})},\ \Eprint {https://arxiv.org/abs/2204.03045} {arXiv:2204.03045 [cond-mat.str-el]} \BibitemShut {NoStop}%
\bibitem [{\citenamefont {Sch\"{a}fer-Nameki}(2024)}]{Schafer-Nameki:2023jdn}%
  \BibitemOpen
  \bibfield  {author} {\bibinfo {author} {\bibfnamefont {S.}~\bibnamefont {Sch\"{a}fer-Nameki}},\ }\bibfield  {title} {\bibinfo {title} {{ICTP lectures on (non-)invertible generalized symmetries}},\ }\href {https://doi.org/10.1016/j.physrep.2024.01.007} {\bibfield  {journal} {\bibinfo  {journal} {Phys. Rept.}\ }\textbf {\bibinfo {volume} {1063}},\ \bibinfo {pages} {1} (\bibinfo {year} {2024})},\ \Eprint {https://arxiv.org/abs/2305.18296} {arXiv:2305.18296 [hep-th]} \BibitemShut {NoStop}%
\bibitem [{\citenamefont {Brennan}\ and\ \citenamefont {Hong}(2023)}]{Brennan:2023mmt}%
  \BibitemOpen
  \bibfield  {author} {\bibinfo {author} {\bibfnamefont {T.~D.}\ \bibnamefont {Brennan}}\ and\ \bibinfo {author} {\bibfnamefont {S.}~\bibnamefont {Hong}},\ }\href@noop {} {\bibinfo {title} {{Introduction to Generalized Global Symmetries in QFT and Particle Physics}}} (\bibinfo {year} {2023}),\ \Eprint {https://arxiv.org/abs/2306.00912} {arXiv:2306.00912 [hep-ph]} \BibitemShut {NoStop}%
\bibitem [{\citenamefont {Bhardwaj}\ \emph {et~al.}(2024)\citenamefont {Bhardwaj}, \citenamefont {Bottini}, \citenamefont {Fraser-Taliente}, \citenamefont {Gladden}, \citenamefont {Gould}, \citenamefont {Platschorre},\ and\ \citenamefont {Tillim}}]{Bhardwaj:2023kri}%
  \BibitemOpen
  \bibfield  {author} {\bibinfo {author} {\bibfnamefont {L.}~\bibnamefont {Bhardwaj}}, \bibinfo {author} {\bibfnamefont {L.~E.}\ \bibnamefont {Bottini}}, \bibinfo {author} {\bibfnamefont {L.}~\bibnamefont {Fraser-Taliente}}, \bibinfo {author} {\bibfnamefont {L.}~\bibnamefont {Gladden}}, \bibinfo {author} {\bibfnamefont {D.~S.~W.}\ \bibnamefont {Gould}}, \bibinfo {author} {\bibfnamefont {A.}~\bibnamefont {Platschorre}},\ and\ \bibinfo {author} {\bibfnamefont {H.}~\bibnamefont {Tillim}},\ }\bibfield  {title} {\bibinfo {title} {{Lectures on generalized symmetries}},\ }\href {https://doi.org/10.1016/j.physrep.2023.11.002} {\bibfield  {journal} {\bibinfo  {journal} {Phys. Rept.}\ }\textbf {\bibinfo {volume} {1051}},\ \bibinfo {pages} {1} (\bibinfo {year} {2024})},\ \Eprint {https://arxiv.org/abs/2307.07547} {arXiv:2307.07547 [hep-th]} \BibitemShut {NoStop}%
\bibitem [{\citenamefont {Luo}\ \emph {et~al.}(2024)\citenamefont {Luo}, \citenamefont {Wang},\ and\ \citenamefont {Wang}}]{Luo:2023ive}%
  \BibitemOpen
  \bibfield  {author} {\bibinfo {author} {\bibfnamefont {R.}~\bibnamefont {Luo}}, \bibinfo {author} {\bibfnamefont {Q.-R.}\ \bibnamefont {Wang}},\ and\ \bibinfo {author} {\bibfnamefont {Y.-N.}\ \bibnamefont {Wang}},\ }\bibfield  {title} {\bibinfo {title} {{Lecture notes on generalized symmetries and applications}},\ }\href {https://doi.org/10.1016/j.physrep.2024.02.002} {\bibfield  {journal} {\bibinfo  {journal} {Phys. Rept.}\ }\textbf {\bibinfo {volume} {1065}},\ \bibinfo {pages} {1} (\bibinfo {year} {2024})},\ \Eprint {https://arxiv.org/abs/2307.09215} {arXiv:2307.09215 [hep-th]} \BibitemShut {NoStop}%
\bibitem [{\citenamefont {Shao}(2023)}]{Shao:2023gho}%
  \BibitemOpen
  \bibfield  {author} {\bibinfo {author} {\bibfnamefont {S.-H.}\ \bibnamefont {Shao}},\ }\href@noop {} {\bibinfo {title} {{What's Done Cannot Be Undone: TASI Lectures on Non-Invertible Symmetries}}} (\bibinfo {year} {2023}),\ \Eprint {https://arxiv.org/abs/2308.00747} {arXiv:2308.00747 [hep-th]} \BibitemShut {NoStop}%
\bibitem [{\citenamefont {Carqueville}\ \emph {et~al.}(2023)\citenamefont {Carqueville}, \citenamefont {Del~Zotto},\ and\ \citenamefont {Runkel}}]{Carqueville:2023jhb}%
  \BibitemOpen
  \bibfield  {author} {\bibinfo {author} {\bibfnamefont {N.}~\bibnamefont {Carqueville}}, \bibinfo {author} {\bibfnamefont {M.}~\bibnamefont {Del~Zotto}},\ and\ \bibinfo {author} {\bibfnamefont {I.}~\bibnamefont {Runkel}},\ }\href@noop {} {\bibinfo {title} {{Topological defects}}} (\bibinfo {year} {2023}),\ \Eprint {https://arxiv.org/abs/2311.02449} {arXiv:2311.02449 [math-ph]} \BibitemShut {NoStop}%
\bibitem [{\citenamefont {Iqbal}(2024)}]{Iqbal:2024pee}%
  \BibitemOpen
  \bibfield  {author} {\bibinfo {author} {\bibfnamefont {N.}~\bibnamefont {Iqbal}},\ }\bibfield  {title} {\bibinfo {title} {{Jena lectures on generalized global symmetries: principles and applications}}\ }(\bibinfo {year} {2024})\ \Eprint {https://arxiv.org/abs/2407.20815} {arXiv:2407.20815 [hep-th]} \BibitemShut {NoStop}%
\bibitem [{\citenamefont {Bhardwaj}\ and\ \citenamefont {Tachikawa}(2018)}]{Bhardwaj:2017xup}%
  \BibitemOpen
  \bibfield  {author} {\bibinfo {author} {\bibfnamefont {L.}~\bibnamefont {Bhardwaj}}\ and\ \bibinfo {author} {\bibfnamefont {Y.}~\bibnamefont {Tachikawa}},\ }\bibfield  {title} {\bibinfo {title} {{On finite symmetries and their gauging in two dimensions}},\ }\href {https://doi.org/10.1007/JHEP03(2018)189} {\bibfield  {journal} {\bibinfo  {journal} {JHEP}\ }\textbf {\bibinfo {volume} {03}},\ \bibinfo {pages} {189}},\ \Eprint {https://arxiv.org/abs/1704.02330} {arXiv:1704.02330 [hep-th]} \BibitemShut {NoStop}%
\bibitem [{\citenamefont {Chang}\ \emph {et~al.}(2019)\citenamefont {Chang}, \citenamefont {Lin}, \citenamefont {Shao}, \citenamefont {Wang},\ and\ \citenamefont {Yin}}]{Chang:2018iay}%
  \BibitemOpen
  \bibfield  {author} {\bibinfo {author} {\bibfnamefont {C.-M.}\ \bibnamefont {Chang}}, \bibinfo {author} {\bibfnamefont {Y.-H.}\ \bibnamefont {Lin}}, \bibinfo {author} {\bibfnamefont {S.-H.}\ \bibnamefont {Shao}}, \bibinfo {author} {\bibfnamefont {Y.}~\bibnamefont {Wang}},\ and\ \bibinfo {author} {\bibfnamefont {X.}~\bibnamefont {Yin}},\ }\bibfield  {title} {\bibinfo {title} {{Topological Defect Lines and Renormalization Group Flows in Two Dimensions}},\ }\href {https://doi.org/10.1007/JHEP01(2019)026} {\bibfield  {journal} {\bibinfo  {journal} {JHEP}\ }\textbf {\bibinfo {volume} {01}},\ \bibinfo {pages} {026}},\ \Eprint {https://arxiv.org/abs/1802.04445} {arXiv:1802.04445 [hep-th]} \BibitemShut {NoStop}%
\bibitem [{\citenamefont {Thorngren}\ and\ \citenamefont {Wang}(2024)}]{Thorngren:2019iar}%
  \BibitemOpen
  \bibfield  {author} {\bibinfo {author} {\bibfnamefont {R.}~\bibnamefont {Thorngren}}\ and\ \bibinfo {author} {\bibfnamefont {Y.}~\bibnamefont {Wang}},\ }\bibfield  {title} {\bibinfo {title} {{Fusion category symmetry. Part I. Anomaly in-flow and gapped phases}},\ }\href {https://doi.org/10.1007/JHEP04(2024)132} {\bibfield  {journal} {\bibinfo  {journal} {JHEP}\ }\textbf {\bibinfo {volume} {04}},\ \bibinfo {pages} {132}},\ \Eprint {https://arxiv.org/abs/1912.02817} {arXiv:1912.02817 [hep-th]} \BibitemShut {NoStop}%
\bibitem [{\citenamefont {Cirac}\ \emph {et~al.}(2021)\citenamefont {Cirac}, \citenamefont {P\'erez-Garc\'ia}, \citenamefont {Schuch},\ and\ \citenamefont {Verstraete}}]{Cirac:2020obd}%
  \BibitemOpen
  \bibfield  {author} {\bibinfo {author} {\bibfnamefont {J.~I.}\ \bibnamefont {Cirac}}, \bibinfo {author} {\bibfnamefont {D.}~\bibnamefont {P\'erez-Garc\'ia}}, \bibinfo {author} {\bibfnamefont {N.}~\bibnamefont {Schuch}},\ and\ \bibinfo {author} {\bibfnamefont {F.}~\bibnamefont {Verstraete}},\ }\bibfield  {title} {\bibinfo {title} {{Matrix product states and projected entangled pair states: Concepts, symmetries, theorems}},\ }\href {https://doi.org/10.1103/RevModPhys.93.045003} {\bibfield  {journal} {\bibinfo  {journal} {Rev. Mod. Phys.}\ }\textbf {\bibinfo {volume} {93}},\ \bibinfo {pages} {045003} (\bibinfo {year} {2021})},\ \Eprint {https://arxiv.org/abs/2011.12127} {arXiv:2011.12127 [quant-ph]} \BibitemShut {NoStop}%
\bibitem [{\citenamefont {Bultinck}\ \emph {et~al.}(2017)\citenamefont {Bultinck}, \citenamefont {Mari\"{e}n}, \citenamefont {Williamson}, \citenamefont {\c{S}ahino\u{g}lu}, \citenamefont {Haegeman},\ and\ \citenamefont {Verstraete}}]{Bultinck:2015bot}%
  \BibitemOpen
  \bibfield  {author} {\bibinfo {author} {\bibfnamefont {N.}~\bibnamefont {Bultinck}}, \bibinfo {author} {\bibfnamefont {M.}~\bibnamefont {Mari\"{e}n}}, \bibinfo {author} {\bibfnamefont {D.~J.}\ \bibnamefont {Williamson}}, \bibinfo {author} {\bibfnamefont {M.~B.}\ \bibnamefont {\c{S}ahino\u{g}lu}}, \bibinfo {author} {\bibfnamefont {J.}~\bibnamefont {Haegeman}},\ and\ \bibinfo {author} {\bibfnamefont {F.}~\bibnamefont {Verstraete}},\ }\bibfield  {title} {\bibinfo {title} {{Anyons and matrix product operator algebras}},\ }\href {https://doi.org/10.1016/j.aop.2017.01.004} {\bibfield  {journal} {\bibinfo  {journal} {Annals Phys.}\ }\textbf {\bibinfo {volume} {378}},\ \bibinfo {pages} {183} (\bibinfo {year} {2017})},\ \Eprint {https://arxiv.org/abs/1511.08090} {arXiv:1511.08090 [cond-mat.str-el]} \BibitemShut {NoStop}%
\bibitem [{\citenamefont {Fidkowski}\ and\ \citenamefont {Kitaev}(2011)}]{Fidkowski_2011}%
  \BibitemOpen
  \bibfield  {author} {\bibinfo {author} {\bibfnamefont {L.}~\bibnamefont {Fidkowski}}\ and\ \bibinfo {author} {\bibfnamefont {A.}~\bibnamefont {Kitaev}},\ }\bibfield  {title} {\bibinfo {title} {Topological phases of fermions in one dimension},\ }\bibfield  {journal} {\bibinfo  {journal} {Physical Review B}\ }\textbf {\bibinfo {volume} {83}},\ \href {https://doi.org/10.1103/physrevb.83.075103} {10.1103/physrevb.83.075103} (\bibinfo {year} {2011}),\ \Eprint {https://arxiv.org/abs/1008.4138} {arXiv:1008.4138 [cond-mat.str-el]} \BibitemShut {NoStop}%
\bibitem [{\citenamefont {Molnar}\ \emph {et~al.}(2022)\citenamefont {Molnar}, \citenamefont {de~Alarc\'on}, \citenamefont {Garre-Rubio}, \citenamefont {Schuch}, \citenamefont {Cirac},\ and\ \citenamefont {P\'erez-Garc\'\i{}a}}]{Molnar:2022nmh}%
  \BibitemOpen
  \bibfield  {author} {\bibinfo {author} {\bibfnamefont {A.}~\bibnamefont {Molnar}}, \bibinfo {author} {\bibfnamefont {A.~R.}\ \bibnamefont {de~Alarc\'on}}, \bibinfo {author} {\bibfnamefont {J.}~\bibnamefont {Garre-Rubio}}, \bibinfo {author} {\bibfnamefont {N.}~\bibnamefont {Schuch}}, \bibinfo {author} {\bibfnamefont {J.~I.}\ \bibnamefont {Cirac}},\ and\ \bibinfo {author} {\bibfnamefont {D.}~\bibnamefont {P\'erez-Garc\'\i{}a}},\ }\href@noop {} {\bibinfo {title} {{Matrix product operator algebras I: representations of weak Hopf algebras and projected entangled pair states}}} (\bibinfo {year} {2022}),\ \Eprint {https://arxiv.org/abs/2204.05940} {arXiv:2204.05940 [quant-ph]} \BibitemShut {NoStop}%
\bibitem [{\citenamefont {Garre-Rubio}\ \emph {et~al.}(2023)\citenamefont {Garre-Rubio}, \citenamefont {Lootens},\ and\ \citenamefont {Moln\'ar}}]{Garre-Rubio:2022uum}%
  \BibitemOpen
  \bibfield  {author} {\bibinfo {author} {\bibfnamefont {J.}~\bibnamefont {Garre-Rubio}}, \bibinfo {author} {\bibfnamefont {L.}~\bibnamefont {Lootens}},\ and\ \bibinfo {author} {\bibfnamefont {A.}~\bibnamefont {Moln\'ar}},\ }\bibfield  {title} {\bibinfo {title} {{Classifying phases protected by matrix product operator symmetries using matrix product states}},\ }\href {https://doi.org/10.22331/q-2023-02-21-927} {\bibfield  {journal} {\bibinfo  {journal} {Quantum}\ }\textbf {\bibinfo {volume} {7}},\ \bibinfo {pages} {927} (\bibinfo {year} {2023})},\ \Eprint {https://arxiv.org/abs/2203.12563} {arXiv:2203.12563 [cond-mat.str-el]} \BibitemShut {NoStop}%
\bibitem [{\citenamefont {Lootens}\ \emph {et~al.}(2023)\citenamefont {Lootens}, \citenamefont {Delcamp}, \citenamefont {Ortiz},\ and\ \citenamefont {Verstraete}}]{Lootens:2021tet}%
  \BibitemOpen
  \bibfield  {author} {\bibinfo {author} {\bibfnamefont {L.}~\bibnamefont {Lootens}}, \bibinfo {author} {\bibfnamefont {C.}~\bibnamefont {Delcamp}}, \bibinfo {author} {\bibfnamefont {G.}~\bibnamefont {Ortiz}},\ and\ \bibinfo {author} {\bibfnamefont {F.}~\bibnamefont {Verstraete}},\ }\bibfield  {title} {\bibinfo {title} {{Dualities in One-Dimensional Quantum Lattice Models: Symmetric Hamiltonians and Matrix Product Operator Intertwiners}},\ }\href {https://doi.org/10.1103/PRXQuantum.4.020357} {\bibfield  {journal} {\bibinfo  {journal} {PRX Quantum}\ }\textbf {\bibinfo {volume} {4}},\ \bibinfo {pages} {020357} (\bibinfo {year} {2023})},\ \Eprint {https://arxiv.org/abs/2112.09091} {arXiv:2112.09091 [quant-ph]} \BibitemShut {NoStop}%
\bibitem [{\citenamefont {Lootens}\ \emph {et~al.}(2024)\citenamefont {Lootens}, \citenamefont {Delcamp},\ and\ \citenamefont {Verstraete}}]{Lootens:2022avn}%
  \BibitemOpen
  \bibfield  {author} {\bibinfo {author} {\bibfnamefont {L.}~\bibnamefont {Lootens}}, \bibinfo {author} {\bibfnamefont {C.}~\bibnamefont {Delcamp}},\ and\ \bibinfo {author} {\bibfnamefont {F.}~\bibnamefont {Verstraete}},\ }\bibfield  {title} {\bibinfo {title} {{Dualities in One-Dimensional Quantum Lattice Models: Topological Sectors}},\ }\href {https://doi.org/10.1103/PRXQuantum.5.010338} {\bibfield  {journal} {\bibinfo  {journal} {PRX Quantum}\ }\textbf {\bibinfo {volume} {5}},\ \bibinfo {pages} {010338} (\bibinfo {year} {2024})},\ \Eprint {https://arxiv.org/abs/2211.03777} {arXiv:2211.03777 [quant-ph]} \BibitemShut {NoStop}%
\bibitem [{\citenamefont {Gorantla}\ \emph {et~al.}(2024)\citenamefont {Gorantla}, \citenamefont {Shao},\ and\ \citenamefont {Tantivasadakarn}}]{Gorantla:2024ocs}%
  \BibitemOpen
  \bibfield  {author} {\bibinfo {author} {\bibfnamefont {P.}~\bibnamefont {Gorantla}}, \bibinfo {author} {\bibfnamefont {S.-H.}\ \bibnamefont {Shao}},\ and\ \bibinfo {author} {\bibfnamefont {N.}~\bibnamefont {Tantivasadakarn}},\ }\href@noop {} {\bibinfo {title} {{Tensor networks for non-invertible symmetries in 3+1d and beyond}}} (\bibinfo {year} {2024}),\ \Eprint {https://arxiv.org/abs/2406.12978} {arXiv:2406.12978 [quant-ph]} \BibitemShut {NoStop}%
\bibitem [{\citenamefont {Fechisin}\ \emph {et~al.}(2023)\citenamefont {Fechisin}, \citenamefont {Tantivasadakarn},\ and\ \citenamefont {Albert}}]{Fechisin:2023dkj}%
  \BibitemOpen
  \bibfield  {author} {\bibinfo {author} {\bibfnamefont {C.}~\bibnamefont {Fechisin}}, \bibinfo {author} {\bibfnamefont {N.}~\bibnamefont {Tantivasadakarn}},\ and\ \bibinfo {author} {\bibfnamefont {V.~V.}\ \bibnamefont {Albert}},\ }\href@noop {} {\bibinfo {title} {{Non-invertible symmetry-protected topological order in a group-based cluster state}}} (\bibinfo {year} {2023}),\ \Eprint {https://arxiv.org/abs/2312.09272} {arXiv:2312.09272 [cond-mat.str-el]} \BibitemShut {NoStop}%
\bibitem [{\citenamefont {Seifnashri}\ and\ \citenamefont {Shao}(2024)}]{Seifnashri:2024dsd}%
  \BibitemOpen
  \bibfield  {author} {\bibinfo {author} {\bibfnamefont {S.}~\bibnamefont {Seifnashri}}\ and\ \bibinfo {author} {\bibfnamefont {S.-H.}\ \bibnamefont {Shao}},\ }\href@noop {} {\bibinfo {title} {{Cluster state as a non-invertible symmetry protected topological phase}}} (\bibinfo {year} {2024}),\ \Eprint {https://arxiv.org/abs/2404.01369} {arXiv:2404.01369 [cond-mat.str-el]} \BibitemShut {NoStop}%
\bibitem [{\citenamefont {Jia}(2024)}]{Jia:2024bng}%
  \BibitemOpen
  \bibfield  {author} {\bibinfo {author} {\bibfnamefont {Z.}~\bibnamefont {Jia}},\ }\href@noop {} {\bibinfo {title} {{Generalized cluster states from Hopf algebras: non-invertible symmetry and Hopf tensor network representation}}} (\bibinfo {year} {2024}),\ \Eprint {https://arxiv.org/abs/2405.09277} {arXiv:2405.09277 [quant-ph]} \BibitemShut {NoStop}%
\bibitem [{\citenamefont {Choi}\ \emph {et~al.}(2024)\citenamefont {Choi}, \citenamefont {Sanghavi}, \citenamefont {Shao},\ and\ \citenamefont {Zheng}}]{Choi:2024rjm}%
  \BibitemOpen
  \bibfield  {author} {\bibinfo {author} {\bibfnamefont {Y.}~\bibnamefont {Choi}}, \bibinfo {author} {\bibfnamefont {Y.}~\bibnamefont {Sanghavi}}, \bibinfo {author} {\bibfnamefont {S.-H.}\ \bibnamefont {Shao}},\ and\ \bibinfo {author} {\bibfnamefont {Y.}~\bibnamefont {Zheng}},\ }\href@noop {} {\bibinfo {title} {{Non-invertible and higher-form symmetries in 2+1d lattice gauge theories}}} (\bibinfo {year} {2024}),\ \Eprint {https://arxiv.org/abs/2405.13105} {arXiv:2405.13105 [cond-mat.str-el]} \BibitemShut {NoStop}%
\bibitem [{\citenamefont {Li}\ and\ \citenamefont {Litvinov}(2024)}]{Li:2024fhy}%
  \BibitemOpen
  \bibfield  {author} {\bibinfo {author} {\bibfnamefont {Y.}~\bibnamefont {Li}}\ and\ \bibinfo {author} {\bibfnamefont {M.}~\bibnamefont {Litvinov}},\ }\href@noop {} {\bibinfo {title} {{Non-invertible SPT, gauging and symmetry fractionalization}}} (\bibinfo {year} {2024}),\ \Eprint {https://arxiv.org/abs/2405.15951} {arXiv:2405.15951 [cond-mat.str-el]} \BibitemShut {NoStop}%
\bibitem [{\citenamefont {Brell}(2015)}]{Brell_2015}%
  \BibitemOpen
  \bibfield  {author} {\bibinfo {author} {\bibfnamefont {C.~G.}\ \bibnamefont {Brell}},\ }\bibfield  {title} {\bibinfo {title} {Generalized cluster states based on finite groups},\ }\href {https://doi.org/10.1088/1367-2630/17/2/023029} {\bibfield  {journal} {\bibinfo  {journal} {New Journal of Physics}\ }\textbf {\bibinfo {volume} {17}},\ \bibinfo {pages} {023029} (\bibinfo {year} {2015})},\ \Eprint {https://arxiv.org/abs/1408.6237} {arXiv:1408.6237 [quant-ph]} \BibitemShut {NoStop}%
\bibitem [{\citenamefont {Inamura}(2022)}]{Inamura:2021szw}%
  \BibitemOpen
  \bibfield  {author} {\bibinfo {author} {\bibfnamefont {K.}~\bibnamefont {Inamura}},\ }\bibfield  {title} {\bibinfo {title} {{On lattice models of gapped phases with fusion category symmetries}},\ }\href {https://doi.org/10.1007/JHEP03(2022)036} {\bibfield  {journal} {\bibinfo  {journal} {JHEP}\ }\textbf {\bibinfo {volume} {03}},\ \bibinfo {pages} {036}},\ \Eprint {https://arxiv.org/abs/2110.12882} {arXiv:2110.12882 [cond-mat.str-el]} \BibitemShut {NoStop}%
\bibitem [{\citenamefont {Etingof}\ \emph {et~al.}(2015)\citenamefont {Etingof}, \citenamefont {Gelaki}, \citenamefont {Nikshych},\ and\ \citenamefont {Ostrik}}]{EGNO2015}%
  \BibitemOpen
  \bibfield  {author} {\bibinfo {author} {\bibfnamefont {P.}~\bibnamefont {Etingof}}, \bibinfo {author} {\bibfnamefont {S.}~\bibnamefont {Gelaki}}, \bibinfo {author} {\bibfnamefont {D.}~\bibnamefont {Nikshych}},\ and\ \bibinfo {author} {\bibfnamefont {V.}~\bibnamefont {Ostrik}},\ }\href {https://doi.org/10.1090/surv/205} {\emph {\bibinfo {title} {{Tensor Categories}}}},\ \bibinfo {series} {Mathematical Surveys and Monographs}, Vol.\ \bibinfo {volume} {205}\ (\bibinfo  {publisher} {American Mathematical Society},\ \bibinfo {address} {Providence, RI},\ \bibinfo {year} {2015})\BibitemShut {NoStop}%
\bibitem [{\citenamefont {Moore}\ and\ \citenamefont {Seiberg}(1989)}]{Moore:1988qv}%
  \BibitemOpen
  \bibfield  {author} {\bibinfo {author} {\bibfnamefont {G.~W.}\ \bibnamefont {Moore}}\ and\ \bibinfo {author} {\bibfnamefont {N.}~\bibnamefont {Seiberg}},\ }\bibfield  {title} {\bibinfo {title} {{Classical and Quantum Conformal Field Theory}},\ }\href {https://doi.org/10.1007/BF01238857} {\bibfield  {journal} {\bibinfo  {journal} {Commun. Math. Phys.}\ }\textbf {\bibinfo {volume} {123}},\ \bibinfo {pages} {177} (\bibinfo {year} {1989})}\BibitemShut {NoStop}%
\bibitem [{\citenamefont {Mac~Lane}(1998)}]{MacLane1998}%
  \BibitemOpen
  \bibfield  {author} {\bibinfo {author} {\bibfnamefont {S.}~\bibnamefont {Mac~Lane}},\ }\href {https://doi.org/10.1007/978-1-4757-4721-8} {\emph {\bibinfo {title} {{Categories for the Working Mathematician}}}},\ \bibinfo {edition} {2nd}\ ed.,\ \bibinfo {series} {Grad. Texts Math.}, Vol.~\bibinfo {volume} {5}\ (\bibinfo  {publisher} {New York, NY: Springer},\ \bibinfo {year} {1998})\BibitemShut {NoStop}%
\bibitem [{\citenamefont {Ostrik}(2003)}]{Ostrik2003}%
  \BibitemOpen
  \bibfield  {author} {\bibinfo {author} {\bibfnamefont {V.}~\bibnamefont {Ostrik}},\ }\bibfield  {title} {\bibinfo {title} {{Module categories, weak Hopf algebras and modular invariants}},\ }\href {https://doi.org/10.1007/s00031-003-0515-6} {\bibfield  {journal} {\bibinfo  {journal} {Transformation Groups}\ }\textbf {\bibinfo {volume} {8}},\ \bibinfo {pages} {177} (\bibinfo {year} {2003})},\ \Eprint {https://arxiv.org/abs/math/0111139} {arXiv:math/0111139 [math.QA]} \BibitemShut {NoStop}%
\bibitem [{\citenamefont {Hayashi}(1999)}]{Hayashi1999}%
  \BibitemOpen
  \bibfield  {author} {\bibinfo {author} {\bibfnamefont {T.}~\bibnamefont {Hayashi}},\ }\href@noop {} {\bibinfo {title} {{A canonical Tannaka duality for finite seimisimple tensor categories}}} (\bibinfo {year} {1999}),\ \Eprint {https://arxiv.org/abs/math/9904073} {arXiv:math/9904073 [math.QA]} \BibitemShut {NoStop}%
\bibitem [{\citenamefont {B\"{o}hm}\ \emph {et~al.}(1999)\citenamefont {B\"{o}hm}, \citenamefont {Nill},\ and\ \citenamefont {Szlach\'{a}nyi}}]{BNS1999}%
  \BibitemOpen
  \bibfield  {author} {\bibinfo {author} {\bibfnamefont {G.}~\bibnamefont {B\"{o}hm}}, \bibinfo {author} {\bibfnamefont {F.}~\bibnamefont {Nill}},\ and\ \bibinfo {author} {\bibfnamefont {K.}~\bibnamefont {Szlach\'{a}nyi}},\ }\bibfield  {title} {\bibinfo {title} {{Weak Hopf Algebras: I. Integral Theory and $C^*$-Structure}},\ }\href {https://doi.org/https://doi.org/10.1006/jabr.1999.7984} {\bibfield  {journal} {\bibinfo  {journal} {Journal of Algebra}\ }\textbf {\bibinfo {volume} {221}},\ \bibinfo {pages} {385} (\bibinfo {year} {1999})},\ \Eprint {https://arxiv.org/abs/math/9805116} {arXiv:math/9805116 [math.QA]} \BibitemShut {NoStop}%
\bibitem [{\citenamefont {Inamura}(2021)}]{Inamura:2021wuo}%
  \BibitemOpen
  \bibfield  {author} {\bibinfo {author} {\bibfnamefont {K.}~\bibnamefont {Inamura}},\ }\bibfield  {title} {\bibinfo {title} {{Topological field theories and symmetry protected topological phases with fusion category symmetries}},\ }\href {https://doi.org/10.1007/JHEP05(2021)204} {\bibfield  {journal} {\bibinfo  {journal} {JHEP}\ }\textbf {\bibinfo {volume} {05}},\ \bibinfo {pages} {204}},\ \Eprint {https://arxiv.org/abs/2103.15588} {arXiv:2103.15588 [cond-mat.str-el]} \BibitemShut {NoStop}%
\bibitem [{\citenamefont {Komargodski}\ \emph {et~al.}(2021)\citenamefont {Komargodski}, \citenamefont {Ohmori}, \citenamefont {Roumpedakis},\ and\ \citenamefont {Seifnashri}}]{Komargodski:2020mxz}%
  \BibitemOpen
  \bibfield  {author} {\bibinfo {author} {\bibfnamefont {Z.}~\bibnamefont {Komargodski}}, \bibinfo {author} {\bibfnamefont {K.}~\bibnamefont {Ohmori}}, \bibinfo {author} {\bibfnamefont {K.}~\bibnamefont {Roumpedakis}},\ and\ \bibinfo {author} {\bibfnamefont {S.}~\bibnamefont {Seifnashri}},\ }\bibfield  {title} {\bibinfo {title} {{Symmetries and strings of adjoint QCD$_{2}$}},\ }\href {https://doi.org/10.1007/JHEP03(2021)103} {\bibfield  {journal} {\bibinfo  {journal} {JHEP}\ }\textbf {\bibinfo {volume} {03}},\ \bibinfo {pages} {103}},\ \Eprint {https://arxiv.org/abs/2008.07567} {arXiv:2008.07567 [hep-th]} \BibitemShut {NoStop}%
\bibitem [{\citenamefont {Huang}\ \emph {et~al.}(2021)\citenamefont {Huang}, \citenamefont {Lin},\ and\ \citenamefont {Seifnashri}}]{Huang:2021zvu}%
  \BibitemOpen
  \bibfield  {author} {\bibinfo {author} {\bibfnamefont {T.-C.}\ \bibnamefont {Huang}}, \bibinfo {author} {\bibfnamefont {Y.-H.}\ \bibnamefont {Lin}},\ and\ \bibinfo {author} {\bibfnamefont {S.}~\bibnamefont {Seifnashri}},\ }\bibfield  {title} {\bibinfo {title} {{Construction of two-dimensional topological field theories with non-invertible symmetries}},\ }\href {https://doi.org/10.1007/JHEP12(2021)028} {\bibfield  {journal} {\bibinfo  {journal} {JHEP}\ }\textbf {\bibinfo {volume} {12}},\ \bibinfo {pages} {028}},\ \Eprint {https://arxiv.org/abs/2110.02958} {arXiv:2110.02958 [hep-th]} \BibitemShut {NoStop}%
\bibitem [{\citenamefont {Tannaka}(1939)}]{Tannaka1939}%
  \BibitemOpen
  \bibfield  {author} {\bibinfo {author} {\bibfnamefont {T.}~\bibnamefont {Tannaka}},\ }\bibfield  {title} {\bibinfo {title} {{\"{U}ber den Dualit\"{a}tssatz der nichtkommutativen topologischen Gruppen}},\ }\href@noop {} {\bibfield  {journal} {\bibinfo  {journal} {Tohoku Mathematical Journal, First Series}\ }\textbf {\bibinfo {volume} {45}},\ \bibinfo {pages} {1} (\bibinfo {year} {1939})}\BibitemShut {NoStop}%
\bibitem [{\citenamefont {Krein}(1949)}]{Krein1949}%
  \BibitemOpen
  \bibfield  {author} {\bibinfo {author} {\bibfnamefont {M.}~\bibnamefont {Krein}},\ }\bibfield  {title} {\bibinfo {title} {{A principle of duality for bicompact groups and quadratic block algebras}},\ }\href@noop {} {\bibfield  {journal} {\bibinfo  {journal} {Doklady AN SSSR}\ }\textbf {\bibinfo {volume} {69}},\ \bibinfo {pages} {725} (\bibinfo {year} {1949})}\BibitemShut {NoStop}%
\bibitem [{\citenamefont {Ulbrich}(1990)}]{Ulbrich1990}%
  \BibitemOpen
  \bibfield  {author} {\bibinfo {author} {\bibfnamefont {K.-H.}\ \bibnamefont {Ulbrich}},\ }\bibfield  {title} {\bibinfo {title} {{On Hopf algebras and rigid monoidal categories}},\ }\href {https://doi.org/10.1007/BF02764622} {\bibfield  {journal} {\bibinfo  {journal} {Israel Journal of Mathematics}\ }\textbf {\bibinfo {volume} {72}},\ \bibinfo {pages} {252} (\bibinfo {year} {1990})}\BibitemShut {NoStop}%
\bibitem [{\citenamefont {Joyal}\ and\ \citenamefont {Street}(1991)}]{JS1991}%
  \BibitemOpen
  \bibfield  {author} {\bibinfo {author} {\bibfnamefont {A.}~\bibnamefont {Joyal}}\ and\ \bibinfo {author} {\bibfnamefont {R.}~\bibnamefont {Street}},\ }\bibfield  {title} {\bibinfo {title} {{An introduction to Tannaka duality and quantum groups}},\ }in\ \href {https://link.springer.com/chapter/10.1007/BFb0084235} {\emph {\bibinfo {booktitle} {Category Theory}}},\ \bibinfo {editor} {edited by\ \bibinfo {editor} {\bibfnamefont {A.}~\bibnamefont {Carboni}}, \bibinfo {editor} {\bibfnamefont {M.~C.}\ \bibnamefont {Pedicchio}},\ and\ \bibinfo {editor} {\bibfnamefont {G.}~\bibnamefont {Rosolini}}}\ (\bibinfo  {publisher} {Springer Berlin Heidelberg},\ \bibinfo {address} {Berlin, Heidelberg},\ \bibinfo {year} {1991})\ pp.\ \bibinfo {pages} {413--492}\BibitemShut {NoStop}%
\bibitem [{\citenamefont {Schauenburg}(1992)}]{Schauenburg1992}%
  \BibitemOpen
  \bibfield  {author} {\bibinfo {author} {\bibfnamefont {P.}~\bibnamefont {Schauenburg}},\ }\bibfield  {title} {\bibinfo {title} {{Tannaka duality for arbitrary Hopf algebras}},\ }\href@noop {} {\bibfield  {journal} {\bibinfo  {journal} {Algebra Berichte}\ }\textbf {\bibinfo {volume} {66}} (\bibinfo {year} {1992})}\BibitemShut {NoStop}%
\bibitem [{\citenamefont {Thorngren}\ and\ \citenamefont {Else}(2018)}]{Thorngren:1612.00846}%
  \BibitemOpen
  \bibfield  {author} {\bibinfo {author} {\bibfnamefont {R.}~\bibnamefont {Thorngren}}\ and\ \bibinfo {author} {\bibfnamefont {D.~V.}\ \bibnamefont {Else}},\ }\bibfield  {title} {\bibinfo {title} {Gauging spatial symmetries and the classification of topological crystalline phases},\ }\bibfield  {journal} {\bibinfo  {journal} {Physical Review X}\ }\textbf {\bibinfo {volume} {8}},\ \href {https://doi.org/10.1103/physrevx.8.011040} {10.1103/physrevx.8.011040} (\bibinfo {year} {2018}),\ \Eprint {https://arxiv.org/abs/1612.00846} {arXiv:1612.00846 [cond-mat.str-el]} \BibitemShut {NoStop}%
\bibitem [{\citenamefont {Thorngren}(2021)}]{Thorngren2021YITP}%
  \BibitemOpen
  \bibfield  {author} {\bibinfo {author} {\bibfnamefont {R.}~\bibnamefont {Thorngren}},\ }\href {https://www.ms.u-tokyo.ac.jp/~yasuyuki/yitp2021.htm} {\bibinfo {title} {Berry phase, diabolical points, and pivot hamiltonians}} (\bibinfo {year} {2021}),\ \bibinfo {note} {talk at Yukawa Institute for Theoretical Physics}\BibitemShut {NoStop}%
\bibitem [{\citenamefont {Bachmann}\ \emph {et~al.}(2024)\citenamefont {Bachmann}, \citenamefont {De~Roeck}, \citenamefont {Fraas},\ and\ \citenamefont {Jappens}}]{Bachmann:2022bhx}%
  \BibitemOpen
  \bibfield  {author} {\bibinfo {author} {\bibfnamefont {S.}~\bibnamefont {Bachmann}}, \bibinfo {author} {\bibfnamefont {W.}~\bibnamefont {De~Roeck}}, \bibinfo {author} {\bibfnamefont {M.}~\bibnamefont {Fraas}},\ and\ \bibinfo {author} {\bibfnamefont {T.}~\bibnamefont {Jappens}},\ }\bibfield  {title} {\bibinfo {title} {{A Classification of $G$-Charge Thouless Pumps in 1D Invertible States}},\ }\href {https://doi.org/10.1007/s00220-024-05010-w} {\bibfield  {journal} {\bibinfo  {journal} {Commun. Math. Phys.}\ }\textbf {\bibinfo {volume} {405}},\ \bibinfo {pages} {157} (\bibinfo {year} {2024})},\ \Eprint {https://arxiv.org/abs/2204.03763} {arXiv:2204.03763 [math-ph]} \BibitemShut {NoStop}%
\bibitem [{\citenamefont {Etingof}\ and\ \citenamefont {Gelaki}(2001)}]{EG2001}%
  \BibitemOpen
  \bibfield  {author} {\bibinfo {author} {\bibfnamefont {P.}~\bibnamefont {Etingof}}\ and\ \bibinfo {author} {\bibfnamefont {S.}~\bibnamefont {Gelaki}},\ }\bibfield  {title} {\bibinfo {title} {{Isocategorical groups}},\ }\href@noop {} {\bibfield  {journal} {\bibinfo  {journal} {International Mathematics Research Notices}\ }\textbf {\bibinfo {volume} {2001}},\ \bibinfo {pages} {59} (\bibinfo {year} {2001})},\ \Eprint {https://arxiv.org/abs/math/0007196} {arXiv:math/0007196 [math.QA]} \BibitemShut {NoStop}%
\bibitem [{\citenamefont {Fannes}\ \emph {et~al.}(1992)\citenamefont {Fannes}, \citenamefont {Nachtergaele},\ and\ \citenamefont {Werner}}]{FNW92}%
  \BibitemOpen
  \bibfield  {author} {\bibinfo {author} {\bibfnamefont {M.}~\bibnamefont {Fannes}}, \bibinfo {author} {\bibfnamefont {B.}~\bibnamefont {Nachtergaele}},\ and\ \bibinfo {author} {\bibfnamefont {R.~F.}\ \bibnamefont {Werner}},\ }\bibfield  {title} {\bibinfo {title} {{Finitely correlated states on quantum spin chains}},\ }\href@noop {} {\bibfield  {journal} {\bibinfo  {journal} {Communications in Mathematical Physics}\ }\textbf {\bibinfo {volume} {144}},\ \bibinfo {pages} {443} (\bibinfo {year} {1992})}\BibitemShut {NoStop}%
\bibitem [{\citenamefont {Verstraete}\ and\ \citenamefont {Cirac}(2006)}]{PhysRevB.73.094423}%
  \BibitemOpen
  \bibfield  {author} {\bibinfo {author} {\bibfnamefont {F.}~\bibnamefont {Verstraete}}\ and\ \bibinfo {author} {\bibfnamefont {J.~I.}\ \bibnamefont {Cirac}},\ }\bibfield  {title} {\bibinfo {title} {Matrix product states represent ground states faithfully},\ }\href {https://doi.org/10.1103/PhysRevB.73.094423} {\bibfield  {journal} {\bibinfo  {journal} {Phys. Rev. B}\ }\textbf {\bibinfo {volume} {73}},\ \bibinfo {pages} {094423} (\bibinfo {year} {2006})},\ \Eprint {https://arxiv.org/abs/cond-mat/0505140} {arXiv:cond-mat/0505140 [cond-mat.str-el]} \BibitemShut {NoStop}%
\bibitem [{\citenamefont {Hastings}(2007)}]{Hastings:2007iok}%
  \BibitemOpen
  \bibfield  {author} {\bibinfo {author} {\bibfnamefont {M.~B.}\ \bibnamefont {Hastings}},\ }\bibfield  {title} {\bibinfo {title} {{An area law for one-dimensional quantum systems}},\ }\href {https://doi.org/10.1088/1742-5468/2007/08/P08024} {\bibfield  {journal} {\bibinfo  {journal} {J. Stat. Mech.}\ }\textbf {\bibinfo {volume} {0708}},\ \bibinfo {pages} {P08024} (\bibinfo {year} {2007})},\ \Eprint {https://arxiv.org/abs/0705.2024} {arXiv:0705.2024 [quant-ph]} \BibitemShut {NoStop}%
\bibitem [{\citenamefont {P\'erez-Garc\'ia}\ \emph {et~al.}(2007)\citenamefont {P\'erez-Garc\'ia}, \citenamefont {Verstraete}, \citenamefont {Wolf},\ and\ \citenamefont {Cirac}}]{Perez-Garcia:2006nqo}%
  \BibitemOpen
  \bibfield  {author} {\bibinfo {author} {\bibfnamefont {D.}~\bibnamefont {P\'erez-Garc\'ia}}, \bibinfo {author} {\bibfnamefont {F.}~\bibnamefont {Verstraete}}, \bibinfo {author} {\bibfnamefont {M.~M.}\ \bibnamefont {Wolf}},\ and\ \bibinfo {author} {\bibfnamefont {J.~I.}\ \bibnamefont {Cirac}},\ }\bibfield  {title} {\bibinfo {title} {{Matrix product state representations}},\ }\href {https://doi.org/10.26421/QIC7.5-6-1} {\bibfield  {journal} {\bibinfo  {journal} {Quant. Inf. Comput.}\ }\textbf {\bibinfo {volume} {7}},\ \bibinfo {pages} {401} (\bibinfo {year} {2007})},\ \Eprint {https://arxiv.org/abs/quant-ph/0608197} {arXiv:quant-ph/0608197} \BibitemShut {NoStop}%
\bibitem [{\citenamefont {Ohyama}\ and\ \citenamefont {Ryu}(2024{\natexlab{a}})}]{OR23}%
  \BibitemOpen
  \bibfield  {author} {\bibinfo {author} {\bibfnamefont {S.}~\bibnamefont {Ohyama}}\ and\ \bibinfo {author} {\bibfnamefont {S.}~\bibnamefont {Ryu}},\ }\bibfield  {title} {\bibinfo {title} {{Higher structures in matrix product states}},\ }\href {https://doi.org/10.1103/PhysRevB.109.115152} {\bibfield  {journal} {\bibinfo  {journal} {Phys. Rev. B}\ }\textbf {\bibinfo {volume} {109}},\ \bibinfo {pages} {115152} (\bibinfo {year} {2024}{\natexlab{a}})},\ \Eprint {https://arxiv.org/abs/2304.05356} {arXiv:2304.05356 [cond-mat.str-el]} \BibitemShut {NoStop}%
\bibitem [{\citenamefont {Sanz}\ \emph {et~al.}(2010)\citenamefont {Sanz}, \citenamefont {P\'erez-Garc\'ia}, \citenamefont {Wolf},\ and\ \citenamefont {Cirac}}]{SP-GWC10}%
  \BibitemOpen
  \bibfield  {author} {\bibinfo {author} {\bibfnamefont {M.}~\bibnamefont {Sanz}}, \bibinfo {author} {\bibfnamefont {D.}~\bibnamefont {P\'erez-Garc\'ia}}, \bibinfo {author} {\bibfnamefont {M.~M.}\ \bibnamefont {Wolf}},\ and\ \bibinfo {author} {\bibfnamefont {J.~I.}\ \bibnamefont {Cirac}},\ }\bibfield  {title} {\bibinfo {title} {{A Quantum Version of Wielandt's Inequality}},\ }\href {https://doi.org/10.1109/TIT.2010.2054552} {\bibfield  {journal} {\bibinfo  {journal} {IEEE Transactions on Information Theory}\ }\textbf {\bibinfo {volume} {56}},\ \bibinfo {pages} {4668} (\bibinfo {year} {2010})},\ \Eprint {https://arxiv.org/abs/0909.5347} {arXiv:0909.5347 [quant-ph]} \BibitemShut {NoStop}%
\bibitem [{\citenamefont {Garre-Rubio}\ and\ \citenamefont {Schuch}(2024)}]{Rubio:2024aiw}%
  \BibitemOpen
  \bibfield  {author} {\bibinfo {author} {\bibfnamefont {J.}~\bibnamefont {Garre-Rubio}}\ and\ \bibinfo {author} {\bibfnamefont {N.}~\bibnamefont {Schuch}},\ }\href@noop {} {\bibinfo {title} {{Fractional domain wall statistics in spin chains with anomalous symmetries}}} (\bibinfo {year} {2024}),\ \Eprint {https://arxiv.org/abs/2405.00439} {arXiv:2405.00439 [cond-mat.str-el]} \BibitemShut {NoStop}%
\bibitem [{\citenamefont {Kitaev}\ and\ \citenamefont {Kong}(2012)}]{KK2011}%
  \BibitemOpen
  \bibfield  {author} {\bibinfo {author} {\bibfnamefont {A.}~\bibnamefont {Kitaev}}\ and\ \bibinfo {author} {\bibfnamefont {L.}~\bibnamefont {Kong}},\ }\bibfield  {title} {\bibinfo {title} {{Models for Gapped Boundaries and Domain Walls}},\ }\href {https://doi.org/10.1007/s00220-012-1500-5} {\bibfield  {journal} {\bibinfo  {journal} {Communications in Mathematical Physics}\ }\textbf {\bibinfo {volume} {313}},\ \bibinfo {pages} {351} (\bibinfo {year} {2012})},\ \Eprint {https://arxiv.org/abs/1104.5047} {arXiv:1104.5047 [cond-mat.str-el]} \BibitemShut {NoStop}%
\bibitem [{\citenamefont {Barter}\ \emph {et~al.}(2022)\citenamefont {Barter}, \citenamefont {Bridgeman},\ and\ \citenamefont {Wolf}}]{Barter_2022}%
  \BibitemOpen
  \bibfield  {author} {\bibinfo {author} {\bibfnamefont {D.}~\bibnamefont {Barter}}, \bibinfo {author} {\bibfnamefont {J.}~\bibnamefont {Bridgeman}},\ and\ \bibinfo {author} {\bibfnamefont {R.}~\bibnamefont {Wolf}},\ }\bibfield  {title} {\bibinfo {title} {{Computing associators of endomorphism fusion categories}},\ }\bibfield  {journal} {\bibinfo  {journal} {SciPost Physics}\ }\textbf {\bibinfo {volume} {13}},\ \href {https://doi.org/10.21468/scipostphys.13.2.029} {10.21468/scipostphys.13.2.029} (\bibinfo {year} {2022}),\ \Eprint {https://arxiv.org/abs/2110.03644} {arXiv:2110.03644 [math.QA]} \BibitemShut {NoStop}%
\bibitem [{\citenamefont {Choi}\ \emph {et~al.}(2023)\citenamefont {Choi}, \citenamefont {Rayhaun}, \citenamefont {Sanghavi},\ and\ \citenamefont {Shao}}]{Choi:2023xjw}%
  \BibitemOpen
  \bibfield  {author} {\bibinfo {author} {\bibfnamefont {Y.}~\bibnamefont {Choi}}, \bibinfo {author} {\bibfnamefont {B.~C.}\ \bibnamefont {Rayhaun}}, \bibinfo {author} {\bibfnamefont {Y.}~\bibnamefont {Sanghavi}},\ and\ \bibinfo {author} {\bibfnamefont {S.-H.}\ \bibnamefont {Shao}},\ }\bibfield  {title} {\bibinfo {title} {{Remarks on boundaries, anomalies, and noninvertible symmetries}},\ }\href {https://doi.org/10.1103/PhysRevD.108.125005} {\bibfield  {journal} {\bibinfo  {journal} {Phys. Rev. D}\ }\textbf {\bibinfo {volume} {108}},\ \bibinfo {pages} {125005} (\bibinfo {year} {2023})},\ \Eprint {https://arxiv.org/abs/2305.09713} {arXiv:2305.09713 [hep-th]} \BibitemShut {NoStop}%
\bibitem [{\citenamefont {C\'{o}rdova}\ \emph {et~al.}(2024{\natexlab{a}})\citenamefont {C\'{o}rdova}, \citenamefont {Garc\'\i{}a-Sep\'ulveda},\ and\ \citenamefont {Holfester}}]{Cordova:2024vsq}%
  \BibitemOpen
  \bibfield  {author} {\bibinfo {author} {\bibfnamefont {C.}~\bibnamefont {C\'{o}rdova}}, \bibinfo {author} {\bibfnamefont {D.}~\bibnamefont {Garc\'\i{}a-Sep\'ulveda}},\ and\ \bibinfo {author} {\bibfnamefont {N.}~\bibnamefont {Holfester}},\ }\bibfield  {title} {\bibinfo {title} {{Particle-soliton degeneracies from spontaneously broken non-invertible symmetry}},\ }\href {https://doi.org/10.1007/JHEP07(2024)154} {\bibfield  {journal} {\bibinfo  {journal} {JHEP}\ }\textbf {\bibinfo {volume} {07}},\ \bibinfo {pages} {154}},\ \Eprint {https://arxiv.org/abs/2403.08883} {arXiv:2403.08883 [hep-th]} \BibitemShut {NoStop}%
\bibitem [{\citenamefont {C\'{o}rdova}\ \emph {et~al.}(2024{\natexlab{b}})\citenamefont {C\'{o}rdova}, \citenamefont {Holfester},\ and\ \citenamefont {Ohmori}}]{Cordova:2024iti}%
  \BibitemOpen
  \bibfield  {author} {\bibinfo {author} {\bibfnamefont {C.}~\bibnamefont {C\'{o}rdova}}, \bibinfo {author} {\bibfnamefont {N.}~\bibnamefont {Holfester}},\ and\ \bibinfo {author} {\bibfnamefont {K.}~\bibnamefont {Ohmori}},\ }\href@noop {} {\bibinfo {title} {{Representation Theory of Solitons}}} (\bibinfo {year} {2024}{\natexlab{b}}),\ \Eprint {https://arxiv.org/abs/2408.11045} {arXiv:2408.11045 [hep-th]} \BibitemShut {NoStop}%
\bibitem [{\citenamefont {Copetti}(2024)}]{Copetti:2024onh}%
  \BibitemOpen
  \bibfield  {author} {\bibinfo {author} {\bibfnamefont {C.}~\bibnamefont {Copetti}},\ }\href@noop {} {\bibinfo {title} {{Defect Charges, Gapped Boundary Conditions, and the Symmetry TFT}}} (\bibinfo {year} {2024}),\ \Eprint {https://arxiv.org/abs/2408.01490} {arXiv:2408.01490 [hep-th]} \BibitemShut {NoStop}%
\bibitem [{\citenamefont {Copetti}\ \emph {et~al.}(2024)\citenamefont {Copetti}, \citenamefont {C\'{o}rdova},\ and\ \citenamefont {Komatsu}}]{Copetti:2024dcz}%
  \BibitemOpen
  \bibfield  {author} {\bibinfo {author} {\bibfnamefont {C.}~\bibnamefont {Copetti}}, \bibinfo {author} {\bibfnamefont {L.}~\bibnamefont {C\'{o}rdova}},\ and\ \bibinfo {author} {\bibfnamefont {S.}~\bibnamefont {Komatsu}},\ }\href@noop {} {\bibinfo {title} {{S-Matrix Bootstrap and Non-Invertible Symmetries}}} (\bibinfo {year} {2024}),\ \Eprint {https://arxiv.org/abs/2408.13132} {arXiv:2408.13132 [hep-th]} \BibitemShut {NoStop}%
\bibitem [{\citenamefont {Zheng}(2024)}]{Zheng2024Oxford}%
  \BibitemOpen
  \bibfield  {author} {\bibinfo {author} {\bibfnamefont {Y.}~\bibnamefont {Zheng}},\ }\href {https://www.youtube.com/watch?v=CPoPVG_th-o} {\bibinfo {title} {{Symmetry TFTs, boundaries, and non-invertible symmetry resolved Affleck-Ludwig-Cardy formula}}} (\bibinfo {year} {2024}),\ \bibinfo {note} {talk at University of Oxford}\BibitemShut {NoStop}%
\bibitem [{\citenamefont {Callan}\ and\ \citenamefont {Harvey}(1985)}]{Callan:1984sa}%
  \BibitemOpen
  \bibfield  {author} {\bibinfo {author} {\bibfnamefont {C.~G.}\ \bibnamefont {Callan}, \bibfnamefont {Jr.}}\ and\ \bibinfo {author} {\bibfnamefont {J.~A.}\ \bibnamefont {Harvey}},\ }\bibfield  {title} {\bibinfo {title} {{Anomalies and Fermion Zero Modes on Strings and Domain Walls}},\ }\href {https://doi.org/10.1016/0550-3213(85)90489-4} {\bibfield  {journal} {\bibinfo  {journal} {Nucl. Phys. B}\ }\textbf {\bibinfo {volume} {250}},\ \bibinfo {pages} {427} (\bibinfo {year} {1985})}\BibitemShut {NoStop}%
\bibitem [{\citenamefont {Turaev}\ and\ \citenamefont {Viro}(1992)}]{Turaev:1992hq}%
  \BibitemOpen
  \bibfield  {author} {\bibinfo {author} {\bibfnamefont {V.~G.}\ \bibnamefont {Turaev}}\ and\ \bibinfo {author} {\bibfnamefont {O.~Y.}\ \bibnamefont {Viro}},\ }\bibfield  {title} {\bibinfo {title} {{State sum invariants of 3 manifolds and quantum 6j symbols}},\ }\href {https://doi.org/10.1016/0040-9383(92)90015-A} {\bibfield  {journal} {\bibinfo  {journal} {Topology}\ }\textbf {\bibinfo {volume} {31}},\ \bibinfo {pages} {865} (\bibinfo {year} {1992})}\BibitemShut {NoStop}%
\bibitem [{\citenamefont {Barrett}\ and\ \citenamefont {Westbury}(1996)}]{Barrett:1993ab}%
  \BibitemOpen
  \bibfield  {author} {\bibinfo {author} {\bibfnamefont {J.~W.}\ \bibnamefont {Barrett}}\ and\ \bibinfo {author} {\bibfnamefont {B.~W.}\ \bibnamefont {Westbury}},\ }\bibfield  {title} {\bibinfo {title} {{Invariants of piecewise linear three manifolds}},\ }\href {https://doi.org/10.1090/S0002-9947-96-01660-1} {\bibfield  {journal} {\bibinfo  {journal} {Trans. Am. Math. Soc.}\ }\textbf {\bibinfo {volume} {348}},\ \bibinfo {pages} {3997} (\bibinfo {year} {1996})},\ \Eprint {https://arxiv.org/abs/hep-th/9311155} {arXiv:hep-th/9311155} \BibitemShut {NoStop}%
\bibitem [{\citenamefont {Levin}\ and\ \citenamefont {Wen}(2005)}]{Levin:2004mi}%
  \BibitemOpen
  \bibfield  {author} {\bibinfo {author} {\bibfnamefont {M.~A.}\ \bibnamefont {Levin}}\ and\ \bibinfo {author} {\bibfnamefont {X.-G.}\ \bibnamefont {Wen}},\ }\bibfield  {title} {\bibinfo {title} {{String net condensation: A Physical mechanism for topological phases}},\ }\href {https://doi.org/10.1103/PhysRevB.71.045110} {\bibfield  {journal} {\bibinfo  {journal} {Phys. Rev. B}\ }\textbf {\bibinfo {volume} {71}},\ \bibinfo {pages} {045110} (\bibinfo {year} {2005})},\ \Eprint {https://arxiv.org/abs/cond-mat/0404617} {arXiv:cond-mat/0404617} \BibitemShut {NoStop}%
\bibitem [{\citenamefont {Fuchs}\ \emph {et~al.}(2013)\citenamefont {Fuchs}, \citenamefont {Schweigert},\ and\ \citenamefont {Valentino}}]{Fuchs:2012dt}%
  \BibitemOpen
  \bibfield  {author} {\bibinfo {author} {\bibfnamefont {J.}~\bibnamefont {Fuchs}}, \bibinfo {author} {\bibfnamefont {C.}~\bibnamefont {Schweigert}},\ and\ \bibinfo {author} {\bibfnamefont {A.}~\bibnamefont {Valentino}},\ }\bibfield  {title} {\bibinfo {title} {{Bicategories for boundary conditions and for surface defects in 3-d TFT}},\ }\href {https://doi.org/10.1007/s00220-013-1723-0} {\bibfield  {journal} {\bibinfo  {journal} {Commun. Math. Phys.}\ }\textbf {\bibinfo {volume} {321}},\ \bibinfo {pages} {543} (\bibinfo {year} {2013})},\ \Eprint {https://arxiv.org/abs/1203.4568} {arXiv:1203.4568 [hep-th]} \BibitemShut {NoStop}%
\bibitem [{\citenamefont {Kapustin}\ and\ \citenamefont {Saulina}(2011)}]{Kapustin:2010hk}%
  \BibitemOpen
  \bibfield  {author} {\bibinfo {author} {\bibfnamefont {A.}~\bibnamefont {Kapustin}}\ and\ \bibinfo {author} {\bibfnamefont {N.}~\bibnamefont {Saulina}},\ }\bibfield  {title} {\bibinfo {title} {{Topological boundary conditions in abelian Chern-Simons theory}},\ }\href {https://doi.org/10.1016/j.nuclphysb.2010.12.017} {\bibfield  {journal} {\bibinfo  {journal} {Nucl. Phys. B}\ }\textbf {\bibinfo {volume} {845}},\ \bibinfo {pages} {393} (\bibinfo {year} {2011})},\ \Eprint {https://arxiv.org/abs/1008.0654} {arXiv:1008.0654 [hep-th]} \BibitemShut {NoStop}%
\bibitem [{\citenamefont {Tachikawa}(2020)}]{Tachikawa:2017gyf}%
  \BibitemOpen
  \bibfield  {author} {\bibinfo {author} {\bibfnamefont {Y.}~\bibnamefont {Tachikawa}},\ }\bibfield  {title} {\bibinfo {title} {{On gauging finite subgroups}},\ }\href {https://doi.org/10.21468/SciPostPhys.8.1.015} {\bibfield  {journal} {\bibinfo  {journal} {SciPost Phys.}\ }\textbf {\bibinfo {volume} {8}},\ \bibinfo {pages} {015} (\bibinfo {year} {2020})},\ \Eprint {https://arxiv.org/abs/1712.09542} {arXiv:1712.09542 [hep-th]} \BibitemShut {NoStop}%
\bibitem [{\citenamefont {Naidu}(2007)}]{naidu2007categorical}%
  \BibitemOpen
  \bibfield  {author} {\bibinfo {author} {\bibfnamefont {D.}~\bibnamefont {Naidu}},\ }\bibfield  {title} {\bibinfo {title} {{Categorical Morita equivalence for group-theoretical categories}},\ }\href@noop {} {\bibfield  {journal} {\bibinfo  {journal} {Communications in Algebra}\ }\textbf {\bibinfo {volume} {35}},\ \bibinfo {pages} {3544} (\bibinfo {year} {2007})},\ \Eprint {https://arxiv.org/abs/math/0605530} {arXiv:math/0605530 [math.QA]} \BibitemShut {NoStop}%
\bibitem [{\citenamefont {Naidu}\ and\ \citenamefont {Nikshych}(2008)}]{naidu2008lagrangian}%
  \BibitemOpen
  \bibfield  {author} {\bibinfo {author} {\bibfnamefont {D.}~\bibnamefont {Naidu}}\ and\ \bibinfo {author} {\bibfnamefont {D.}~\bibnamefont {Nikshych}},\ }\bibfield  {title} {\bibinfo {title} {{Lagrangian subcategories and braided tensor equivalences of twisted quantum doubles of finite groups}},\ }\href@noop {} {\bibfield  {journal} {\bibinfo  {journal} {Communications in mathematical physics}\ }\textbf {\bibinfo {volume} {279}},\ \bibinfo {pages} {845} (\bibinfo {year} {2008})},\ \Eprint {https://arxiv.org/abs/0705.0665} {arXiv:0705.0665 [math.QA]} \BibitemShut {NoStop}%
\bibitem [{\citenamefont {Uribe~Jongbloed}(2017)}]{uribe2017classification}%
  \BibitemOpen
  \bibfield  {author} {\bibinfo {author} {\bibfnamefont {B.}~\bibnamefont {Uribe~Jongbloed}},\ }\bibfield  {title} {\bibinfo {title} {{On the classification of pointed fusion categories up to weak Morita equivalence}},\ }\href@noop {} {\bibfield  {journal} {\bibinfo  {journal} {Pacific Journal of Mathematics}\ }\textbf {\bibinfo {volume} {290}},\ \bibinfo {pages} {437} (\bibinfo {year} {2017})},\ \Eprint {https://arxiv.org/abs/1511.05522} {arXiv:1511.05522 [math.AT]} \BibitemShut {NoStop}%
\bibitem [{\citenamefont {Hsin}\ \emph {et~al.}(2020)\citenamefont {Hsin}, \citenamefont {Kapustin},\ and\ \citenamefont {Thorngren}}]{Hsin:2020cgg}%
  \BibitemOpen
  \bibfield  {author} {\bibinfo {author} {\bibfnamefont {P.-S.}\ \bibnamefont {Hsin}}, \bibinfo {author} {\bibfnamefont {A.}~\bibnamefont {Kapustin}},\ and\ \bibinfo {author} {\bibfnamefont {R.}~\bibnamefont {Thorngren}},\ }\bibfield  {title} {\bibinfo {title} {{Berry Phase in Quantum Field Theory: Diabolical Points and Boundary Phenomena}},\ }\href {https://doi.org/10.1103/PhysRevB.102.245113} {\bibfield  {journal} {\bibinfo  {journal} {Phys. Rev. B}\ }\textbf {\bibinfo {volume} {102}},\ \bibinfo {pages} {245113} (\bibinfo {year} {2020})},\ \Eprint {https://arxiv.org/abs/2004.10758} {arXiv:2004.10758 [cond-mat.str-el]} \BibitemShut {NoStop}%
\bibitem [{\citenamefont {Aasen}\ \emph {et~al.}(2022)\citenamefont {Aasen}, \citenamefont {Wang},\ and\ \citenamefont {Hastings}}]{Aasen:2022cdu}%
  \BibitemOpen
  \bibfield  {author} {\bibinfo {author} {\bibfnamefont {D.}~\bibnamefont {Aasen}}, \bibinfo {author} {\bibfnamefont {Z.}~\bibnamefont {Wang}},\ and\ \bibinfo {author} {\bibfnamefont {M.~B.}\ \bibnamefont {Hastings}},\ }\bibfield  {title} {\bibinfo {title} {{Adiabatic paths of Hamiltonians, symmetries of topological order, and automorphism codes}},\ }\href {https://doi.org/10.1103/PhysRevB.106.085122} {\bibfield  {journal} {\bibinfo  {journal} {Phys. Rev. B}\ }\textbf {\bibinfo {volume} {106}},\ \bibinfo {pages} {085122} (\bibinfo {year} {2022})},\ \Eprint {https://arxiv.org/abs/2203.11137} {arXiv:2203.11137 [quant-ph]} \BibitemShut {NoStop}%
\bibitem [{\citenamefont {Hsin}\ and\ \citenamefont {Wang}(2023)}]{Hsin:2022iug}%
  \BibitemOpen
  \bibfield  {author} {\bibinfo {author} {\bibfnamefont {P.-S.}\ \bibnamefont {Hsin}}\ and\ \bibinfo {author} {\bibfnamefont {Z.}~\bibnamefont {Wang}},\ }\bibfield  {title} {\bibinfo {title} {{On topology of the moduli space of gapped Hamiltonians for topological phases}},\ }\href {https://doi.org/10.1063/5.0136906} {\bibfield  {journal} {\bibinfo  {journal} {J. Math. Phys.}\ }\textbf {\bibinfo {volume} {64}},\ \bibinfo {pages} {041901} (\bibinfo {year} {2023})},\ \Eprint {https://arxiv.org/abs/2211.16535} {arXiv:2211.16535 [cond-mat.str-el]} \BibitemShut {NoStop}%
\bibitem [{\citenamefont {{P. Etingof, D. Nikshych, V. Ostrik, with an appendix by E. Meir}}(2010)}]{ENO2010}%
  \BibitemOpen
  \bibfield  {author} {\bibinfo {author} {\bibnamefont {{P. Etingof, D. Nikshych, V. Ostrik, with an appendix by E. Meir}}},\ }\bibfield  {title} {\bibinfo {title} {{Fusion categories and homotopy theory}},\ }\href@noop {} {\bibfield  {journal} {\bibinfo  {journal} {Quantum topology}\ }\textbf {\bibinfo {volume} {1}},\ \bibinfo {pages} {209} (\bibinfo {year} {2010})},\ \Eprint {https://arxiv.org/abs/0909.3140} {arXiv:0909.3140 [math.QA]} \BibitemShut {NoStop}%
\bibitem [{\citenamefont {Kitaev}(2006)}]{Kitaev:2005hzj}%
  \BibitemOpen
  \bibfield  {author} {\bibinfo {author} {\bibfnamefont {A.}~\bibnamefont {Kitaev}},\ }\bibfield  {title} {\bibinfo {title} {{Anyons in an exactly solved model and beyond}},\ }\href {https://doi.org/10.1016/j.aop.2005.10.005} {\bibfield  {journal} {\bibinfo  {journal} {Annals Phys.}\ }\textbf {\bibinfo {volume} {321}},\ \bibinfo {pages} {2} (\bibinfo {year} {2006})},\ \Eprint {https://arxiv.org/abs/cond-mat/0506438} {arXiv:cond-mat/0506438} \BibitemShut {NoStop}%
\bibitem [{\citenamefont {Baez}\ and\ \citenamefont {Lauda}(2004)}]{Baez:2003yaq}%
  \BibitemOpen
  \bibfield  {author} {\bibinfo {author} {\bibfnamefont {J.~C.}\ \bibnamefont {Baez}}\ and\ \bibinfo {author} {\bibfnamefont {A.~D.}\ \bibnamefont {Lauda}},\ }\bibfield  {title} {\bibinfo {title} {{Higher-Dimensional Algebra V: 2-Groups}},\ }\href@noop {} {\bibfield  {journal} {\bibinfo  {journal} {Theory and Applications of Categories}\ }\textbf {\bibinfo {volume} {12}},\ \bibinfo {pages} {423} (\bibinfo {year} {2004})},\ \Eprint {https://arxiv.org/abs/math/0307200} {arXiv:math/0307200 [math.QA]} \BibitemShut {NoStop}%
\bibitem [{\citenamefont {Kapustin}\ and\ \citenamefont {Thorngren}(2017)}]{Kapustin:2013uxa}%
  \BibitemOpen
  \bibfield  {author} {\bibinfo {author} {\bibfnamefont {A.}~\bibnamefont {Kapustin}}\ and\ \bibinfo {author} {\bibfnamefont {R.}~\bibnamefont {Thorngren}},\ }\bibfield  {title} {\bibinfo {title} {{Higher Symmetry and Gapped Phases of Gauge Theories}},\ }\href {https://doi.org/10.1007/978-3-319-59939-7_5} {\bibfield  {journal} {\bibinfo  {journal} {Prog. Math.}\ }\textbf {\bibinfo {volume} {324}},\ \bibinfo {pages} {177} (\bibinfo {year} {2017})},\ \Eprint {https://arxiv.org/abs/1309.4721} {arXiv:1309.4721 [hep-th]} \BibitemShut {NoStop}%
\bibitem [{\citenamefont {Benini}\ \emph {et~al.}(2019)\citenamefont {Benini}, \citenamefont {C\'{o}rdova},\ and\ \citenamefont {Hsin}}]{Benini:2018reh}%
  \BibitemOpen
  \bibfield  {author} {\bibinfo {author} {\bibfnamefont {F.}~\bibnamefont {Benini}}, \bibinfo {author} {\bibfnamefont {C.}~\bibnamefont {C\'{o}rdova}},\ and\ \bibinfo {author} {\bibfnamefont {P.-S.}\ \bibnamefont {Hsin}},\ }\bibfield  {title} {\bibinfo {title} {{On 2-Group Global Symmetries and their Anomalies}},\ }\href {https://doi.org/10.1007/JHEP03(2019)118} {\bibfield  {journal} {\bibinfo  {journal} {JHEP}\ }\textbf {\bibinfo {volume} {03}},\ \bibinfo {pages} {118}},\ \Eprint {https://arxiv.org/abs/1803.09336} {arXiv:1803.09336 [hep-th]} \BibitemShut {NoStop}%
\bibitem [{\citenamefont {Wigner}(1970)}]{wigner1970algebraic}%
  \BibitemOpen
  \bibfield  {author} {\bibinfo {author} {\bibfnamefont {D.}~\bibnamefont {Wigner}},\ }\bibfield  {title} {\bibinfo {title} {Algebraic cohomology of topological groups},\ }\href@noop {} {\bibfield  {journal} {\bibinfo  {journal} {Bulletin of the American Mathematical Society}\ }\textbf {\bibinfo {volume} {76}},\ \bibinfo {pages} {825} (\bibinfo {year} {1970})}\BibitemShut {NoStop}%
\bibitem [{\citenamefont {Baez}\ and\ \citenamefont {Stevenson}(2009)}]{BS2009}%
  \BibitemOpen
  \bibfield  {author} {\bibinfo {author} {\bibfnamefont {J.~C.}\ \bibnamefont {Baez}}\ and\ \bibinfo {author} {\bibfnamefont {D.}~\bibnamefont {Stevenson}},\ }\bibinfo {title} {{The Classifying Space of a Topological 2-Group}},\ in\ \href {https://doi.org/10.1007/978-3-642-01200-6_1} {\emph {\bibinfo {booktitle} {{Algebraic Topology: The Abel Symposium 2007}}}},\ \bibinfo {editor} {edited by\ \bibinfo {editor} {\bibfnamefont {N.}~\bibnamefont {Baas}}, \bibinfo {editor} {\bibfnamefont {E.~M.}\ \bibnamefont {Friedlander}}, \bibinfo {editor} {\bibfnamefont {B.}~\bibnamefont {Jahren}},\ and\ \bibinfo {editor} {\bibfnamefont {P.~A.}\ \bibnamefont {{\O}stv{\ae}r}}}\ (\bibinfo  {publisher} {Springer Berlin Heidelberg},\ \bibinfo {address} {Berlin, Heidelberg},\ \bibinfo {year} {2009})\ pp.\ \bibinfo {pages} {1--31},\ \Eprint {https://arxiv.org/abs/0801.3843} {arXiv:0801.3843 [math.AT]} \BibitemShut {NoStop}%
\bibitem [{\citenamefont {Wockel}(2011)}]{Wockel2011}%
  \BibitemOpen
  \bibfield  {author} {\bibinfo {author} {\bibfnamefont {C.}~\bibnamefont {Wockel}},\ }\bibfield  {title} {\bibinfo {title} {{Principal 2-bundles and their gauge 2-groups}},\ }\bibfield  {journal} {\bibinfo  {journal} {Forum Mathematicum}\ }\textbf {\bibinfo {volume} {23}},\ \href {https://doi.org/10.1515/form.2011.020} {10.1515/form.2011.020} (\bibinfo {year} {2011}),\ \Eprint {https://arxiv.org/abs/0803.3692} {arXiv:0803.3692 [math.DG]} \BibitemShut {NoStop}%
\bibitem [{\citenamefont {Bartels}(2006)}]{Bartels2006}%
  \BibitemOpen
  \bibfield  {author} {\bibinfo {author} {\bibfnamefont {T.}~\bibnamefont {Bartels}},\ }\href {https://arxiv.org/abs/math/0410328} {\bibinfo {title} {{Higher gauge theory I: 2-Bundles}}} (\bibinfo {year} {2006}),\ \Eprint {https://arxiv.org/abs/math/0410328} {arXiv:math/0410328 [math.CT]} \BibitemShut {NoStop}%
\bibitem [{\citenamefont {Jur\v{c}o}(2011)}]{Jurco2011}%
  \BibitemOpen
  \bibfield  {author} {\bibinfo {author} {\bibfnamefont {B.}~\bibnamefont {Jur\v{c}o}},\ }\bibfield  {title} {\bibinfo {title} {{Crossed Module Bundle Gerbes; Classification, String Group And Differential Geometry}},\ }\href {https://doi.org/10.1142/s0219887811005555} {\bibfield  {journal} {\bibinfo  {journal} {International Journal of Geometric Methods in Modern Physics}\ }\textbf {\bibinfo {volume} {08}},\ \bibinfo {pages} {1079} (\bibinfo {year} {2011})},\ \Eprint {https://arxiv.org/abs/math/0510078} {arXiv:math/0510078 [math.DG]} \BibitemShut {NoStop}%
\bibitem [{\citenamefont {Baas}\ \emph {et~al.}(2012)\citenamefont {Baas}, \citenamefont {B\"{o}kstedt},\ and\ \citenamefont {Kro}}]{BBK2012}%
  \BibitemOpen
  \bibfield  {author} {\bibinfo {author} {\bibfnamefont {N.~A.}\ \bibnamefont {Baas}}, \bibinfo {author} {\bibfnamefont {M.}~\bibnamefont {B\"{o}kstedt}},\ and\ \bibinfo {author} {\bibfnamefont {T.~A.}\ \bibnamefont {Kro}},\ }\bibfield  {title} {\bibinfo {title} {Two-categorical bundles and their classifying spaces},\ }\href {https://doi.org/10.1017/is012001012jkt181} {\bibfield  {journal} {\bibinfo  {journal} {Journal of K-Theory}\ }\textbf {\bibinfo {volume} {10}},\ \bibinfo {pages} {299} (\bibinfo {year} {2012})},\ \Eprint {https://arxiv.org/abs/math/0612549} {arXiv:math/0612549 [math.AT]} \BibitemShut {NoStop}%
\bibitem [{\citenamefont {Berwick-Evans}\ \emph {et~al.}(2021)\citenamefont {Berwick-Evans}, \citenamefont {Cliff}, \citenamefont {Murray}, \citenamefont {Nakade},\ and\ \citenamefont {Phillips}}]{Berwickevans2021}%
  \BibitemOpen
  \bibfield  {author} {\bibinfo {author} {\bibfnamefont {D.}~\bibnamefont {Berwick-Evans}}, \bibinfo {author} {\bibfnamefont {E.}~\bibnamefont {Cliff}}, \bibinfo {author} {\bibfnamefont {L.}~\bibnamefont {Murray}}, \bibinfo {author} {\bibfnamefont {A.}~\bibnamefont {Nakade}},\ and\ \bibinfo {author} {\bibfnamefont {E.}~\bibnamefont {Phillips}},\ }\href {https://arxiv.org/abs/2110.07571} {\bibinfo {title} {{String structures, 2-group bundles, and a categorification of the Freed-Quinn line bundle}}} (\bibinfo {year} {2021}),\ \Eprint {https://arxiv.org/abs/2110.07571} {arXiv:2110.07571 [math.AT]} \BibitemShut {NoStop}%
\bibitem [{\citenamefont {Thorngren}(2018)}]{Thorngren:2018ziu}%
  \BibitemOpen
  \bibfield  {author} {\bibinfo {author} {\bibfnamefont {R.~G.}\ \bibnamefont {Thorngren}},\ }\emph {\bibinfo {title} {{Combinatorial Topology and Applications to Quantum Field Theory}}},\ \href {https://escholarship.org/uc/item/7r44w49f} {Ph.D. thesis},\ \bibinfo  {school} {UC, Berkeley (main)} (\bibinfo {year} {2018})\BibitemShut {NoStop}%
\bibitem [{\citenamefont {Kapustin}(2014)}]{Kapustin:2014tfa}%
  \BibitemOpen
  \bibfield  {author} {\bibinfo {author} {\bibfnamefont {A.}~\bibnamefont {Kapustin}},\ }\href@noop {} {\bibinfo {title} {{Symmetry Protected Topological Phases, Anomalies, and Cobordisms: Beyond Group Cohomology}}} (\bibinfo {year} {2014}),\ \Eprint {https://arxiv.org/abs/1403.1467} {arXiv:1403.1467 [cond-mat.str-el]} \BibitemShut {NoStop}%
\bibitem [{\citenamefont {Kapustin}\ \emph {et~al.}(2015)\citenamefont {Kapustin}, \citenamefont {Thorngren}, \citenamefont {Turzillo},\ and\ \citenamefont {Wang}}]{Kapustin:2014dxa}%
  \BibitemOpen
  \bibfield  {author} {\bibinfo {author} {\bibfnamefont {A.}~\bibnamefont {Kapustin}}, \bibinfo {author} {\bibfnamefont {R.}~\bibnamefont {Thorngren}}, \bibinfo {author} {\bibfnamefont {A.}~\bibnamefont {Turzillo}},\ and\ \bibinfo {author} {\bibfnamefont {Z.}~\bibnamefont {Wang}},\ }\bibfield  {title} {\bibinfo {title} {{Fermionic Symmetry Protected Topological Phases and Cobordisms}},\ }\href {https://doi.org/10.1007/JHEP12(2015)052} {\bibfield  {journal} {\bibinfo  {journal} {JHEP}\ }\textbf {\bibinfo {volume} {12}},\ \bibinfo {pages} {052}},\ \Eprint {https://arxiv.org/abs/1406.7329} {arXiv:1406.7329 [cond-mat.str-el]} \BibitemShut {NoStop}%
\bibitem [{\citenamefont {Freed}\ and\ \citenamefont {Hopkins}(2021)}]{Freed:2016rqq}%
  \BibitemOpen
  \bibfield  {author} {\bibinfo {author} {\bibfnamefont {D.~S.}\ \bibnamefont {Freed}}\ and\ \bibinfo {author} {\bibfnamefont {M.~J.}\ \bibnamefont {Hopkins}},\ }\bibfield  {title} {\bibinfo {title} {{Reflection positivity and invertible topological phases}},\ }\href {https://doi.org/10.2140/gt.2021.25.1165} {\bibfield  {journal} {\bibinfo  {journal} {Geom. Topol.}\ }\textbf {\bibinfo {volume} {25}},\ \bibinfo {pages} {1165} (\bibinfo {year} {2021})},\ \Eprint {https://arxiv.org/abs/1604.06527} {arXiv:1604.06527 [hep-th]} \BibitemShut {NoStop}%
\bibitem [{\citenamefont {Kitaev}(2019)}]{Kitaev2019UTA}%
  \BibitemOpen
  \bibfield  {author} {\bibinfo {author} {\bibfnamefont {A.}~\bibnamefont {Kitaev}},\ }\href {https://web.ma.utexas.edu/topqft/talkslides/kitaev.pdf} {\bibinfo {title} {Differential forms on the space of statistical mechanics models}} (\bibinfo {year} {2019}),\ \bibinfo {note} {talk at University of Texas at Austin}\BibitemShut {NoStop}%
\bibitem [{\citenamefont {Kapustin}\ and\ \citenamefont {Spodyneiko}(2020{\natexlab{b}})}]{Kapustin:2020eby}%
  \BibitemOpen
  \bibfield  {author} {\bibinfo {author} {\bibfnamefont {A.}~\bibnamefont {Kapustin}}\ and\ \bibinfo {author} {\bibfnamefont {L.}~\bibnamefont {Spodyneiko}},\ }\bibfield  {title} {\bibinfo {title} {{Higher-dimensional generalizations of Berry curvature}},\ }\href {https://doi.org/10.1103/PhysRevB.101.235130} {\bibfield  {journal} {\bibinfo  {journal} {Phys. Rev. B}\ }\textbf {\bibinfo {volume} {101}},\ \bibinfo {pages} {235130} (\bibinfo {year} {2020}{\natexlab{b}})},\ \Eprint {https://arxiv.org/abs/2001.03454} {arXiv:2001.03454 [cond-mat.str-el]} \BibitemShut {NoStop}%
\bibitem [{\citenamefont {Artymowicz}\ \emph {et~al.}(2024)\citenamefont {Artymowicz}, \citenamefont {Kapustin},\ and\ \citenamefont {Sopenko}}]{Artymowicz:2023erv}%
  \BibitemOpen
  \bibfield  {author} {\bibinfo {author} {\bibfnamefont {A.}~\bibnamefont {Artymowicz}}, \bibinfo {author} {\bibfnamefont {A.}~\bibnamefont {Kapustin}},\ and\ \bibinfo {author} {\bibfnamefont {N.}~\bibnamefont {Sopenko}},\ }\bibfield  {title} {\bibinfo {title} {{Quantization of the Higher Berry Curvature and the Higher Thouless Pump}},\ }\href {https://doi.org/10.1007/s00220-024-05026-2} {\bibfield  {journal} {\bibinfo  {journal} {Commun. Math. Phys.}\ }\textbf {\bibinfo {volume} {405}},\ \bibinfo {pages} {191} (\bibinfo {year} {2024})},\ \Eprint {https://arxiv.org/abs/2305.06399} {arXiv:2305.06399 [math-ph]} \BibitemShut {NoStop}%
\bibitem [{\citenamefont {C\'ordova}\ \emph {et~al.}(2020{\natexlab{a}})\citenamefont {C\'ordova}, \citenamefont {Freed}, \citenamefont {Lam},\ and\ \citenamefont {Seiberg}}]{Cordova:2019jnf}%
  \BibitemOpen
  \bibfield  {author} {\bibinfo {author} {\bibfnamefont {C.}~\bibnamefont {C\'ordova}}, \bibinfo {author} {\bibfnamefont {D.~S.}\ \bibnamefont {Freed}}, \bibinfo {author} {\bibfnamefont {H.~T.}\ \bibnamefont {Lam}},\ and\ \bibinfo {author} {\bibfnamefont {N.}~\bibnamefont {Seiberg}},\ }\bibfield  {title} {\bibinfo {title} {{Anomalies in the Space of Coupling Constants and Their Dynamical Applications I}},\ }\href {https://doi.org/10.21468/SciPostPhys.8.1.001} {\bibfield  {journal} {\bibinfo  {journal} {SciPost Phys.}\ }\textbf {\bibinfo {volume} {8}},\ \bibinfo {pages} {001} (\bibinfo {year} {2020}{\natexlab{a}})},\ \Eprint {https://arxiv.org/abs/1905.09315} {arXiv:1905.09315 [hep-th]} \BibitemShut {NoStop}%
\bibitem [{\citenamefont {C\'ordova}\ \emph {et~al.}(2020{\natexlab{b}})\citenamefont {C\'ordova}, \citenamefont {Freed}, \citenamefont {Lam},\ and\ \citenamefont {Seiberg}}]{Cordova:2019uob}%
  \BibitemOpen
  \bibfield  {author} {\bibinfo {author} {\bibfnamefont {C.}~\bibnamefont {C\'ordova}}, \bibinfo {author} {\bibfnamefont {D.~S.}\ \bibnamefont {Freed}}, \bibinfo {author} {\bibfnamefont {H.~T.}\ \bibnamefont {Lam}},\ and\ \bibinfo {author} {\bibfnamefont {N.}~\bibnamefont {Seiberg}},\ }\bibfield  {title} {\bibinfo {title} {{Anomalies in the Space of Coupling Constants and Their Dynamical Applications II}},\ }\href {https://doi.org/10.21468/SciPostPhys.8.1.002} {\bibfield  {journal} {\bibinfo  {journal} {SciPost Phys.}\ }\textbf {\bibinfo {volume} {8}},\ \bibinfo {pages} {002} (\bibinfo {year} {2020}{\natexlab{b}})},\ \Eprint {https://arxiv.org/abs/1905.13361} {arXiv:1905.13361 [hep-th]} \BibitemShut {NoStop}%
\bibitem [{\citenamefont {Choi}\ and\ \citenamefont {Ohmori}(2022)}]{Choi:2022odr}%
  \BibitemOpen
  \bibfield  {author} {\bibinfo {author} {\bibfnamefont {Y.}~\bibnamefont {Choi}}\ and\ \bibinfo {author} {\bibfnamefont {K.}~\bibnamefont {Ohmori}},\ }\bibfield  {title} {\bibinfo {title} {{Higher Berry phase of fermions and index theorem}},\ }\href {https://doi.org/10.1007/JHEP09(2022)022} {\bibfield  {journal} {\bibinfo  {journal} {JHEP}\ }\textbf {\bibinfo {volume} {09}},\ \bibinfo {pages} {022}},\ \Eprint {https://arxiv.org/abs/2205.02188} {arXiv:2205.02188 [hep-th]} \BibitemShut {NoStop}%
\bibitem [{\citenamefont {Ohyama}\ \emph {et~al.}(2024)\citenamefont {Ohyama}, \citenamefont {Terashima},\ and\ \citenamefont {Shiozaki}}]{Ohyama:2023suc}%
  \BibitemOpen
  \bibfield  {author} {\bibinfo {author} {\bibfnamefont {S.}~\bibnamefont {Ohyama}}, \bibinfo {author} {\bibfnamefont {Y.}~\bibnamefont {Terashima}},\ and\ \bibinfo {author} {\bibfnamefont {K.}~\bibnamefont {Shiozaki}},\ }\bibfield  {title} {\bibinfo {title} {{Discrete higher Berry phases and matrix product states}},\ }\href {https://doi.org/10.1103/PhysRevB.110.035114} {\bibfield  {journal} {\bibinfo  {journal} {Phys. Rev. B}\ }\textbf {\bibinfo {volume} {110}},\ \bibinfo {pages} {035114} (\bibinfo {year} {2024})},\ \Eprint {https://arxiv.org/abs/2303.04252} {arXiv:2303.04252 [cond-mat.str-el]} \BibitemShut {NoStop}%
\bibitem [{\citenamefont {Qi}\ \emph {et~al.}(2023)\citenamefont {Qi}, \citenamefont {Stephen}, \citenamefont {Wen}, \citenamefont {Spiegel}, \citenamefont {Pflaum}, \citenamefont {Beaudry},\ and\ \citenamefont {Hermele}}]{Qi:2023ysw}%
  \BibitemOpen
  \bibfield  {author} {\bibinfo {author} {\bibfnamefont {M.}~\bibnamefont {Qi}}, \bibinfo {author} {\bibfnamefont {D.~T.}\ \bibnamefont {Stephen}}, \bibinfo {author} {\bibfnamefont {X.}~\bibnamefont {Wen}}, \bibinfo {author} {\bibfnamefont {D.}~\bibnamefont {Spiegel}}, \bibinfo {author} {\bibfnamefont {M.~J.}\ \bibnamefont {Pflaum}}, \bibinfo {author} {\bibfnamefont {A.}~\bibnamefont {Beaudry}},\ and\ \bibinfo {author} {\bibfnamefont {M.}~\bibnamefont {Hermele}},\ }\href@noop {} {\bibinfo {title} {{Charting the space of ground states with tensor networks}}} (\bibinfo {year} {2023}),\ \Eprint {https://arxiv.org/abs/2305.07700} {arXiv:2305.07700 [cond-mat.str-el]} \BibitemShut {NoStop}%
\bibitem [{\citenamefont {Sommer}\ \emph {et~al.}(2024{\natexlab{a}})\citenamefont {Sommer}, \citenamefont {Wen},\ and\ \citenamefont {Vishwanath}}]{Sommer:2024dtb}%
  \BibitemOpen
  \bibfield  {author} {\bibinfo {author} {\bibfnamefont {O.~E.}\ \bibnamefont {Sommer}}, \bibinfo {author} {\bibfnamefont {X.}~\bibnamefont {Wen}},\ and\ \bibinfo {author} {\bibfnamefont {A.}~\bibnamefont {Vishwanath}},\ }\href@noop {} {\bibinfo {title} {{Higher Berry Curvature from the Wave Function I: Schmidt Decomposition and Matrix Product States}}} (\bibinfo {year} {2024}{\natexlab{a}}),\ \Eprint {https://arxiv.org/abs/2405.05316} {arXiv:2405.05316 [cond-mat.str-el]} \BibitemShut {NoStop}%
\bibitem [{\citenamefont {Ohyama}\ and\ \citenamefont {Ryu}(2024{\natexlab{b}})}]{Ohyama:2024jsg}%
  \BibitemOpen
  \bibfield  {author} {\bibinfo {author} {\bibfnamefont {S.}~\bibnamefont {Ohyama}}\ and\ \bibinfo {author} {\bibfnamefont {S.}~\bibnamefont {Ryu}},\ }\href@noop {} {\bibinfo {title} {{Higher Berry Connection for Matrix Product States}}} (\bibinfo {year} {2024}{\natexlab{b}}),\ \Eprint {https://arxiv.org/abs/2405.05327} {arXiv:2405.05327 [cond-mat.str-el]} \BibitemShut {NoStop}%
\bibitem [{\citenamefont {Douglas}\ and\ \citenamefont {Reutter}(2018)}]{Douglas:2018qfz}%
  \BibitemOpen
  \bibfield  {author} {\bibinfo {author} {\bibfnamefont {C.~L.}\ \bibnamefont {Douglas}}\ and\ \bibinfo {author} {\bibfnamefont {D.~J.}\ \bibnamefont {Reutter}},\ }\href@noop {} {\bibinfo {title} {{Fusion 2-categories and a state-sum invariant for 4-manifolds}}} (\bibinfo {year} {2018}),\ \Eprint {https://arxiv.org/abs/1812.11933} {arXiv:1812.11933 [math.QA]} \BibitemShut {NoStop}%
\bibitem [{\citenamefont {Barrett}\ \emph {et~al.}(2024)\citenamefont {Barrett}, \citenamefont {Meusburger},\ and\ \citenamefont {Schaumann}}]{barrett2024gray}%
  \BibitemOpen
  \bibfield  {author} {\bibinfo {author} {\bibfnamefont {J.~W.}\ \bibnamefont {Barrett}}, \bibinfo {author} {\bibfnamefont {C.}~\bibnamefont {Meusburger}},\ and\ \bibinfo {author} {\bibfnamefont {G.}~\bibnamefont {Schaumann}},\ }\href@noop {} {\bibinfo {title} {Gray categories with duals and their diagrams}} (\bibinfo {year} {2024}),\ \Eprint {https://arxiv.org/abs/1211.0529} {arXiv:1211.0529 [math.QA]} \BibitemShut {NoStop}%
\bibitem [{\citenamefont {Gordon}\ \emph {et~al.}(1995)\citenamefont {Gordon}, \citenamefont {Power},\ and\ \citenamefont {Street}}]{gordon1995coherence}%
  \BibitemOpen
  \bibfield  {author} {\bibinfo {author} {\bibfnamefont {R.}~\bibnamefont {Gordon}}, \bibinfo {author} {\bibfnamefont {A.~J.}\ \bibnamefont {Power}},\ and\ \bibinfo {author} {\bibfnamefont {R.}~\bibnamefont {Street}},\ }\href {https://www.ams.org/books/memo/0558/} {\emph {\bibinfo {title} {{Coherence for tricategories}}}},\ Vol.\ \bibinfo {volume} {558}\ (\bibinfo  {publisher} {American Mathematical Soc.},\ \bibinfo {year} {1995})\BibitemShut {NoStop}%
\bibitem [{\citenamefont {Schommer-Pries}(2014)}]{schommerpries2014classification}%
  \BibitemOpen
  \bibfield  {author} {\bibinfo {author} {\bibfnamefont {C.~J.}\ \bibnamefont {Schommer-Pries}},\ }\href@noop {} {\bibinfo {title} {The classification of two-dimensional extended topological field theories}} (\bibinfo {year} {2014}),\ \Eprint {https://arxiv.org/abs/1112.1000} {arXiv:1112.1000 [math.AT]} \BibitemShut {NoStop}%
\bibitem [{\citenamefont {Gaiotto}\ and\ \citenamefont {Johnson-Freyd}(2019)}]{Gaiotto:2019xmp}%
  \BibitemOpen
  \bibfield  {author} {\bibinfo {author} {\bibfnamefont {D.}~\bibnamefont {Gaiotto}}\ and\ \bibinfo {author} {\bibfnamefont {T.}~\bibnamefont {Johnson-Freyd}},\ }\href@noop {} {\bibinfo {title} {{Condensations in higher categories}}} (\bibinfo {year} {2019}),\ \Eprint {https://arxiv.org/abs/1905.09566} {arXiv:1905.09566 [math.CT]} \BibitemShut {NoStop}%
\bibitem [{\citenamefont {Johnson-Freyd}(2022)}]{Johnson_Freyd_2022}%
  \BibitemOpen
  \bibfield  {author} {\bibinfo {author} {\bibfnamefont {T.}~\bibnamefont {Johnson-Freyd}},\ }\bibfield  {title} {\bibinfo {title} {{On the Classification of Topological Orders}},\ }\href {https://doi.org/10.1007/s00220-022-04380-3} {\bibfield  {journal} {\bibinfo  {journal} {Communications in Mathematical Physics}\ }\textbf {\bibinfo {volume} {393}},\ \bibinfo {pages} {989 } (\bibinfo {year} {2022})},\ \Eprint {https://arxiv.org/abs/2003.06663} {arXiv:2003.06663 [math.CT]} \BibitemShut {NoStop}%
\bibitem [{\citenamefont {Kong}\ and\ \citenamefont {Zheng}(2022)}]{Kong:2020iek}%
  \BibitemOpen
  \bibfield  {author} {\bibinfo {author} {\bibfnamefont {L.}~\bibnamefont {Kong}}\ and\ \bibinfo {author} {\bibfnamefont {H.}~\bibnamefont {Zheng}},\ }\bibfield  {title} {\bibinfo {title} {{Categories of quantum liquids I}},\ }\href {https://doi.org/10.1007/JHEP08(2022)070} {\bibfield  {journal} {\bibinfo  {journal} {JHEP}\ }\textbf {\bibinfo {volume} {08}},\ \bibinfo {pages} {070}},\ \Eprint {https://arxiv.org/abs/2011.02859} {arXiv:2011.02859 [hep-th]} \BibitemShut {NoStop}%
\bibitem [{\citenamefont {Inamura}\ and\ \citenamefont {Ohmori}(2024)}]{Inamura:2023qzl}%
  \BibitemOpen
  \bibfield  {author} {\bibinfo {author} {\bibfnamefont {K.}~\bibnamefont {Inamura}}\ and\ \bibinfo {author} {\bibfnamefont {K.}~\bibnamefont {Ohmori}},\ }\bibfield  {title} {\bibinfo {title} {{Fusion Surface Models: 2+1d Lattice Models from Fusion 2-Categories}},\ }\href {https://doi.org/10.21468/SciPostPhys.16.6.143} {\bibfield  {journal} {\bibinfo  {journal} {SciPost Phys.}\ }\textbf {\bibinfo {volume} {16}},\ \bibinfo {pages} {143} (\bibinfo {year} {2024})},\ \Eprint {https://arxiv.org/abs/2305.05774} {arXiv:2305.05774 [cond-mat.str-el]} \BibitemShut {NoStop}%
\bibitem [{\citenamefont {D\'{e}coppet}(2023)}]{Decoppet2023Morita}%
  \BibitemOpen
  \bibfield  {author} {\bibinfo {author} {\bibfnamefont {T.~D.}\ \bibnamefont {D\'{e}coppet}},\ }\bibfield  {title} {\bibinfo {title} {{The Morita Theory of Fusion 2-Categories}},\ }\href {https://doi.org/10.21136/hs.2023.07} {\bibfield  {journal} {\bibinfo  {journal} {Higher Structures}\ }\textbf {\bibinfo {volume} {7}},\ \bibinfo {pages} {234} (\bibinfo {year} {2023})},\ \Eprint {https://arxiv.org/abs/2208.08722} {arXiv:2208.08722 [math.CT]} \BibitemShut {NoStop}%
\bibitem [{\citenamefont {Freed}\ and\ \citenamefont {Teleman}(2021)}]{Freed:2020qfy}%
  \BibitemOpen
  \bibfield  {author} {\bibinfo {author} {\bibfnamefont {D.~S.}\ \bibnamefont {Freed}}\ and\ \bibinfo {author} {\bibfnamefont {C.}~\bibnamefont {Teleman}},\ }\bibfield  {title} {\bibinfo {title} {{Gapped Boundary Theories in Three Dimensions}},\ }\href {https://doi.org/10.1007/s00220-021-04192-x} {\bibfield  {journal} {\bibinfo  {journal} {Commun. Math. Phys.}\ }\textbf {\bibinfo {volume} {388}},\ \bibinfo {pages} {845} (\bibinfo {year} {2021})},\ \Eprint {https://arxiv.org/abs/2006.10200} {arXiv:2006.10200 [math.QA]} \BibitemShut {NoStop}%
\bibitem [{\citenamefont {D\'ecoppet}\ and\ \citenamefont {Xu}(2024)}]{Decoppet:2023rlx}%
  \BibitemOpen
  \bibfield  {author} {\bibinfo {author} {\bibfnamefont {T.~D.}\ \bibnamefont {D\'ecoppet}}\ and\ \bibinfo {author} {\bibfnamefont {H.}~\bibnamefont {Xu}},\ }\bibfield  {title} {\bibinfo {title} {{Local modules in braided monoidal 2-categories}},\ }\href {https://doi.org/10.1063/5.0172042} {\bibfield  {journal} {\bibinfo  {journal} {J. Math. Phys.}\ }\textbf {\bibinfo {volume} {65}},\ \bibinfo {pages} {061702} (\bibinfo {year} {2024})},\ \Eprint {https://arxiv.org/abs/2307.02843} {arXiv:2307.02843 [math.CT]} \BibitemShut {NoStop}%
\bibitem [{\citenamefont {Ohyama}\ and\ \citenamefont {Ryu}(2024{\natexlab{c}})}]{Ohyama:2024ytt}%
  \BibitemOpen
  \bibfield  {author} {\bibinfo {author} {\bibfnamefont {S.}~\bibnamefont {Ohyama}}\ and\ \bibinfo {author} {\bibfnamefont {S.}~\bibnamefont {Ryu}},\ }\href@noop {} {\bibinfo {title} {{Higher Berry Phase from Projected Entangled Pair States in (2+1) dimensions}}} (\bibinfo {year} {2024}{\natexlab{c}}),\ \Eprint {https://arxiv.org/abs/2405.05325} {arXiv:2405.05325 [cond-mat.str-el]} \BibitemShut {NoStop}%
\bibitem [{\citenamefont {Sommer}\ \emph {et~al.}(2024{\natexlab{b}})\citenamefont {Sommer}, \citenamefont {Vishwanath},\ and\ \citenamefont {Wen}}]{Sommer:2024lzp}%
  \BibitemOpen
  \bibfield  {author} {\bibinfo {author} {\bibfnamefont {O.~E.}\ \bibnamefont {Sommer}}, \bibinfo {author} {\bibfnamefont {A.}~\bibnamefont {Vishwanath}},\ and\ \bibinfo {author} {\bibfnamefont {X.}~\bibnamefont {Wen}},\ }\href@noop {} {\bibinfo {title} {{Higher Berry Curvature from the Wave function II: Locally Parameterized States Beyond One Dimension}}} (\bibinfo {year} {2024}{\natexlab{b}}),\ \Eprint {https://arxiv.org/abs/2405.05323} {arXiv:2405.05323 [cond-mat.str-el]} \BibitemShut {NoStop}%
\bibitem [{\citenamefont {Bhardwaj}\ \emph {et~al.}(2022)\citenamefont {Bhardwaj}, \citenamefont {Sch\"{a}fer-Nameki},\ and\ \citenamefont {Wu}}]{Bhardwaj:2022lsg}%
  \BibitemOpen
  \bibfield  {author} {\bibinfo {author} {\bibfnamefont {L.}~\bibnamefont {Bhardwaj}}, \bibinfo {author} {\bibfnamefont {S.}~\bibnamefont {Sch\"{a}fer-Nameki}},\ and\ \bibinfo {author} {\bibfnamefont {J.}~\bibnamefont {Wu}},\ }\bibfield  {title} {\bibinfo {title} {{Universal Non-Invertible Symmetries}},\ }\href {https://doi.org/10.1002/prop.202200143} {\bibfield  {journal} {\bibinfo  {journal} {Fortsch. Phys.}\ }\textbf {\bibinfo {volume} {70}},\ \bibinfo {pages} {2200143} (\bibinfo {year} {2022})},\ \Eprint {https://arxiv.org/abs/2208.05973} {arXiv:2208.05973 [hep-th]} \BibitemShut {NoStop}%
\bibitem [{\citenamefont {Bartsch}\ \emph {et~al.}(2022)\citenamefont {Bartsch}, \citenamefont {Bullimore}, \citenamefont {Ferrari},\ and\ \citenamefont {Pearson}}]{Bartsch:2022mpm}%
  \BibitemOpen
  \bibfield  {author} {\bibinfo {author} {\bibfnamefont {T.}~\bibnamefont {Bartsch}}, \bibinfo {author} {\bibfnamefont {M.}~\bibnamefont {Bullimore}}, \bibinfo {author} {\bibfnamefont {A.~E.~V.}\ \bibnamefont {Ferrari}},\ and\ \bibinfo {author} {\bibfnamefont {J.}~\bibnamefont {Pearson}},\ }\href@noop {} {\bibinfo {title} {{Non-invertible Symmetries and Higher Representation Theory I}}} (\bibinfo {year} {2022}),\ \Eprint {https://arxiv.org/abs/2208.05993} {arXiv:2208.05993 [hep-th]} \BibitemShut {NoStop}%
\bibitem [{\citenamefont {Delcamp}\ and\ \citenamefont {Tiwari}(2024)}]{Delcamp:2023kew}%
  \BibitemOpen
  \bibfield  {author} {\bibinfo {author} {\bibfnamefont {C.}~\bibnamefont {Delcamp}}\ and\ \bibinfo {author} {\bibfnamefont {A.}~\bibnamefont {Tiwari}},\ }\bibfield  {title} {\bibinfo {title} {{Higher categorical symmetries and gauging in two-dimensional spin systems}},\ }\href {https://doi.org/10.21468/SciPostPhys.16.4.110} {\bibfield  {journal} {\bibinfo  {journal} {SciPost Phys.}\ }\textbf {\bibinfo {volume} {16}},\ \bibinfo {pages} {110} (\bibinfo {year} {2024})},\ \Eprint {https://arxiv.org/abs/2301.01259} {arXiv:2301.01259 [hep-th]} \BibitemShut {NoStop}%
\bibitem [{\citenamefont {Barkeshli}\ \emph {et~al.}(2023)\citenamefont {Barkeshli}, \citenamefont {Chen}, \citenamefont {Huang}, \citenamefont {Kobayashi}, \citenamefont {Tantivasadakarn},\ and\ \citenamefont {Zhu}}]{Barkeshli:2022wuz}%
  \BibitemOpen
  \bibfield  {author} {\bibinfo {author} {\bibfnamefont {M.}~\bibnamefont {Barkeshli}}, \bibinfo {author} {\bibfnamefont {Y.-A.}\ \bibnamefont {Chen}}, \bibinfo {author} {\bibfnamefont {S.-J.}\ \bibnamefont {Huang}}, \bibinfo {author} {\bibfnamefont {R.}~\bibnamefont {Kobayashi}}, \bibinfo {author} {\bibfnamefont {N.}~\bibnamefont {Tantivasadakarn}},\ and\ \bibinfo {author} {\bibfnamefont {G.}~\bibnamefont {Zhu}},\ }\bibfield  {title} {\bibinfo {title} {{Codimension-2 defects and higher symmetries in (3+1)D topological phases}},\ }\href {https://doi.org/10.21468/SciPostPhys.14.4.065} {\bibfield  {journal} {\bibinfo  {journal} {SciPost Phys.}\ }\textbf {\bibinfo {volume} {14}},\ \bibinfo {pages} {065} (\bibinfo {year} {2023})},\ \Eprint {https://arxiv.org/abs/2208.07367} {arXiv:2208.07367 [cond-mat.str-el]} \BibitemShut {NoStop}%
\bibitem [{\citenamefont {Barkeshli}\ \emph {et~al.}(2024)\citenamefont {Barkeshli}, \citenamefont {Chen}, \citenamefont {Hsin},\ and\ \citenamefont {Kobayashi}}]{Barkeshli:2022edm}%
  \BibitemOpen
  \bibfield  {author} {\bibinfo {author} {\bibfnamefont {M.}~\bibnamefont {Barkeshli}}, \bibinfo {author} {\bibfnamefont {Y.-A.}\ \bibnamefont {Chen}}, \bibinfo {author} {\bibfnamefont {P.-S.}\ \bibnamefont {Hsin}},\ and\ \bibinfo {author} {\bibfnamefont {R.}~\bibnamefont {Kobayashi}},\ }\bibfield  {title} {\bibinfo {title} {{Higher-group symmetry in finite gauge theory and stabilizer codes}},\ }\href {https://doi.org/10.21468/SciPostPhys.16.4.089} {\bibfield  {journal} {\bibinfo  {journal} {SciPost Phys.}\ }\textbf {\bibinfo {volume} {16}},\ \bibinfo {pages} {089} (\bibinfo {year} {2024})},\ \Eprint {https://arxiv.org/abs/2211.11764} {arXiv:2211.11764 [cond-mat.str-el]} \BibitemShut {NoStop}%
\bibitem [{\citenamefont {Tambara}\ and\ \citenamefont {Yamagami}(1998)}]{TY1998}%
  \BibitemOpen
  \bibfield  {author} {\bibinfo {author} {\bibfnamefont {D.}~\bibnamefont {Tambara}}\ and\ \bibinfo {author} {\bibfnamefont {S.}~\bibnamefont {Yamagami}},\ }\bibfield  {title} {\bibinfo {title} {{Tensor Categories with Fusion Rules of Self-Duality for Finite Abelian Groups}},\ }\href {https://doi.org/https://doi.org/10.1006/jabr.1998.7558} {\bibfield  {journal} {\bibinfo  {journal} {Journal of Algebra}\ }\textbf {\bibinfo {volume} {209}},\ \bibinfo {pages} {692 } (\bibinfo {year} {1998})}\BibitemShut {NoStop}%
\bibitem [{\citenamefont {Kac}\ and\ \citenamefont {Paljutkin}(1966)}]{KP1966}%
  \BibitemOpen
  \bibfield  {author} {\bibinfo {author} {\bibfnamefont {G.~I.}\ \bibnamefont {Kac}}\ and\ \bibinfo {author} {\bibfnamefont {V.~G.}\ \bibnamefont {Paljutkin}},\ }\href@noop {} {\bibinfo {title} {{Finite ring groups}}},\ \bibinfo {howpublished} {\href{http://www.mathnet.ru/php/archive.phtml?wshow=paper&jrnid=mmo&paperid=170&option_lang=eng}{Trans. {Mosc}. {Math}. {Soc}. 15 251-294}; {Translation} from {Tr}. {Mosk}. {Mat}. {Obs}. 15 224-261.} (\bibinfo {year} {1966})\BibitemShut {NoStop}%
\bibitem [{\citenamefont {Tambara}(2000)}]{Tam2000}%
  \BibitemOpen
  \bibfield  {author} {\bibinfo {author} {\bibfnamefont {D.}~\bibnamefont {Tambara}},\ }\bibfield  {title} {\bibinfo {title} {{Representations of tensor categories with fusion rules of self-duality for abelian groups}},\ }\href {https://doi.org/10.1007/BF02803515} {\bibfield  {journal} {\bibinfo  {journal} {Israel Journal of Mathematics}\ }\textbf {\bibinfo {volume} {118}},\ \bibinfo {pages} {29} (\bibinfo {year} {2000})}\BibitemShut {NoStop}%
\bibitem [{\citenamefont {Fuchs}\ \emph {et~al.}(2002)\citenamefont {Fuchs}, \citenamefont {Runkel},\ and\ \citenamefont {Schweigert}}]{Fuchs:2002cm}%
  \BibitemOpen
  \bibfield  {author} {\bibinfo {author} {\bibfnamefont {J.}~\bibnamefont {Fuchs}}, \bibinfo {author} {\bibfnamefont {I.}~\bibnamefont {Runkel}},\ and\ \bibinfo {author} {\bibfnamefont {C.}~\bibnamefont {Schweigert}},\ }\bibfield  {title} {\bibinfo {title} {{TFT construction of RCFT correlators 1. Partition functions}},\ }\href {https://doi.org/10.1016/S0550-3213(02)00744-7} {\bibfield  {journal} {\bibinfo  {journal} {Nucl. Phys. B}\ }\textbf {\bibinfo {volume} {646}},\ \bibinfo {pages} {353} (\bibinfo {year} {2002})},\ \Eprint {https://arxiv.org/abs/hep-th/0204148} {arXiv:hep-th/0204148} \BibitemShut {NoStop}%
\bibitem [{\citenamefont {Fuchs}\ and\ \citenamefont {Stigner}(2009)}]{FS2008}%
  \BibitemOpen
  \bibfield  {author} {\bibinfo {author} {\bibfnamefont {J.}~\bibnamefont {Fuchs}}\ and\ \bibinfo {author} {\bibfnamefont {C.}~\bibnamefont {Stigner}},\ }\bibfield  {title} {\bibinfo {title} {{On Frobenius algebras in rigid monoidal categories}},\ }\href@noop {} {\bibfield  {journal} {\bibinfo  {journal} {Arabian Journal for Science and Engineering}\ }\textbf {\bibinfo {volume} {33-2C}},\ \bibinfo {pages} {175} (\bibinfo {year} {2009})},\ \Eprint {https://arxiv.org/abs/0901.4886} {arXiv:0901.4886 [math.CT]} \BibitemShut {NoStop}%
\bibitem [{\citenamefont {Andruskiewitsch}\ and\ \citenamefont {Mombelli}(2007)}]{AM2007}%
  \BibitemOpen
  \bibfield  {author} {\bibinfo {author} {\bibfnamefont {N.}~\bibnamefont {Andruskiewitsch}}\ and\ \bibinfo {author} {\bibfnamefont {J.~M.}\ \bibnamefont {Mombelli}},\ }\bibfield  {title} {\bibinfo {title} {{On module categories over finite-dimensional Hopf algebras}},\ }\href {https://doi.org/https://doi.org/10.1016/j.jalgebra.2007.04.006} {\bibfield  {journal} {\bibinfo  {journal} {Journal of Algebra}\ }\textbf {\bibinfo {volume} {314}},\ \bibinfo {pages} {383} (\bibinfo {year} {2007})},\ \Eprint {https://arxiv.org/abs/math/0608781} {arXiv:math/0608781 [math.QA]} \BibitemShut {NoStop}%
\end{thebibliography}%

\end{document}